\documentclass[a4paper,10pt,pointlessnumbers,headtopline,headsepline]{scrbook}

\usepackage[ngerman,english]{babel}
\usepackage[utf8]{inputenc}
\usepackage{pdfsync}
\usepackage[T1]{fontenc}

\usepackage{color}
\usepackage{graphicx}
\usepackage{epstopdf}
\usepackage{appendix}
\usepackage{wrapfig}
\usepackage{enumerate}
\usepackage{tabularx}
\usepackage{hyperref}

\usepackage{./supplemental/max_mathcommands} 
\usepackage{diagxy}

\usepackage[charter]{mathdesign}
\usepackage{helvet}

\usepackage{ntheorem_settings}

\providecommand{\FourierTrafo}{\mathfrak{F}}
\providecommand{\sympForm}{\sigma}
\providecommand{\Pspace}{\R^d_x \times \R^d_p}
\providecommand{\Cont}{\mathcal{C}}
\providecommand{\BCont}{\mathcal{BC}}

\providecommand{\Fourier}{\FourierTrafo}
\providecommand{\spec}{\mathrm{spec} \, }
\providecommand{\image}{\mathrm{im} \, }

\providecommand{\Op}{\mathfrak{Op}}

\providecommand{\Fs}{\FourierTrafo_{\sympForm}}

\providecommand{\Qe}{\mathsf{Q}}
\providecommand{\Pe}{\mathsf{P}}
\providecommand{\Weyl}{\sharp}

\providecommand{\wastlim}{{\mathrm{w}^{\ast}\mbox{-}\lim}}
\providecommand{\slim}{\mathrm{s}\mbox{-}\lim}
\providecommand{\WeylSys}{W}
\providecommand{\pspace}{\Xi}

\providecommand{\spec}{\mathrm{spec}}

\providecommand{\Qe}{{Q_{\eps}}}

\providecommand{\ordere}[1]{\mathcal{O}(\eps^{#1})}

\providecommand{\order}{\mathcal{O}}

\providecommand{\Piref}{\hat{\Pi}_{\mathrm{ref}}}
\providecommand{\piref}{\pi_{\mathrm{ref}}}

\providecommand{\Hfast}{\mathcal{H}_{\mathrm{fast}}}
\providecommand{\Hslow}{\mathcal{H}_{\mathrm{slow}}}
\providecommand{\Hphys}{\mathcal{H}_{\mathrm{phys}}}

\providecommand{\Fs}{\mathcal{F}_{\sigma}}

\providecommand{\Weyl}{\star}

\providecommand{\Op}{\mathrm{Op}}

\providecommand{\WignerTrafo}{{\mathcal{W}}}

\providecommand{\Schwartz}{\mathcal{S}}

\providecommand{\rel}{\mathrm{rel}}

\providecommand{\Hil}{\mathcal{H}}

\providecommand{\Index}{\mathcal{J}}

\providecommand{\opHBO}{\hat{H}_{\mathrm{BO}}}

\providecommand{\He}{H_{\mathrm{e}}}

\title{Weyl Quantization and Semiclassics}
\subtitle{Lecture Notes\\ {\small WS 2009/2010}}

\author{Max Lein\footnote{\texttt{lein@ma.tum.de}}}
\date{\today} 

\publishers{\vspace{20mm}
\includegraphics{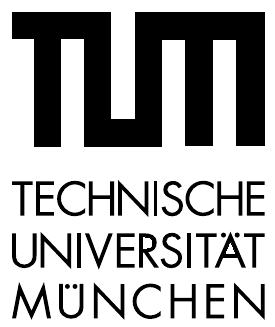} \\ 
Zentrum Mathematik \\
Lehrstuhl M5 \\
Boltzmannstr. 3 \\
85747 Garching 
} 

\begin{document}

\frontmatter
\maketitle

\chapter*{Acknowledgements}

I would like to thank Herbert Spohn for his encouragement, support and cutting the red tape. Chapters~\ref{hilbert_spaces} and \ref{operators} follow closely his `Quantendynamik' from 2005. Furthermore the physical arguments presented in  Chpater~\ref{semiclassics} to arrive at the proper formulation of the semiclassical limit (which I feel is the most important part) in is based on his lecture. 

Furthermore, I would like to express my gratitude towards Martin Fürst and David Sattlegger who were very helpful in improving the lecture. 

\tableofcontents

\mainmatter

\chapter{Introduction} 
\label{intro}

This lecture will elaborate on what a quantization is. If classical phase space, \ie the space of possible positions and momenta, is $T^* \R^d_x \cong \Pspace$, we give a concrete solution. Roughly speaking, we generalize Dirac's prescription to replace $x$ by ``multiplication with $x$,'' 
\begin{align*}
	(\hat{x} \varphi)(x) := x \, \varphi(x) 
	, 
	&& \varphi \in L^2(\R^d) 
	, 
\end{align*}
and momentum by a ``derivative with respect to $x$'' 
\begin{align*}
	(\hat{p} \varphi)(x) := - i \hbar \nabla_x \varphi(x) 
	, 
	&& \varphi \in L^2(\R^d) \cap \Cont^1(\R^d)
	.  
\end{align*}
Dirac's recipe only works if we are quantizing functions of $x$ \emph{or} $p$ and their linear combinations, \eg $H(x,p) = \frac{1}{2m} p^2 + V(x)$ becomes $\hat{H} = H(\hat{x},\hat{p}) = - \frac{\hbar^2}{2m} \Delta_x + V(\hat{x})$, and it needs to be supplemented with an ``operator ordering prescription:'' what is the quantization of $x \cdot p$? Possible solutions are 
\begin{align*}
	\hat{x} \cdot \hat{p} 
	, 
	&&
	\hat{p} \cdot \hat{x} = \hat{x} \cdot \hat{p} - i \hbar \, \id_{L^2}
	, 
	&&
	\tfrac{1}{2} \bigl ( \hat{x} \cdot \hat{p} + \hat{p} \cdot \hat{x} \bigr ) 
	, 
\end{align*}
all of which are \emph{different} operators and the first two are not even symmetric (= hermitian). Choosing an operator ordering is necessary, because $\hat{x}$ and $\hat{p}$ \emph{do not commute}, 
\begin{align*}
	i [ \hat{p}_l , \hat{x}_j ] = \hbar \, \delta_{lj} 
	. 
\end{align*}
The physical constant $\hbar$ has a \emph{fixed value} and units and it measures the ``degree of non-commutativity'' of position and momentum. On the other hand, can we find the function whose quantization is, say, $\hat{x} \cdot \hat{p} = \sum_{l = 1}^d \hat{x}_l \, \hat{p}_l$? One way to look at it is to find a product $\Weyl_{\hbar}$ which \emph{emulates the operator product} on the level of functions, \ie 
\begin{align*}
	\sum_{l = 1}^d \widehat{x_l \Weyl_{\hbar} p_l} = \sum_{l = 1}^d \hat{x}_l \, \hat{p}_l 
	. 
\end{align*}
On the one hand, this product must inherit the non-commutativity of the operator product, but on the other, if ``$\hbar$ is small,'' $x_l \Weyl p_l$ should be approximately given by $x_l \, p_l$. One of the major reasons to study $\Weyl_{\hbar}$ is that we can \emph{use it to derive corrections in perturbation expansions}: if we can find an expansion of the product $\Weyl_{\hbar}$ in $\hbar$, we can ``expand the operator product:'' 
\begin{align*}
	\hat{f} \, \hat{g} = \widehat{f \Weyl_{\hbar} g} = \sum_{n = 0}^{\infty} \hbar^n \widehat{(f \Weyl_{\hbar} g)}_{(n)} + \mbox{small error}
\end{align*}
We will make all of the above mathematically rigorous in this lecture and explain the symbols properly. This idea was first proposed by Littlejohn and Weigert \cite{LittlejohnWeigert:diagonalizationMultiWave:1993} (who are theoretical physicists) and used both in a mathematically rigorous fashion and in theoretical physics to derive currents in crystals that are subjected to electromagnetic fields \cite{PST:effDynamics:2003}, piezocurrents \cite{PanatiSparberTeufel:polarization:2006,Lein:polarization:2005}, guiding center motion of particles in electromagnetic fields \cite{Mueller:productRuleGaugeInvariantWeylSymbols:1999}, the non- and semirelativistic limit of the Dirac equation \cite{FuerstLein:nonRelLimitDirac:2008} and many others. 
\medskip

\noindent
The second main point of this lecture is the semiclassical limit (``$\hbar \rightarrow 0$''): for $\hbar = 0$, $\hat{p}$ and $\hat{x}$ commute. However, \emph{a priori} it is not at all clear how and why this implies classical behavior. Let $H(x,p) = \frac{1}{2m} p^2 + V(x)$ be a classical hamiltonian. If we start at $t = 0$ at the point $(x_0 , p_0)$ in phase space, Hamilton's equations of motion 
\begin{align*}
	\dot{x} &= + \nabla_p H = \tfrac{p}{m} \\
	\dot{p} &= - \nabla_x H = - \nabla_x V 
\end{align*}
with initial conditions $(x_0,p_0)$ determine a unique classical trajectory $\bigl ( x(t),p(t) \bigr )$. The quantum dynamics of a particle whose initial state is described by the wave function $\psi_0$ is given by the \emph{Schrödinger equation}, 
\begin{align}
	i \hbar \frac{\partial}{\partial t} \psi^{\hbar}(t) = \hat{H} \psi^{\hbar}(t) = \bigl ( - \tfrac{1}{2m} (- i \hbar \nabla_x)^2 + V(\hat{x}) \bigr ) \psi^{\hbar}(t) 
	, 
	&& \psi^{\hbar}(0) = \psi_0 \in L^2(\R^d)
	. 
\end{align}
If in order to ``take the semiclassical $\hbar \rightarrow 0$,'' we naïvely set $\hbar = 0$ in the Schrödinger equation, we get the unintelligible equation 
\begin{align*}
	0 = V(\hat{x}) \psi^0(t) 
\end{align*}
which implies trivial dynamics. The reason is that the actual solution requires a bit more thought. There are simple ansätze such as WKB functions, coherent states and other wave packet approaches\footnote{A wave packet is a \emph{sharply peaked} wave function whose envelope function varies slowly, just like in the case of a WKB function. The center of the wave packet is interpreted as ``semiclassical position'' and the time-derivative of this center gives the classical velocity.}, but this suggests that only special (``good'') initial states have a good semiclassical limit. First of all, why should nature ``choose'' such rather special initial conditions? Wave packet approaches are usually based on the Ehrenfest theorem which roughly states that ``quantum expectation values behave like classical observables.'' Although this is always correct, it is \emph{not always interesting}: consider the case of two wave packets on the line $\R$ that are identical in every respect, but run in opposite directions. The average position is always $0$, but the dynamics is far from trivial. 

We will show that such \emph{special initial states are not necessary} for semiclassical behavior, it is a generic phenomenon. The core is the \emph{Egorov theorem}: 
\begin{thm}[Egorov]
	Let $H(x,p) = \tfrac{1}{2m} p^2 + V(x)$ be a hamiltonian such that $V \in \Cont^{\infty}(\R^d_x)$ and $\babs{\partial_x^a V(x)} \leq C_a \sqrt{1 + x^2}^{[2-\abs{a}]_+}$ for $a \in \N_0^d$. Furthermore, let $f \in \BCont^{\infty}(\Pspace)$ be an observable. Then for $\hbar \ll 1$, we have 
	\begin{align*}
		e^{+ i \frac{t}{\hbar} \hat{H}} \hat{f} e^{- i \frac{t}{\hbar} \hat{H}} - \widehat{f \circ \Phi_t} = \mathcal{O}(\hbar) 
	\end{align*}
	where $\hat{f}$ denotes the Weyl quantization of $f$ and $\Phi_t$ the classical flow generated by $H$. 
\end{thm}
The second goal of this lecture is to understand the Egorov theorem. Once all the symbols have been explained, the result is rather intuitive: it says that for ``good'' observables and ``good'' hamiltonians, the Heisenberg observable $F_{\mathrm{qm}}(t) := e^{+ i \frac{t}{\hbar} \hat{H}} \hat{f} e^{- i \frac{t}{\hbar} \hat{H}}$ can be approximated for \emph{macroscopic times} by $F_{\mathrm{cl}} := \widehat{f \circ \Phi_t}$. One can show that this is the case for typical hamiltonians whose kinetic and potential energy are smooth and grow at most quadratically. This is by no means the Egorov theorem in its most general form; in particular, it \emph{does not depend on the form of the hamiltonian}. \marginpar{\small 2009.10.20}

\chapter{Physical Frameworks} 
\label{frameworks}

Understanding of quantization requires knowledge of classical \emph{and} quantum mechanics. A quantization procedure is not merely a method to ``consistently assign operators on $L^2(\R^d)$ to functions on phase space,'' but rather a collection of procedures. A nice overview of the two frameworks can be found in the first few sections of chapter~5 in \cite{Waldmann:deformationQuantization:2008} and we will give a condensed account here: physical theories consist roughly of three parts: 
\begin{enumerate}[(i)]
	\item \emph{A state space: }
	states describe the current configuration of the system and need to be encoded in a mathematical structure. 
	\item \emph{Observables: }
	they represent quantities physicists would like to measure. Related to this is the idea of spectrum as the set of possible outcomes of measurements as well as expectation values (if ones deals with distributions of states). 
	\item \emph{An evolution equation: }
	usually, one is interested in the time evolution of states as well as observables. As energy is the observable conjugate to time, energy functions generate time evolution. 
\end{enumerate}

\section{Hamiltonian framework of classical mechanics} 
\label{frameworks:classical_mechanics}

For simplicity, we only treat classical mechanics of a \emph{spinless point particle moving in $\R^d$}, for the more general theory, we refer to \cite{MarsdenRatiu:introMechanicsSymmetry1999,Arnold:classical_mechanics:1989}. In any standard lecture on classical mechanics, at least two points of views are covered: lagrangian and hamiltonian mechanics. 

We will only treat hamiltonian mechanics here: the dynamics is generated by the so-called hamilton function $H : \Pspace \longrightarrow \R$ which describes the energy of the system for a given configuration. Here, $\Pspace$ is also known as \emph{phase space}. Since only \emph{energy differences} are measurable, the hamiltonian $H' := H + E_0$, $E_0 \in \R$, generates the \emph{same dynamics} as $H$. This is obvious from the \emph{hamiltonian equations of motion}, 
\begin{align}
	\dot{x}(t) &= + \nabla_p H \bigl ( x(t),p(t) \bigr ) 
	\label{frameworks:classical_mechanics:eqn:hamiltons_eom} \\ 
	\dot{p}(t) &= - \nabla_x H \bigl ( x(t),p(t) \bigr ) 
	, 
	\notag
\end{align}
which can be rewritten in matrix notation as 
\begin{align}
	J 
	\left (
	\begin{matrix}
		\dot{x}(t) \\
		\dot{p}(t) \\
	\end{matrix}
	\right ) 
	:= 
	\left (
	\begin{matrix}
		0 & - \id_{\R^d} \\
		+ \id_{\R^d} & 0 \\
	\end{matrix}
	\right )
	\left (
	\begin{matrix}
		\dot{x}(t) \\
		\dot{p}(t) \\
	\end{matrix}
	\right ) = \left (
	\begin{matrix}
		\nabla_x \\
		\nabla_p \\
	\end{matrix}
	\right ) H \bigl ( x(t),p(t) \bigr )
	\label{frameworks:classical_mechanics:eqn:hamiltons_eom} 
	. 
\end{align}
The matrix $J$ appearing on the left-hand side is often called \emph{symplectic form} and leads to a geometric point of view of classical mechanics. For fixed initial condition $(x_0,p_0) \in \Pspace$ at time $t_0 = 0$, \ie initial position and momentum, the \emph{hamiltonian flow} 
\begin{align}
	\Phi : \R_t \times \Pspace \longrightarrow \Pspace 
\end{align}
maps $(x_0,p_0)$ onto the trajectory which solves the hamiltonian equations of motion, 
\begin{align*}
	\Phi_t(x_0,p_0) = \bigl ( x(t),p(t) \bigr )
	, 
	&&
	\bigl ( x(0),p(0) \bigr ) = (x_0,p_0) 
	. 
\end{align*}
If the flow exists for all $t \in \R_t$, it has the following nice properties: for all $t , t' \in \R_t$ and $(x_0,p_0) \in \Pspace$, we have 
\begin{enumerate}[(i)]
	\item $\Phi_t \bigl ( \Phi_{t'}(x_0,p_0) \bigr ) = \Phi_{t + t'}(x_0,p_0)$, 
	\item $\Phi_0(x_0,p_0) = (x_0,p_0)$, and 
	\item $\Phi_{t} \bigl ( \Phi_{-t}(x_0,p_0) \bigr ) = \Phi_{t - t}(x_0,p_0) = (x_0,p_0)$. 
\end{enumerate}
Mathematically, this means $\Phi$ is a \emph{group action} of $\R_t$ (with respect to time translations) on phase space $\Pspace$. This is a fancy way of saying: 
\begin{enumerate}[(i)]
	\item If we first evolve for time $t$ and then for time $t'$, this is the same as evolving for time $t + t'$. 
	\item If we do not evolve at all in time, nothing changes. 
	\item The system can be evolved forwards or backwards in time. 
\end{enumerate}
The existence of the flow is usually proven via the Theorem of Picard and Lindelöf. It holds in much broader generality and is no easier to prove if we specialize to $X = (x,p)$ and $\R^{2d} = \Pspace$. 
\begin{thm}[Picard-Lindelöf]\label{frameworks:classical_mechanics:thm:Picard_Lindeloef}
	Let $F$ be a continuous vector field, $F \in \Cont(U,\R^n)$, $U \subseteq \R^n$ open, which defines a system of differential equations, 
	\begin{align}
		\dot{X} &= F(X) 
		\label{frameworks:classical_mechanics:eqn:general_deq}
		. 
	\end{align}
	Assume for a certain initial condition $X_0 \in U$ there exists a ball $B_{\rho}(X_0) := \bigl \{ X \in \R^n \; \vert \; \sabs{X - X_0} < \rho \bigr \} \subseteq U$, $\rho > 0$, such that $F$ is Lipschitz on $B_{\rho}(X_0)$, \ie there exists $L > 0$ which satisfies 
	\begin{align*}
		\babs{F(X) - F(X')} \leq L \babs{X - X'} 
	\end{align*}
	for all $X , X' \in B_{\rho}(X_0)$. Then the initial value problem, equation~\eqref{frameworks:classical_mechanics:eqn:general_deq} with $X(0) = X_0$, has a unique solution $t \mapsto X(t)$ for times $\abs{t} \leq T := \min \bigl ( \nicefrac{\rho}{V_{\mathrm{max}}} , \nicefrac{1}{2L} \bigr )$ where the maximal velocity is defined as $V_{\mathrm{max}} := \sup_{X \in B_{\rho}(X_0)} \sabs{F(X)}$. 
\end{thm}
Multiplying the hamiltonian equations of motion~\eqref{frameworks:classical_mechanics:eqn:hamiltons_eom} by the inverse of the symplectic form, we get the vector field $F$ explicitly in terms of the gradients of $H$ with respect to $x$ and $p$. 
\begin{proof}
	We only sketch the proof here since it is part of the standard course on analysis (for~physicists) or differential equations (for mathematicians). It consists of three steps: 
	
	We can rewrite the initial value problem equation~\eqref{frameworks:classical_mechanics:eqn:general_deq} with $X(0) = X_0$ as 
	\begin{align*}
		X(t) = X_0 + \int_0^t \dd s \, F(X(s)) 
	\end{align*}
	This equation can be solved iteratively: we define $X_0(t) := X_0$ and the $n+1$th iteration 
	$X_{n+1}(t) := \bigl ( P (X_n) \bigr )(t)$ in terms of the so-called Picard map 
	\begin{align*}
		\bigl ( P (X_n) \bigr )(t) := X_0 + \int_0^t \dd s \, F(X_n(t)) 
		. 
	\end{align*}
	If $t \in [-T,+T]$ and $T > 0$ is chosen to be small enough, $P : \mathcal{X} \longrightarrow \mathcal{X}$ is a contraction on the space of trajectories which start at $X_0$, 
	\begin{align*}
		\mathcal{X} := \Bigl \{ Y \in \Cont \bigl ([-T,+T] , B_{\rho}(X_0) \bigr ) \; \big \vert \; Y(0) = X_0 \Bigr \} 
		. 
	\end{align*}
	$\mathcal{X}$ is a \emph{complete} metric space if we use 
	\begin{align*}
		\mathrm{d} (Y,Z) := \sup_{t \in [-T,+T]} \babs{Y(t) - Z(t)} 
		&& 
		Y,Z \in \mathcal{X} 
	\end{align*}
	to measure distances between trajectories. A map $P : \mathcal{X} \longrightarrow \mathcal{X}$ is a contraction if for all $Y,Z$ there exists $C < 1$ such that 
	\begin{align*}
		\mathrm{d} \bigl ( P(Y) , P(Z) \bigr ) \leq C \, \mathrm{d} (Y,Z) 
		. 
	\end{align*}
	%
	If the Picard iteration is a contraction on $\mathcal{X}$, since it is complete, the sequence $(X_n)$ converges to the \emph{unqiue} solution of \eqref{frameworks:classical_mechanics:eqn:general_deq} $X = P(X)$ by the Banach fix point theorem. 
	
	To show that $P$ is really a contraction, we need to treat both possible choices of $T$: the first, $T \leq \nicefrac{\rho}{V_{\mathrm{max}}}$ implies that the trajectory cannot leave the ball $B_{\rho}(X_0)$. For any $Y \in \mathcal{X}$, we have 
	\begin{align*}
		\babs{P(Y) - X_0} = \abs{\int_0^t \dd s \, F(Y(s))} \leq t \, V_{\mathrm{max}} \leq T \, V_{\mathrm{max}} < \rho 
		. 
	\end{align*}
	The second condition, $T \leq \nicefrac{1}{2 L}$, together with the Lipschitz property ensure that it is also a contraction with $C = \nicefrac{1}{2}$: for any $Y,Z \in \mathcal{X}$, we have 
	\begin{align*}
		\mathrm{d} \bigl ( P(Y) , P(Z) \bigr ) &= \sup_{t \in [-T,+T]} \abs{\int_0^t \dd s \, \bigl [ \bigl ( F(Y) \bigr )(t) - \bigl ( F(Z) \bigr )(t) \bigr ]} 
		\\
		&\leq T \, L \sup_{t \in [-T,+T]} \babs{Y(t) - Z(t)} 
		\leq \tfrac{1}{2L} L \, \mathrm{d} (Y,Z) = \tfrac{1}{2} \mathrm{d} (Y,Z)
	\end{align*}
	This concludes the proof. 
\end{proof}
\begin{cor}\label{frameworks:classical_mechanics:cor:existence_flow}
	If the vector field $F$ satisfies the Lipschitz condition \emph{globally}, \ie there exists $L > 0$ such that 
	\begin{align*}
		\babs{F(X) - F(X')} \leq L \babs{X - X'} 
	\end{align*}
	holds for all $X , X' \in \R^n$, then $t \mapsto \Phi_t(X_0)$ exists globally for all $t \in \R_t$ and $X_0 \in \R^n$. 
\end{cor}
%
%
\begin{proof}
	For every $X_0 \in \R^n$, we can solve the initial value problem at least for $\abs{t} \leq \nicefrac{1}{2 L}$. Since 
	\begin{align*}
		\babs{F(X_0)} \leq \babs{F(0)} + L \babs{X_0} 
	\end{align*}
	implies that if we choose our ball large enough, \ie for any $\rho > \abs{X_0} + \nicefrac{\abs{F(0)}}{L}$, the condition $\abs{t} \leq \nicefrac{\rho}{V_{\mathrm{max}}}$ is automatically satisfied, 
	\begin{align*}
		\frac{\rho}{V_{\mathrm{max}}} \geq \frac{\rho}{\sabs{F (0)} + L(\abs{X_0} + \rho)} \geq \frac{1}{2 L} 
		. 
	\end{align*}
	Hence, we can patch \emph{local} solutions together using the group property $\Phi_{t} \circ \Phi_{t'} = \Phi_{t + t'}$ of the flow to obtain \emph{global} solutions. 
\end{proof}
Another important fact is that the flow $\Phi$ inherits the smoothness of the vector field which generates it. 
\marginpar{\small 2009.10.21}%
\begin{thm}\label{frameworks:classical_mechanics:thm:smoothness_flow}
	Assume the vector field $F$ is $k$ times continuously differentiable, $F \in \Cont^k ( U , \R^n )$, $U \subseteq \R^n$. Then the flow $\Phi$ associated to \eqref{frameworks:classical_mechanics:eqn:general_deq} is also $k$ times continuously differentiable, \ie $\Phi \in \Cont^k \bigl ( [-T,+T] \times V , U \bigr )$ where $V \subset U$ is suitable. 
\end{thm}
\begin{proof}
	We refer to Chapter 3, Section 7.3 in \cite{Arnold:ode:1992}. 
\end{proof}
The above results immediately apply to the hamiltonian equations of motion: 
\begin{cor}
	Let $H = \tfrac{1}{2m} p^2 + V(x)$ be the hamiltonian which generates the dynamics according to equation~\eqref{frameworks:classical_mechanics:eqn:hamiltons_eom} such that $\nabla_x V$ satisfies a global Lipschitz condition 
	\begin{align*}
		\babs{\nabla_x V(x) - \nabla_x V(x')} \leq L \babs{x - x'} 
		&& \forall x , x' \in \R^d_x 
		. 
	\end{align*}
	Then the hamiltonian flow $\Phi$ exists for all $t \in \R_t$. 
\end{cor}

\paragraph{Classical states} 
\label{frameworks:classical_mechanics:classical_states}

Pure states in classical mechanics are simply points in phase space: a point particle's state at time $t$ is characterized by its position $x(t)$ and momentum $p(t)$. More generally, one can consider \emph{distributions of initial conditions} which are relevant in statistical mechanics, for instance. 
\begin{defn}[Classical states]
	A classical state is a probability measure $\mu$ on phase space, that is a positive Borel measure\footnote{Unfortunately we do not have time to define Borel sets and Borel measures in this context. We refer the interested reader to chapter~1 of \cite{LiebLoss:Analysis:2001}. Essentially, a Borel measure assigns a ``volume'' to ``nice'' sets, \ie Borel sets. } which is normed to $1$, 
	\begin{align*}
		\mu(U) &\geq 0 && \mbox{for all Borel sets $U \subseteq \Pspace$} \\
		\mu(\Pspace) &= 1 
		. 
	\end{align*}
	Pure states are \emph{point measures}, \ie if $(x_0,p_0) \in \Pspace$, then the associated pure state is given by $\mu_{(x_0,p_0)}(\cdot) := \delta_{(x_0,p_0)}(\cdot) = \delta(\cdot - (x_0,p_0))$.\footnote{Here, $\delta$ is the Dirac distribution which we will consider in detail in Chapter~\ref{S_and_Sprime}. } 
\end{defn}
%

\paragraph{Observables} 
\label{frameworks:classical_mechanics:observables}

Observables $f$ such as position, momentum, angular momentum and energy are smooth functions on phase space with values in $\R$, $f \in \Cont^{\infty}(\Pspace,\R)$. If we compose observables with the flow, we can work with \emph{time-evolved observables} 
\begin{align*}
	f(t) := f \circ \Phi_t : \Pspace \longrightarrow \R 
\end{align*}
where $\Phi$ is the flow generated by some suitable hamiltonian $H$. If $\bigl ( x(t) , p(t) \bigr )$ is the trajectory associated to the initial conditions $(x_0,p_0)$, then $\bigl ( f(t) \bigr ) (x_0,p_0) = f \bigl ( \Phi_t(x_0,p_0) \bigr ) = f \bigl ( x(t) , p(t) \bigr )$ gives the value of the observable at time $t$. This is analogous to the \emph{Heisenberg picture} in quantum mechanics: observables are evolved in time and states stay constant. The possible outcomes of measurements is given by the 
\begin{defn}[Spectrum of an observable]
	The spectrum of a classical observables, \ie the set of possible outcomes of measurements, is given by 
	\begin{align*}
		\spec f := f(\Pspace) = \image f 
		. 
	\end{align*}
\end{defn}
If we are given a classical state $\mu$, then the \emph{expected value $\mathbb{E}_{\mu}$} of an observable $f$ for the distribution of initial conditions $\mu$ is given by 
\begin{align*}
	\mathbb{E}_{\mu} \bigl ( f(t) \bigr ) :=& \int_{\Pspace} \dd \mu(x,p) \, f \bigl ( \Phi_t(x,p) \bigr ) 
	= \int_{\Pspace} \dd \mu \bigl ( \Phi_{-t}(x,p) \bigr ) \, f(x,p) \\
	=& \mathbb{E}_{\mu(t)}(f) 
	. 
\end{align*}
The right-hand side corresponds to the \emph{Schrödinger picture} where states are evolved in time and not observables: the time-evolved state $\mu(t)$ associated to $\mu$ is defined by $\mu(t) := \mu \circ \Phi_{-t}$. 


\paragraph{Time evolution} 
\label{frameworks:classical_mechanics:time_evolution}

Now to the dynamics: we have already defined the time-evolved observable $f(t)$. Which equation generates its time evolution? 
\begin{prop}\label{frameworks:classical_mechanics:prop:equations_of_motion_Poisson_bracket}
	Let $f \in \Cont^{\infty}(\Pspace,\R)$ be an observable and $\Phi$ the hamiltonian flow which solves the equations of motion~\eqref{frameworks:classical_mechanics:eqn:hamiltons_eom} associated to a hamiltonian $H \in \Cont^{\infty}(\Pspace,\R)$ which we assume to exist globally in time for all $(x_0,p_0) \in \Pspace$. Then 
	\begin{align}
		\frac{\dd }{\dd t} f(t) = \bigl \{ H , f(t) \bigr \} 
		\label{frameworks:classical_mechanics:eqn:eom_observables}
	\end{align}
	%
	holds where $\bigl \{ f , g \bigr \} := \sum_{l = 1}^d \bigl ( \partial_{p_l} f \, \partial_{x_l} g - \partial_{x_l} f \, \partial_{p_l} g \bigr )$ is the so-called \emph{Poisson bracket}. 
\end{prop}
\begin{proof}
	Theorem~\ref{frameworks:classical_mechanics:thm:smoothness_flow} implies the smoothness of the flow from the smoothness of the hamiltonian. This means $f(t) \in \Cont^{\infty}(\Pspace,\R)$ is again a classical observable. By assumption, all initial conditions lead to trajectories that exist globally in time.\footnote{A slightly more sophisticated argument shows that the Proposition holds if the hamiltonian flow exists only locally in time. } For $(x_0,p_0)$, we compute the time derivative of $f(t)$ to be 
	\begin{align*}
		\left ( \frac{\dd }{\dd t} f(t) \right )(x_0,p_0) &= \frac{\dd }{\dd t} f \bigl ( x(t) , p(t) \bigr ) 
		\\
		&= \sum_{l = 1}^d \Bigl ( \partial_{x_l} f \circ \Phi_t(x_0,p_0) \, \dot{x}_l(t) + \partial_{p_l} f \circ \Phi_t(x_0,p_0) \, \dot{p}_l(t) \Bigr ) 
		\\
		&\overset{\ast}{=} \sum_{l = 1}^d \Bigl ( \partial_{x_l} f \circ \Phi_t(x_0,p_0) \, \partial_{p_l} H \circ \Phi_t(x_0,p_0) 
		+ \\
		&\qquad \qquad 
		+ \partial_{p_l} f \circ \Phi_t(x_0,p_0) \, \bigl ( - \partial_{x_l} H \circ \Phi_t(x_0,p_0) \bigr ) \Bigr ) 
		\\
		&= \bigl \{ H(t) , f(t) \bigr \} 
		. 
	\end{align*}
	In the step marked with $\ast$, we have inserted the hamiltonian equations of motion. Compared to equation~\eqref{frameworks:classical_mechanics:eqn:eom_observables}, we have $H$ instead of $H(t)$ as argument in the Poisson bracket. However, by setting $f(t) = H(t)$ in the above equation, we see that energy is a \emph{conserved quantity}, 
	\begin{align*}
		\frac{\dd }{\dd t} H(t) = \bigl \{ H(t) , H(t) \bigr \} = 0 
		. 
	\end{align*}
	Hence, we can replace $H(t)$ by $H$ in the Poisson bracket with $f$ and obtain equation~\eqref{frameworks:classical_mechanics:eqn:eom_observables}. 
\end{proof}
The proof immediately leads to the notion of conserved quantity: 
\begin{defn}[Conserved quantity/constant of motion]\label{frameworks:classical:defn:conserved_quantity}
	An observable $f \in \Cont^{\infty}(\Pspace,\R)$ which is invariant under the flow $\Phi$ generated by the hamiltonian $H \in \Cont^{\infty}(\Pspace,\R)$, \ie 
	\begin{align*}
		f(t) = f(0) 
		, 
	\end{align*}
	or equivalently satisfies 
	\begin{align*}
		\frac{\dd }{\dd t} f(t) = \bigl \{ H , f(t) \bigr \} = 0 
		, 
	\end{align*}
	is called \emph{conserved quantity} or \emph{constant of motion}. 
\end{defn}
As is very often in physics and mathematics, we have completed the circle: starting from the hamiltonian equations of motion, we have proven that the time evolution of observables is given by the Poisson bracket. Alternatively, we could have \emph{started} by postulating 
\begin{align*}
	\frac{\dd}{\dd t} f(t) = \bigl \{ H , f(t) \bigr \} 
\end{align*}
for observables and we would have \emph{arrived} at the hamiltonian equations of motion by plugging in $x$ and $p$ as observables. Another important fact is Liouville's Theorem which states that the hamiltonian flow \emph{preserves phase space volume}: 
%
\begin{thm}[Liouville]\label{frameworks:classical:thm:Liouville}
	The hamiltonian vector field is divergence free, \ie the hamiltonian flow preserves volume in phase space of bounded subsets $V$ of $\Pspace$ with smooth boundary $\partial V$. In particular, the functional determinant of the flow is constant and equal to 
	\begin{align*}
		\mathrm{det} \, \bigl ( D \Phi_t(x,p) \bigr ) = 1 
	\end{align*}
	for all $t \in \R_t$ and $(x,p) \in \Pspace$. 
\end{thm}
\begin{figure}
	\hfil\includegraphics[height=4cm]{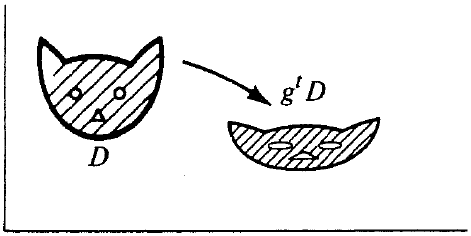}\hfil
	\caption{Phase space volume is preserved under the hamiltonian flow. }
\end{figure}
\begin{remark}\label{frameworks:classical:remark:Liouville}
	We will need a fact from the theory of dynamical systems: if $\Phi_t$ is the flow associated to a differential equation $\dot{X} = F(X)$ with $F \in \Cont^1(\R^d,R^d)$, then 
	\begin{align*}
		\frac{\dd }{\dd t} D \Phi_t(X) = DF(X) \, D \Phi_t(X) 
		, 
		&& D \Phi_t \big \vert_{t = 0} = \id_{\R^d} 
		, 
	\end{align*}
	holds for the differential of the flow. As a consequence, one can prove 
	\begin{align*}
		\frac{\dd }{\dd t} \bigl ( \mathrm{det} \, D \Phi_t(X) \bigr ) &= \mathrm{tr} \, \bigl ( DF \bigl (\Phi_t(X) \bigr ) \bigr ) \, \mathrm{\det} \, \bigl ( D \Phi_t(X) \bigr ) 
		\\
		&= \mathrm{div} \, F \bigl ( \Phi_t(X) \bigr ) \, \mathrm{\det} \, \bigl ( D \Phi_t(X) \bigr ) 
		, 
	\end{align*}
	and $\mathrm{det} \, \bigl ( D \Phi_t \bigr ) \big \vert_{t = 0} = 1$. We refer to \cite[Chapter~1.6]{Dreher:dynamische_systeme:2007} for details and proofs. 
\end{remark}
\begin{proof}
	Let $H$ be the hamiltonian which generates the flow $\Phi_t$. Let us denote the hamiltonian vector field by 
	\begin{align*}
		X_H = \left (
		\begin{matrix}
			+ \nabla_p H \\
			- \nabla_x H \\
		\end{matrix}
		\right ) 
		. 
	\end{align*}
	Then a direct calculation yields 
	\begin{align*}
		\mathrm{div} \, X_H = \sum_{l = 1}^d \Bigl ( \partial_{x_l} \bigl ( + \partial_{p_l} H \bigr ) + \partial_{p_l} \bigl ( - \partial_{x_l} H \bigr ) \Bigr ) = 0 
	\end{align*}
	and the hamiltonian vector field is divergence free. This implies the hamiltonian flow $\Phi_t$ preserves volumes in phase space: let $V \subseteq \Pspace$ be a bounded region in phase space (a Borel subset) with smooth boundary. Then for all $- T \leq t \leq T$ for which the flow exists, we have 
	\begin{align*}
		\frac{\dd}{\dd t} \mathrm{Vol} \, \bigl ( \Phi_t(V) \bigr ) &= \frac{\dd}{\dd t} \int_{\Phi_t(V)} \dd x \, \dd p 
		\\
		&= \frac{\dd}{\dd t} \int_{V} \dd x' \, \dd p' \, \mathrm{det} \, \bigl ( D \Phi_t(x',p') \bigr )
		. 
	\end{align*}
	Since $V$ is bounded, we can bound $\mathrm{det} \, \bigl ( D \Phi_t \bigr )$ and its time derivative uniformly. Thus, can interchange integration and differentiation and apply Remark~\ref{frameworks:classical:remark:Liouville}, 
	\begin{align*}
		\frac{\dd}{\dd t} \int_{V} \dd x' \, \dd p' \, \mathrm{det} \, \bigl ( D \Phi_t(x',p') \bigr ) &= \int_V \dd x' \, \dd p' \, \frac{\dd}{\dd t} \mathrm{det} \, \bigl ( D \Phi_t(x',p') \bigr ) 
		\\
		&
		= \int_V \dd x' \, \dd p' \, \underbrace{\mathrm{div} \, X_H \bigl ( \Phi_t(x',p') \bigr )}_{= 0} \, \mathrm{det} \, \bigl ( D \Phi_t(x',p') \bigr ) 
		\\
		&
		= 0 
		. 
	\end{align*}
	Hence $\frac{\dd}{\dd t} \mathrm{Vol} \, (V) = 0$ and the hamiltonian flow conserves phase space volume. The functional determinant of the flow is constant as the time derivative vanishes, 
	\begin{align*}
		\frac{\dd}{\dd t} \mathrm{det} \, \bigl ( D \Phi_t(x',p') \bigr ) = 0 
		, 
	\end{align*}
	and equal to $1$, 
	\begin{align*}
		\mathrm{det} \, \bigl ( D \Phi_t (x',p') \bigr ) \big \vert_{t = 0} = \mathrm{det} \, \id_{\Pspace} = 1 
		. 
	\end{align*}
	This concludes the proof. 
\end{proof}
With a different proof relying on alternating multilinear forms, the requirements on $V$ can be lifted, see \eg \cite[Satz~29.3 and Satz~29.5]{Knauf:klassische_Mechanik:2004}. 
\begin{cor}\label{frameworks:classical:states_stay_states}
	Let $\mu$ be a state on phase space $\Pspace$ and $\Phi_t$ the flow generated by a hamiltonian $H \in \Cont^{\infty}(\Pspace)$ which we assume to exist for $\abs{t} \leq T$ where $0 < T \leq \infty$ is suitable. Then $\mu(t)$ is again a state. 
\end{cor}
\begin{proof}
	Since $\Phi_t$ is continuous, it is also measurable. Thus $\mu(t) = \mu \circ \Phi_{-t}$ is also a Borel measure on $\Pspace$ ($\Phi_{-t}$ exists by assumption on $t$). In fact, $\Phi_t$ is a diffeomorphism on phase space. Liouville's theorem not only ensures that the measure $\mu(t)$ remains positive, but also that it is normalized to 1: let $U \subseteq \Pspace$ be a Borel subset. Then we conclude 
	\begin{align*}
		\bigl ( \mu(t) \bigr ) (U) &= \int_{U} \dd \bigl ( \mu(t) \bigr )(x,p) 
		= \int_{U} \dd \mu \bigl ( \Phi_{-t}(x,p) \bigr ) 
		\\ 
		&= \int_{\Phi_{-t}(U)} \dd \mu (x,p) \geq 0 
	\end{align*}
	where we have used the positivity of $\mu$ and the fact that $\Phi_{-t}(U)$ is again a Borel set by continuity of $\Phi_{-t}$. If we set $U = \Pspace$ and use the fact that the flow is a diffeomorphism, we see that $\Pspace$ is mapped onto itself, $\Phi_{-t}(\Pspace) = \Pspace$, and the normalization of $\mu$ leads to 
	\begin{align*}
		\bigl ( \mu(t) \bigr ) (\Pspace) &= \int_{\Phi_{-t}(\Pspace)} \dd \mu (x,p) = \int_{\Pspace} \dd \mu (x,p) = 1
	\end{align*}
	This concludes the proof. 
\end{proof}
%


\section{Quantum mechanics} 
\label{frameworks:quantum_mechanics}
While position and momentum characterize the state of a classical particle, both are \emph{not simultaneously measurable} with arbitrary precision in a quantum system. This is forbidden by Heisenberg's uncertainty principle, 
\begin{align*}
	\Delta_{\psi} \hat{x} \, \Delta_{\psi} \hat{p} \geq \frac{\hbar}{2} 
	, 
\end{align*}
where 
\begin{align*}
	\Delta_{\psi} \hat{x} &:= \sqrt{\mathrm{Var}_{\psi}(\hat{x})} := \sqrt{\bscpro{\psi}{(\hat{x} - \sscpro{\psi}{\hat{x} \psi})^2 \psi}}
\end{align*}
and $\Delta_{\psi} \hat{p} := \sqrt{\mathrm{Var}(\hat{p})}$ are the standard deviations of position and momentum with respect to the state $\psi$. They quantify how sharply the wave function $\psi$ is peaked in position and momentum space. Either we can pinpoint the location with increased precision (reduce the standard deviation $\Delta_{\psi} \hat{x}$) or measure the particle's momentum more accurately (reduce $\Delta_{\psi} \hat{p}$).\marginpar{\small 2009.10.22} 
In contrast, the variance (and thus the standard deviation) of a \emph{classical} observable \emph{with respect to a pure state vanishes}, 
\begin{align*}
	\mathrm{Var}_{\delta_{(x_0,p_0)}} \, (f) :=& \mathbb{E}_{\delta_{(x_0,p_0)}} \bigl ( (f - \mathbb{E}_{\delta_{(x_0,p_0)}}(f) )^2) \bigr ) 
	= \mathbb{E}_{\delta_{(x_0,p_0)}}(f^2) - \mathbb{E}_{\delta_{(x_0,p_0)}}(f)^2 
	\\
	=& f^2(x_0,p_0) - f(x_0,p_0)^2 = 0 
	. 
\end{align*}
Quantum particles simultaneously have wave and particle character: the Schrödinger equation 
\begin{align*}
	i \hbar \frac{\partial}{\partial t} \psi(t) = \hat{H} \psi(t) 
	, 
	&& \psi(t) \in L^2(\R^d), \psi(0) = \psi_0 
	, 
\end{align*}
is structurally very similar to a wave equation. $\hat{H}$ is the so-called hamilton operator (hamiltonian for short) and is typically given by 
\begin{align*}
	\hat{H} = \frac{1}{2m} (- i \hbar \nabla_x)^2 + V(\hat{x}) 
	. 
\end{align*}
Here $m > 0$ is the mass of the particle and $V$ the potential it has been subjected to. The physical constant $\hbar$ relates the energy of a particle with the associated wavelength and has units of $[\mbox{energy} \cdot \mbox{time}]$. 

Pure states are described by wave functions, \ie complex-valued, square integrable functions 
\begin{align}
	\psi \in L^2(\R^d) := \Bigl \{ \psi : \R^d \longrightarrow \C \; \big \vert \; \mbox{$\psi$ measurable, } \int_{\R^d} \dd x \, \abs{\psi(x)}^2 < \infty \Bigr \} / \sim 
	. 
\end{align}
$L^2(\R^d)$ with the usual scalar product 
\begin{align}
	\scpro{\phi}{\psi} := \int_{\R^d} \dd x \, \phi(x)^* \, \psi(x) 
\end{align}
and norm $\norm{\psi} := \sqrt{\scpro{\psi}{\psi}}$ is the prototype of a \emph{Hilbert space.} $\sim$ means we identify two functions if they agree almost everywhere. We will discuss this space in much more detail in Chapter~\ref{hilbert_spaces} and we focus on the physical content rather than mathematics. In physics text books, one usually encounters the the \emph{bra-ket} notation: here $\ket{\psi}$ is a state and $\ipro{x}{\psi}$ is $\psi(x)$. The scalar product of $\phi,\psi \in L^2(\R^d)$ is denoted by $\ipro{\phi}{\psi}$ and corresponds to $\scpro{\phi}{\psi}$. Although bra-ket notation can be ambiguous, it is sometimes useful and in fact used in mathematics every once in a while. 

Physically, $\abs{\psi(x,t)}^2$ is interpreted as the \emph{probability to measure a particle at time $t$ in (an infinitesimally small box located in) location $x$}. If we are interested in the probability that we can measure a particle in a region $\Lambda \subseteq \R^d$, we have to integrate $\abs{\psi(x,t)}^2$ over $\Lambda$, 
\begin{align}
	\mathbb{P}(X(t) \in \Lambda) = \int_{\Lambda} \dd x \, \abs{\psi(x,t)}^2 
	. 
\end{align}
If we want to interpret $\abs{\psi}^2$ as \emph{probability} density, the wave function has to be \emph{normalized}, \ie 
\begin{align*}
	\norm{\psi} = \sqrt{\int_{\R^d} \dd x \, \abs{\psi(x)}^2} = 1 
	. 
\end{align*}
%
This point of view is called \emph{Born rule}: $\abs{\psi}^2$ could either be a \emph{mass} or \emph{charge
 density} -- or a \emph{probability density}. To settle this, physicists have performed the double slit experiment with an electron source of low flux. If $\abs{\psi}^2$ were a density, one would see the whole interference pattern building up slowly. Instead, one \emph{measures} ``single impacts'' of electrons and the result is similar to the data obtained from experiments in statistics (\eg the Dalton board). Hence, we speak of \emph{particles}.

\paragraph{Quantum observables} 
\label{frameworks:quantum:quantum_observables}

Quantities that can be measured are represented by symmetric (hermitian) operators $\hat{A}$ on the Hilbert space (typically $L^2(\R^d)$), \ie special linear maps 
\begin{align*}
	\hat{A} : \mathcal{D}(\hat{A}) \subseteq L^2(\R^d) \longrightarrow L^2(\R^d) 
	. 
\end{align*}
Here, $\mathcal{D}(\hat{A})$ is the \emph{domain} of the operator since \emph{typical} observables are not defined for all $\psi \in L^2(\R^d)$. \emph{This is not a mathematical subtlety with no physical content}, quite the contrary: consider the observabble energy, typically given by 
\begin{align*}
	\hat{H} = \frac{1}{2m} (- i \hbar \nabla_x)^2 + V(\hat{x}) 
	, 
\end{align*}
then states in the domain 
\begin{align*}
	\mathcal{D}(\hat{H}) := \Bigl \{ \psi \in L^2(\R^d) \; \big \vert \; \hat{H} \psi \in L^2(\R^d) \Bigr \} \subseteq L^2(\R^d)
\end{align*}
are those of \emph{finite energy}. For all $\psi$ in the domain of the hamiltonian $\mathcal{D}(\hat{H}) \subseteq L^2(\R^d)$, the expectation value 
\begin{align*}
	\bscpro{\psi}{\hat{H} \psi} < \infty
\end{align*}
is bounded. Well-defined observables have domains that are \emph{dense} in $L^2(\R^d)$. Similarly, states in the domain $\mathcal{D}(\hat{x}_l)$ of the $l$th component of the position operator are those that are ``localized in a finite region'' in the sense of expectation values. 

Physically, results of measurements are real which is reflected in the selfadjointness of operators,\footnote{A symmetric operator is selfadjoint if $\hat{H}^{\ast} = \hat{H}$ and $\mathcal{D}(\hat{H}^{\ast}) = \mathcal{D}(\hat{H})$, see~Chapter~\ref{operators:selfadjoint_operators}. } 
\begin{align*}
	\hat{H}^* = \hat{H} 
	. 
\end{align*}
The spectrum $\sigma(\hat{H}) \subseteq \R$ is the set of all possible outcomes of measurements. Unfortunately, we cannot define the spectrum now, but we will do so in Chapter~\ref{operators:bounded}. 

Typically one ``guesses'' quantum observables from classical observables: in $d = 3$, the angular momentum operator is given by 
\begin{align*}
	\hat{L} = \hat{x} \wedge \hat{p} = \hat{x} \wedge  (- i \hbar \nabla_x) 
	. 
\end{align*}
In the simplest case, one uses Dirac's recipe (replace $x$ by $\hat{x}$ and $p$ by $\hat{p} = - i \hbar \nabla_x$) on the classical observable angular momentum $L(x,p) = x \wedge p$. In other words, \emph{many quantum observables are obtained as quantizations of classical observables}: examples are position, momentum and energy. Moreover, the \emph{interpretation} of, say, the angular momentum operator as angular momentum is taken from classical mechanics. This is an indication how important quantizations are conceptually. 

In the definition of the domain, we have already used the definition of expectation value: the expectation value of an observable $\hat{A}$ with respect to a state $\psi$ (which we assume to be normalized, $\norm{\psi} = 1$) is given by 
\begin{align}
	\mathbb{E}_{\psi}(\hat{A}) := \bscpro{\psi}{\hat{A} \psi} 
	. 
\end{align}
The expectation value is finite if the state $\psi$ is in the domain $\mathcal{D}(\hat{A})$. The Born rule of quantum mechanics tells us that if we repeat an experiment measuring the observable $\hat{A}$ many times for a particle that is prepared in the state $\psi$ each time, the statistical average calculated according to the relative frequencies converges to the expectation value $\mathbb{E}_{\psi}(\hat{A})$. 

Hence, quantum observables, selfadjoint operators on Hilbert spaces, are \emph{bookkeeping devices} that have two components: 
\begin{enumerate}[(i)]
	\item a \emph{set of possible outcomes} of measurements, the spectrum $\sigma(\hat{A})$, and 
	\item \emph{statistics}, \ie how often a possible outcome occurs. 
\end{enumerate}
If for two quantum observables $\hat{A}$ and $\hat{B}$ have the same spectrum and statistics, they have the chance of being ``equivalent.'' 


\paragraph{Quantum states} 
\label{frameworks:quantum:states}

Pure states are wave functions $\psi \in L^2(\R^d)$, or rather, wave functions up to a total phase: just like one can measure only energy \emph{differences}, only phase \emph{shifts} are accessible to measurements. Hence, one can think of pure states as orthogonal \emph{projections} 
\begin{align*}
	P_{\psi} := \sopro{\psi}{\psi} = \sscpro{\psi}{\cdot} \, \psi 
	. 
\end{align*}
if $\psi$ is normalized to $1$, $\norm{\psi} = 1$. 
Here, one can see the elegance of bra-ket notation vs. the notation that is ``mathematically proper.'' A generalization of this concept are \emph{density operators} $\hat{\rho}$ (often called density matrices): density matrices are defined via the trace. If $\hat{\rho}$ is a suitable linear operator and $\{ \varphi_n \}_{n \in \N}$ and orthonormal basis of $L^2(\R^d)$, then we define 
\begin{align*}
	\mathrm{tr} \, \hat{\rho} := \sum_{n \in \N} \sscpro{\varphi_n}{\hat{\rho} \varphi_n} 
	. 
\end{align*}
One can easily check that this definition is independent of the choice of basis and we will show this in a later chapter. Clearly, $P_{\psi}$ has trace $1$ and it is also positive in the sense that 
\begin{align*}
	\bscpro{\varphi}{P_{\psi} \varphi} \geq 0 
\end{align*}
for all $\varphi \in L^2(\R^d)$. This is also the good definition for quantum states: 
\begin{defn}[Quantum state]
	A quantum state (or density operator/matrix) $\hat{\rho}$ is a positive operator of trace $1$, \ie 
	\begin{align*}
		\bscpro{\psi}{\hat{\rho} \psi} &\geq 0 
		, \\
		\mathrm{tr} \, \hat{\rho} &= 1 
		, 
		&& \forall \psi \in L^2(\R^d) 
		. 
	\end{align*}
	If $\hat{\rho}$ is also an orthogonal projection, \ie $\hat{\rho}^2 = \hat{\rho}$, it is a pure state.\footnote{Note that the condition $\mathrm{tr} \, \hat{\rho} = 1$ implies that $\hat{\rho}$ is a \emph{bounded} operator while the positivity implies the selfadjointness. Hence, if $\hat{\rho}$ is a projection, \ie ${\hat{\rho}}^2 = \hat{\rho}$, it is automatically also an orthogonal projection. } Otherwise $\hat{\rho}$ is a mixed state. 
\end{defn}
\begin{example}
	Let $\psi_j \in L^2(\R^d)$ be two wave functions normalized to $1$. Then for any $0 < \alpha < 1$
	\begin{align*}
		\hat{\rho} = \alpha P_{\psi_1} + (1 - \alpha) P_{\psi_2} = \alpha \sopro{\psi_1}{\psi_1} + (1 - \alpha) \sopro{\psi_2}{\psi_2}
	\end{align*}
	is a mixed state as 
	\begin{align*}
		\hat{\rho}^2 
		&= \alpha^2 \sopro{\psi_1}{\psi_1} + (1 - \alpha)^2 \sopro{\psi_2}{\psi_2} 
		+ \\
		&\qquad \qquad 
		+ \alpha (1 - \alpha) \bigl ( \sopro{\psi_1}{\psi_1} \sopro{\psi_2}{\psi_2} + \sopro{\psi_2}{\psi_2} \sopro{\psi_1}{\psi_1} \bigr )
		\\
		&\neq \hat{\rho}
		. 
	\end{align*}
	Even if $\psi_1$ and $\psi_2$ are orthogonal to each other, since $\alpha^2 \neq \alpha$ and similarly $(1 - \alpha)^2 \neq (1 - \alpha)$, $\hat{\rho}$ cannot be a projection. Nevertheless, it is a state since $\mathrm{tr} \, \hat{\rho} = \alpha + (1 - \alpha) = 1$. Keep in mind that $\hat{\rho}$ does not project on $\alpha \psi_1 + (1 - \alpha) \psi_2$! 
\end{example}
%


\paragraph{Time evolution} 
\label{frameworks:quantum:time_evolution}

The time evolution is determined through the \emph{Schrödinger equation}, 
\begin{align}
	i \hbar \frac{\partial}{\partial t} \psi(t) = \hat{H} \psi(t)
	, 
	&& \psi(t) \in L^2(\R^d), \; \psi(0) = \psi_0, \; \norm{\psi_0} = 1 
	. 
\end{align}
Alternatively, one can write $\psi(t) = U(t) \psi_0$ with $U(0) = \id_{L^2}$. Then, we have 
\begin{align*}
	i \hbar \frac{\partial}{\partial t} U(t) = \hat{H} U(t) 
	, 
	&& U(0) = \id_{L^2} 
	. 
\end{align*}
If $\hat{H}$ \emph{were} a number, one would immediately use the ansatz 
\begin{align}
	U(t) = e^{- i \frac{t}{\hbar} \hat{H}} 
\end{align}
as solution to the Schrödinger equation. If $\hat{H}$ is a selfadjoint operator, this is \emph{still true}, but takes a lot of work to justify rigorously if the domain of $\hat{H}$ is not all of $L^2(\R^d)$ (the case of unbounded operators, the \emph{generic} case). 

As has already been mentioned, we can evolve either states or observables in time and one speaks of the Schrödinger or Heisenberg picture, respectively. In the Schrödinger picture, states evolve according to 
\begin{align*}
	\psi(t) = U(t) \psi_0 
\end{align*}
while observables remain fixed. Conversely, in the Heisenberg picture, states are kept fixed in time and observables evolve according to 
\begin{align}
	\hat{A}(t) := U(t)^* \, \hat{A} \, U(t) = e^{+ i \frac{t}{\hbar} \hat{H}} \hat{A} e^{- i \frac{t}{\hbar} \hat{H}} 
	. 
\end{align}
Heisenberg observables satisfy 
\begin{align}
	\frac{\dd}{\dd t} \hat{A}(t) = \frac{i}{\hbar} \bigl [ \hat{H} , \hat{A}(t) \bigr ] 
\end{align}
which\marginpar{\small 2009.10.27} can be checked by plugging in the definition of $\hat{A}(t)$ and elementary \emph{formal} manipulations. It is no coincidence that this equation looks structurally similar to equation~\eqref{frameworks:classical_mechanics:eqn:eom_observables}! 

As a last point, we mention the conservation of probability: if $\psi(t)$ solves the Schrödinger equation for some selfadjoint $\hat{H}$, then we can check at least formally that the time evolution is unitary and thus preserves probability, 
\begin{align*}
	\frac{\dd}{\dd t} \bnorm{\psi(t)}^2 &= \frac{\dd}{\dd t} \bscpro{\psi(t)}{\psi(t)} 
	= \bscpro{\tfrac{1}{i \hbar} \hat{H} \psi(t)}{\psi(t)} + \bscpro{\psi(t)}{\tfrac{1}{i \hbar} \hat{H} \psi(t)} 
	\\
	&
	= \frac{i}{\hbar} \Bigl ( \bscpro{\psi(t)}{\hat{H}^* \psi(t)} - \bscpro{\psi(t)}{\hat{H} \psi(t)} \Bigr ) 
	\\
	&
	= \frac{i}{\hbar} \bscpro{\psi(t)}{(\hat{H}^* - \hat{H}) \psi(t)} 
	= 0 
	. 
\end{align*}
Conservation of probability is reminiscent of Corollary~\ref{frameworks:classical:states_stay_states}. We see that the condition $\hat{H}^* = \hat{H}$ is the key here: selfadjoint operators generate unitary evolution groups. As a matter of fact, there are cases when one \emph{wants} to violate conservation of proability: one has to introduce so-called \emph{optical potentials} which simulate particle creation and annihilation. 

The time evolution $e^{- i \frac{t}{\hbar} \hat{H}}$ is not the only unitary group of interest, other commonly used examples are \emph{translations} in position or momentum which are generated by the momentum and position operator, respectively (the order is reversed!), as well as rotations which are generated by the angular momentum operators. 



\section{Comparison of the two frameworks} 
\label{frameworks:comparison}
Now that we have an understanding of the structures of classical and quantum mechanics, juxtaposed in Table~\ref{frameworks:comparison:table:overview_frameworks}, we can elaborate on the differences and similarities of both theories. 
%
\begin{table}
	\begin{tabularx}{\textwidth}{>{\small\raggedright\hsize=3cm}X | >{\small\raggedright}X >{\small\raggedright}X >{\hsize=0cm}X} 
		\makebox[1cm]{}  & \textit{Classical} & \textit{Quantum} & \\ \hline 
		\textit{Observables} & $f \in \Cont^{\infty}(\Pspace,\R)$ & selfadjoint operators acting on Hilbert space $L^2(\R^d)$ & \\ [0.5ex] 
		\textit{Building block observables} & position $x$ and momentum $p$ & position $\hat{x}$ and momentum $\hat{p}$ operators & \\ [0.5ex] 
		\textit{Possible results of measurements} & $\mathrm{im}(f)$ & $\sigma(\hat{A})$ & \\ [0.5ex] 
		\textit{States} & probability measures $\mu$ on phase space $\Pspace$ & density operators $\hat{\rho}$ on $L^2(\R^d)$ & \\ [0.5ex] 
		\textit{Pure states} & points in phase space $\Pspace$ & wave functions $\psi \in L^2(\R^d)$ & \\ [0.5ex] 
		\textit{Generator of evolution} & hamiltonian function $H : \Pspace \longrightarrow \R$ & hamiltonian operator $\hat{H}$ & \\ [0.5ex] 
		\textit{Infinitesimal time evolution equation} & $\frac{\dd}{\dd t} f(t) = \{ H , f(t) \}$ & $\frac{\dd }{\dd t} \hat{A}(t) = \frac{i}{\hbar} [ \hat{H} , \hat{A}(t) ]$ & \\ [0.5ex]
		\textit{Integrated time evolution} & hamiltonian flow $\phi_t$ & $e^{+ i \frac{t}{\hbar} \hat{H}} \, \Box \, e^{- i \frac{t}{\hbar} \hat{H}}$ & \\ 
	\end{tabularx}
	\caption{Comparison of classical and quantum framework}
	\label{frameworks:comparison:table:overview_frameworks}
\end{table}
For instance, observables form an \emph{algebra} (a vector space with multiplication): in classical mechanics, we use the \emph{pointwise product} of functions, 
\begin{align*}
	\cdot :& \, \Cont^{\infty}(\Pspace) \times \Cont^{\infty}(\Pspace) \longrightarrow \Cont^{\infty}(\Pspace) , \; (f , g) \mapsto f \cdot g 
	\\
	& (f \cdot g)(x,p) := f(x,p) \, g(x,p) 
	, 
\end{align*}
which is obviously commutative. We also admit \emph{complex}-valued functions and add \emph{complex conjugation} as involution (\ie $f^{\ast \ast} = f$). Lastly, we add the Poisson bracket to make $\Cont^{\infty}(\Pspace)$ into a so-called Poisson algebra. As we have seen, the notion of Poisson bracket gives rise to dynamics as soon as we choose an energy function (hamiltonian). 

On the quantum side, bounded operators (see~Chapter~\ref{operators:bounded}) form an algebra. This algebra is non-commutative, \ie 
\begin{align*}
	\hat{A} \cdot \hat{B} \neq \hat{B} \cdot \hat{A} 
	. 
\end{align*}
\emph{Exactly this is what makes quantum mechanics different.} Taking adjoints is the involution here and the commutator plays the role of the Poisson bracket. Again, once a hamiltonian (operator) is chosen, the dynamics of Heisenberg observables $\hat{A}(t)$ is determined by the commutator of the $\hat{A}(t)$ with the hamiltonian $\hat{H}$. If an operator commutes with the hamiltonian, \emph{it is a constant of motion}. This is in analogy with Definition~\ref{frameworks:classical:defn:conserved_quantity} where a classical observable is a constant of motion if and only if its Poisson bracket with the hamiltonian (function) vanishes. 


\section{Properties of quantizations} 
\label{frameworks:properties_of_quantizations}

Quantization is a problem of theoretical and mathematical physics, \emph{not a problem of nature.} Quantum mechanics is simply \emph{more fundamental} than classical mechanics and the reason we can \emph{guess} the correct quantum mechanics comes from the fact that most observables we are interested are macroscopic and have a good semiclassical limit. In a way, quantization is going backwards again. That is the reason why not every quantum observable has a classical analog. Spin, for instance, is purely a quantum mechanical concept with no semiclassical limit. 

Since quantum mechanics is more deep and classical mechanics emerges as an approximation under certain conditions, there should be ``more'' quantum observables than classical observables (although it is usually not possible to qualify what we mean by ``more,'' the frameworks and constructions are too different). In the simplest case, we start with an algebra of classical observables and ``quantize them,'' \ie we associate them to operators on a Hilbert space in a systematic fashion. Although this is what most people think \emph{is} a quantization, we emphasize that this is only \emph{part} of what really is a consistent quantization procedure. 

Let us call the quantization map $\Op$ which promotes suitable functions on phase space to operators on $L^2(\R^d)$. A number of requirements on $\Op$ seem natural:

\paragraph{Linearity} 
\label{frameworks:properties_of_quantizations:linearity}

The map $\Op$ should be linear, \ie for two classical observables $f , g \in \mathfrak{A}_{\mathrm{cl}}$ taken from the algebra of classical observables and $\alpha , \beta \in \C$, we should have 
\begin{align*}
	\Op \bigl ( \alpha f + \beta g \bigr ) = \alpha \, \Op(f) + \beta \, \Op(g) \in \mathfrak{A}_{\mathrm{qm}} 
	. 
\end{align*}
Here, $\mathfrak{A}_{\mathrm{qm}}$ is an algebra of quantum observables. 


\paragraph{Compatibility with involution} 
\label{frameworks:properties_of_quantizations:involution}

$\Op$ should intertwine complex conjugation and taking adjoints, \ie for all $f \in \mathfrak{A}_{\mathrm{cl}}$ 
\begin{align*}
	\Op (f^*) = \Op(f)^* \in \mathfrak{A}_{\mathrm{qm}} 
	. 
\end{align*}
%


\paragraph{Products} 
\label{frameworks:properties_of_quantizations:products}

The two products \emph{cannot} be equivalent: the operator product is noncommutative and hence for general $f , g \in \mathfrak{A}_{\mathrm{cl}}$ 
\begin{align*}
	\Op(f \cdot g) \neq \Op(f) \cdot \Op(g) 
	. 
\end{align*}
Instead, for \emph{suitable} functions $f,g$, we can define a \emph{non-commutative} product $\Weyl$ on the level of functions on phase space such that 
\begin{align*}
	\Op(f \Weyl g) := \Op(f) \cdot \Op(g) \in \mathfrak{A}_{\mathrm{qm}} 
	. 
\end{align*}
\emph{A priori} it is not at all clear whether $f \Weyl g \in \mathfrak{A}_{\mathrm{cl}}$. 


\paragraph{Poisson bracket and commutator} 
\label{frameworks:properties_of_quantizations:brackets}

Roughly, Poisson bracket and commutators play similar roles and thus are analogs of one another, 
\begin{align*}
	\bigl \{ f , g \bigr \} \leftrightsquigarrow \frac{i}{\hbar} \bigl [ \Op(f) , \Op(g) \bigr ] 
	. 
\end{align*}
Just like with the product, the quantization of the Poisson bracket usually does not coincide with with $\nicefrac{1}{i \hbar}$ times the commutator. \marginpar{\small 2009.10.28}


\paragraph{Small parameters} 
\label{frameworks:properties_of_quantizations:epsilon}

One of the most important (and potentially confusing) aspects of the semiclassical limit is the question of the small parameter. Planck's constant $\hbar$ is a \textbf{physical constant}, \ie its value is \emph{fixed} and it has \emph{units}. \emph{Any good small parameter should not have units, \eg the fine structure constant $\alpha \simeq \nicefrac{1}{137}$ is a good small parameter.} Hence, we cannot take the limit $\hbar \rightarrow 0$. Other interpretations of the semiclassical limit point in the right direction: in its earliest version, the semiclassical limit was the limit of ``large quantum numbers.'' Translated to a more modern setting, this means that the \emph{ratio of typical energies to the energy level spacing} is large. Rydberg states of hydrogen-like atoms, \ie states for which the main quantum number is in the range $n \gtrsim 50$, are probably the best-known example. Although this particular mechanism is limited to hamiltonians with non-trivial discrete spectrum, it suggests the root cause of (semi)classical behavior is the \emph{existence of two scales} and the semiclassical parameter \textbf{which from now on we will denote with $\eps$} is related to the ratio of these scales. We give two examples: 
\begin{enumerate}[(i)]
	\item \emph{Separation of spatial scales. }
	Quantum objects are typically much, much smaller than ``classical'' objects and the ratio of typical quantum to typical classical length scales is small. External electromagnetic fields, for instance, vary on the \emph{macroscopic} scale, \ie they \emph{change slowly on the microscopic scale}. In solid state physics, the microscopic length scale is of the order of $\unit[10-100]{\mbox{Å}}$ (\eg measured in terms of the typical lattice spacing of a crystal or the localization length of a particle) while the macroscopic length scale is around $\unit[10^{0}-10^{-2}]{mm}$; hence, the separation of scales is approximately $\eps \lesssim 10^{-3}$. 
	\item \emph{Ratio of masses/Born-Oppenheimer-type systems. }
	Born-Oppenheimer hamiltonians are simply many-body hamiltonians which are used to model molecules. Here, there are two types of particles: heavy (slow) nucleons and light (fast) electrons. The square root of the ratio of electronic and nucleonic mass is small, 
	\begin{align*}
		\eps := \sqrt{\frac{m_{\mathrm{e}}}{m_{\mathrm{nuc}}}} \lesssim \sqrt{\frac{1}{2000}} 
		, 
	\end{align*}
	and this is the reason why nucleons tend to behave classically. Born-Oppenheimer systems will be the main focus of Chapter~\ref{multiscale}. 
\end{enumerate}
However, we caution that the existence of small parameters -- even in the right places -- implies a semiclassical limit: if we take the semirelativistic limit of the Dirac equation, then $\eps = \nicefrac{v_0}{c}$ appears naturally in front of the gradient where $v_0$ is a typical velocity and $c$ is the speed of light. If $v_0$ is well below $c$ (the electron is much slower than the speed of light), no electron-positron pairs are created (since the kinetic energy is below the pair creation threshold of $2 m c^2$) and \emph{electrons and positrons decouple}. Hence, the $\eps \rightarrow 0$ limit need not have the interpretation as semiclassical limit. 
\medskip

\noindent
We see that the reason for different small parameters is that \textbf{semiclassical behavior can be due to very different physical mechanisms}! 
%
\medskip

\noindent
If the semiclassical parameter $\eps$ is small, we expect 
\begin{align*}
	\Op(f \cdot g) = \Op(f) \cdot \Op(g) + \ordere{} 
\end{align*}
as well as 
\begin{align*}
	\Op \bigl ( \{ f , g \} \bigr ) = \frac{1}{i \eps} \bigl [ \Op(f) , \Op(g) \bigr ] + \ordere{} 
	. 
\end{align*}
The latter equation is key in the derivation of the semiclassical limit (the Egorov Theorem~7.1) whose validity can be tested by experiment. 
\medskip

\noindent
For completeness, we mention that we could have quantized classical \emph{states} instead of observables. This is the point of view of Berezin \cite{Berezin:quantization:1975}. 


\chapter{Hilbert spaces} 
\label{hilbert_spaces}

We will give a few basic facts on Hilbert spaces. This is not intended to be a replacement for a lecture on functional analysis. For a more in depth look on the subject, we refer to \cite{Reed_Simon:bibel_1:1981,Werner:Funktionalanalysis:2005,LiebLoss:Analysis:2001}

\section{Prototypical Hilbert spaces: $L^2(\R^d)$ and $\ell^2(\Z^d)$} 
\label{hilbert_spaces:L2}

We have already introduced the space of square integrable functions on $\R^d$, 
\begin{align*}
	\mathcal{L}^2(\R^d) := \Bigl \{ \varphi : \R^d \longrightarrow \C \; \big \vert \; \varphi \mbox{ measurable, } \int_{\R^d} \dd x \, \abs{\varphi(x)}^2 < \infty \Bigr \} 
	, 
\end{align*}
when talking about wave functions. The Born rule states that $\abs{\psi(x)}^2$ is to be interpreted as a probability density on $\R^d$ for position. Hence, we are interested in solutions to the Schrödinger equation which are also square integrable. When we say integrable, we mean integrable with respect to the Lebesgue measure \cite[p.~6~ff.]{LiebLoss:Analysis:2001}. $\mathcal{L}^2(\R^d)$ is a $\C$-vector space, but 
\begin{align*}
	\norm{\varphi}^2 := \int_{\R^d} \dd x \, \abs{\varphi(x)}^2 
\end{align*}
is not a norm: there are functions $\varphi \neq 0$ for which $\norm{\varphi} = 0$. Instead, $\norm{\varphi} = 0$ only ensures 
\begin{align*}
	\varphi(x) = 0 \mbox{ almost everywhere (with respect to the Lebesgue measure $\dd x$).} 
\end{align*}
Almost everywhere is sometimes abbreviated with a.~e. and the terms ``almost surely'' and ``for almost all $x \in \R^d$'' can be used synonymously. If we introduce the equivalence relation 
\begin{align*}
	\varphi \sim \psi :\Leftrightarrow \norm{\varphi - \psi} = 0 
	, 
\end{align*}
then we can define the vector space $L^2(\R^d)$: 
\begin{defn}[$L^2(\R^d)$]
	We define $L^2(\R^d)$ as 
	\begin{align*}
		\mathcal{L}^2(\R^d) / \sim
	\end{align*}
	where $\sim$ is the equivalence relation that identifies $\varphi$ and $\psi$ if $\norm{\varphi - \psi} = 0$. 
\end{defn}
If $\varphi_1 \sim \varphi_2$ are two normalized functions in $\mathcal{L}^2(\R^d)$, then we get the same probabilities for both: if $\Lambda \subseteq \R^d$ is a measurable set, then 
\begin{align*}
	\mathbb{P}_1(X \in \Lambda) = \int_{\Lambda} \dd x \, \abs{\varphi_1(x)}^2 = \int_{\Lambda} \dd x \, \abs{\varphi_2(x)}^2 = \mathbb{P}_2(X \in \Lambda) 
	. 
\end{align*}
This is proven via the triangle inequality and the Cauchy-Schwartz inequality (we will prove the latter in the next chapter): 
\begin{align*}
	0 &\leq \babs{\mathbb{P}_1(X \in \Lambda) - \mathbb{P}_2(X \in \Lambda)} 
	= \abs{\int_{\Lambda} \dd x \, \abs{\varphi_1(x)}^2 - \int_{\Lambda} \dd x \, \abs{\varphi_2(x)}^2} 
	\\
	&
	= \abs{\int_{\Lambda} \dd x \, \bigl ( \varphi_1(x) - \varphi_2(x) \bigr )^* \, \varphi_1(x) - \int_{\Lambda} \dd x \, \varphi_2(x)^* \, \bigl ( \varphi_1(x) - \varphi_2(x) \bigr )} 
	\\
	&\leq \int_{\Lambda} \dd x \, \babs{\varphi_1(x) - \varphi_2(x)} \, \babs{\varphi_1(x)} - \int_{\Lambda} \dd x \, \babs{\varphi_2(x)} \, \babs{\varphi_1(x) - \varphi_2(x)} 
	\\
	&\leq \norm{\varphi_1 - \varphi_2} \, \norm{\varphi_1} + \norm{\varphi_2} \, \norm{\varphi_1 - \varphi_2} 
	= 0
\end{align*}
Very often, another space is used in applications (\eg in tight-binding models): 
\begin{defn}[$\ell^2(S)$]
	Let $S$ be a countable set. Then 
	\begin{align*}
		\ell^2(S) := \Bigl \{ c : S \longrightarrow \C \; \big \vert \; \mbox{$\sum_{j \in S} c_j^* c_j < \infty$} \Bigr \} 
	\end{align*}
	is the space of square-summable sequences. 
\end{defn}
On $\ell^2(S)$ the scalar product $\scpro{c}{c'} := \sum_{j \in S} c_j^* c'_j$ induces the norm $\norm{c} := \sqrt{\scpro{c}{c'}}$. With respect to this norm, $\ell^2(S)$ is complete. \marginpar{\small 2009.11.03}


\section{Abstract Hilbert spaces} 
\label{hilbert_spaces:abstract}

We will only consider vector spaces over $\C$, although much of what we do works just fine if the field of scalars is $\R$. 
\begin{defn}[Metric space]
	Let $\mathcal{X}$ be a set. A mapping $d : \mathcal{X} \times \mathcal{X} \longrightarrow [0,+\infty)$ with properties 
	\begin{enumerate}[(i)]
		\item $d(x,y) = 0$ exactly if $x = y$, 
		\item $d(x,y) = d(y,x)$ (symmetry), and 
		\item $d(x,z) \leq d(x,y) + d(y,z)$ (triangle inequality), 
	\end{enumerate}
	for all $x,y,z \in \mathcal{X}$ is called metric. We refer to $(\mathcal{X},d)$ as metric space (often only denoted as $\mathcal{X}$). A metric space $(\mathcal{X},d)$ is called complete if all Cauchy sequences $(x_n)$ (with respect to the metric) converge to some $x \in \mathcal{X}$. 
\end{defn}
A metric gives a notion of distance -- and thus a notion of convergence and open sets (a topology): quite naturally, one considers the topology generated by open balls defined in terms of $d$. There are more general ways to study convergence and alternative topologies (\eg Fréchet topologies or weak topologies) can be both useful and necessary. 
\begin{example}
	\begin{enumerate}[(i)]
		\item Let $\mathcal{X}$ be a set and define 
		\begin{align*}
			d(x,y) := 
			\begin{cases}
				1 \qquad x \neq y \\
				0 \qquad x = y \\
			\end{cases}
			. 
		\end{align*}
		It is easy to see $d$ satisfies the axioms of a metric and $\mathcal{X}$ is complete with respect to $d$. This particular choice leads to the discrete topology. 
		\item Let $\mathcal{X} = \Cont([a,b],\C)$ be the space of continuous functions on an interval. Then one naturally considers the metric  
		\begin{align*}
			d_{\infty}(f,g) := \sup_{x \in [a,b]} \abs{f(x) - g(x)} = \max_{x \in [a,b]} \abs{f(x) - g(x)}
		\end{align*}
		with respect to which $\Cont([a,b],\C)$ is complete. 
	\end{enumerate}
\end{example}
A more peculiar way of measuring distances -- one which is adapted to the linear structure of vector spaces already -- are \emph{norms}. 
\begin{defn}[Normed space]
	Let $\mathcal{X}$ be a vector space. A mapping $\norm{\cdot} : \mathcal{X} \longrightarrow [0,+\infty)$ with properties 
	\begin{enumerate}[(i)]
		\item $\norm{x} = 0$ if and only if $x = 0$, 
		\item $\norm{\alpha x} = \abs{\alpha} \, \norm{x}$, and 
		\item $\norm{x + y} \leq \norm{x} + \norm{y}$, 
	\end{enumerate}
	for all $x,y \in \mathcal{X}$, $\alpha \in \C$, is called norm. The pair $(\mathcal{X},\norm{\cdot})$ is then referred to as normed space. 
\end{defn}
A norm on $\mathcal{X}$ quite naturally induces a metric by setting 
\begin{align*}
	d(x,y) := \norm{x-y}
\end{align*}
for all $x,y \in \mathcal{X}$. Unless specifically mentioned otherwise, one always works with the metric induced by the norm. 
\begin{defn}[Banach space]
	A complete normed space is a Banach space. 
\end{defn}
\begin{example}
	\begin{enumerate}[(i)]
		\item The space $\mathcal{X} = \Cont([a,b],\C)$ from the previous list of examples has a norm, the sup norm 
		\begin{align*}
			\norm{f}_{\infty} = \sup_{x \in [a,b]} \abs{f(x)} 
			. 
		\end{align*}
		Since $\Cont([a,b],\C)$ is complete, it is a Banach space. 
		\item Another important example are the $L^p$ spaces we will get to know in Chapter~\ref{hilbert_spaces:facts_on_Lp}. 
	\end{enumerate}
\end{example}
If we add even more structure, we arrive at the notion of 
\begin{defn}[pre-Hilbert space and Hilbert space]
	A pre-Hilbert space is a complex vector space $\Hil$ with scalar product 
	\begin{align*}
		\scpro{\cdot}{\cdot} : \Hil \times \Hil \longrightarrow \C 
		, 
	\end{align*}
	\ie a mapping with properties 
	\begin{enumerate}[(i)]
		\item $\scpro{\varphi}{\varphi} \geq 0$ and $\scpro{\varphi}{\varphi} = 0$ implies $\varphi = 0$ (positive definiteness), 
		\item $\scpro{\varphi}{\psi}^* = \scpro{\psi}{\varphi}$, and 
		\item $\scpro{\varphi}{\alpha \psi + \chi} = \alpha \scpro{\varphi}{\psi} + \scpro{\varphi}{\chi}$ 
	\end{enumerate}
	for all $\varphi , \psi , \chi \in \Hil$ and $\alpha \in \C$. This induces a natural norm $\norm{\varphi} := \sqrt{\sscpro{\varphi}{\varphi}}$ and metric $d(\varphi,\psi) := \norm{\varphi - \psi}$, $\varphi , \psi \in \Hil$. If $\Hil$ is complete with respect to the induced metric, it is a Hilbert space. 
\end{defn}
\begin{example}
	\begin{enumerate}[(i)]
		\item $\C^n$ with scalar product
		\begin{align*}
			\scpro{z}{w} := \sum_{l = 1}^n z_j^* \, w_j 
		\end{align*}
		is a Hilbert space. 
		\item $\Cont([a,b],\C)$ with scalar product 
		\begin{align*}
			\scpro{f}{g} := \int_a^b \dd x \, f(x)^* \, g(x) 
		\end{align*}
		is just a pre-Hilbert space, since it is not complete. 
	\end{enumerate}
\end{example}
%


\section{Orthonormal bases and orthogonal subspaces} 
\label{hilbert_spaces:onb}

Hilbert spaces have the important notion of orthonormal vectors and sequences which do not exist in Banach spaces. 
\begin{defn}[Orthonormal set]
	Let $\mathcal{I}$ be a countable index set. A family of vectors $\{ \varphi_k \}_{k \in \mathcal{I}}$ is called orthonormal set if for all $k,j \in \mathcal{I}$
	\begin{align*}
		\scpro{\varphi_k}{\varphi_j} = \delta_{kj} 
	\end{align*}
	holds. 
\end{defn}
As we will see, all vectors in a separable Hilbert spaces can be written in terms of a countable orthonormal basis. Especially when we want to approximate elements in a Hilbert space by elements in a proper closed subspace, the vector of best approximation can be written as a linear combination of basis vectors. 
\begin{defn}[Orthonormal basis]
	Let $\mathcal{I}$ be a countable index set. An orthonormal set of vectors $\{ \varphi_k \}_{k \in \mathcal{I}}$ is called orthonormal basis if and only if for all $\psi \in \Hil$, we have 
	\begin{align*}
		\psi = \sum_{k \in \mathcal{I}} \sscpro{\varphi_k}{\psi} \, \varphi_k 
		. 
	\end{align*}
	If $\mathcal{I}$ is countably infinite, $\mathcal{I} \cong \N$, then this means the sequence $\psi_n := \sum_{j = 1}^n \sscpro{\varphi_j}{\psi} \, \varphi_j$ of partial converges in norm to $\psi$, 
	\begin{align*}
		\lim_{n \rightarrow \infty} \Bnorm{\psi - \mbox{$\sum_{j = 1}^n$} \sscpro{\varphi_j}{\psi} \, \varphi_j} = 0
	\end{align*}
\end{defn}
With this general notion of orthogonality, we have a Pythagorean theorem: 
\begin{thm}[Pythagoras]
	Given a finite orthonormal family $\{ \varphi_1 , \ldots , \varphi_n \}$ in a pre-Hilbert space $\Hil$ and $\varphi \in \Hil$, we have 
	\begin{align*}
		\bnorm{\varphi}^2 = \mbox{$\sum_{k = 1}^n$} \babs{\sscpro{\varphi_k}{\varphi}}^2 + \bnorm{\varphi - \mbox{$\sum_{k = 1}^n$} \sscpro{\varphi_k}{\varphi} \, \varphi_k}^2 
		. 
	\end{align*}
\end{thm}
\begin{proof}
	It is easy to check that $\psi := \sum_{k = 1}^n \sscpro{\varphi_k}{\varphi} \, \varphi_k$ and $\psi^{\perp} := \varphi - \sum_{k = 1}^n \sscpro{\varphi_k}{\varphi} \, \varphi_k$ are orthogonal and $\varphi = \psi + \psi^{\perp}$. Hence, we obtain 
	\begin{align*}
		\norm{\varphi}^2 &= \sscpro{\varphi}{\varphi} 
		= \sscpro{\psi+\psi^{\perp}}{\psi+\psi^{\perp}} 
		= \sscpro{\psi}{\psi} + \sscpro{\psi^{\perp}}{\psi^{\perp}} 
		\\
		&= \bnorm{\mbox{$\sum_{k = 1}^n$} \sscpro{\varphi_k}{\varphi} \, \varphi_k}^2 + \bnorm{\varphi - \mbox{$\sum_{k = 1}^n$} \sscpro{\varphi_k}{\varphi} \, \varphi_k}^2 
		. 
	\end{align*}
	This concludes the proof. 
\end{proof}
A simple corollary are Bessel's inequality and the Cauchy-Schwarz inequality. 
\begin{thm}
	Let $\Hil$ be a pre-Hilbert space. 
	\begin{enumerate}[(i)]
		\item Bessel's inequality holds: let $\bigl \{ \varphi_1 , \ldots \varphi_n \bigr \}$ be a finite orthonormal sequence. Then 
		\begin{align*}
			\snorm{\psi}^2 \geq \sum_{j = 1}^n \sabs{\sscpro{\varphi_j}{\psi}}^2 
			. 
		\end{align*}
		holds for all $\psi \in \Hil$. 
		\item The Cauchy-Schwarz inequality holds, \ie 
		\begin{align*}
			\sabs{\sscpro{\varphi}{\psi}} \leq \snorm{\varphi} \snorm{\psi} 
		\end{align*}
		is valid for all $\varphi , \psi \in \Hil$
	\end{enumerate}
\end{thm}
\begin{proof}
	\begin{enumerate}[(i)]
		\item This follows trivially from the previous Theorem as $\snorm{\psi^{\perp}}^2 \geq 0$. 
		\item Pick $\varphi , \psi \in \Hil$. In case $\varphi = 0$, the inequality holds. So assume $\varphi \neq 0$ and define 
		\begin{align*}
			\varphi_1 := \frac{\varphi}{\norm{\varphi}} 
		\end{align*}
		which has norm $1$. We can apply (i) for $n = 1$ to conclude 
		\begin{align*}
			\norm{\psi}^2 \geq \abs{\sscpro{\varphi_1}{\psi}}^2 
			= \frac{1}{\norm{\varphi}^2} \abs{\sscpro{\varphi}{\psi}}^2 
			. 
		\end{align*}
		This is equivalent to the Cauchy-Schwarz inequality. 
	\end{enumerate}
\end{proof}
An important corollary says that the scalar product is continuous with respect to the norm topology. This is not at all surprising, after all the norm is induced by the scalar product! 
\begin{cor}\label{hilbert_spaces:onb:cor:continuity_scalar_product}
	Let $\Hil$ be a Hilbert space. Then the scalar product is continuous with respect to the norm topology, \ie for two sequences $(\varphi_n)_{n \in \N}$ and $(\psi_m)_{m \in \N}$ that converge to $\varphi$ and $\psi$, respectively, we have 
	\begin{align*}
		\lim_{n,m \rightarrow \infty} \sscpro{\varphi_n}{\psi_m} &= \sscpro{\varphi}{\psi} 
		. 
	\end{align*}
\end{cor}
\begin{proof}
	Let $(\varphi_n)_{n \in \N}$ and $(\psi_m)_{m \in \N}$ be two sequences in $\Hil$ that converge to $\varphi$ and $\psi$, respectively. Then by Cauchy-Schwarz, we have 
	\begin{align*}
		\lim_{n,m \rightarrow \infty} \babs{\sscpro{\varphi}{\psi} - \sscpro{\varphi_n}{\psi_m}} &= \lim_{n,m \rightarrow \infty} \babs{\sscpro{\varphi - \varphi_n}{\psi} - \sscpro{\varphi_n}{\psi_m - \psi}} 
		\\
		&\leq \lim_{n,m \rightarrow \infty} \babs{\sscpro{\varphi - \varphi_n}{\psi}} + \lim_{n,m \rightarrow 0} \babs{\sscpro{\varphi_n}{\psi_m - \psi}} 
		\\
		&\leq \lim_{n,m \rightarrow \infty} \snorm{\varphi - \varphi_n} \, \snorm{\psi} + \lim_{n,m \rightarrow \infty} \snorm{\varphi_n} \, \snorm{\psi_m - \psi} 
		= 0 
	\end{align*}
	since there exists some $C > 0$ such that $\norm{\varphi_n} \leq C$ for all $n \in \N$. 
\end{proof}
\begin{defn}[Separable Hilbert space]
	A Hilbert space $\Hil$ is called separable if there exists a countable dense subset. 
\end{defn}
Before we prove that a Hilbert space is separable exactly if it admits a \emph{countable} basis, we need to introduce the notion of orthogonal complement: if $A$ is a subset of a pre-Hilbert space $\Hil$, then we define 
\begin{align*}
	A^{\perp} := \bigl \{ \varphi \in \Hil \; \vert \; \sscpro{\varphi}{\psi} = 0 \; \forall \psi \in A \bigr \} 
	. 
\end{align*}
The following few properties of the orthogonal complement follow immediately from its definition: 
\begin{enumerate}[(i)]
	\item $\{ 0 \}^{\perp} = \Hil$ and $\Hil^{\perp} = \{ 0 \}$. 
	\item $A^{\perp}$ is a closed linear subspace of $\Hil$ for \emph{any} subset $A \subseteq \Hil$. 
	\item If $A \subseteq B$, then $B^{\perp} \subseteq A^{\perp}$. 
	\item If we denote the sub vector space spanned by the elements in $A$ by $\mathrm{span} \, A$, we have 
	\begin{align*}
		A^{\perp} = \bigl ( \mathrm{span} \, A \bigr )^{\perp} = \bigl ( \overline{\mathrm{span} \, A} \bigr )^{\perp}
	\end{align*}
	where $\overline{\mathrm{span} \, A}$ is the completion of $\mathrm{span} \, A$ with respect to the norm topology. 
\end{enumerate}
If $(\Hil,d)$ is a metric space, we can define the distance between a point $\varphi \in \Hil$ and a subset $A \subseteq \Hil$ as 
\begin{align*}
	d(\varphi,A) := \inf_{\psi \in A} d(\varphi,\psi) 
	. 
\end{align*}
If there exists $\varphi_0 \in A$ which minimizes the distance, \ie $d(\varphi,A) = d(\varphi,\varphi_0)$, then $\varphi_0$ is called \emph{element of best approximation} for $\varphi$ in $A$. This notion is helpful to understand why and how elements in an infinite-dimensional Hilbert space can be approximated by finite linear combinations -- something that is used in numerics all the time.  

If $A \subset \Hil$ is a convex subset of a Hilbert space $\Hil$, then one can show that there always exists an element of best approximation. In case $A$ is a linear subspace of $\Hil$, it is given by projecting an arbitrary $\psi \in \Hil$ down to the subspace $A$. 
\begin{thm}\label{hilbert_spaces:onb:thm:best_approx}
	Let $A$ be a closed convex subset of a Hilbert space $\Hil$. Then there exists for each $\varphi \in \Hil$ exactly one $\varphi_0 \in A$ such that 
	\begin{align*}
		d(\varphi,A) = d(\varphi,\varphi_0) 
		. 
	\end{align*}
\end{thm}
\begin{proof}
	We choose a sequence $(\psi_n)_{n \in \N}$ in $A$ with $d(\varphi,\psi_n) = \norm{\varphi - \psi_n} \rightarrow d(x,A)$. This sequence is also a Cauchy sequence: we add and subtract $\varphi$ to get 
	\begin{align*}
		\bnorm{\psi_n - \psi_m}^2 &= \bnorm{(\psi_n - \varphi) + (\varphi - \psi_m)}^2 
		. 
	\end{align*}
	If $\Hil$ were a normed space, we could have to use the triangle inequality to estimate the right-hand side from above. However, $\Hil$ is a Hilbert space and by using the parallelogram identity,\footnote{For all $\varphi , \psi \in \Hil$, the identity $2 \norm{\varphi}^2 + 2 \norm{\psi}^2 = \norm{\varphi + \psi}^2 + \norm{\varphi - \psi}^2$ holds. } we see that the right-hand side is actually \emph{equal} to 
	\begin{align*}
		\bnorm{\psi_n - \psi_m}^2 &= 2 \bnorm{\psi_n - \varphi}^2 + 2 \bnorm{\psi_m - \varphi}^2 - \bnorm{\psi_n + \psi_m - 2 \varphi}^2 
		\\
		&
		= 2 \bnorm{\psi_n - \varphi}^2 + 2 \bnorm{\psi_m - \varphi}^2 - 4 \bnorm{\tfrac{1}{2} (\psi_n + \psi_m) - \varphi}^2 
		\\
		&\leq 2 \bnorm{\psi_n - \varphi}^2 + 2 \bnorm{\psi_m - \varphi}^2 - 4 d(\varphi,A)
		\\
		&
		\xrightarrow{n \rightarrow \infty} 2 d(\varphi,A) + 2 d(\varphi,A) - 4 d(\varphi,A) 
		= 0
		. 
	\end{align*}
	By convexity, $\frac{1}{2} (\psi_n + \psi_m)$ is again an element of $A$. This is crucial once again for the uniqueness argument. Letting $n,m \rightarrow \infty$, we see that $(\psi_n)_{n \in \N}$ is a Cauchy sequence in $A$ which converges in $A$ as it is a closed subset of $\Hil$. Let us call the limit point $\varphi_0 := \lim_{n \rightarrow \infty} \psi_n$. Then $\varphi_0$ is \emph{an} element of best approximation, 
	\begin{align*}
		\bnorm{\varphi - \varphi_0} = \lim_{n \rightarrow \infty} \bnorm{\varphi - \psi_n} = d(\varphi,A)
		. 
	\end{align*}
	To show uniqueness, we assume that there exists another element of best approximation $\varphi_0' \in A$. Define the sequence $(\tilde{\psi}_n)_{n \in \N}$ by $\tilde{\psi}_{2n} := \varphi_0$ for even indices and $\tilde{\psi}_{2n + 1} := \varphi_0'$ for odd indices. By assumption, we have $\norm{\varphi - \varphi_0} = d(\varphi,A) = \norm{\varphi - \varphi_0'}$ and thus, by repeating the steps above, we conclude $(\tilde{\psi}_n)_{n \in \N}$ is a Cauchy sequence that converges to some element. However, since the sequence is alternating, the two elements $\varphi'_0 = \varphi_0$ are in fact identical. 
\end{proof}
As we have seen, the condition that the set is convex and closed is crucial in the proof. Otherwise the minimizer may not be unique or even contained in the set. 
\begin{cor}
	Let $E$ be a closed subvector space of the Hilbert space $\Hil$. Then for any $\varphi \in \Hil$, there exists $\varphi_0 \in E$ such that $d(\varphi,E) = d(\varphi,\varphi_0)$. 
	\marginpar{\small 2009.11.04}
\end{cor}
This is all very abstract. For the case of a closed subvector space $E \subseteq \Hil$, we can express the element of best approximation in terms of the basis: not surprisingly, it is given by the projection of $\varphi$ onto $E$. 
\begin{thm}\label{hilbert_spaces:onb:thm:best_approx_explicit}
	Let $E \subseteq \Hil$ be a closed subspace of a Hilbert space that is spanned by countably many orthonormal basis vectors $\{ \varphi_k \}_{k \in \mathcal{I}}$. Then for any $\varphi \in \Hil$, the element of best approximation $\varphi_0 \in E$ is given by 
	\begin{align*}
		\varphi_0 = \sum_{k \in \mathcal{I}} \sscpro{\varphi_k}{\varphi} \, \varphi_k 
		. 
	\end{align*}
\end{thm}
\begin{proof}
	It is easy to show that $\varphi - \varphi_0$ is orthogonal to any $\psi = \sum_{k \in \mathcal{I}} \lambda_k \, \varphi_k \in E$: we focus on the more difficult case when $E$ is not finite dimensional: then, we have to approximate $\varphi_0$ and $\psi$ by finite linear combinations and take limits. We call $\varphi_0^{(n)} := \sum_{k = 1}^n \sscpro{\varphi_k}{\varphi} \, \varphi_k$ and $\psi^{(m)} := \sum_{l = 1}^m \lambda_l \, \varphi_l$. With that, we have 
	\begin{align*}
		\bscpro{\varphi - \varphi_0^{(n)}}{\psi^{(m)}} &= \Bscpro{\varphi - \mbox{$\sum_{k = 1}^n$} \sscpro{\varphi_k}{\varphi} \, \varphi_k}{\mbox{$\sum_{l = 1}^m$} \lambda_l \, \varphi_l} \
		\\
		&
		= \sum_{l = 1}^m \lambda_l \, \sscpro{\varphi}{\varphi_l} - \sum_{k = 1}^n \sum_{l = 1}^m \lambda_l \, \sscpro{\varphi_k}{\varphi}^* \, \sscpro{\varphi_k}{\varphi_l} 
		\\
		&= \sum_{l = 1}^m \lambda_l \, \sscpro{\varphi}{\varphi_l} \, \Bigl ( 1 - \mbox{$\sum_{k = 1}^n$} \, \delta_{kl} \Bigr ) 
		. 
	\end{align*}
	By continuity of the scalar product, Corollary~\ref{hilbert_spaces:onb:cor:continuity_scalar_product}, we can take the limit $n,m \rightarrow \infty$. The term in parentheses containing the sum is $0$ exactly when $l \in \{ 1 , \ldots , m \}$ and $1$ otherwise. Specifically, if $n \geq m$, the right-hand side vanishes identically. Hence, we have 
	\begin{align*}
		\bscpro{\varphi - \varphi_0}{\psi} &= \lim_{n,m \rightarrow \infty} \bscpro{\varphi - \varphi_0^{(n)}}{\psi^{(m)}} 
		= 0 
		, 
	\end{align*}
	in other words $\varphi - \varphi_0 \in E^{\perp}$. This, in turn, implies by the Pythagorean theorem that 
		\begin{align*}
			\norm{\varphi - \psi}^2 &= \norm{\varphi - \varphi_0}^2 + \norm{\varphi_0 - \psi}^2 
			\geq \norm{\varphi - \varphi_0}^2 
		\end{align*}
		and hence $\norm{\varphi - \varphi_0} = d(\varphi,E)$. Put another way, $\varphi_0$ is \emph{an} element of best approximation. Let us now show uniqueness. Assume, there exists another element of best approximation $\varphi'_0 = \sum_{k \in \mathcal{I}} \lambda'_k \, \varphi_k$. Then we know by repeating the previous calculation backwards that $\varphi - \varphi'_0 \in E^{\perp}$ and the scalar product with respect to any of the basis vectors $\varphi_k$ which span $E$ has to vanish, 
		\begin{align*}
			0 = \bscpro{\varphi_k}{\varphi - \varphi'_0} &= \sscpro{\varphi_k}{\varphi} - \sum_{l \in \mathcal{I}} {\lambda'_l} \, \sscpro{\varphi_k}{\varphi_l} 
			= \sscpro{\varphi_k}{\varphi} - \sum_{l \in \mathcal{I}} {\lambda'_l} \, \delta_{kl} 
			\\
			&= \sscpro{\varphi_k}{\varphi} - {\lambda'_k} 
			. 
		\end{align*}
		This means the coefficients with respect to the basis $\{ \varphi_k \}_{k \in \mathcal{I}}$ all agree with those of $\varphi_0$. Hence, the element of approximation is unique, $\varphi_0 = \varphi_0'$, and given by the projection of $\varphi$ onto $E$. 
\end{proof}
\begin{thm}\label{hilbert_spaces:onb:thm:direct_sum_decomp}
	Let $E$ be a closed linear subspace of a Hilbert space $\Hil$. Then 
	\begin{enumerate}[(i)]
		\item $\Hil = E \oplus E^{\perp}$, \ie every vector $\varphi \in \Hil$ can be uniquely decomposed as $\varphi = \psi + \psi^{\perp}$ with $\psi \in E$, $\psi^{\perp} \in E^{\perp}$. 
		\item $E^{\perp \perp} = E$. 
	\end{enumerate}
\end{thm}
\begin{proof}
	\begin{enumerate}[(i)]
		\item By Theorem~\ref{hilbert_spaces:onb:thm:best_approx}, for each $\varphi \in \Hil$, there exists $\varphi_0 \in E$ such that $d(\varphi,E) = d(\varphi,\varphi_0)$. From the proof of the previous theorem, we see that $\varphi_0^{\perp} := \varphi - \varphi_0 \in E^{\perp}$. Hence, $\varphi = \varphi_0 + \varphi_0^{\perp}$ is \emph{a} decomposition of $\varphi$. To show that it is unique, assume $\varphi'_0 + {\varphi'_0}^{\perp} = \varphi = \varphi_0 + \varphi_0^{\perp}$ is another decomposition. Then by subtracting, we are led to conclude that 
		\begin{align*}
			E \ni \varphi'_0 - \varphi_0 = \varphi_0^{\perp} - {\varphi'_0}^{\perp} \in E^{\perp} 
		\end{align*}
		holds. On the other hand, $E \cap E^{\perp} = \{ 0 \}$ and thus $\varphi_0 = {\varphi'_0}$ and $\varphi_0^{\perp} = {\varphi'_0}^{\perp}$, the decomposition is unique. 
		\item It is easy to see that $E \subseteq E^{\perp \perp}$. Let $\tilde{\varphi} \in E^{\perp \perp}$. By the same arguments as above, we can decompose $\tilde{\varphi} \in E^{\perp \perp} \subseteq \Hil$ into 
		\begin{align*}
			\tilde{\varphi} = \tilde{\varphi}_0 + {\tilde{\varphi}}_0^{\perp} 
		\end{align*}
		with $\tilde{\varphi}_0 \in E \subseteq E^{\perp \perp}$ and ${\tilde{\varphi}}_0^{\perp} \in E^{\perp}$. Hence, ${\tilde{\varphi}} - \tilde{\varphi}_0 \in E^{\perp \perp} \cap E^{\perp} = (E^{\perp})^{\perp} \cap E^{\perp} = \{ 0 \}$ and thus $\tilde{\varphi} = \tilde{\varphi}_0 \in E$. 
	\end{enumerate}
\end{proof}
Now we are in a position to prove the following important Proposition: 
\begin{prop}
	A Hilbert space $\Hil$ is separable if and only if there exists a countable orthonormal basis. 
\end{prop}
\begin{proof}
	$\Leftarrow$: 
	The set generated by the orthonormal basis $\{ \varphi_j \}_{j \in \mathcal{I}}$, $\mathcal{I}$ countable, and coefficients $z = q + i p$, $q , p \in \Q$, is dense in $\Hil$, 
	\begin{align*}
		\Bigl \{ \mbox{$\sum_{j = 1}^n$} z_j \varphi_j \in \Hil \; \big \vert \; \N \ni n \leq \abs{\mathcal{I}}, \; \varphi_j \in \{ \varphi_k \}_{k \in \N} , \; z_j = q_j + i p_j , \; q_j , p_j \in \Q \Bigr \} 
		. 
	\end{align*}
	\medskip
	
	\noindent
	$\Rightarrow$: Assume there exists a countable dense subset $\mathcal{D}$, \ie $\overline{\mathcal{D}} = \Hil$. If $\Hil$ is finite dimensional, the induction terminates after finitely many steps and the proof is simpler. Hence, we will assume $\Hil$ to be infinite dimensional. Pick a vector $\tilde{\varphi}_1 \in \mathcal{D} \setminus \{ 0 \}$ and normalize it. The normalized vector is then called $\varphi_1$. Note that $\varphi_1$ need not be in $\mathcal{D}$. By Theorem~\ref{hilbert_spaces:onb:thm:direct_sum_decomp}, we can split any $\psi \in \mathcal{D}$ into $\psi_1$ and $\psi_1^{\perp}$ such that $\psi_1 \in \mathrm{span} \, \{ \varphi_1 \} := E_1$, $\psi_1^{\perp} \in \mathrm{span} \, \{ \varphi_1 \}^{\perp} := E_1^{\perp}$ and 
	\begin{align*}
		\psi = \psi_1 + \psi_1^{\perp} 
		. 
	\end{align*}
	pick a second $\tilde{\varphi}_2 \in \mathcal{D} \setminus E_1$ (which is non-empty). Now we apply Theorem~\ref{hilbert_spaces:onb:thm:best_approx_explicit} (which is in essence Gram-Schmidt orthonormalization) to $\tilde{\varphi}_2$, \ie we pick the part which is orthogonal to $\varphi_1$, 
	\begin{align*}
		\varphi'_2 := \tilde{\varphi}_2 - \sscpro{\varphi_1}{\tilde{\varphi}_2} \, \varphi_2
	\end{align*}
	and normalize to $\varphi_2$, 
	\begin{align*}
		\varphi_2 := \frac{\varphi'_2}{\snorm{\varphi'_2}} 
		. 
	\end{align*}
	This defines $E_2 := \mathrm{span} \, \{ \varphi_1 , \varphi_2 \}$ and $\Hil = E_2 \oplus E_2^{\perp}$. 
	
	Now we proceed by induction: assume we are given $E_n = \mathrm{span} \, \{ \varphi_1 , \ldots , \varphi_n \}$. Take $\tilde{\varphi}_{n+1} \in  \mathcal{D} \setminus E_n$ and apply Gram-Schmidt once again to yield $\varphi_{n+1}$ which is the obtained from normalizing the vector 
	\begin{align*}
		\varphi'_{n+1} := \tilde{\varphi}_{n+1} - \sum_{k = 1}^n \sscpro{\varphi_k}{\tilde{\varphi}_{n+1}} \, \varphi_k
		. 
	\end{align*}
	This induction yields an orthonormal sequence $\{ \varphi_n \}_{n \in \N}$ which is by definition an orthonormal basis of $E_{\infty} := \overline{\mathrm{span} \, \{ \varphi_n \}_{n \in \N}}$ a closed subspace of $\Hil$. If $E_{\infty} \subsetneq \Hil$, we can split the Hilbert space into $\Hil = E_{\infty} \oplus E_{\infty}^{\perp}$. Then either $\mathcal{D} \cap (\Hil \setminus E_{\infty}) = \emptyset$ -- in which case $\mathcal{D}$ cannot be dense in $\Hil$ -- or $\mathcal{D} \cap (\Hil \setminus E_{\infty}) \neq \emptyset$. But then we have terminated the induction prematurely. 
\end{proof}
%

\section{Subspaces ($\oplus$) and product spaces ($\otimes$)} 
\label{hilbert_spaces:sub_product_spaces}

There are several ways to split Hilbert spaces: in direct sums and direct products. With the same techniques, we can construct new ones from existing Hilbert spaces. In Theorem~\ref{hilbert_spaces:onb:thm:direct_sum_decomp}, we have shown that if $E$ is a closed subspace, then $\Hil$ decomposes into a direct sum 
\begin{align*}
	\Hil = E \oplus E^{\perp} 
	. 
\end{align*}
That means any vector $\varphi = \psi + \psi_{\perp} \in \Hil$ can be \emph{uniquely} decomposed into $\psi \in E$ and $\psi^{\perp} \in E^{\perp}$. We now define the direct sum of two Hilbert spaces: 
\begin{defn}[Direct sum $\oplus$]
	Let $\Hil_1$ and $\Hil_2$ be Hilbert spaces with scalar products $\sscpro{\cdot}{\cdot}_1$ and $\sscpro{\cdot}{\cdot}_2$. Then we define $\Hil_1 \oplus \Hil_2$ as the carteisan product $\Hil_1 \times \Hil_2$ of vector spaces endowed with the structure of a vector space in the following way: for any $\varphi = (\varphi_1 , \varphi_2) , \psi = (\psi_1 , \psi_2) \in \Hil_1 \times \Hil_2$ and $\alpha \in \C$, we define 
	\begin{enumerate}[(i)]
		\item addition component-wise, $\varphi + \psi = (\varphi_1,\varphi_2) + (\psi_1,\psi_2) := (\varphi_1 + \psi_1 , \varphi_2 + \psi_2)$, 
		\item scalar multiplication component-wise, $\alpha \varphi = \alpha (\varphi_1 , \varphi_2) := (\alpha \varphi_1 , \alpha \varphi_2)$, 
		\item and the scalar product on $\Hil_1 \oplus \Hil_2$ is the sum of the two scalar products, 
		\begin{align*}
			\sscpro{\varphi}{\psi} := \sscpro{\varphi_1}{\psi_1}_1 + \sscpro{\varphi_2}{\psi_2}_2 
			. 
		\end{align*}
	\end{enumerate}
\end{defn}
\begin{prop}
	The direct sum $\Hil_1 \oplus \Hil_2$ of two Hilbert spaces is a Hilbert space. 
\end{prop}
\begin{proof}
	One immediately checks that $\Hil_1 \oplus \Hil_2$ is a vector space. Completeness also follows from the completeness of the components: let $\{ \varphi^{(n)} \}_{n \in \N}$ be a Cauchy sequence with respect to the norm induced by $\sscpro{\cdot}{\cdot}$. Writing out the definition, it is clear that this also means each component $\{ \varphi_j^{(n)} \}_{n \in \N}$, $j = 1 , 2$, is a Cauchy sequence in $\Hil_j$ which converges to some $\varphi_j \in \Hil_j$. Hence, $\varphi^{(n)} \longrightarrow (\varphi_1,\varphi_2)$ as $n \rightarrow \infty$. 
\end{proof}
The\marginpar{\small 2009.11.10} other way to construct new Hibert spaces is taking tensor products: if $\varphi_1 \in \Hil_1$ and $\varphi_2 \in \Hil_2$ are two vectors from two vector spaces, we can characterize $\varphi_1 \otimes \varphi_2$, the tensor product of $\varphi_1$ and $\varphi_2$, by the following defining properties: for any $\varphi_1 , \psi_1 \in \Hil_1$ and $\varphi_2, \psi_2 \in \Hil_2$ 
\begin{align*}
	\varphi_1 \otimes ( \varphi_2 + \psi_2 ) &= \varphi_1 \otimes \varphi_2 + \varphi_1 \otimes \psi_2 
	\\
	( \varphi_1 + \psi_1 ) \otimes \varphi_2 &= \varphi_1 \otimes \varphi_2 + \psi_1 \otimes \varphi_2 	
\end{align*}
holds. Scalars can be pushed back and forth between factors, 
\begin{align*}
	\alpha (\varphi_1 \otimes \varphi_2) = (\alpha \varphi_1) \otimes \varphi_2 = \varphi_1 \otimes (\alpha \varphi_2) 
	. 
\end{align*}
The formal definition is a lot more complicated: one has to construct a big space where vectors such as $(\alpha \varphi_1) \otimes \varphi_2$ and $\varphi_1 \otimes (\alpha \varphi_2)$ are distinct and then use equivalence relations to implement the above characteristics. The constructed space is defined only up to isomorphism. 
\begin{defn}[Tensor product $\otimes$]
	Let $\Hil_1$ and $\Hil_2$ be two Hilbert spaces with scalar products $\sscpro{\cdot}{\cdot}_1$ and $\sscpro{\cdot}{\cdot}_2$. Then the tensor product $\Hil_1 \otimes \Hil_2$ is defined as the completion of the algebraic tensor product 
	\begin{align*}
		\Hil_1 \odot \Hil_2 := \Bigl \{ \mbox{$\sum_{k = 1}^n$} \lambda_k \varphi_{1 \, k} \otimes \varphi_{2 \, k} \; \big \vert \; n \in \N , \, \varphi_{1 \, k} \in \Hil_1 , \, \varphi_{2 \, k} \in \Hil_2, \, \lambda_k \in \C \, \forall 1 \leq k \leq n \Bigr \} 
	\end{align*}
	with respect to the norm induced by the scalar product 
	\begin{align*}
		\sscpro{\varphi_1 \otimes \varphi_2}{\psi_1 \otimes \psi_2} := \sscpro{\varphi_1}{\psi_1}_1 \, \sscpro{\varphi_2}{\psi_2}_2 
		, 
		&&
		\forall \varphi_1 , \psi_1 \in \Hil_1 , \; \varphi_2 , \psi_2 \in \Hil_2  
		. 
	\end{align*}
\end{defn}
If $\{ \varphi_{1 \, k} \}_{k \in \mathcal{I}_1}$ and $\{ \varphi_{2 \, j} \}_{j \in \mathcal{I}_2}$ are orthonormal bases of two separable Hilbert spaces $\Hil_1$ and $\Hil_2$, respectively, then 
\begin{align*}
	\bigl \{ \varphi_{1 \, k} \otimes \varphi_{2 \, j} \bigr \}_{k \in \mathcal{I}_1, j \in \mathcal{I}_2}
\end{align*}
is an orthonormal basis in $\Hil_1 \otimes \Hil_2$, \ie every vector $\Psi \in \Hil_1 \otimes \Hil_2$ can be written as 
\begin{align*}
	\sum_{\substack{k \in \mathcal{I}_1 \\ j \in \mathcal{I}_2}} \bscpro{\varphi_{1 \, k} \otimes \varphi_{2 \, j}}{\Psi} \, \varphi_{1 \, k} \otimes \varphi_{2 \, j} 
	. 
\end{align*}
\begin{example}
	\begin{enumerate}[(i)]
		\item A non-relativistic spin-$\nicefrac{1}{2}$ particle lives in the Hilbert space $L^2(\R^d,\C^2)$. An easy, but very helpful exercise is to show the following equivalence (which correspond to different physical points of views): 
		\begin{align*}
			L^2(\R^d,\C^2) \cong L^2(\R^d) \oplus L^2(\R^d) \cong L^2(\R^d) \otimes \C^2 
		\end{align*}
		Depending on the physical situation, these identification may be very helpful in solving a problem. 
		\item Consider $L^2(\R^d) \otimes L^2(\R^d)$. This is the Hilbert space of two particles. If they are identical, we have to restrict ourselves to the symmetric and antisymmetric subspace, depending on whether the particle in question is a boson or a fermion. Keep in mind that in general, elements $\Psi \in L^2(\R^d) \otimes L^2(\R^d)$ cannot be written as the product of two wave functions $\varphi_1 , \varphi_2 \in L^2(\R^d)$, 
		\begin{align*}
			\Psi \neq \varphi_1 \otimes \varphi_2 
			. 
		\end{align*}
		We will show in an exercise that $L^2(\R^d) \otimes L^2(\R^d) \cong L^2(\R^d \times \R^d)$. 
	\end{enumerate}
\end{example}
%


\section{Linear functionals, dual space and weak convergence} 
\label{hilbert_spaces:dual_space}

A very important notion is that of a functional. We have already gotten to know the free energy functional 
\begin{align*}
	E_{\mathrm{free}} : &\mathcal{D}(E_{\mathrm{free}}) \subset L^2(\R^d) \longrightarrow [0,+\infty) \subset \C 
	, 
	\\
	&\varphi \mapsto E_{\mathrm{free}}(\varphi) = \frac{1}{2m} \sum_{l = 1}^d \bscpro{(- i \hbar \partial_{x_l} \varphi)}{(- i \hbar \partial_{x_l} \varphi)} 
	. 
\end{align*}
This functional, however, is not linear, and it is not defined for all $\varphi \in L^2(\R^d)$. Let us restrict ourselves to a smaller class of functionals: 
\begin{defn}[Bounded linear functional]
	Let $\mathcal{X}$ be a normed space. Then a map 
	\begin{align*}
		 L : \mathcal{X} \longrightarrow \C 
	\end{align*}
	is a bounded linear functional if and only if 
	\begin{enumerate}[(i)]
		\item there exists $C>0$ such that $\abs{L(x)} \leq C \norm{x}$ and 
		\item $L(x + \mu y) = L(x) + \mu L(y)$ 
	\end{enumerate}
	hold for all $x,y \in \mathcal{X}$ and $\mu \in \C$. 
\end{defn}
A very basic fact is that boundedness of a linear functional is equivalent to its continuity. 
\begin{thm}\label{hilbert_spaces:dual_space:thm:bounded_functionals_continuous}
	Let $L : \mathcal{X} \longrightarrow \C$ be a linear functional on the normed space $\mathcal{X}$. Then the following statements are equivalent: 
	\begin{enumerate}[(i)]
		\item $L$ is continuous at $x_0 \in \mathcal{X}$. 
		\item $L$ is continuous. 
		\item $L$ is bounded. 
	\end{enumerate}
\end{thm}
\begin{proof}
	(i) $\Leftrightarrow$ (ii): This follows immediately from the linearity. 
	\medskip
	
	\noindent
	(ii) $\Rightarrow$ (iii): Assume $L$ to be continuous. Then it is continuous at $0$ and for $\eps = 1$, we can pick $\delta > 0$ such that %
	\begin{align*}
		\abs{L(x)} \leq \eps = 1 
	\end{align*}
	for all $x \in \mathcal{X}$ with $\norm{x} \leq \delta$. By linearity, this implies for any $y \in \mathcal{X} \setminus \{ 0 \}$ that 
	\begin{align*}
		\babs{L \bigl ( \tfrac{\delta}{\norm{y}} y \bigr )} = \tfrac{\delta}{\norm{y}} \, \babs{L(y)} \leq 1 
		. 
	\end{align*}
	Hence, $L$ is bounded with bound $\nicefrac{1}{\delta}$, 
	\begin{align*}
		\babs{L(y)} \leq \tfrac{1}{\delta} \norm{y} 
		. 
	\end{align*}
	(iii) $\Rightarrow$ (ii): Conversely, if $L$ is bounded by $C > 0$, 
	\begin{align*}
		\babs{L(x) - L(y)} \leq C \norm{x-y} 
		, 
	\end{align*}
	holds for all $x,y \in \mathcal{X}$. This means, $L$ is continuous: for $\eps > 0$ pick $\delta = \nicefrac{\eps}{C}$ so that 
	\begin{align*}
		\babs{L(x) - L(y)} \leq C \norm{x-y} \leq C \tfrac{\eps}{C} = \eps 
	\end{align*}
	holds for all $x,y \in \mathcal{X}$ such that $\norm{x - y} \leq \nicefrac{\eps}{C}$. 
\end{proof}
\begin{defn}[Dual space]
	Let $\mathcal{X}$ be a normed space. The dual space $\mathcal{X}^*$ is the vector space of bounded linear functionals endowed with the norm 
	\begin{align*}
		\norm{L}_* := \sup_{x \in \mathcal{X} \setminus \{ 0 \}} \frac{\abs{L(x)}}{\norm{x}} = \sup_{\substack{x \in \mathcal{X} \\ \norm{x} = 1}} \abs{L(x)} 
		. 
	\end{align*}
\end{defn}
Independently of whether $\mathcal{X}$ is complete, $\mathcal{X}^*$ is a Banach space. 
\begin{prop}\label{hilbert_spaces:dual_space:prop:completeness_Xstar}
	The dual space to a normed linear space $\mathcal{X}$ is a Banach space. 
\end{prop}
\begin{proof}
	Let $( L_n )_{n \in \N}$ be a Cauchy sequence in $\mathcal{X}^*$, \ie a sequence for which 
	\begin{align*}
		\norm{L_k - L_j}_* \xrightarrow{k,j \rightarrow \infty} 0
		. 
	\end{align*}
	We have to show that $( L_n )_{n \in \N}$ converges to some $L \in \mathcal{X}^*$. For any $\eps > 0$, there exists $N(\eps) \in \N$ such that 
	\begin{align*}
		\norm{L_k - L_j}_* < \eps 
	\end{align*}
	for all $k,j \geq N(\eps)$. This also implies that for any $x \in \mathcal{X}$, $\bigl ( L_n(x) \bigr )_{n \in \N}$ converges as well, 
	\begin{align*}
		\babs{L_k(x) - L_j(x)} \leq \bnorm{L_k - L_j}_* \, \norm{x} < \eps \norm{x} 
		. 
	\end{align*}
	The field of complex numbers is complete and $\bigl ( L_n(x) \bigr )_{n \in \N}$ converges to some $L(x) \in \C$. We now \emph{define} 
	\begin{align*}
		L(x) := \lim_{n \rightarrow \infty} L_n(x) 
	\end{align*}
	for any $x \in \mathcal{X}$. Clearly, $L$ inherits the linearity of the $(L_n)_{n \in \N}$. The map $L$ is also bounded: for any $\eps > 0$, there exists $N(\eps) \in \N$ such that $\norm{L_j - L_n}_* < \eps$ for all $j , n \geq N(\eps)$. Then 
	\begin{align*}
		\babs{(L - L_n)(x)} &= \lim_{j \rightarrow \infty} \babs{(L_j - L_n)(x)} \leq \lim_{j \to \infty} \bnorm{L_j - L_n}_* \, \bnorm{x} 
		\\
		&< \eps \norm{x}
	\end{align*}
	holds for all $n \geq N(\eps)$. Since we can write $L$ as $L = L_n + (L - L_n)$, we can estimate the norm of the linear map $L$ by $\snorm{L}_* \leq \snorm{L_n}_* + \eps < \infty$. This means $L$ is a bounded linear functional on $\mathcal{X}$. 
\end{proof}
In case of Hilbert spaces, the dual $\Hil^*$ can be canonically identified with $\Hil$ itself: 
\begin{thm}[Riesz' Lemma]\label{hilbert_spaces:dual_space:thm:Riesz_Lemma}
	Let $\Hil$ be a Hilbert space. Then for all $L \in \Hil^*$ there exist $\psi_L \in \Hil$ such that 
	\begin{align*}
		L(\varphi) = \sscpro{\psi_L}{\varphi} 
		. 
	\end{align*}
	In particular, we have $\norm{L}_* = \snorm{\psi_L}$. 
\end{thm}
\begin{proof}
	Let $\ker L := \bigl \{ \varphi \in \Hil \; \vert \; L(\varphi) = 0 \bigr \}$ be the kernel of the functional $L$ and as such is a closed linear subspace of $\Hil$. If $\ker L = \Hil$, then $0 \in \Hil$ is the associated vector, 
	\begin{align*}
		L(\varphi) = 0 = \sscpro{0}{\varphi} 
		. 
	\end{align*}
	So assume $\ker L \subsetneq \Hil$ is a proper subspace. Then we can split $\Hil = \ker L \oplus (\ker L)^{\perp}$. Pick $\varphi_0 \in (\ker L)^{\perp}$, \ie $L(\varphi_0) \neq 0$. Then define 
	\begin{align*}
		\psi_L := \frac{L(\varphi_0)^*}{\snorm{\varphi_0}^2} \, \varphi_0 
		. 
	\end{align*}
	We will show that $L(\varphi) = \sscpro{\psi_L}{\varphi}$. If $\varphi \in \ker L$, then $L(\varphi) = 0 = \sscpro{\psi_L}{\varphi}$. One easily shows that for $\varphi = \alpha \, \varphi_0$, $\alpha \in \C$, 
	\begin{align*}
		L(\varphi) &= L(\alpha \, \varphi_0) = \alpha \, L(\varphi_0) 
		\\
		&= \sscpro{\psi_L}{\varphi} = \Bscpro{\tfrac{L(\varphi_0)^*}{\snorm{\varphi_0}^2} \varphi_0}{\alpha \, \varphi_0} 
		\\
		&= \alpha \, L(\varphi_0) \, \frac{\sscpro{\varphi_0}{\varphi_0}}{\snorm{\varphi_0}^2} 
		= \alpha \, L(\varphi_0) 
		. 
	\end{align*}
	Every $\varphi \in \Hil$ can be written as 
	\begin{align*}
		\varphi = \biggl ( \varphi - \frac{L(\varphi)}{L(\varphi_0)} \, \varphi_0 \biggr ) + \frac{L(\varphi)}{L(\varphi_0)} \, \varphi_0 
		. 
	\end{align*}
	Then the first term is in the kernel of $L$ while the second one is in the orthogonal complement of $\ker L$. Hence, $L(\varphi) = \sscpro{\psi_L}{\varphi}$ for all $\varphi \in \Hil$. If there exists a second $\psi_L' \in \Hil$, then for any $\varphi \in \Hil$
	\begin{align*}
		0 = L(\varphi) - L(\varphi) = \sscpro{\psi_L}{\varphi} - \sscpro{\psi_L'}{\varphi} 
		= \sscpro{\psi_L - \psi_L'}{\varphi} 
		. 
	\end{align*}
	This implies $\psi_L' = \psi_L$ and thus the element $\psi_L$ is unique. 
	
	To show $\snorm{L}_* = \snorm{\psi_L}$, assume $L \neq 0$. Then, we have 
	\begin{align*}
		\norm{L}_* &= \sup_{\norm{\varphi} = 1} \babs{L(\varphi)} \geq \babs{L \bigl ( \tfrac{\psi_L}{\snorm{\psi_L}} \bigr )} 
		\\
		&= \bscpro{\psi_L}{\tfrac{\psi_L}{\snorm{\psi_L}}} = \snorm{\psi_L} 
		. 
	\end{align*}
	On the other hand, the Cauchy-Schwarz inequality yields 
	\begin{align*}
		\norm{L}_* &= \sup_{\norm{\varphi} = 1} \babs{L(\varphi)} 
		= \sup_{\norm{\varphi} = 1} \babs{\sscpro{\psi_L}{\varphi}} 
		\\
		&\leq \sup_{\norm{\varphi} = 1} \snorm{\psi_L} \snorm{\varphi} 
		= \snorm{\psi_L} 
		. 
	\end{align*}
	Putting these two together, we conclude $\snorm{L}_* = \snorm{\psi_L}$. 
\end{proof}
\marginpar{\small 2009.11.11} 
\begin{remark}
	The bidual of a Hilbert space $\Hil^{**}$ can be canonically identified with $\Hil$ itself, \ie Hilbert spaces are reflexive. 
\end{remark}
\begin{defn}[Weak convergence]
	Let $\mathcal{X}$ be a Banach space. Then a sequence $(x_n)_{n \in \N}$ in $\mathcal{X}$ is said to converge weakly to $x \in \mathcal{X}$ if for all $L \in \mathcal{X}^*$ 
	\begin{align*}
		L(x_n) \xrightarrow{n \rightarrow \infty} L(x) 
	\end{align*}
	holds. In this case, one also writes $x_n \rightharpoonup x$. 
\end{defn}
Weak convergence, as the name suggests, is really weaker than convergence in norm. The reason why ``more'' sequences converge is that, a sense, uniformity is lost. If $\mathcal{X}$ is a Hilbert space, then applying a functional is the same as computing the inner product with respect to some vector $\psi_L$. If the ``non-convergent part'' lies in the orthogonal complement to $\{ \psi_L \}$, then this particular functional does not notice that the sequence has not converged yet. 
\begin{example}
	Let $\Hil$ be a separable infinite-dimensional Hilbert space and $\{ \varphi_n \}_{n \in \N}$ an ortho\-normal basis. Then the sequence $(\varphi_n)_{n \in \N}$ does not converge in norm, for as long as $n \neq k$ 
	\begin{align*}
		\norm{\varphi_n - \varphi_k} = \sqrt{2} 
		, 
	\end{align*}
	but it does converge weakly to $0$: for any functional $L = \sscpro{\psi_L}{\cdot}$, we see that $\bigl ( \abs{L(\varphi_n)} \bigr )_{n \in \N}$ is a sequence in $\R$ that converges to $0$. Since $\{ \varphi_n \}_{n \in \N}$ is a basis, we can write 
	\begin{align*}
		\psi_L = \sum_{n = 1}^{\infty} \sscpro{\varphi_n}{\psi_L} \, \varphi_n 
	\end{align*}
	and for the sequence of partial sums to converge to $\psi_L$, the sequence of coefficients 
	\begin{align*}
		\bigl ( \sscpro{\varphi_n}{\psi_L} \bigr )_{n \in \N} = \bigl ( L(\varphi_n)^* \bigr )_{n \in \N} 
	\end{align*}
	must converge to $0$. Since this is true for any $L \in \Hil^*$, we have proven that $\varphi_n \rightharpoonup 0$ (\ie $\varphi_n \rightarrow 0$ weakly). 
\end{example}

In case of $\mathcal{X} = L^p(\R^d)$, there are three basic mechanisms for when a sequence of functions $(f_k)$ does not converge in norm, but only weakly: 
\begin{enumerate}[(i)]
	\item \emph{$f_k$ oscillates to death:} take $f_k(x) = \sin (kx)$ for $0 \leq x \leq 1$ and zero otherwise. 
	\item \emph{$f_k$ goes up the spout:} pick $g \in L^p(\R)$ and define $f_k(x) := k^{\nicefrac{1}{p}} \, g(kx)$. This sequence explodes near $x = 0$ for large $k$. 
	\item \emph{$f_k$ wanders off to infinity:} this is the case when for some $g \in L^p(\R)$, we define $f_k(x) := g(x+k)$. 
\end{enumerate}
All of these sequences converge weakly to $0$, but do not converge in norm. 


\section{Important facts on $L^p(\R^d)$} 
\label{hilbert_spaces:facts_on_Lp}

For future reference, we collect a few facts on $L^p(\R^d)$ spaces. In particular, we will make use of dominated convergence frequently. We will give them without proof, they can be found in standard text books on analysis, see \eg \cite{LiebLoss:Analysis:2001}. 
\begin{defn}[$L^p(\R^d)$]
	Let $1 \leq p < \infty$. Then we define 
	\begin{align*}
		\mathcal{L}^p(\R^d) := \Bigl \{ f : \R^d \longrightarrow \C \; \big \vert \; \mbox{$f$ measurable, } \int_{\R^d} \dd x \, \abs{f(x)}^p < \infty \Bigr \} 
	\end{align*}
	as the vector space of functions whose $p$th power is integrable. Then $L^p(\R^d)$ is the vector space 
	\begin{align*}
		L^p(\R^d) := \mathcal{L}^p(\R^d) / \sim 
	\end{align*}
	consisting of equivalence classes of functions that agree almost everywhere. With the $p$ norm 
	\begin{align*}
		\norm{f}_p := \biggl ( \int_{\R^d} \dd x \, \abs{f(x)}^p \biggr )^{\nicefrac{1}{p}} 
	\end{align*}
	it forms a normed space. 
\end{defn}
In case $p = \infty$, we have to modify the definition a little bit. 
\begin{defn}[$L^{\infty}(\R^d)$]
	We define 
	\begin{align*}
		\mathcal{L}^{\infty}(\R^d) := \Bigl \{ f : \R^d \longrightarrow \C \; \big \vert \; \mbox{$f$ measurable, } \exists 0 < K < \infty : \abs{f(x)} \leq K \mbox{ almost everywhere} \Bigr \} 
	\end{align*}
	to be the space of functions that are bounded almost everywhere and 
	\begin{align*}
		\norm{f}_{\infty} := \mathrm{ess} \sup_{x \in \R^d} \babs{f(x)} 
		:= \inf \bigl \{ K \geq 0 \; \big \vert \abs{f(x)} \leq K \mbox{ for almost all $x \in \R^d$} \bigr \} 
		. 
	\end{align*}
	Then the space $L^{\infty}(\R^d) := \mathcal{L}^{\infty}(\R^d) / \sim$ is defined as the vector space of equivalence classes where two functions are identified if they agree almost everywhere. 
\end{defn}
\begin{thm}[Riesz-Fischer]
	For any $1 \leq p \leq \infty$, $L^p(\R^d)$ is complete with respect to the $\norm{\cdot}_p$ norm and thus a Banach space. If $p = 2$, $L^2(\R^d)$ is also a Hilbert space with scalar product 
	\begin{align*}
		\scpro{f}{g} = \int_{\R^d} \dd x \, f(x)^* \, g(x) 
		. 
	\end{align*}
\end{thm}
\begin{thm}
	For any $1 \leq p \lneq \infty$, the Banach space $L^p(\R^d)$ is separable. 
\end{thm}
\begin{proof}
	We refer to \cite[Lemma~2.17]{LiebLoss:Analysis:2001} for an explicit construction. The idea is to approximate arbitrary functions by functions which are constant on cubes and take only values in the rational complex numbers. 
\end{proof}
\begin{thm}[Monotone Convergence]
	Let $(f_k)_{k \in \N}$ be a sequence of non-decreasing functions in $L^1(\R^d)$ with pointwise limit $f$ defined almost everywhere. Define $I_k := \int_{\R^d} \dd x \, f_k(x)$; then the sequence $(I_k)$ is non-decreasing as well. If $I := \lim_{k \to \infty} I_k < \infty$, then $I = \int_{\R^d} \dd x \, f(x)$, \ie 
	\begin{align*}
		\lim_{k \to \infty} \int_{\R^d} \dd x \, f_k(x) = \int_{\R^d} \dd x \, \lim_{k \to \infty} f_k(x) 
		= \int_{\R^d} \dd x \, f(x) 
	\end{align*}
	holds. 
\end{thm}
\begin{thm}[Dominated Convergence]
	Let $(f_k)_{k \in \N}$ be a sequence of functions in $L^1(\R^d)$ that converges almost everywhere pointwise to some $f : \R^d \longrightarrow \C$. If there exists a non-negative $g \in L^1(\R^d)$ such that $\abs{f_k(x)} \leq g(x)$ holds almost everywhere for all $k \in \N$, then $g$ also bounds $\abs{f}$, \ie $\abs{f(x)} \leq g(x)$ almost everywhere, and $f \in L^1(\R^d)$. Furthermore, the limit $k \to \infty$ and integration with respect to $x$ commute and we have 
	\begin{align*}
		\lim_{k \to \infty} \int_{\R^d} \dd x \, f_k(x) = \int_{\R^d} \dd x \, \lim_{k \to \infty} f_k(x) = \int_{\R^d} \dd x \, f(x) 
		. 
	\end{align*}
\end{thm}
%

\chapter{Linear operators} 
\label{operators}

We will give a rough overview of operator theory. The most important section is that on selfadjointness which illuminates the connection between unitary evolution groups and selfadjoint operators.

\section{Bounded operators} 
\label{operators:bounded}

The simplest operators are bounded operators. 
\begin{defn}[Bounded operator]
	Let $\mathcal{X}$ and $\mathcal{Y}$ be normed spaces. A linear operator $T : \mathcal{X} \longrightarrow \mathcal{Y}$ is called bounded if there exists $M \geq 0$ with $\norm{T x}_{\mathcal{Y}} \leq M \norm{x}_{\mathcal{X}}$. 
\end{defn}
Just as in the case of linear functionals, we have 
\begin{thm}
	Let $T : \mathcal{X} \longrightarrow \mathcal{Y}$ be a linear operator between two normed spaces $\mathcal{X}$ and $\mathcal{Y}$. Then the following statements are equivalent: 
	\begin{enumerate}[(i)]
		\item $T$ is continuous at $x_0 \in \mathcal{X}$. 
		\item $T$ is continuous. 
		\item $T$ is bounded. 
	\end{enumerate}
\end{thm}
\begin{proof}
	We leave it to the reader to modify the proof of Theorem~\ref{hilbert_spaces:dual_space:thm:bounded_functionals_continuous}. 
\end{proof}
We can introduce a norm on the operators which leads to a natural notion of convergence: 
\begin{defn}[Operator norm]
	Let $T : \mathcal{X} \longrightarrow \mathcal{Y}$ be a bounded linear operator between normed spaces. Then we define the operator norm of $T$ as 
	\begin{align*}
		\norm{T} := \sup_{\substack{x \in \mathcal{X} \\ \norm{x} = 1}} \norm{T x}_{\mathcal{Y}} 
		. 
	\end{align*}
	The space of all bounded linear operators between $\mathcal{X}$ and $\mathcal{Y}$ is denoted by $\mathcal{B}(\mathcal{X},\mathcal{Y})$. 
\end{defn}
One can show that $\norm{T}$ coincides with 
\begin{align*}
	\inf \bigl \{ M \geq 0 \; \vert \; \norm{T x}_{\mathcal{Y}} \leq M \norm{x}_{\mathcal{X}} \; \forall x \in \mathcal{X} \bigr \} = \norm{T} 
	. 
\end{align*}
The product of two bounded operators $T \in \mathcal{B}(\mathcal{Y},\mathcal{Z})$ and $S \in \mathcal{B}(\mathcal{X},\mathcal{Y})$ is again a bounded operator and its norm can be estimated from above by 
\begin{align*}
	\norm{T S} \leq \norm{T} \norm{S} 
	. 
\end{align*}
If $\mathcal{Y} = \mathcal{X} = \mathcal{Z}$, this implies that the product is jointly continuous with respect to the norm topology on $\mathcal{X}$. For Hilbert spaces, the following useful theorem holds: 
\begin{thm}[Hellinger-Toeplitz]
	Let $A$ be a linear operator on a Hilbert space $\Hil$ with dense domain $\mathcal{D}(A)$ such that $\scpro{\psi}{A \varphi} = \scpro{A \psi}{\varphi}$ holds for all $\varphi , \psi \in \mathcal{D}(A)$. Then $\mathcal{D}(A) = \Hil$ if and only if $A$ is bounded. 
\end{thm}
\begin{proof}
	$\Leftarrow$: If $A$ is bounded, then $\norm{A \varphi} \leq M \norm{\varphi}$ for some $M \geq 0$ and all $\varphi \in \Hil$ by definition of the norm. Hence, the domain of $A$ is all of $\Hil$. 
	\medskip
	
	\noindent
	$\Rightarrow$: This direction relies on a rather deep result of functional analysis, the so-called Open Mapping Theorem and its corollary, the Closed Graph Theorem. The interested reader may look it up in Chapter~III.5 of \cite{Reed_Simon:bibel_1:1981}. 
\end{proof}
Let $T,S$ be bounded linear operators between the normed spaces $\mathcal{X}$ and $\mathcal{Y}$. If we define 
\begin{align*}
	(T + S) x := T x + S x 
\end{align*}
as addition and 
\begin{align*}
	\bigl ( \lambda \cdot T  \bigr ) x := \lambda T x 
\end{align*}
as scalar multiplication, the set of bounded linear operators forms a vector space. 
\begin{prop}
	The vector space $\mathcal{B}(\mathcal{X},\mathcal{Y})$ of bounded linear operators between normed spaces $\mathcal{X}$ and $\mathcal{Y}$ with operator norm forms a normed space. If $\mathcal{Y}$ is complete, $\mathcal{B}(\mathcal{X},\mathcal{Y})$ is a Banach space. 
\end{prop}
\begin{proof}
	The fact $\mathcal{B}(\mathcal{X},\mathcal{Y})$ is a normed vector space follows directly from the definition. To show that $\mathcal{B}(\mathcal{X},\mathcal{Y})$ is a Banach space whenever $\mathcal{Y}$ is, one has to modify the proof of Theorem~\ref{hilbert_spaces:dual_space:prop:completeness_Xstar} to suit the current setting. This is left as an exercise. 
\end{proof}
Very often, it is easy to \emph{define} an operator $T$ on a ``nice'' dense subset $\mathcal{D} \subseteq \mathcal{X}$. Then the next theorem tells us that if the operator is bounded, there is a unique bounded extension of the operator to the whole space $\mathcal{X}$. For instance, this allows us to instantly extend the Fourier transform from Schwartz functions to $L^2(\R^d)$ functions (see~Proposition~\ref{S_and_Sprime:schwartz_functions:thm:Fourier_unitary_on_L2}). 
\begin{thm}\label{operators:bounded:thm:extensions_bounded_operators}
	Let $\mathcal{D} \subseteq \mathcal{X}$ be a dense subset of a normed space and $\mathcal{Y}$ be a Banach space. Furthermore, let $T : \mathcal{D} \longrightarrow \mathcal{Y}$ be a bounded linear operator. Then there exists a unique bounded linear extension $\tilde{T} : \mathcal{X} \longrightarrow \mathcal{Y}$ and $\snorm{\tilde{T}} = \norm{T}_{\mathcal{D}}$. 
\end{thm}
\begin{proof}
	We construct $\tilde{T}$ explicitly: let $x \in \mathcal{X}$ be arbitrary. Since $\mathcal{D}$ is dense in $\mathcal{X}$, there exists a sequence $(x_n)_{n \in \N}$ in $\mathcal{D}$ which converges to $x$. Then we set 
	\begin{align*}
		\tilde{T} x := \lim_{n \to \infty} T x_n 
		. 
	\end{align*}
	First of all, $\tilde{T}$ is linear. It is also well-defined: $(T x_n)_{n \in \N}$ is a Cauchy sequence in $\mathcal{Y}$, 
	\begin{align*}
		\bnorm{T x_n - T x_k}_{\mathcal{Y}} \leq \snorm{T}_{\mathcal{D}} \, \snorm{x_n - x_k}_{\mathcal{X}} \xrightarrow{n,k \to \infty} 0 
		, 
	\end{align*}
	where the norm of $T$ is defined as 
	\begin{align*}
		\snorm{T}_{\mathcal{D}} := \sup_{x \in \mathcal{D} \setminus \{ 0 \}} \frac{\norm{T x}_{\mathcal{Y}}}{\norm{x}_{\mathcal{X}}} 
		. 
	\end{align*}
	This Cauchy sequence in $\mathcal{Y}$ converges to some unique $y \in \mathcal{Y}$ as the target space is complete. Let $(x_n')_{n \in \N}$ be a second sequence in $\mathcal{D}$ that converges to $x$ and assume the sequence $(T x_n')_{n \in \N}$ converges to some $y' \in \mathcal{Y}$. We define a third sequence $(z_n)_{n \in \N}$ which alternates between elements of the first sequence $(x_n)_{n \in \N}$ and the second sequence $(x_n')_{n \in \N}$, \ie 
	\begin{align*}
		z_{2n - 1} &:= x_n 
		\\
		z_{2n} &:= x_n' 
	\end{align*}
	for all $n \in \N$. Then $(z_n)_{n \in \N}$ also converges to $x$ and $\bigl ( T z_n \bigr )$ forms a Cauchy sequence that converges to, say, $\zeta \in \mathcal{Y}$. Subsequences of convergent sequences are also convergent and they must converge to the same limit point. Hence, we conclude that 
	\begin{align*}
		\zeta &= \lim_{n \to \infty} T z_n = \lim_{n \to \infty} T z_{2n} = \lim_{n \to \infty} T x_n = y 
		\\
		&= \lim_{n \to \infty} T z_{2n - 1} = \lim_{n \to \infty} T x_n' = y' 
	\end{align*}
	holds and $\tilde{T} x$ does not depend on the particular choice of sequence which approximates $x$ in $\mathcal{D}$. It remains to show that $\snorm{\tilde{T}} = \snorm{T}_{\mathcal{D}}$: we can calculate the norm of $\tilde{T}$ on the dense subset $\mathcal{D}$ and use that $\tilde{T} \vert_{\mathcal{D}} = T$ to obtain 
	\begin{align*}
		\snorm{\tilde{T}} &= \sup_{\substack{x \in \mathcal{X} \\ \norm{x} = 1}} \snorm{\tilde{T} x} 
		= \sup_{x \in \mathcal{X} \setminus \{ 0 \}} \frac{\snorm{\tilde{T} x}}{\norm{x}} 
		= \sup_{x \in \mathcal{D} \setminus \{ 0 \}} \frac{\snorm{\tilde{T} x}}{\norm{x}} 
		\\
		&= \sup_{x \in \mathcal{D} \setminus \{ 0 \}} \frac{\norm{T x}}{\norm{x}} 
		. 
	\end{align*}
	Hence, the norm of the extension $\tilde{T}$ is equal to the norm of the original operator $T$. 
\end{proof}
The spectrum of an operator has been related to the set of possible outcomes of measurements (if the operator is selfadjoint). 
\begin{defn}[Spectrum]
	Let $T \in \mathcal{B}(\mathcal{X})$ be a bounded linear operator on a Banach space $\mathcal{X}$. We define: 
	\begin{enumerate}[(i)]
		\item The resolvent of $T$ is the set $\rho(T) := \bigl \{ \lambda \in \C \; \vert \; T - \lambda \id \mbox{ is bijective} \bigr \}$. 
		\item The spectrum $\sigma(T) := \C \setminus \rho(T)$ is the complement of $\rho(T)$ in $\C$. 
	\end{enumerate}
\end{defn}
One can show that for all $\lambda \in \rho(T)$, the map $(T - \lambda \id)^{-1}$ is a bounded operator and the spectrum is a closed subset of $\C$. One can show its $\sigma(T)$ is \emph{compact} and contained in $\bigl \{ \lambda \in \C \; \vert \; \abs{\lambda} \leq \norm{T} \bigr \} \subset \C$. 


\section{Adjoint operator} 
\label{operators:adjoint}

If $\mathcal{X}$ is a normed space, then we have defined $\mathcal{X}^*$, the space of bounded linear functionals on $\mathcal{X}$. If $T : \mathcal{X} \longrightarrow \mathcal{Y}$ is a bounded linear operator between two normed spaces, it naturally defines the \emph{adjoint operator} $T' : \mathcal{Y}^* \longrightarrow \mathcal{X}^*$ via 
\begin{align}
	(T' L)(x) := L(Tx) 
	\label{operators:adjoint:eqn:adjoint_operator_functional}
\end{align}
for all $x \in \mathcal{X}$ and $L \in \mathcal{Y}^*$. In case of Hilbert spaces, one can associate the \emph{Hilbert space adjoint}. We will almost exclusively work with the latter and thus drop ``Hilbert space'' most of the time. 
\begin{defn}[Hilbert space adjoint]
	Let $\Hil$ be a Hilbert space and $A \in \mathcal{B}(\Hil)$ be a bounded linear operator on $\Hil$. The antilinear isomorphism $C : \Hil \longrightarrow \Hil^*$ taken from Theorem~\ref{hilbert_spaces:dual_space:thm:Riesz_Lemma} maps functionals on $\Hil$ onto the corresponding vectors, \ie $C \psi := \sscpro{\psi}{\cdot} = L_{\psi}$. Then the Hilbert space adjoint is defined as 
	\begin{align*}
		A^* := C^{-1} A' C 
		, 
	\end{align*}
	or put differently 
	\begin{align*}
		\sscpro{A^* \varphi}{\psi} := \sscpro{\varphi}{A \psi} 
	\end{align*}
	for all $\varphi , \psi \in \Hil$. \marginpar{\small 2009.11.17} 
\end{defn}
\begin{prop}
	Let $A , B \in \mathcal{B}(\Hil)$ be two bounded linear operators on a Hilbert space $\Hil$ and $\alpha \in \C$. Then, we have: 
	\begin{enumerate}[(i)]
		\item $(A + B)^* = A^* + B^*$ 
		\item $(\alpha A)^* = \alpha^* \, A^*$
		\item $(A B)^* = B^* A^*$
		\item $\norm{A^*} = \norm{A}$
		\item $A^{**} = A$
		\item $\norm{A^* A} = \norm{A A^*} = \norm{A}^2$
		\item $\ker A = (\im A^*)^{\perp}$, $\ker A^* = (\im A)^{\perp}$
	\end{enumerate}
\end{prop}
\begin{proof}
	Properties (i)-(iii) follow directly from the defintion. 
	
	To show (iv), we note that $\norm{A} \leq \norm{A^*}$ follows from 
	\begin{align*}
		\norm{A \varphi} &= \abs{\Bscpro{\tfrac{A \varphi}{\snorm{A \varphi}}}{A \varphi}} 
		\overset{*}{=} \sup_{\norm{L}_* = 1} \abs{L(A \varphi)} 
		\\
		&= \sup_{\norm{\psi_L} = 1} \abs{\scpro{A^* \psi_L}{\varphi}} 
		\leq \norm{A^*} \norm{\varphi} 
	\end{align*}
	where in the step marked with $\ast$, we have used that we can calculate the norm from picking the functional associated to $\tfrac{A \varphi}{\snorm{A \varphi}}$: for a functional with norm 1, $\norm{L}_* =  1$, the norm of $L(A \varphi)$ cannot exceed that of $A \varphi$
	\begin{align*}
		\abs{L(A \varphi)} &= \sabs{\sscpro{\psi_L}{A \varphi}} 
		\leq \snorm{\psi_L} \snorm{A \varphi} 
		= \snorm{A \varphi} 
		. 
	\end{align*}
	Here, $\psi_L$ is the vector such that $L = \sscpro{\psi_L}{\cdot}$ which exists by Theorem~\ref{hilbert_spaces:dual_space:thm:Riesz_Lemma}. This theorem also ensures $\norm{L}_* = \snorm{\psi_L}$. On the other hand, from 
	\begin{align*}
		\bnorm{A^* \psi_L} &= \bnorm{L_{A^* \psi_L}}_* 
		= \sup_{\norm{\varphi} = 1} \babs{\bscpro{A^* \psi_L}{\varphi}} 
		\\
		&\leq \sup_{\norm{\varphi} = 1} \norm{\psi_L} \norm{A \varphi} 
		= \norm{A} \norm{L}_* 
		= \norm{A} \snorm{\psi_L} 
	\end{align*}
	we conclude $\norm{A^*} \leq \norm{A}$. Hence, $\norm{A^*} = \norm{A}$. 
	
	(v) is clear. For (vi), we remark 
	\begin{align*}
		\norm{A}^2 &= \sup_{\norm{\varphi} = 1} \norm{A \varphi}^2 
		= \sup_{\norm{\varphi} = 1} \bscpro{\varphi}{A^* A \varphi} 
		\\
		&\leq \sup_{\norm{\varphi} = 1} \norm{A^* A \varphi} 
		= \norm{A^* A} 
		. 
	\end{align*}
	This means 
	\begin{align*}
		\norm{A}^2 \leq \norm{A^* A} \leq \snorm{A^*} \norm{A} = \norm{A}^2 
		. 
	\end{align*}
	which combined with (iv), 
	\begin{align*}
		\norm{A}^2 = \norm{A^*}^2 \leq \norm{A A^*} \leq \snorm{A} \norm{A^*} = \norm{A}^2 
	\end{align*}
	implies $\norm{A^* A}= \norm{A}^2 = \norm{A A^*}$. (vii) is left as an exercise. 
\end{proof}
\begin{defn}
	Let $\Hil$ be a Hilbert space and $A \in \mathcal{B}(\Hil)$. Then $A$ is called 
	\begin{enumerate}[(i)]
		\item normal if $A^* A = A A^*$. 
		\item selfadjoint (or hermitian) if $A^* = A$. 
		\item unitary if $A^* A = \id_{\Hil} = A A^*$. 
		\item an orthogonal projection if $A^2 = A$ and $A^* = A$. 
		\item positive if $\bscpro{\varphi}{A \varphi} \geq 0$ for all $\varphi \in \Hil$. 
	\end{enumerate}
\end{defn}
%

\section{Unitary operators} 
\label{operators:unitary}

Unitary operators $U$ have the nice property that 
\begin{align*}
	\scpro{U \varphi}{U \psi} = \scpro{\varphi}{U^* U \psi} = \scpro{\varphi}{\psi} 
\end{align*}
for all $\varphi , \psi \in \Hil$. In case of quantum mechanics, we are interested in solutions to the Schrödinger equation 
\begin{align*}
	i \frac{\dd }{\dd t} \psi(t) = H \psi(t) 
	, 
	&& 
	\psi(t) = \psi_0 
	, 
\end{align*}
for a hamilton operator which satisfies $H^* = H$. Assume that $H$ is bounded (this is really the case for many simple quantum systems). Then the unitary group generated by $H$, 
\begin{align*}
	U(t) = e^{-i t H} 
	, 
\end{align*}
can be written as a power series, 
\begin{align*}
	e^{-i t H} = \sum_{n = 0}^{\infty} \frac{1}{n!} (-i t)^n \, H^n 
\end{align*}
where $H^0 := \id$ by convention. The sequence of partial sums converges in the operator norm to $e^{- i t H}$, 
\begin{align*}
	\sum_{n = 0}^N \frac{1}{n!} (-i t)^n \, H^n \xrightarrow{N \to \infty} e^{-i t H} 
	, 
\end{align*}
since we can make the simple estimate 
\begin{align*}
	\norm{\sum_{n = 0}^{\infty} \frac{1}{n!} (-i t)^n \, H^n \psi} &\leq \sum_{n = 0}^{\infty} \frac{1}{n!} \abs{t}^n \norm{H^n \psi} 
	\leq \sum_{n = 0}^{\infty} \frac{1}{n!} \abs{t}^n \norm{H}^n \, \norm{\psi} 
	\\ 
	&= e^{\abs{t} \norm{H}} \norm{\psi} < \infty 
	. 
\end{align*}
This shows that the power series of the exponential converges in the operator norm independently of the choice of $\psi$ to a bounded operator. Given a unitary evolution group, it is suggestive to obtain the hamiltonian which generates it by deriving $U(t) \psi$ with respect to time. This is indeed the correct idea. The left-hand side of the Schrödinger equation (modulo a factor of $i$) can be expressed as a limit 
\begin{align*}
	\frac{\dd}{\dd t} \psi(t) = \lim_{\delta \to 0} \tfrac{1}{\delta} \bigl ( \psi(t+\delta) - \psi(t) \bigr ) 
	. 
\end{align*}
This limit really exists, but before we compute it, we note that since 
\begin{align*}
	\psi(t+\delta) - \psi(t) = e^{- i (t+\delta) H} \psi_0 - e^{- i t H} \psi_0 
	= e^{- i t H} \bigl ( e^{- i \delta H} - 1 \bigr ) \psi_0 
	, 
\end{align*}
it suffices to consider differentiability at $t = 0$: taking limits in norm of $\Hil$, we get 
\begin{align*}
	\frac{\dd}{\dd t} \psi(0) &= \lim_{\delta \to 0} \, \tfrac{1}{\delta} \bigl ( \psi(\delta) - \psi_0 \bigr ) = \lim_{\delta \to 0} \, \frac{1}{\delta} \left ( \sum_{n = 0}^{\infty} \frac{(-i)^n}{n!} \delta^n H^n \psi_0 - \psi_0 \right ) 
	\\
	&
	= \lim_{\delta \to 0} \sum_{n = 1}^{\infty} \frac{(-i)^n}{n!} \delta^{n-1} H^n \psi_0 
	= -i H \psi_0 
	.  
\end{align*}
Hence, we have established that $e^{-i t H} \psi_0$ solves the Schrödinger condition with $\psi(0) = \psi_0$, 
\begin{align*}
	i \frac{\dd}{\dd t} \psi(t) = H \psi(t) 
	. 
\end{align*}
However, this procedure \emph{does not work} if $H$ is unbounded (\ie the generic case)! Before we proceed, we need to introduce several different notions of convergence of sequences of operators which are necessary to define derivatives of $U(t)$. 
\begin{defn}[Convergence of operators]
	Let $A_n \in \mathcal{B}(\Hil)$ be a sequence of bounded operators. We say that the sequence converges to $A \in \mathcal{B}(\Hil)$
	\begin{enumerate}[(i)]
		\item uniformly/in norm if $\lim_{n \to \infty} \bnorm{A_n - A} = 0$. 
		\item strongly if $\lim_{n \to \infty} \bnorm{A_n \psi - A \psi} = 0$ for all $\psi \in \Hil$. 
		\item weakly if $\lim_{n \to \infty} \bscpro{\varphi}{A_n \psi - A \psi} = 0$ for all $\varphi , \psi \in \Hil$. 
	\end{enumerate}
\end{defn}
Convergence of a sequence of operators in norm implies strong and weak convergence, but not the other way around. In the tutorials, we will also show explicitly that weak convergence does not necessarily imply strong convergence. 
\begin{example}
	With the arguments above, we have shown that if $H = H^*$ is selfadjoint and bounded, then $t \mapsto e^{-i t H}$ is \emph{uniformly} continuous. 
\end{example}
If $\norm{H} = \infty$ on the other hand, uniform continuity is too strong a requirement. If $H = - \frac{1}{2} \Delta_x$ is the free Schrödinger operator on $L^2(\R^d_x)$, then the Fourier transform $\Fourier$ links the position representation on $L^2(\R^d_x)$ to the momentum representation on $L^2(\R^d_k)$. In this representation, the free Schrödinger operator $H$ simplifies to the multiplication operator 
\begin{align*}
	\hat{H} = \tfrac{1}{2} \hat{k}^2 
\end{align*}
acting on $L^2(\R^d_k)$. More elaborate mathematical arguments show that for any $t \in \R$, the norm of the difference between $\hat{U}(t) = e^{- i t \frac{1}{2} \hat{k}^2}$ and $\hat{U}(0) = \id$
\begin{align*}
	\bnorm{\hat{U}(t) - \id} = \sup_{k \in \R^d} \babs{e^{- i t \frac{1}{2} k^2} - 1} = 2 
\end{align*}
is exactly $2$ and $\hat{U}(t)$ \emph{cannot} be uniformly continuous in $t$. However, if $\hat{\psi} \in L^2(\R^d_k)$ is a wave function, the estimate 
\begin{align*}
	\bnorm{\hat{U}(t) \hat{\psi} - \hat{\psi}}^2 &= \int_{\R^d_k} \dd k \, \babs{e^{- i t \frac{1}{2} k^2} - 1}^2 \, \babs{\hat{\psi}(k)}^2 
	\\
	&\leq 2^2 \int_{\R^d_k} \dd k \, \babs{\hat{\psi}(k)}^2 
	= 4 \bnorm{\hat{\psi}}^2 
\end{align*}
shows we can invoke the Theorem of Dominated Convergence to conclude $\hat{U}(t)$ is \emph{strongly continuous} in $t \in \R$. 
\begin{defn}[Strongly continuous one-parameter unitary group]
	A family of unitary operators $\{ U(t) \}_{t \in \R}$ on a Hilbert space $\Hil$ is called a strongly continuous one-parameter unitary group -- or unitary group for short -- if 
	\begin{enumerate}[(i)]
		\item $t \mapsto U(t)$ is strongly continuous and 
		\item $U(t) U(t') = U(t+t')$ as well as $U(0) = \id_{\Hil}$ 
	\end{enumerate}
	hold for all $t,t' \in \R_t$. 
\end{defn}
This is again a \emph{group representation of $\R_t$} just as in the case of the classical flow $\Phi$. The form of the Schrödinger equation, 
\begin{align*}
	i \frac{\dd }{\dd t} \psi(t) = H \psi(t) 
	, 
\end{align*}
also suggests that strong continuity/differentiability is the correct notion. Let us once more consider the free hamiltonian $H = - \frac{1}{2} \Delta_x$ on $L^2(\R^d_x)$. We have shown in the tutorials that its domain is 
\begin{align*}
	\mathcal{D}(H) = \bigl \{ \varphi \in L^2(\R^d_x) \; \vert \; - \Delta_x \varphi \in L^2(\R^d_x) \bigr \} 
	. 
\end{align*}
In Chapter~\ref{S_and_Sprime}, we will learn that $\mathcal{D}(H)$ is mapped by the Fourier transform onto 
\begin{align*}
	\mathcal{D}(\hat{H}) = \bigl \{ \hat{\psi} \in L^2(\R^d_k) \; \vert \; \hat{k}^2 \hat{\psi} \in L^2(\R^d_k) \bigr \} 
	. 
\end{align*}
Dominated Convergence can once more be used to make the following claims rigorous: for any $\hat{\psi} \in \mathcal{D}(\hat{H})$, we have 
\begin{align}
	\lim_{t \to 0} &\bnorm{\tfrac{i}{t} \bigl ( \hat{U}(t) - \id) \hat{\psi} - \tfrac{1}{2} \hat{k}^2 \hat{\psi}} 
	\leq \lim_{t \to 0} \bnorm{\tfrac{i}{t} \bigl ( \hat{U}(t) - \id) \hat{\psi}} +  \bnorm{\tfrac{1}{2} \hat{k}^2 \hat{\psi}}
	\label{operators:unitary:eqn:free_evolution_strongly_continuous}
	. 
\end{align}
The second term is finite since $\hat{\psi} \in \mathcal{D}(\hat{H})$ and we have to focus on the first term. On the level of functions, 
\begin{align*}
	\lim_{t \to 0} \tfrac{i}{t} \bigl ( e^{- i t \frac{1}{2} k^2} - 1 \bigr ) = i \frac{\dd}{\dd t} e^{-i t \frac{1}{2} k^2} \Big \vert_{t = 0} = \tfrac{1}{2} k^2 
\end{align*}
holds pointwise. Furthermore, by the mean value theorem, for any finite $t \in \R$ with $\abs{t} \leq 1$, for instance, then there exists $0 \leq t_0 \leq t$ such that 
\begin{align*}
	\tfrac{1}{t} \bigl ( e^{- i t \frac{1}{2} k^2} - 1 \bigr ) = \partial_t e^{- i t \frac{1}{2} k^2} \big \vert_{t = t_0} 
	= -i \tfrac{1}{2} k^2 \, e^{- i t_0 \frac{1}{2} k^2} 
	. 
\end{align*}
This can be bounded uniformly in $t$ by $\frac{1}{2} k^2$. Thus, also the first term can be bounded by $\bnorm{\frac{1}{2} \hat{k}^2 \hat{\psi}}$ uniformly. By Dominated Convergence, we can interchange the limit $t \to 0$ and integration with respect to $k$ on the left-hand side of equation~\eqref{operators:unitary:eqn:free_evolution_strongly_continuous}. But then the integrand is zero and thus the domain where the free evolution group is differentiable coincides with the domain of the Fourier transformed hamiltonian, \marginpar{2009.11.18}
\begin{align*}
	\lim_{t \to 0} \norm{\tfrac{i}{t} \bigl ( \hat{U}(t) - \id) \hat{\psi} - \tfrac{1}{2} \hat{k}^2 \hat{\psi}} = 0 
	. 
\end{align*}
This suggests to use the following definition: 
\begin{defn}[Generator of a unitary group]
	A densely defined linear operator on a Hilbert space $\Hil$ with domain $\mathcal{D}(H) \subseteq \Hil$ is called generator of a unitary evolution group $U(t)$, $t \in \R$, if 
	\begin{enumerate}[(i)]
		\item the domain coincides with 
		\begin{align*}
			\widetilde{\mathcal{D}(H)} = \Bigl \{ \varphi \in \Hil \; \big \vert \;  t \mapsto U(t) \varphi \mbox{ differentiable} \Bigr \} = \mathcal{D}(H)
		\end{align*}
		\item and for all $\psi \in \mathcal{D}(H)$, the Schrödinger equation holds, 
		\begin{align*}
			i \frac{\dd}{\dd t} U(t) \psi = H U(t) \psi 
			. 
		\end{align*}
	\end{enumerate}
\end{defn}
This is only one of the two implications: usually we are given a hamiltonian $H$ and we would like to know under which circumstances this operator generates a unitary evolution group. We will answer this question conclusively in the next section with Stone's Theorem. 
\begin{thm}
	Let $H$ be the generator of a strongly continuous evolution group $U(t)$, $t \in \R$. Then the following holds: 
	\begin{enumerate}[(i)]
		\item $\mathcal{D}(H)$ is invariant under the action of $U(t)$, \ie $U(t) \mathcal{D}(H) = \mathcal{D}(H)$ for all $t \in \R$. 
		\item $H$ commutes with $U(t)$, \ie $[U(t) , H] \psi := U(t) \, H \psi - H \, U(t) \psi = 0$ for all $t \in \R$ and $\psi \in \mathcal{D}(H)$. 
		\item $H$ is symmetric, \ie $\scpro{H \varphi}{\psi} = \scpro{\varphi}{H \psi}$ holds for all $\varphi , \psi \in \mathcal{D}(H)$. 
		\item $U(t)$ is uniquely determined by $H$. 
		\item $H$ is uniquely determined by $U(t)$. 
	\end{enumerate}
\end{thm}
\begin{proof}
	\begin{enumerate}[(i)]
		\item Let $\psi \in \mathcal{D}(H)$. To show that $U(t) \psi$ is still in the domain, we have to show that the norm of $H U(t) \psi$ is finite. Since $H$ is the generator of $U(t)$, it is equal to 
		\begin{align*}
			H \psi = i \frac{\dd}{\dd s} U(s) \psi \bigg \vert_{s = 0} = \lim_{s \to 0} \tfrac{i}{s} \bigl ( U(s) - \id \bigr ) \psi 
			. 
		\end{align*}
		Let us start with $s > 0$ and omit the limit. Then 
		\begin{align*}
			\Bnorm{\tfrac{i}{s} \bigl ( U(s) - \id \bigr ) U(t) \psi} &= \Bnorm{U(t) \tfrac{i}{s} \bigl ( U(s) - \id \bigr ) \psi} 
			= \Bnorm{\tfrac{i}{s} \bigl ( U(s) - \id \bigr ) \psi} < \infty 
		\end{align*}
		holds for all $s > 0$. Taking the limit on left and right-hand side yields that we can estimate the norm of $H U(t) \psi$ by the norm of $H \psi$ -- which is finite since $\psi$ is in the domain. This means $U(t) \mathcal{D}(H) \subseteq \mathcal{D}(H)$. To show the converse, we repeat the proof for $U(-t) = U(t)^{-1} = U(t)^*$ to obtain 
		\begin{align*}
			\mathcal{D}(H) = U(-t) U(t) \mathcal{D}(H) \subseteq U(t) \mathcal{D}(H) 
			. 
		\end{align*}
		Hence, $U(t) \mathcal{D}(H) = \mathcal{D}(H)$. 
		\item This follows from an extension of the proof of (i): since the domain $\mathcal{D}(H)$ coincides with the set of vectors on which $U(t)$ is strongly differentiable and is left invariant by $U(t)$, taking limits on left- and right-hand side of 
		\begin{align*}
			\Bnorm{\tfrac{i}{s} \bigl ( U(s) - \id \bigr ) U(t) \psi - U(t) \tfrac{i}{s} \bigl ( U(s) - \id \bigr ) \psi} = 0 
		\end{align*}
		leads to $[H,U(t)] \psi = 0$. 
		\item This follows from differentiating $\scpro{U(t) \varphi}{U(t) \psi}$ for arbitrary $\varphi , \psi \in \mathcal{D}(H)$ and using $\bigl [ U(t) , H \bigr ] = 0$ as well as the unitarity of $U(t)$ for all $t \in \R$. 
		\item Assume that both unitary evolution groups, $U(t)$ and $\tilde{U}(t)$, have $H$ as their generator. For any $\psi \in \mathcal{D}(H)$, we can calculate the time derivative of $\bnorm{(U(t) - \tilde{U}(t)) \psi}^2$, 
		\begin{align*}
			\frac{\dd }{\dd t} \bnorm{(U(t) - \tilde{U}(t)) \psi}^2 &= 2 \frac{\dd }{\dd t} \bigl ( \norm{\psi}^2 - \Re \bscpro{U(t) \psi}{\tilde{U}(t) \psi} \bigr ) 
			\\
			&= - 2 \Re \bigl ( \bscpro{- i H U(t) \psi}{\tilde{U}(t) \psi} + \bscpro{U(t) \psi}{- i H \tilde{U}(t) \psi} \bigr ) 
			\\
			&= 0 
			. 
		\end{align*}
		Since $U(0) = \id = \tilde{U}(0)$, this means $U(t)$ and $\tilde{U}(t)$ agree at least on $\mathcal{D}(H)$. Using the fact that there is only bounded extension of a bounded operator to all of $\Hil$, Theorem~\ref{operators:bounded:thm:extensions_bounded_operators}, we conclude they must be equal on all of $\Hil$. 
		\item This follows from the definition of the generator and the density of the domain. 
	\end{enumerate}
\end{proof}
Now that we have collected a few facts on unitary evolution groups, one could think that \emph{symmetric} operators generate evolution groups, but \emph{this is false!} The standard example to showcase this fact is the group of translations on $L^2([0,1])$. Since we would like $T(t)$ to conserve ``mass'' -- or more accurately, probability, we define for $\varphi \in L^2([0,1])$ and $0 \leq t < 1$
\begin{align*}
	\bigl ( T(t) \varphi \bigr )(x) := \left \{
	\begin{matrix}
		\varphi(x - t) & x - t \in [0,1] \\
		\varphi(x - t + 1) & x - t + 1 \in [0,1] \\
	\end{matrix}
	\right . 
	. 
\end{align*}
For all other $t \in \R$, we extend this operator periodically, \ie we plug in the fractional part of $t$. Clearly, $\bscpro{T(t) \varphi}{T(t) \psi} = \bscpro{\varphi}{\psi}$ holds for all $\varphi , \psi \in L^2([0,1])$. Locally, the infinitesimal generator is $- i \partial_x$ as a simple calculation shows: 
\begin{align*}
	\biggl ( i \frac{\dd}{\dd t} \bigl ( T(t) \varphi \bigr ) \biggr )(x) \bigg \vert_{t = 0} &= i \frac{\dd}{\dd t} \varphi(x - t) \bigg \vert_{t = 0} 
	= - i \partial_x \varphi(x) 
\end{align*}
However, $T(t)$ does not respect the maximal domain of $- i \partial_x$, 
\begin{align*}
	\mathcal{D}_{\max}(- i \partial_x) = \bigl \{ \varphi \in L^2([0,1]) \; \vert \; -i \partial_x \varphi \in L^2([0,1]) \bigr \} 
	. 
\end{align*}
Any element of the maximal domain has a continuous representative, but if $\varphi(0) \neq \varphi(1)$, then for $t > 0$, $T(t) \varphi$ will have a discontinuity at $t$. We will denote the operator $- i \partial_x$ on $\mathcal{D}_{\max}(- i \partial_x)$ with $\mathsf{P}_{\max}$. Let us check whether $\mathsf{P}_{\max}$ is symmetric: for any $\varphi , \psi \in \mathcal{D}_{\max}(- i \partial_x)$, we compute 
\begin{align}
	\bscpro{\varphi}{-i \partial_x \psi} &= \int_0^1 \dd x \, \varphi^*(x) \, (- i \partial_x \psi)(x) 
	= \Bigl [ - i \varphi^*(x) \, \psi(x) \Bigr ]_0^1 - \int_0^1 \dd x \, (-i) \partial_x \varphi^*(x) \, \psi(x) 
	\notag \\
	&
	= i \bigl ( \varphi^*(0) \, \psi(0) - \varphi^*(1) \, \psi(1) \bigr ) + \int_0^1 \dd x \, (- i \partial_x \varphi)^*(x) \, \psi(x) 
	\notag \\
	&= i \bigl ( \varphi^*(0) \, \psi(0) - \varphi^*(1) \, \psi(1) \bigr ) + \bscpro{- i \partial_x \varphi}{\psi} 
	\label{operators:unitary:eqn:symmetry_translations_interval}
	. 
\end{align}
In general, the boundary terms do not disappear and the maximal domain is ``too large'' for $- i \partial_x$ to be symmetric. Thus it is not at all surprising, $T(t)$ does not leave $\mathcal{D}_{\max}(- i \partial_x)$ invariant. Let us try another domain: one way to make the boundary terms disappear is to choose 
\begin{align*}
	\mathcal{D}_{\min}(- i \partial_x) := \Bigl \{ \varphi \in L^2([0,1]) \; \big \vert \; - i \partial_x \varphi \in L^2([0,1]) , \; \varphi(0) = 0 = \varphi(1) \Bigr \} 
	. 
\end{align*}
We denote $- i \partial_x$ on this ``minimal'' domain with $\mathsf{P}_{\min}$. In this case, the boundary terms in equation~\eqref{operators:unitary:eqn:symmetry_translations_interval} vanish which tells us that $\mathsf{P}_{\min}$ is symmetric. Alas, the domain is still not invariant under translations $T(t)$, even though $\mathsf{P}_{\min}$ is symmetric. This is an example of a symmetric operator which \emph{does not} generate a unitary group. 

There is another thing we have missed so far: the translations allow for an additional phase factor, \ie for $\varphi , \psi \in L^2([0,1])$ and $\vartheta \in [0,2\pi)$, we define for $0 \leq t < 1$
\begin{align*}
	\bigl ( T_{\vartheta}(t) \varphi \bigr )(x) := \left \{
	\begin{matrix}
		\varphi(x - t) & x - t \in [0,1] \\
		e^{i \vartheta} \varphi(x - t + 1) & x - t + 1 \in [0,1] \\
	\end{matrix}
	\right . 
	. 
\end{align*}
while for all other $t$, we plug in the fractional part of $t$. The additional phase factor cancels in the inner product, $\bscpro{T_{\vartheta}(t) \varphi}{T_{\vartheta}(t) \psi} = \bscpro{\varphi}{\psi}$ still holds true for all $\varphi , \psi \in L^2([0,1])$. In general $T_{\vartheta}(t) \neq T_{\vartheta'}(t)$ if $\vartheta \neq \vartheta'$ and the unitary groups are genuinely different. Repeating the simple calculation from before, we see that the local generator still is $- i \partial_x$ and it would seem we can generate a family of unitary evolutions from a \emph{single} generator. The confusion is resolved if we focus on \emph{invariant domains}: choosing $\vartheta \in [0,2\pi)$, we define $\mathsf{P}_{\vartheta}$ to be the operator $- i \partial_x$ on the domain 
\begin{align*}
	\mathcal{D}_{\vartheta}(- i \partial_x) := \Bigl \{ \varphi \in L^2([0,1]) \; \big \vert \; - i \partial_x \varphi \in L^2([0,1]) , \; \varphi(0) = e^{- i \vartheta} \varphi(1) \Bigr \} 
	. 
\end{align*}
A quick look at equation~\eqref{operators:unitary:eqn:symmetry_translations_interval} reassures us that $\mathsf{P}_{\vartheta}$ is symmetric and a quick calculation shows it is also \emph{invariant} under the action of $T_{\vartheta}(t)$. Hence, $\mathsf{P}_{\vartheta}$ is the generator of $T_{\vartheta}$, and the \emph{definition of an unbounded operator is incomplete without spelling out its domain}. 
\begin{example}
	Another example where the domain is crucial in the properties is the wave equation on $[0,L]$, 
	\begin{align*}
		\partial_t^2 u(x,t) - \partial_x^2 u(x,t) = 0 
		, 
		&&
		u \in \Cont^2([0,L] \times \R_t)
		. 
	\end{align*}
	Here, $u$ is the amplitude of the vibration, \ie the lateral deflection. If we choose Dirichlet boundary conditions at both ends, \ie $u(0) = 0 = u(L)$, we model a closed pipe, if we choose Dirichlet boundary conditions on one end, $u(0) = 0$, and von Neumann boundary conditions on the other, $u'(L) = 0$, we model a half-closed pipe. Choosing domains is a question of physics! 
\end{example}
%


\section{Selfadjoint operators} 
\label{operators:selfadjoint_operators}

Although we do not have time to explore this very far, the crucial difference between $\mathsf{P}_{\min}$ and $\mathsf{P}_{\vartheta}$ is that the former is only symmetric while the latter is also selfadjoint. We first recall the definition of the adjoint of a possibly unbounded operator: 
\begin{defn}[Adjoint operator]
	Let $A$ be a densely defined linear operator on a Hilbert space $\Hil$ with domain $\mathcal{D}(A)$. Let $\mathcal{D}(A^*)$ be the set of $\varphi \in \Hil$ for which there exists $\phi \in \Hil$ with 
	\begin{align*}
		\scpro{A \psi}{\varphi} = \scpro{\psi}{\phi} 
		&& 
		\forall \psi \in \mathcal{D}(A) 
		. 
	\end{align*}
	For each $\varphi \in \mathcal{D}(A^*)$, we define $A^* \varphi := \phi$ and $A^*$ is called the adjoint of $A$. 
\end{defn}
\begin{remark}
	By Riesz Lemma, $\varphi$ belongs to $\mathcal{D}(A^*)$ if and only if 
	\begin{align*}
		\babs{\scpro{A \psi}{\varphi}} \leq C \norm{\psi} 
		&&
		\forall \psi \in \mathcal{D}(A) 
		. 
	\end{align*}
	This is equivalent to saying $\varphi \in \mathcal{D}(A^*)$ if and only if $\psi \mapsto \sscpro{A \psi}{\varphi}$ is continuous on $\mathcal{D}(A)$. As a matter of fact, we could have used to latter to \emph{define} the adjoint operator. 
\end{remark}
One word of caution: even if $A$ is densely defined, $A^*$ need not be. 
\begin{example}
	Let $f \in L^{\infty}(\R)$, but $f \not\in L^2(\R)$, and pick $\psi_0 \in L^2(\R)$. Define 
	\begin{align*}
		\mathcal{D}(T_f) := \Bigl \{ \psi \in L^2(\R) \; \vert \; \int_{\R} \dd x \, \abs{f(x) \, \psi(x)} < \infty \Bigr \} 
		. 
	\end{align*}
	Then the adjoint of the operator 
	\begin{align*}
		T_f \psi := \sscpro{f}{\psi} \, \psi_0 
		, 
		&&
		\psi \in \mathcal{D}(T_f) 
		, 
	\end{align*}
	has domain $\mathcal{D}(T_f^*) = \{ 0 \}$. Let $\psi \in \mathcal{D}(T_f)$. Then for any $\varphi \in \mathcal{D}(T_f^*)$ 
	\begin{align*}
		\bscpro{T_f \psi}{\varphi} &= \bscpro{\sscpro{f}{\psi} \, \psi_0}{\varphi} = \bscpro{\psi}{f} \, \bscpro{\psi_0}{\varphi} 
		\\
		&= \bscpro{\psi}{\sscpro{\psi_0}{\varphi} f} 
		. 
	\end{align*}
	Hence $T_f^* \varphi = \sscpro{\psi_0}{\varphi} f$. However $f \not\in L^2(\R)$ and thus $\varphi = 0$ is the only possible choice for which $T_f^* \varphi$ is well defined. 
\end{example}
Symmetric operators, however, are special: since $\scpro{A \varphi}{\psi} = \scpro{\varphi}{A \psi}$ holds by definition for all $\varphi , \psi \in \Hil$, the domain of $A^*$ is contained in that of $A$, $\mathcal{D}(A^*) \supseteq \mathcal{D}(A)$. In particular, $\mathcal{D}(A^*) \subseteq \Hil$ is also dense. Thus, $A^*$ is an \emph{extension of $A$}. 
\begin{defn}[Selfadjoint operator]
	Let $H$ be a symmetric operator on a Hilbert space $\Hil$ with domain $\mathcal{D}(H)$. $H$ is called selfadjoint, $H^* = H$, if $\mathcal{D}(H^*) = \mathcal{D}(H)$. 
\end{defn}
One word regarding notation: if we write $A^* = A$, we do not just imply that the ``operator prescription'' of $A$ and $A^*$ is the same, but that the \emph{domains} of both coincide. 
\begin{example}
	In this sense, $\mathsf{P}_{\min}^* \neq \mathsf{P}_{\min}$. 
\end{example}
The central theorem of this section is Stone's Theorem: 
\begin{thm}[Stone]
	 To every strongly continuous one-parameter unitary group $U$ on a Hilbert space $\Hil$, there exists a selfadjoint operator $H = H^*$ which generates $U(t) = e^{- i t H}$. Conversely, every selfadjoint operator $H$ generates the unitary evolution group $U(t) = e^{- i t H}$. 
\end{thm}
A complete proof \cite[Chapter~VIII.3]{Reed_Simon:bibel_1:1981} is beyond our capabilities. \marginpar{\small 2009.11.24}


\chapter{Schwartz functions and tempered distributions} 
\label{S_and_Sprime}

Schwartz functions are a space of test functions, \ie a space of ``very nicely behaved functions.'' The dual of this space of test functions, the tempered distributions, allow us to extend common operations such as Fourier transforms and derivatives to objects which may not even be functions.

\section{Schwartz functions} 
\label{S_and_Sprime:schwartz_functions}

The motivation to define Schwartz functions on $\R^d$ comes from dealing with Fourier transforms: our class of test functions $\Schwartz(\R^d)$ has three defining properties: 
\begin{enumerate}[(i)]
	\item $\Schwartz(\R^d)$ forms a vector space. 
	\item \emph{Stability under derivation}, $\partial_x^{\alpha} \Schwartz(\R^d) \subset \Schwartz(\R^d)$: for all multiindices $\alpha \in \N_0^d$ and $f \in \Schwartz(\R^d)$, we have $\partial_x^{\alpha} f \in \Schwartz(\R^d)$. 
	\item \emph{Stability under Fourier transform}, $\Fourier \Schwartz(\R^d) \subseteq \Schwartz(\R^d)$: for all $f \in \Schwartz(\R^d)$, the Fourier transform 
	\begin{align}
		\Fourier^{\pm 1} f : \xi \mapsto \frac{1}{(2\pi)^{\nicefrac{d}{2}}} \int_{\R^d} \dd x \, e^{\mp i x \cdot \xi} \, f(x) \in \Schwartz(\R^d)
		\label{S_and_Sprime:schwartz_functions:eqn:Fourier_transform}
	\end{align}
	is also a test function. 
\end{enumerate}
These relatively simple requirements have surprisingly rich implications: 
\begin{enumerate}[(i)]
	\item $\Schwartz(\R^d) \subset L^1(\R^d)$, \ie any $f \in \Schwartz(\R^d)$ and all of its derivatives are integrable. 
	\item $\Fourier : \Schwartz(\R^d) \longrightarrow \Schwartz(\R^d)$ acts bijectively: if $f \in \Schwartz(\R^d) \subset L^1(\R^d)$, then $\Fourier f \in \Schwartz(\R^d) \subset L^1(\R^d)$. 
	\item For all $\alpha \in \N_0^d$, we have $\Fourier \bigl ( (i \partial_x)^{\alpha} f \bigr ) = x^{\alpha} \, \Fourier f \in \Schwartz(\R^d)$. This holds as all derivatives are integrable. 
	\item Hence, for all $a , \alpha \in \N_0^d$, we have $x^a \partial_x^{\alpha} f \in \Schwartz$. 
	\item Translations of Schwartz functions are again Schwartz functions, $f( \cdot - x_0) \in \Schwartz(\R^d)$; this follows from $\Fourier f(\cdot - x_0) = e^{- i \xi \cdot x_0} \, \Fourier f \in \Schwartz(\R^d)$ for all $x_0 \in \R^d$. 
\end{enumerate}
This leads to the following definition: 
\begin{defn}[Schwartz functions]
	The space of Schwartz functions 
	\begin{align*}
		\Schwartz(\R^d) := \Bigl \{ f \in \Cont^{\infty}(\R^d) \; \big \vert \; \forall a , \alpha \in \N_0^d : \norm{f}_{a \alpha} < \infty \Bigr \} 
	\end{align*}
	is defined in terms of the family of seminorms\footnote{A seminorm has all properties of a norm except that $\norm{f} = 0$ does not necessarily imply $f = 0$. } indexed by $a , \alpha \in \N_0^d$ 
	\begin{align*}
		\norm{f}_{a \alpha} := \sup_{x \in \R^d} \babs{x^a \partial_x^{\alpha} f(x)} 
		, 
		&& f \in \Cont^{\infty}(\R^d) 
		. 
	\end{align*}
\end{defn}
The family of seminorms defines a so-called \emph{Fréchet topology:} put in simple terms, to make sure that sequences in $\Schwartz(\R^d)$ converge to rapidly decreasing smooth functions, we need to control all derivatives as well as the decay. This is also the reason why there is \emph{no norm on $\Schwartz(\R^d)$} which generates the same topology as the family of seminorms. However, $\norm{f}_{a \alpha} = 0$ for all $a , \alpha \in \N_0^d$ ensures $f = 0$, all seminorms put together can distinguish points. 
\begin{example}
	Two simple examples of Schwartz functions are 
	\begin{align*}
		f(x) = e^{- a x^2} 
		, 
		&& 
		a > 0 
		, 
	\end{align*}
	and 
	\begin{align*}
		g(x) = \left \{
		\begin{matrix}
			e^{- \frac{1}{1 - x^2} + 1} & \abs{x} < 1 \\
			0 & \abs{x} \geq 1 \\
		\end{matrix}
		\right . 
		. 
	\end{align*}
	The second one even has compact support. 
\end{example}
The first major fact we will establish is completeness. 
\begin{thm}
	The space of Schwartz functions endowed with 
	\begin{align*}
		\mathrm{d}(f,g) := \sum_{n = 0}^{\infty} 2^{-n} \sup_{\abs{a} + \abs{\alpha} = n} \frac{\norm{f - g}_{a \alpha}}{1 + \norm{f - g}_{a \alpha}} 
	\end{align*}
	is a complete metric space. 
\end{thm}
\begin{proof}
	$\mathrm{d}$ is positive and symmetric. It also satisfies the triangle inequality as $x \mapsto \frac{x}{1 + x}$ is concave on $\R^+_0$ and all of the seminorms satisfy the triangle inequality. Hence, $\bigl ( \Schwartz(\R^d) , \mathrm{d} \bigr )$ is a metric space. 
	
	To show completeness, take a Cauchy sequence $(f_n)$ with respect to $\mathrm{d}$. By definition and positivity, this means $(f_n)$ is also a Cauchy sequence with respect to all of the seminorms $\norm{\cdot}_{a \alpha}$. Each of the $\bigl ( x^a \partial_x^{\alpha} f_n \bigr )$ converge to some $g_{a \alpha}$ as the space of bounded continuous functions $\BCont(\R^d)$ with $\sup$ norm is complete. It remains to show that $g_{a \alpha} = x^a \partial_x^{\alpha} g_{00}$. Clearly, only taking derivatives is problematic: we will prove this for $\abs{\alpha} = 1$, the general result follows from a simple induction. Assume we are interested in the sequence $(\partial_{x_k} f_n)$, $k \in \{ 1 , \ldots , d \}$. With $\alpha_k := (0 , \ldots , 0 , 1 , 0, \ldots)$ as the multiindex that has a $1$ in the $k$th entry and $e_k := (0 , \ldots , 0 , 1 , 0, \ldots) \in \R^d$ as the $k$th canonical base vector, we know that 
	\begin{align*}
		f_n(x) = f_n(x - x_k e_k) + \int_0^{x_k} \dd s \, \partial_{x_k} f_n \bigl ( x + (s - x_k) e_k \bigr ) 
	\end{align*}
	as well as 
	\begin{align*}
		g_{00}(x) = g_{00}(x - x_k e_k) + \int_0^{x_k} \dd s \, \partial_{x_k} g_{00} \bigl ( x + (s - x_k) e_k \bigr ) 
	\end{align*}
	hold since $f_n \to g_{00}$ and $\partial_{x_k} f_n \to g_{0\alpha_k}$ uniformly. Hence, $g_{00} \in \Cont^1(\R^d)$ and the derivative of $g_{00}$ coincides with $g_{0\alpha_k}$, $\partial_{x_k} g_{00} = g_{0 \alpha_k}$. We then proceed by induction to show $g_{00} \in \Cont^{\infty}(\R^d)$. This means $\mathrm{d}(f_n,g_{00}) \to 0$ as $n \to \infty$ and $\Schwartz(\R^d)$ is complete. 
\end{proof}
The $L^p$ norm of each element in $\Schwartz(\R^d)$ can be dominated by two seminorms: 
\begin{lem}\label{S_and_Sprime:schwartz_functions:lem:Lp_estimate}
	Let $f \in \Schwartz(\R^d)$. Then for each $1 \leq p < \infty$, the $L^p$ norm of $f$ can be dominated by a finite number of seminorms, 
	\begin{align*}
		\norm{f}_{L^p(\R^d)} \leq C_1(d) \norm{f}_{00} + C_2(d) \max_{\abs{a} = 2n(d)} \norm{f}_{a 0} 
		, 
	\end{align*}
	where $C_1(d) , C_2(d) \in \R^+$ and $n(d) \in \N_0$ only depend on the dimension of $\R^d$. Hence, $f \in L^p(\R^d)$. 
\end{lem}
\begin{proof}
	We split the integral on $\R^d$ into an integral over the unit ball centered at the origin and its complement: let $B_n := \max_{\abs{a} = 2n} \norm{f}_{a0}$, then 
	\begin{align*}
		\norm{f}_{L^p(\R^d)} &= \biggl ( \int_{\R^d} \dd x \, \abs{f(x)}^p \biggr )^{\nicefrac{1}{p}} 
		\leq \biggl ( \int_{\abs{x} \leq 1} \dd x \, \abs{f(x)}^p \biggr )^{\nicefrac{1}{p}} + \biggl ( \int_{\abs{x} > 1} \dd x \, \abs{f(x)}^p \biggr )^{\nicefrac{1}{p}} 
		\\
		&\leq \norm{f}_{00} \, \biggl ( \int_{\abs{x} \leq 1} \dd x \, 1 \biggr )^{\nicefrac{1}{p}} +  \biggl ( \int_{\abs{x} > 1} \dd x \, \abs{f(x)}^p \frac{\abs{x}^{2np}}{\abs{x}^{2np}} \biggr )^{\nicefrac{1}{p}} 
		\\
		&\leq \mathrm{Vol}(B_1(0))^{\nicefrac{1}{p}} \, \norm{f}_{00} + B_n \, \biggl ( \int_{\abs{x} > 1} \dd x \, \frac{1}{\abs{x}^{2np}} \biggr )^{\nicefrac{1}{p}} 
		. 
	\end{align*}
	If we choose $n$ large enough, $\abs{x}^{-2np}$	is integrable and can be computed explicitly, and we get 
	\begin{align*}
		\norm{f}_{L^p(\R^d)} \leq C_1(d) \, \norm{f}_{00} + C_2(d) \, \max_{\abs{a} = 2n} \norm{f}_{a0} 
		. 
	\end{align*}
	This concludes the proof. 
\end{proof}
\begin{lem}\label{S_and_Sprime:schwartz_functions:lem:density_of_Cinfty_compact}
	The smooth functions with compact support $\Cont^{\infty}_c(\R^d)$ are dense in $\Schwartz(\R^d)$. 
\end{lem}
\begin{proof}
	Take any $f \in \Schwartz(\R^d)$ and take 
	\begin{align*}
		g(x) = \left \{
		\begin{matrix}
			e^{- \frac{1}{1 - x^2} + 1} & \abs{x} \leq 1 \\
			0 & \abs{x} > 1 \\
		\end{matrix}
		\right . 
		. 
	\end{align*}
	Then $f_n := g(\nicefrac{\cdot}{n}) \, f$ converges to $f$ in $\Schwartz(\R^d)$, \ie 
	\begin{align*}
		\lim_{n \to \infty} \bnorm{f_n - f}_{a \alpha} = 0 
	\end{align*}
	holds for all $a , \alpha \in \N_0^d$. 
\end{proof}
Next, we will show that $\Fourier : \mathcal{S}(\R^d) \longrightarrow \mathcal{S}(\R^d)$ is a continuous and bijective map from $\mathcal{S}(\R^d)$ onto itself. 
\begin{thm}\label{S_and_Sprime:schwartz_functions:thm:Fourier_is_bijection}
	The Fourier transform $\Fourier$ as defined by equation~\eqref{S_and_Sprime:schwartz_functions:eqn:Fourier_transform} maps $\Schwartz(\R^d)$ continuously and bijectively onto itself. The inverse $\Fourier^{-1}$ is continuous as well.  Furthermore, for all $f \in \Schwartz(\R^d)$ and $a , \alpha \in \N_0^d$, we have 
	\begin{align}
		\Fourier \bigl ( x^a (+i \partial_x)^{\alpha} f \bigr ) = (+i \partial_{\xi})^a \xi^{\alpha} \Fourier f 
		. 
		\label{S_and_Sprime:schwartz_functions:eqn:Fourier_is_bijection}
	\end{align}
\end{thm}
\begin{proof}
	We need to prove $\Fourier \bigl ( x^a (+i \partial_x)^{\alpha} f \bigr ) = (+i \partial_{\xi})^a \xi^{\alpha} \Fourier f$ first: since $x^{\alpha} \partial_x^a f$ is integrable, its Fourier transform exists and is continuous by Dominated Convergence. For any $a , \alpha \in \N_0^d$, we have 
	\begin{align*}
		\Bigl ( \Fourier \bigl ( x^a (+i \partial_x)^{\alpha} f \bigr ) \Bigr )(\xi) &= 
		\frac{1}{(2\pi)^{\nicefrac{d}{2}}} \int_{\R^d} \dd x \, e^{-i x \cdot \xi} \, x^a \, (+i \partial_x)^{\alpha} f(x) 
		\\
		&= \frac{1}{(2\pi)^{\nicefrac{d}{2}}} \int_{\R^d} \dd x \, (+i \partial_{\xi})^a e^{- i x \cdot \xi} \, (+i \partial_x )^{\alpha} f(x) 
		\\
		&
		\overset{\ast}{=} \frac{1}{(2\pi)^{\nicefrac{d}{2}}} (+i \partial_{\xi})^a \int_{\R^d} \dd x \, e^{- i x \cdot \xi} \, (+i \partial_x)^{\alpha} f(x) 
		. 
	\end{align*}
	In the step marked with $\ast$, we have used Dominated Convergence to interchange integration and differentiation. Now we integrate partially $\abs{\alpha}$ times and use that the boundary terms vanish, 
	\begin{align*}
		\Bigl ( \Fourier \bigl ( x^a (+i \partial_x)^{\alpha} f \bigr ) \Bigr )(\xi) &= \frac{1}{(2\pi)^{\nicefrac{d}{2}}} (+i \partial_{\xi})^a \int_{\R^d} \dd x \, (+i \partial_x)^{\alpha} e^{- i x \cdot \xi} \, f(x) 
		\\
		&
		= \frac{1}{(2\pi)^{\nicefrac{d}{2}}} (+i \partial_{\xi})^a \int_{\R^d} \dd x \, \xi^{\alpha} e^{- i x \cdot \xi} \, f(x) 
		\\
		&= \bigl ( (+i \partial_{\xi})^a \xi^{\alpha} \Fourier f \bigr )(\xi) 
		. 
	\end{align*}
	To show $\Fourier$ is continuous, we need to estimate the seminorms of $\Fourier f$ by those of $f$: for any $a , \alpha \in \N_0^d$, it holds 
	\begin{align*}
		\bnorm{\Fourier f}_{a \alpha} &= \sup_{\xi \in \R^d} \babs{\bigl ( \xi^a \partial_{\xi}^{\alpha} \Fourier f \bigr )(\xi)} = \sup_{\xi \in \R^d} \Babs{\Bigl ( \Fourier \bigl ( \partial_x^a x^{\alpha} f \bigr ) \Bigr )(x)} 
		\\
		&\leq \frac{1}{(2\pi)^{\nicefrac{d}{2}}} \bnorm{\partial_x^a x^{\alpha} f}_{L^1(\R^d)} 
		. 
	\end{align*}
	In particular, this implies $\Fourier f \in \Schwartz(\R^d)$. Since $\partial_x^a x^{\alpha} f \in \Schwartz(\R^d)$, we can apply Lemma~\ref{S_and_Sprime:schwartz_functions:lem:Lp_estimate} and estaimte the right-hand side by a finite number of seminorms of $f$. Hence, $\Fourier$ is continuous: if $f_n$ is a Cauchy sequence in $\Schwartz(\R^d)$ that converges to $f$, then $\Fourier f_n$ has to converge to $\Fourier f \in \Schwartz(\R^d)$. 
	
	To show that $\Fourier$ is a bijection with continuous inverse, we note that it suffices to prove $\Fourier^{-1} \Fourier f = f$ for functions $f$ in a dense subset, namely $\Cont^{\infty}_c(\R^d)$ (see Lemma~\ref{S_and_Sprime:schwartz_functions:lem:density_of_Cinfty_compact}). Pick $f$ so that the support of is contained in a cube $W_n$ with sides of length $2n$. We can write $f$ on $W_n$ as a uniformly convergent Fourier series, 
	\begin{align*}
		f(x) = \sum_{\xi \in \frac{\pi}{n} \Z^d} \hat{f}(\xi) e^{i x \cdot \xi} 
		, 
	\end{align*}
	with 
	\begin{align*}
		\hat{f}(\xi) &= \frac{1}{\mathrm{Vol}(W_n)} \int_{W_n} \dd x \, e^{- i x \cdot \xi} \, f(x) 
		= \frac{(2 \pi)^{\nicefrac{d}{2}}}{(2n)^d} \frac{1}{(2\pi)^{\nicefrac{d}{2}}} \int_{\R^d} \dd x \, e^{- i x \cdot \xi} \, f(x) 
		. 
	\end{align*}
	Hence, $f$ can be expressed as 
	\begin{align*}
		f(x) = \sum_{\xi \in \frac{\pi}{n} \Z^d} \frac{1}{(2\pi)^{\nicefrac{d}{2}}} \frac{\pi^d}{n^d} \, (\Fourier f)(\xi) \, e^{i x \cdot \xi} 
	\end{align*}
	which is a Riemann sum that converges to 
	\begin{align*}
		f(x) &= \frac{1}{(2\pi)^{\nicefrac{d}{2}}} \int_{\R^d} \dd x \, e^{i x \cdot \xi} \, (\Fourier f)(\xi) 
		= \bigl ( \Fourier^{-1} \Fourier f \bigr )(x)
	\end{align*}
	as $\Fourier f \in \Schwartz$. This concludes the proof. 
\end{proof}
Hence, we have shown that $\Schwartz(\R^d)$ has the defining properties that suggested its motivation in the first place. The Schwartz functions also have other nice properties whose proofs are left as an exercise. \marginpar{\small 2009.11.25}
\begin{prop}
	The Schwartz functions have the following properties: 
	\begin{enumerate}[(i)]
		\item With pointwise multiplication $\cdot : \Schwartz(\R^d) \times \Schwartz(\R^d) \longrightarrow \Schwartz(\R^d)$, the space of $\Schwartz(\R^d)$ forms a Fréchet algebra (\ie the multiplication is continuous in both arguments). 
		\item For all $a , \alpha$, the map $f \mapsto x^a \partial_{\xi}^{\alpha} f$ is continuous on $\Schwartz(\R^d)$. 
		\item For any $x_0 \in \R^d$, the map $\tau_{x_0} : f \mapsto f(\cdot - x_0)$ continuous on $\Schwartz(\R^d)$. 
		\item For any $f \in \Schwartz$, $\frac{1}{h} \bigl ( \tau_{h e_k} f - f)$ converges to $\partial_{x_k} f$ as $h \to 0$ where $e_k$ is the $k$th canonical base vector of $\R^d$. 
	\end{enumerate}
\end{prop}
The next important fact will be mentioned without proof: 
\begin{thm}\label{S_and_Sprime:schwartz_functions:thm:S_dense_in_Lp}
	$\Schwartz(\R^d)$ is dense in $L^p(\R^d)$, $1 \leq p < \infty$. 
\end{thm}
This means, we can approximate any $L^p(\R^d)$ function by a test function. We will use this and the next theorem to extend the Fourier transform to $L^2(\R^d)$. 
\begin{thm}\label{S_and_Sprime:schwartz_functions:thm:unitarity_Fourier_on_S}
	For all $f , g \in \Schwartz(\R^d)$, we have 
	\begin{align*}
		\int_{\R^d} \dd x \, (\Fourier f)(x) \, g(x) = \int_{\R^d} \dd x \, f(x) \, (\Fourier g)(x) 
		. 
	\end{align*}
	This implies $\sscpro{\Fourier f}{g} = \sscpro{f}{\Fourier^{-1} g}$ and $\sscpro{\Fourier f}{\Fourier g} = \sscpro{f}{g}$ where $\scpro{\cdot}{\cdot}$ is the usual scalar product on $L^2(\R^d)$. 
\end{thm}
\begin{proof}
	Using Fubini's theorem, we conclude we can first integrate with respect to $\xi$ instead of $x$, 
	\begin{align*}
		\int_{\R^d} \dd x \, (\Fourier f)(x) \, g(x) &= \int_{\R^d} \dd x \, \frac{1}{(2\pi)^{\nicefrac{d}{2}}} \int_{\R^d} \dd \xi \, e^{- i x \cdot \xi} \, f(\xi) \, g(x) 
		\\
		&
		= \int_{\R^d} \dd \xi \, f(\xi) \, \frac{1}{(2\pi)^{\nicefrac{d}{2}}} \int_{\R^d} \dd x \, e^{- i x \cdot \xi} g(x) 
		= \int_{\R^d} \dd \xi \, f(\xi) \, (\Fourier g)(\xi) 
		. 
	\end{align*}
	To prove the second part, we remark that compared to the scalar product on $L^2(\R^d)$, we are missing a complex conjugation of the first function. Furthermore, $(\Fourier f)^* = \Fourier^{-1} f^*$ holds. From this, it follows that $\sscpro{\Fourier f}{g} = \sscpro{f}{\Fourier^{-1} g}$ and upon replacing $g$ with $\Fourier g$, that $\sscpro{\Fourier f}{\Fourier g} = \sscpro{f}{\Fourier^{-1} \Fourier g} = \sscpro{f}{g}$. 
\end{proof}
This has a nice consequence: since $\Schwartz(\R^d) \subset L^2(\R^d)$ is dense and the Fourier transform acts unitarily on $\Schwartz(\R^d)$ with respect to the scalar product, \ie the Fourier transform is bounded,  Theorem~\ref{operators:bounded:thm:extensions_bounded_operators} allows us to extend $\Fourier : \Schwartz(\R^d) \longrightarrow \Schwartz(\R^d)$ to a unitary operator $\Fourier : L^2(\R^d) \longrightarrow L^2(\R^d)$. Hence, we have just proven 

\begin{prop}\label{S_and_Sprime:schwartz_functions:thm:Fourier_unitary_on_L2}
	The Fourier transform $\Fourier : \Schwartz(\R^d) \longrightarrow \Schwartz(\R^d)$ extends to a unitary operator $\Fourier : L^2(\R^d) \longrightarrow L^2(\R^d)$. 
\end{prop}
Now we will apply this to the free Schrödinger operator $H = - \tfrac{1}{2} \Delta_x$. First of all, we conclude from Theorem~\ref{S_and_Sprime:schwartz_functions:thm:S_dense_in_Lp} that the domain of $H$, 
\begin{align*}
	\Schwartz(\R^d) \subset \mathcal{D}(H) = \bigl \{ \varphi \in L^2(\R^d) \; \vert \; - \Delta_x \varphi \in L^2(\R^d) \bigr \} \subset L^2(\R^d)
	, 
\end{align*}
is dense. So let us pick initial conditions such that $\psi_0 \in \Schwartz(\R^d) \subset L^2(\R^d)$. Then we can Fourier transform the free Schrödinger equation, 
\begin{align*}
	\Fourier \bigl ( i \partial_t \psi(t) \bigr ) &= i \partial_t \Fourier \psi(t) = \Fourier \bigl ( - \tfrac{1}{2} \Delta_x \psi(t) \bigr ) 
	\\
	&
	= \Fourier \bigl ( - \tfrac{1}{2} \Delta_x \bigr ) \Fourier^{-1} \Fourier \psi(t) 
	. 
\end{align*}
If we compute the right-hand side at $t = 0$, by Theorem~\ref{S_and_Sprime:schwartz_functions:thm:Fourier_is_bijection} this leads to 
\begin{align*}
	\Fourier \bigl ( - \tfrac{1}{2} \Delta_x \psi_0 \bigr ) (\xi) = \tfrac{1}{2} \xi^2 \, (\Fourier \psi_0)(\xi) 
	. 
\end{align*}
Denoting the Fourier transform of $\psi$ by $\hat{\psi}$, this means the free Schrödinger equation in the momentum representation reads 
\begin{align*}
	i \frac{\dd }{\dd t} \hat{\psi}(t) = \tfrac{1}{2} \hat{\xi}^2 \hat{\psi}(t) 
	, 
	&&
	\hat{\psi}(0) = \hat{\psi}_0 
	. 
\end{align*}
We have already proven in section~\ref{operators:unitary} that the unitary time evolution generated by $\hat{H} = \frac{1}{2} \hat{\xi}^2$ is $e^{-i t \frac{1}{2} \hat{\xi}^2}$. Sandwiching the unitary time evolution in the momentum representation in between Fourier transforms (which are themselves unitary on $L^2(\R^d)$) yields the time evolution in the position representation. 
\begin{prop}\label{S_and_Sprime:schwartz_functions:prop:free_Schroedinger}
	Let $\psi_0 \in \Schwartz(\R^d) \subset L^2(\R^d)$. Then for $t \neq 0$ the global solution of the free Schrödinger equation with initial condition $\psi_0$ is given by 
	\begin{align}
		\psi(x,t) &= \frac{1}{(2 \pi i t)^{\nicefrac{d}{2}}} \int_{\R^d} \dd y \, e^{i \frac{(x-y)^2}{2 t}} \psi_0(y) 
		=: \int_{\R^d} \dd y \, p(x,y,t) \, \psi_0(y) 
		. 
		\label{S_and_Sprime:schwartz_functions:eqn:free_propagator}
	\end{align}
	This expression converges in the $L^2$ norm to $\psi_0$ as $t \to 0$. 
\end{prop}
\begin{proof}
	We denote the Fourier transform of $e^{-i t \frac{1}{2} \hat{\xi}^2}$ by 
	\begin{align*}
		U(t) := \Fourier^{-1} e^{- i t \frac{1}{2} \hat{\xi}^2} \Fourier 
		. 
	\end{align*}
	If $t = 0$, then the bijectivity of the Fourier transform on $\Schwartz(\R^d)$, Theorem~\ref{S_and_Sprime:schwartz_functions:thm:Fourier_is_bijection}, yields $U(0) = \id_{\Schwartz}$. 
	
	So let $t \neq 0$. If $\hat{\psi}_0$ is a Schwartz function, so is $e^{- i t \frac{1}{2} \xi^2} \hat{\psi}_0$. As the Fourier transform is a unitary map on $L^2(\R^d)$ (Theorem~\ref{S_and_Sprime:schwartz_functions:thm:Fourier_unitary_on_L2}) and maps Schwartz functions onto Schwartz functions (Theorem~\ref{S_and_Sprime:schwartz_functions:thm:Fourier_is_bijection}), $U(t)$ also maps Schwartz functions onto Schwartz functions. Plugging in the definition of the Fourier transform, for any $\psi_0 \in \Schwartz(\R^d)$ and $t \neq 0$ we can write out $U(t) \psi_0$ as 
	\begin{align}
		\bigl ( \Fourier^{-1} &e^{- i t \frac{1}{2} \hat{\xi}^2} \Fourier \psi_0 \bigr )(x) = \frac{1}{(2\pi)^d} \int_{\R^d} \dd \xi \, e^{+ i x \cdot \xi} \, e^{- i t \frac{1}{2} \xi^2} \, \int_{\R^d} \dd y \, e^{- i y \cdot \xi} \, \psi_0(y) 
		\notag \\
		&= \frac{1}{(2 \pi)^{\nicefrac{d}{2}}} \int_{\R^d} \dd \xi \int_{\R^d} \dd y \, e^{i \frac{(x - y)^2}{2 t}} \, \left ( \frac{1}{(2 \pi)^{\nicefrac{d}{2}}} \, e^{- i \frac{t}{2} ( \xi - \nicefrac{(x - y)}{t})^2} \right ) \, \psi_0(y) 
		. 
		\label{S_and_Sprime:schwartz_functions:eqn:proof_free_propagator}
	\end{align}
	We need to regularize the integral: if we write the right-hand side of the above as 
	\begin{align*}
		\mbox{r. h. s.} &= \lim_{\eps \searrow 0} \frac{1}{(2 \pi)^{\nicefrac{d}{2}}} \int_{\R^d} \dd \xi \int_{\R^d} \dd y \, e^{i \frac{(x - y)^2}{2 t}} \, \left ( \frac{1}{(2 \pi)^{\nicefrac{d}{2}}} \, e^{- (\eps + i) \frac{t}{2} ( \xi - \nicefrac{(x - y)}{t})^2} \right ) \, \psi_0(y) 
		\\
		&= \lim_{\eps \searrow 0} \frac{1}{(2 \pi)^{\nicefrac{d}{2}}} \int_{\R^d} \dd y \, e^{i \frac{(x - y)^2}{2 t}} \, \left ( \frac{1}{(2 \pi)^{\nicefrac{d}{2}}} \int_{\R^d} \dd \xi \, e^{- (\eps + i) \frac{t}{2} ( \xi - \nicefrac{(x - y)}{t})^2} \right ) \, \psi_0(y) 
		, 
	\end{align*}
	we can use Fubini to change the order of integration. The inner integral can be computed by interpreting it as an integral in the complex plane, 
	\begin{align*}
		\frac{1}{(2 \pi)^{\nicefrac{d}{2}}} \int_{\R^d} \dd \xi \, e^{- (\eps + i) \frac{t}{2} ( \xi - \nicefrac{(x - y)}{t})^2} = \frac{1}{\bigl ( (\eps + i) t \bigr )^{\nicefrac{d}{2}}} 
		. 
	\end{align*}
	Plugged back into equation~\eqref{S_and_Sprime:schwartz_functions:eqn:proof_free_propagator} and combined with the Theorem of Monotone Convergence, this yields equation~\eqref{S_and_Sprime:schwartz_functions:eqn:free_propagator}. 
\end{proof}
Later on, we will need the convolution, 
\begin{align*}
	(f \ast g)(x) := \int_{\R^d} \dd y \, f(x - y) \, g(y) 
	, 
\end{align*}
for the Moyal product the product which emulates the operator product on the level of functions. The convolution is an abelian product on $\Schwartz(\R^d)$, \ie $f \ast g = g \ast f$ is satisfied for all $f , g \in \Schwartz(\R^d)$. 
\begin{prop}
	$\Schwartz(\R^d) \ast \Schwartz(\R^d) \subseteq \Schwartz(\R^d)$ 
\end{prop}
\begin{proof}
	Let $f , g \in \Schwartz(\R^d)$. As $\Fourier (f \ast g) = (2\pi)^{\nicefrac{d}{2}} \Fourier f \, \Fourier g$, 
	\begin{align*}
		\bigl ( \Fourier (f \ast g) \bigr )(\xi) &= \frac{1}{(2 \pi)^{\nicefrac{d}{2}}} \int_{\R^d} \dd x \, e^{- i x \cdot \xi} \, \int_{\R^d} \dd y \, f(x - y) \, g(y) 
		\\
		&= \frac{1}{(2 \pi)^{\nicefrac{d}{2}}} \int_{\R^d} \dd x \, e^{- i (x - y) \cdot \xi} \, \int_{\R^d} \dd y \, f(x - y) \, e^{- i y \cdot \xi} \, g(y) 
		\\
		&
		= (2\pi)^{\nicefrac{d}{2}} (\Fourier f)(\xi) \, (\Fourier g)(\xi) 
		, 
	\end{align*}
	we see that $f \ast g = (2 \pi)^{\nicefrac{d}{2}} \Fourier^{-1} \bigl ( \Fourier f \, \Fourier g \bigr ) \in \Schwartz(\R^d)$. 
\end{proof}
%

\section{Tempered distributions} 
\label{S_and_Sprime:Sprime}

Tempered distributions are linear functionals on $\Schwartz(\R^d)$. 
\begin{defn}[Tempered distributions]
	The tempered distributions $\Schwartz'(\R^d)$ are the continuous linear functions on the Schwartz functions $\Schwartz(\R^d)$. If $L \in \Schwartz'(\R^d)$ is a linear functional, we will often write 
	\begin{align*}
		\bigl ( L , f \bigr ) := L(f) 
		&&
		\forall f \in \Schwartz(\R^d) 
		. 
	\end{align*}
\end{defn}
\begin{example}
	\begin{enumerate}[(i)]
		\item The $\delta$ distribution defined via 
		\begin{align*}
			\delta(f) := f(0) 
		\end{align*}
		is a linear continuous functional on $\Schwartz(\R^d)$. (See exercise sheet~6.) 
		\item Let $g \in L^p(\R^d)$, $1 \leq p < \infty$, then for $f \in \Schwartz(\R^d)$, we define 
		\begin{align}
			L_g(f) = \int_{\R^d} \dd x \, g(x) \, f(x) =: \bigl ( g , f \bigr ) 
			. 
			\label{S_and_Sprime:Sprime:eqn:distributions_of_functions}
		\end{align}
		As $f \in \Schwartz(\R^d) \subset L^q(\R^d)$, $\frac{1}{p} + \frac{1}{q} = 1$, by Hölder's inequality, we have 
		\begin{align*}
			\babs{\bigl ( g , f \bigr )} \leq \norm{g}_p \, \norm{f}_q 
			. 
		\end{align*}
		Since $\norm{f}_q$ can be bounded by a finite linear combination of Fréchet seminorms of $f$, $L_g$ is continuous. 
	\end{enumerate}
\end{example}
Equation~\eqref{S_and_Sprime:Sprime:eqn:distributions_of_functions} is the \emph{canonical way to interpret less nice functions as distributions}: we identify a suitable function $g : \R^d \longrightarrow \C$ with the distribution $L_g$. For instance, polynomially bounded smooth functions (think of $h(x,\xi) = \frac{1}{2} \xi^2 + V(x)$) define continuous linear functionals in this manner since for any $h \in \Cont^{\infty}_{\mathrm{pol}}(\R^d)$, there exists $n \in \N_0$ such that $\sqrt{1 + x^2}^{-n} h(x)$ is bounded. Hence, for any $f \in \Schwartz(\R^d)$, Hölder's inequality yields 
\begin{align*}
	\babs{\bigl ( h , f \bigr )} &= \abs{\int_{\R^d} \dd x \, h(x) \, f(x)} 
	= \abs{\int_{\R^d} \dd x \, \sqrt{1 + x^2}^{-n} h(x) \, \sqrt{1 + x^2}^{n} f(x)} 
	\\
	&
	\leq \bnorm{\sqrt{1 + x^2}^{-n} h(x)}_{L^{\infty}} \, \bnorm{\sqrt{1 + x^2}^{n} f(x)}_{L^1} 
	. 
\end{align*}
Later on, we will see that this point of view, interpreting ``not so nice'' functions as distributions, helps us extend operations from test functions to much broader classes of functions. 

Similar to the case of normed spaces, we see that continuity implies ``boundedness.'' 
\begin{prop}
	A linear functional $L : \Schwartz(\R^d) \longrightarrow \C$ is a tempered distribution if and only if there exist constants $C > 0$ and $k,n \in \N_0$ such that 
	\begin{align*}
		\babs{L(f)} \leq C \sum_{\substack{\abs{a} \leq k \\
		\abs{\alpha} \leq n}} \norm{f}_{a \alpha} 
	\end{align*}
	for all $f \in \Schwartz(\R^d)$. 
\end{prop}
As mentioned before, we can interpret suitable functions $g$ as tempered distributions. In particular, every Schwartz function $g \in \Schwartz(\R^d)$ defines a tempered distribution so that 
\begin{align*}
	\bigl ( \partial_{x_k} g , f \bigr ) &= \int_{\R^d} \dd x \, \partial_{x_k} g(x) \, f(x) = - \int_{\R^d} \dd x \, g(x) \, \partial_{x_k} f(x) = \bigr ( g , - \partial_{x_k} f \bigr ) 
\end{align*}
holds for any $f \in \Schwartz(\R^d)$. We can use the right-hand side to \emph{define} derivatives of distributions: 
\begin{defn}[Weak derivative]
	For $\alpha \in \N_0^d$ and $L \in \Schwartz'(\R^d)$, we define the weak or distributional derivative of $L$ as 
	\begin{align*}
		\bigl ( \partial_x^{\alpha} L , f \bigr ) := \bigl ( L , (-1)^{\abs{\alpha}} \partial_x^{\alpha} f \bigr ) 
		, 
		&&
		\forall f \in \Schwartz(\R^d) 
		. 
	\end{align*}
\end{defn}
\begin{example}
	\begin{enumerate}[(i)]
		\item The weak derivative of $\delta$ is 
		\begin{align*}
			\bigl ( \partial_{x_k} \delta , f \bigr ) &= \bigl ( \delta , - \partial_{x_k} f \bigr ) 
			= - \partial_{x_k} f(0) 
			. 
		\end{align*}
		\item Let $g \in \Cont^{\infty}_{\mathrm{pol}}(\R^d)$. Then the weak derivative coincides with the usual derivative, by partial integration, we get 
		\begin{align*}
			\bigl ( \partial_{x_k} g , f \bigr ) &= - \bigl ( g , \partial_{x_k} f \bigr ) 
			= - \int_{\R^d} \dd x \, g(x) \, \partial_{x_k} f(x) 
			\\
			&
			= + \int_{\R^d} \dd x \, \partial_{x_k} g(x) \, f(x) 
			. 
		\end{align*}
	\end{enumerate}
\end{example}
Similarly,\marginpar{\small 2009.12.01} the Fourier transform can be extended to a bijection $\Schwartz'(\R^d) \longrightarrow \Schwartz'(\R^d)$. Theorem~\ref{S_and_Sprime:schwartz_functions:thm:unitarity_Fourier_on_S} tells us that if $g , f \in \Schwartz(\R^d)$, then 
\begin{align*}
	\bigl ( \Fourier g , f \bigr ) = \bigl ( g , \Fourier f \bigr ) 
\end{align*}
holds. If we replace $g$ with an arbitrary tempered distribution, the right-hand side again serves as \emph{definition} of the left-hand side: 
\begin{defn}[Fourier transform on $\mathcal{S}'(\R^d)$]
	For any tempered distribution $L \in \Schwartz'(\R^d)$, we define its Fourier transform to be 
	\begin{align*}
		\bigl ( \Fourier L , f \bigr ) := \bigl ( L , \Fourier f \bigr ) 
		&&
		\forall f \in \Schwartz(\R^d) 
		. 
	\end{align*}
\end{defn}
\begin{example}
	\begin{enumerate}[(i)]
		\item The Fourier transform of $\delta$ is the constant function $(2\pi)^{-\nicefrac{d}{2}}$, 
		\begin{align*}
			\bigl ( \Fourier \delta , f \bigr ) &= \bigl ( \delta , \Fourier f \bigr ) 
			= \Fourier f (0) 
			= \frac{1}{(2\pi)^{\nicefrac{d}{2}}}\int_{\R^d} \dd x \, f(x) 
			\\
			&= \bigl ( (2\pi)^{-\nicefrac{d}{2}} , f \bigr ) 
			. 
		\end{align*}
		\item The Fourier transform of $x^2$ makes sense as a tempered distribution on $\R$: $x^2$ is a polynomially bounded function and thus defines a tempered distribution via equation~\eqref{S_and_Sprime:Sprime:eqn:distributions_of_functions}: 
		\begin{align*}
			\bigl ( \Fourier x^2 , f \bigr ) &= \bigl ( x^2 , \Fourier f \bigr ) 
			= \int_{\R} \dd x \, x^2 \, \frac{1}{(2\pi)^{\nicefrac{1}{2}}} \int_{\R} \dd \xi \, e^{- i x \cdot \xi} \, f(\xi) 
			\\
			&= \frac{1}{(2\pi)^{\nicefrac{1}{2}}} \int_{\R} \dd \xi \int_{\R} \dd x \, (+i)^2 \partial_{\xi}^2 e^{- i x \cdot \xi} \, f(\xi) 
			\\
			&= (-1)^2 \cdot (-1) \, \int_{\R} \dd \xi \left ( \frac{1}{(2\pi)^{\nicefrac{1}{2}}} \int_{\R} \dd x \, e^{- i x \cdot \xi} \right ) \, \partial_{\xi}^2 f(\xi) 
			\\
			&= - \int_{\R} \dd \xi \, (2\pi)^{\nicefrac{1}{2}} \, \delta(\xi) \, \partial_{\xi}^2 f(\xi) 
			\\
			&= - \partial_{\xi}^2 f(0) 
			= \bigl ( (2\pi)^{\nicefrac{1}{2}} \delta , - \partial_{\xi}^2 f \bigr ) 
			= \bigl ( - (2\pi)^{\nicefrac{1}{2}} \delta'' , f \bigr ) 
		\end{align*}
		This is consistent with what we have shown earlier in Theorem~\ref{S_and_Sprime:schwartz_functions:thm:Fourier_is_bijection}, namely 
		\begin{align*}
			\Fourier \bigl ( x^2 f \bigr ) = (+ i \partial_{\xi}^2) \Fourier f = - \partial_{\xi}^2 \Fourier f 
			. 
		\end{align*}
	\end{enumerate}
\end{example}
We have just computed Fourier transforms of functions that do not have Fourier transforms in the usual sense. We can apply the idea we have used to define the derivative and Fourier transform on $\Schwartz'(\R^d)$ to other operators initially defined on $\Schwartz(\R^d)$. Before we do that though, we need to introduce the appropriate notion of continuity on $\Schwartz'(\R^d)$. 
\begin{defn}[Weak-$\ast$ convergence]
	Let $\Schwartz$ be a metric space with dual $\Schwartz'$. A sequence $(L_n)$ in $\Schwartz'$ is said to converge to $L \in \Schwartz'$ in the weak-$\ast$ sense if 
	\begin{align*}
		L_n(f) \xrightarrow{n \to \infty} L(f) 
	\end{align*}
	holds for all $f \in \Schwartz$. We will write $\wastlim_{n \to \infty} L_n = L$. 
\end{defn}
This notion of convergence implies a notion of continuity and is crucial for the next theorem. 
\begin{thm}\label{S_and_Sprime:Sprime:thm:weak-star_continuity}
	Let $A : \Schwartz(\R^d) \longrightarrow \Schwartz(\R^d)$ be a linear continuous map. Then for all $L \in \Schwartz'(\R^d)$, the map $A' : \Schwartz'(\R^d) \longrightarrow \Schwartz'(\R^d)$
	\begin{align}
		\bigl ( A' L , f \bigr ) := \bigl ( L , A f \bigr ) 
		, 
		&& 
		f \in \Schwartz(\R^d) 
		, 
	\end{align}
	defines a weak-$\ast$ continuous linear map. 
\end{thm}
Put in the terms of Chapter~\ref{operators:adjoint}, $A'$ is the \emph{adjoint} of $A$. 
\begin{proof}
	First of all, $A'$ is linear and well-defined, $A' L$ maps $f \in \Schwartz(\R^d)$ onto $\C$. To show continuity, let $(L_n)$ be a sequence of tempered distributions which converges in the weak-$\ast$ sense to $L \in \Schwartz'(\R^d)$. Then 
	\begin{align*}
		\bigl ( A' L_n , f \bigr ) &= \bigl ( L_n , A f \bigr ) 
		\xrightarrow{n \to \infty} \bigl ( L , A f \bigr ) = \bigl ( A' L , f \bigr ) 
	\end{align*}
	holds for all $f \in \Schwartz(\R^d)$ and $A'$ is weak-$\ast$ continuous. 
\end{proof}
As a last consequence, we can extend the convolution from $\ast : \Schwartz(\R^d) \times \Schwartz(\R^d) \longrightarrow \Schwartz(\R^d)$ to 
\begin{align*}
	\ast &: \Schwartz'(\R^d) \times \Schwartz(\R^d) \longrightarrow \Schwartz'(\R^d) 
	\\ 
	\ast &: \Schwartz(\R^d) \times \Schwartz'(\R^d) \longrightarrow \Schwartz'(\R^d) 
	. 
\end{align*}
For any $f , g , h \in \Schwartz(\R^d)$, we can push the convolution from one argument of the duality bracket to the other, 
\begin{align*}
	\bigl ( f \ast g , h \bigr ) &= \bigl ( g \ast f , h \bigr ) 
	= \int_{\R^d} \dd y \, (f \ast g)(y) \, h(y) 
	= \int_{\R^d} \dd y \int_{\R^d} \dd x \, f(x) \, g(y-x) \, h(y) 
	\\
	&
	= \int_{\R^d} \dd x \, f(x) \, (g(- \; \cdot) \ast h)(x) 
	= \bigl ( f , g(- \; \cdot) \ast h \bigr ) 
	. 
\end{align*}
Thus, we define 
\begin{defn}[Convolution on $\Schwartz'(\R^d)$]
	Let $L \in \Schwartz'(\R^d)$ and $g \in \Schwartz(\R^d)$. Then the convolution of $L$ and $g$ is defined as 
	\begin{align*}
		\bigl ( L \ast g , f \bigr ) := \bigl ( L , g(- \; \cdot) \ast f \bigr ) 
		&&
		\forall f \in \Schwartz(\R^d) 
		. 
	\end{align*}
\end{defn}
By Theorem~\ref{S_and_Sprime:Sprime:thm:weak-star_continuity}, this extension of the convolution is weak-$\ast$ continuous. 


\chapter{Weyl calculus} 
\label{weyl_calculus}
%
After much preparation, we are finally in a position to formulate the tenets of Weyl calculus. If all we have to quantize are functions of positions or momentum alone, life is simple: the kinetic energy function $T(\xi) = \frac{1}{2m} \xi^2$ can be quantized via the Fourier transform: according to Theorem~\ref{S_and_Sprime:schwartz_functions:thm:Fourier_is_bijection}: for any $\varphi \in \Schwartz(\R^d) \subset L^2(\R^d)$, we have 
\begin{align*}
	\tfrac{1}{2m} \xi^2 (\Fourier \varphi)(\xi) = \Fourier \bigl ( \tfrac{1}{2m} (- i \nabla_x)^2 \varphi \bigr )(\xi) 
\end{align*}
and we could define 
\begin{align*}
	\hat{T} \varphi := \Fourier^{-1} T(\hat{\xi}) \Fourier \varphi 
	. 
\end{align*}
The reason why this works is that different components of momenta commute. However, not all functions are of this form, \eg if one adds a magnetic field, we need to quantize $h(x,\xi) = \tfrac{1}{2m} \bigl ( \xi - A(x) \bigr )^2 + V(x)$. In other words, we need a quantization procedure that encapsulates the commutation relations 
\begin{align*}
	i \bigl [ \Qe_l , \Qe_j \bigr ] = 0 
	&& 
	i \bigl [ \Pe_l , \Pe_j \bigr ] = 0 
	&& 
	i \bigl [ \Pe_l , \Qe_j \bigr ] = \eps \delta_{lj} 
\end{align*}
of the building block operators position $\Qe$, 
\begin{align*}
	(\Qe \varphi)(x) := x \, \varphi(x) 
	, 
	&& 
	\varphi \in \Schwartz(\R^d) \subset L^2(\R^d) 
	, 
\end{align*}
and momentum $\Pe$
\begin{align*}
	(\Pe \varphi)(x) := - i \eps \nabla_x \varphi(x) 
	, 
	&& 
	\varphi \in \Schwartz(\R^d) \subset L^2(\R^d) 
	. 
\end{align*}
We have added once more the semiclassical parameter $\eps$, although for much of what we do, this parameter is optional. Before we reformulate the commutation relations in a different way, we need some notation: in the hamiltonian framework, the motion of a classical particle is on phase space $\R^d_x \times \R^d_{\xi} =: \pspace$. $x , y , z \in \R^d_x$ are the position variables associated to the momenta $\xi , \eta , \zeta \in \R^d_{\xi}$. We will often use the capital letters $X = (x,\xi) , Y = (y , \eta) , Z = (z , \zeta) \in \pspace$ to denote coordinates in phase space.

\section{The Weyl system} 
\label{weyl_calculus:weyl_system}
Instead of working with $\Qe$ and $\Pe$, each of which defines a selfadjoint operator on its maximal domain 
\begin{align*}
	\mathcal{D}(\Qe_l) &= \bigl \{ \varphi \in L^2(\R^d) \; \vert \; \Qe_l \varphi \in L^2(\R^d) \bigr \} \supset \Schwartz(\R^d) 
	\\
	\mathcal{D}(\Pe_l) &= \bigl \{ \varphi \in L^2(\R^d) \; \vert \; \Pe_l \varphi \in L^2(\R^d) \bigr \} \supset \Schwartz(\R^d) 
	, 
\end{align*}
we work with the corresponding evolution groups, namely translations in momentum 
\begin{align*}
	V(\eta) := e^{- i \eta \cdot \Qe} , \; \bigl ( V(\eta) \varphi \bigr )(x) = e^{- i \eta \cdot x} \varphi(x) 
	, 
	&&
	\varphi \in L^2(\R^d) 
	, 
\end{align*}
and position 
\begin{align*}
	U(y) := e^{- i y \cdot \Pe} , \; \bigl ( U(y) \varphi \bigr )(x) = \varphi(x - \eps y) 
	, 
	&&
	\varphi \in L^2(\R^d) 
	. 
\end{align*}
These are not \emph{one}-parameter groups, but \emph{strongly continuous} $d$-parameter groups: one can check that 
\begin{align*}
	V(\xi) V(\eta) &= V(\xi + \eta) 
	, 
	&& 
	V(\xi)^* = V(- \xi) 
	, 
	\\
	U(x) U(y) &= U(x + y) 
	, 
	&&
	U(x)^* = U(- x) 
	, 
\end{align*}
hold as all the different components of position and momentum commute. But what about combinations of $U$ and $V$? Clearly the origin of the fact 
\begin{align*}
	\bigl ( U(y) V(\eta) \varphi \bigr )(x) &= \bigl ( V(\eta) \varphi \bigr ) (x - \eps y) 
	= e^{- i \eta \cdot ( x - \eps y)} \varphi(x - \eps y) 
	\\
	&\neq e^{- i \eta \cdot x} \varphi(x - \eps y) 
	= \bigl ( V(\eta) U(y) \varphi \bigr ) (x) 
\end{align*}
can be traced back to the noncommutativity of the generators of $U$ and $V$. Instead, we have just shown 
\begin{align*}
	U(y) V(\eta) = e^{+ i \eps y \cdot \eta} \, V(\eta) U(y) 
	&&
	\forall y \in \R^d_x , \eta \in \R^d_{\xi} 
	. 
\end{align*}
It turns out that picking an operator ordering can be rephrased into picking an ordering of $U$ and $V$. Since we would like real-valued functions to be quantized to potentially selfadjoint operators, we choose the ``symmetrically'' defined 
\begin{defn}[Weyl system]
	For all $X \in \pspace$, we define 
	\begin{align*}
		\WeylSys(X) := e^{- i (\xi \cdot \Qe - x \cdot \Pe)} =: e^{- i \sigma((x,\xi),(\Qe,\Pe))}
	\end{align*}
	where $\sigma(X,Y) := \xi \cdot y - x \cdot \eta$. 
\end{defn}
We note that the small parameter $\eps$ is contained in the definition of the operators $\Qe$ and $\Pe$. The next Lemma tells us how this operator acts on wave functions: 
\begin{lem}\label{weyl_calculus:weyl_system:lem:application_Weyl_system}
	For all $Y \in \pspace$  and $\varphi \in L^2(\R^d)$, $\bigl ( \WeylSys(Y) \varphi \bigr )(x) = e^{- i \eta \cdot (x + \frac{\eps}{2} y)} \varphi(x + \eps y)$ holds. The map $Y \mapsto \WeylSys(Y)$ is strongly continuous. 
\end{lem}
\begin{proof}
	We use the Trotter product formula (\cite[Theorem~VIII.31]{Reed_Simon:bibel_1:1981}) to write $\WeylSys(Y)$ as 
	\begin{align*}
		\slim_{n \to \infty} &\Bigl ( e^{- \frac{i}{n} \eta \cdot \Qe} e^{+ \frac{i}{n} y \cdot \Pe} \Bigr )^n 
		= \\
		&
		= \slim_{n \to \infty} \Bigl ( e^{- \frac{i}{n} \eta \cdot \Qe} e^{+ \frac{i}{n} y \cdot \Pe} e^{- \frac{i}{n} \eta \cdot \Qe} e^{+ \frac{i}{n} y \cdot \Pe} \cdots e^{- \frac{i}{n} \eta \cdot \Qe} e^{+ \frac{i}{n} y \cdot \Pe} e^{- \frac{i}{n} \eta \cdot \Qe} e^{+ \frac{i}{n} y \cdot \Pe}  \Bigr ) 
		\\
		&= \slim_{n \to \infty} \Bigl ( e^{- \frac{i}{n} \eta \cdot \Qe} e^{+ \frac{i}{n} y \cdot \Pe} e^{- \frac{i}{n} \eta \cdot \Qe} e^{- \frac{i}{n} y \cdot \Pe} e^{+ i \frac{2}{n} y \cdot \Pe} e^{- \frac{i}{n} \eta \cdot \Qe} e^{- i \frac{2}{n} y \cdot \Pe} e^{+ i \frac{3}{n} y \cdot \Pe} 
		\cdots \\
		&\qquad \qquad \qquad \cdots 
		e^{+ i \frac{n-1}{n} y \cdot \Pe} e^{- \frac{i}{n} \eta \cdot \Qe} e^{- i \frac{n-1}{n} y \cdot \Pe} e^{+ i \frac{n}{n} y \cdot \Pe} \Bigr ) 
		. 
	\end{align*}
	The terms can be simplified using 
	\begin{align*}
		\Bigl ( e^{+ i \frac{k}{n} y \cdot \Pe} e^{- \frac{i}{n} \eta \cdot \Qe} e^{- i \frac{k}{n} y \cdot \Pe} \varphi \Bigr )(x) &= \Bigl ( e^{- \frac{i}{n} \eta \cdot \Qe} e^{- i \frac{k}{n} y \cdot \Pe} \varphi \Bigr ) \bigl (x + \eps \tfrac{k}{n} y \bigr ) 
		\\
		&
		= e^{- \frac{i}{n} \eta \cdot (x + \eps \frac{k}{n} y)} \, \bigl ( e^{- i \frac{k}{n} y \cdot \Pe} \varphi \bigr ) \bigl (x + \eps \tfrac{k}{n} y \bigr ) 
		\\
		&= e^{- \frac{i}{n} \eta \cdot (x + \eps \frac{k}{n} y)} \, \varphi(x) 
		= \bigl ( e^{- \frac{i}{n} \eta \cdot (\Qe + \eps \frac{k}{n} y)} \varphi \bigr )(x) 
	\end{align*}
	where $k \in \{ 0 , 1 , \ldots , n - 1 \}$. This means the above expression can be rewritten as 
	\begin{align*}
		\slim_{n \to \infty} \Bigl ( e^{- i \frac{1}{n} \eta \cdot \Qe} &e^{- i \tfrac{1}{n} \eta \cdot (\Qe + \eps \frac{1}{n} y)} \cdots e^{- i \frac{1}{n} \eta \cdot (\Qe + \eps \frac{n-1}{n} y)} e^{+ i y \cdot \Pe} \Bigr ) 
		= \\
		&
		= \slim_{n \to \infty} e^{- i \frac{1}{n} \eta \cdot \mbox{$\sum_{k = 0}^{n - 1}$} (\Qe + \eps \frac{k}{n} y)} e^{+ i y \cdot \Pe} 
		. 
	\end{align*}
	The sum in the exponential is a Riemann sum that converges to 
	\begin{align*}
		\frac{1}{n} \sum_{k = 0}^{n-1} \bigl ( x + \eps \tfrac{k}{n} y \bigr ) \xrightarrow{n \to \infty} \int_0^1 \dd s \, \bigl ( x + s \eps y \bigr ) = x + \tfrac{\eps}{2} y
	\end{align*}
	so that we get 
	\begin{align*}
		\WeylSys(Y) = e^{- i \eta \cdot ( \Qe + \frac{\eps}{2} y)} e^{+ i y \cdot \Pe} 
		. 
	\end{align*}
	The strong continuity of $Y \mapsto \WeylSys(Y)$ follows from the strong continuity of $y \mapsto U(y)$ and $\eta \mapsto V(\eta)$. 
\end{proof}
\begin{remark}
	Another way to see this is via the Baker-Campbell-Hausdorff formula: formally, we have 
	\begin{align*}
		\WeylSys(Y) = e^{- i (\eta \cdot \Qe - y \cdot \Pe)} &= e^{- i \eta \cdot \Qe} e^{+ i y \cdot \Pe} e^{- \frac{i^2}{2} [\eta \cdot \Qe , - y \cdot \Pe]} 
		= e^{- i \eta \cdot \Qe} e^{+ i y \cdot \Pe} e^{- i \frac{\eps}{2} \eta \cdot y} 
		\\
		&= e^{- i \eta \cdot ( \Qe + \frac{\eps}{2} y)} e^{+ i y \cdot \Pe} 
		. 
	\end{align*}
	However, it is a bit tricky to rigorously control commutators of unbounded operators \cite[Chapter~VIII.5]{Reed_Simon:bibel_1:1981}. \marginpar{\small 2009.12.08}
\end{remark}
The Weyl system combines translations in real and reciprocal space. If one of the two components is $0$, then $\WeylSys$ reduces to $U$ or $V$, 
\begin{align*}
	\WeylSys(x,0) &= U(-x) = U(x)^* 
	, \\
	\WeylSys(0,\xi) &= V(\xi) 
	. 
\end{align*}
The Weyl system completely encodes the commutation relations: 
\begin{prop}\label{weyl_calculus:weyl_system:prop:composition_law_weyl_system}
	Let $Y , Z \in \pspace$. Then, the Weyl system obeys the following composition law: 
	\begin{align*}
		\WeylSys(Y) \WeylSys(Z) = e^{i \frac{\eps}{2} \sigma(Y,Z)} \, \WeylSys(Y + Z) 
	\end{align*}
\end{prop}
This contains the commutation relations: if we combine a translation along $e_k$ in real space and $e_j$ in reciprocal space, 
\begin{align*}
	\WeylSys(e_k , 0) \WeylSys(0 , e_j) = e^{- i \frac{\eps}{2} \delta_{kj}} \, \WeylSys(e_k , e_j) 
	, 
\end{align*}
then the extra factor reduces to the exponential of half of $[\Qe_k , \Pe_j] = - i \eps \delta_{kj}$. 
\begin{proof}
	This follows from direct calculation: for all $\varphi \in L^2(\R^d)$, we have 
	\begin{align*}
		\bigl ( \WeylSys(Y) \WeylSys(Z) \varphi \bigr )(x) &= e^{- i (x + \frac{\eps}{2} y) \cdot \eta} \bigl ( \WeylSys(Z) \varphi \bigr )(x + \eps y) 
		\\
		&
		= e^{- i (x + \frac{\eps}{2} y) \cdot \eta} e^{- i (x + \eps y + \frac{\eps}{2} z) \cdot \zeta} \varphi(x + \eps y + \eps z) 
		\\
		&= e^{i \frac{\eps}{2} (- y \cdot \zeta + z \cdot \eta)} e^{- i (x + \frac{\eps}{2}(y + z)) \cdot (\eta + \zeta)} \varphi \bigl ( x + \eps (y + z) \bigr ) 
		\\
		&
		= e^{i \frac{\eps}{2} \sigma(Y,Z)} \, \bigl ( \WeylSys(Y + Z) \varphi \bigr )(x) 
		. 
	\end{align*}
	Hence, the claim follows. 
\end{proof}
%

\section{Weyl quantization} 
\label{weyl_calculus:weyl_quantization}

To define the Weyl quantization, we first introduce a convenient variant of the Fourier transform, the symplectic Fourier transform on $\pspace$ 
\begin{align*}
	(\Fs f)(X) := \frac{1}{(2\pi)^d} \int_{\pspace} \dd Y \, e^{i \sigma(X,Y)} \, f(Y) 
	, 
	&& f \in \Schwartz(\pspace)
	. 
\end{align*}
$\Fs$ has the nice property of being an involution, \ie $\Fs^2 = \id_{\Schwartz}$. If we pretend for a moment that $\Qe$ and $\Pe$ are \emph{variables} and \emph{not} operators, we have 
\begin{align*}
	f(\Qe,\Pe) &= (\Fs^2 f)(\Qe,\Pe) 
	= \frac{1}{(2 \pi)^d} \int_{\pspace} \dd X \, e^{i \sigma((\Qe,\Pe),X)} \, \frac{1}{(2\pi)^d} \int_{\pspace} \dd Y \, e^{i \sigma(X,Y)} \, f(Y) 
	\\
	&
	= \frac{1}{(2\pi)^d} \int_{\pspace} \dd X \, (\Fs f)(X) \, e^{- i \sigma(X,(\Qe,\Pe))} 
	\\
	&= \frac{1}{(2\pi)^{2d}} \int_{\pspace} \dd X \int_{\pspace} \dd Y \, e^{i \sigma(X,Y-(\Qe,\Pe))} \, f(Y) 
	= \int_{\pspace} \dd Y \, \delta \bigl ( Y - (\Qe,\Pe) \bigr ) \, f(Y) 
	. 
\end{align*}
Even though the above does not make sense if $\Qe$ and $\Pe$ are operators, it gives an intuition why we define Weyl quantization the way we do: 
\begin{defn}[Weyl quantization]
	Let $f \in \Schwartz(\pspace)$. Then the Weyl quantization of $f$ is defined as 
	\begin{align*}
		\Op(f) \varphi := \frac{1}{(2\pi)^d} \int_{\pspace} \dd X \, (\Fs f)(X) \, \WeylSys(X) \varphi 
		, 
		&& 
		\forall \varphi \in \Schwartz(\R^d) 
		. 
	\end{align*}
\end{defn}
Lemma~\ref{weyl_calculus:weyl_system:lem:application_Weyl_system} essentially already tells us how $\Op(f)$ acts on wave functions: 
\begin{prop}\label{weyl_calculus:weyl_quantization:prop:application_Weyl_quantization}
	Let $f \in \Schwartz(\pspace)$ and $\varphi \in L^2(\R^d)$. Then $\Op(f)$ defines a bounded operator on $L^2(\R^d)$ and its action on $\varphi$ is given by 
	\begin{align}
		\bigl ( \Op(f) \varphi \bigr )(x) &= \frac{1}{(2 \pi)^{d}} \int_{\R^d_x} \dd y \int_{\R^d_{\xi}} \dd \eta \, e^{- i(y - x) \cdot \eta} \, f \bigl ( \tfrac{1}{2}(x+y) , \eps \eta \bigr ) \, \varphi(y) 
		\label{weyl_calculus:weyl_system:eqn:application_Weyl_quantization} \\
		&= \frac{1}{(2 \pi)^{\nicefrac{d}{2}}} \int_{\R^d_x} \dd y \, \eps^{-d} \, (\Fourier_2 f) \bigl ( \tfrac{1}{2}(x+y) , \tfrac{y-x}{\eps} \bigr ) \, \varphi(y) 
		\notag \\
		&=: \frac{1}{(2\pi)^{\nicefrac{d}{2}}} \int_{\R^d_x} \dd y \, K_f(x,y) \, \varphi(y) 
		=: \bigl ( \mathfrak{Int}(K_f) \varphi \bigr )(x)
		. 
		\notag 
	\end{align}
\end{prop}
\begin{proof}
	Let $\varphi \in \Schwartz(\R^d) \subset L^2(\R^d)$. Then we can bound $\Op(f) \varphi$ by 
	\begin{align*}
		\bnorm{\Op(f) \varphi} &\leq \frac{1}{(2\pi)^d} \int_{\pspace} \dd X \, \babs{(\Fs f)(X)} \, \bnorm{\WeylSys(X) \varphi} 
		= \left ( \frac{1}{(2\pi)^d} \int_{\pspace} \dd X \, \babs{(\Fs f)(X)} \right ) \, \bnorm{\varphi} 
		\\
		&= (2\pi)^{-d} \, \norm{\Fs f}_{L^1(\pspace)} \, \norm{\varphi} 
		. 
	\end{align*}
	The $L^1$-norm of $\Fs f$ is certainly bounded as $\Fs f$ is a Schwartz function and thus integrable. Hence, $\Op(f)$ defines a bounded operator on $\Schwartz(\R^d) \subset L^2(\R^d)$. Since the Schwartz functions are dense in $L^2(\R^d)$ by Lemma~\ref{S_and_Sprime:schwartz_functions:thm:S_dense_in_Lp}, we invoke Theorem~\ref{operators:bounded:thm:extensions_bounded_operators} which ensures the existence of a unique extension of $\Op(f)$ to all of $L^2(\R^d)$. 
	
	Equation~\eqref{weyl_calculus:weyl_system:eqn:application_Weyl_quantization} follows from direct computation and Lemma~\ref{weyl_calculus:weyl_system:lem:application_Weyl_system}: for any $\varphi \in L^2(\R^d)$, we have 
	\begin{align*}
		\bigl ( \Op(f) \varphi \bigr )(x) &= \frac{1}{(2\pi)^{2d}} \int_{\pspace} \dd Y \int_{\pspace} \dd Z \, e^{i \sigma(Y,Z)} \, f(Z) \, \bigl ( \WeylSys(Y) \varphi \bigr )(x) 
		\\
		&
		= \frac{1}{(2\pi)^{2d}} \int_{\R^d_x} \dd y \int_{\R^d_{\xi}} \dd \eta \int_{\R^d_x} \dd z \int_{\R^d_{\xi}} \dd \zeta \, e^{i ( \eta \cdot z - y \cdot \zeta )} \, f(z,\zeta) 
		\cdot \\
		&\qquad \qquad \qquad \qquad \qquad \qquad \qquad \qquad \cdot 
		e^{- i \eta \cdot (x + \frac{\eps}{2} y)} \, \varphi \bigl ( x + \tfrac{\eps}{2} y \bigr ) 
		\\
		&
		= \frac{\eps^{-d}}{(2\pi)^{2d}} \int_{\R^d_x} \dd y \int_{\R^d_{\xi}} \dd \eta \int_{\R^d_x} \dd z \int_{\R^d_{\xi}} \dd \zeta \, e^{i \eta \cdot (z - \frac{1}{2}(x+y))} e^{- \frac{i}{\eps}(y - x) \cdot \zeta} \, f(z,\zeta) \, \varphi(y) 
	\end{align*}
	If we integrate over $\eta$, we get a $\delta$ that can be used to kill the integral with respect to $z$. 
	\begin{align*}
		\bigl ( \Op(f) \varphi \bigr )(x) 
		&
		= \frac{\eps^{-d}}{(2\pi)^{2d}} \int_{\R^d_x} \dd y \int_{\R^d_x} \dd z \int_{\R^d_{\xi}} \dd \zeta \, (2 \pi)^d \delta \bigl ( z - \tfrac{1}{2}(x+y) \bigr ) \, e^{- \frac{i}{\eps}(y - x) \cdot \zeta} 
		\cdot \\
		&\qquad \qquad \qquad \qquad \qquad \qquad \quad \cdot 
		f(z,\zeta) \, \varphi(y) 
		\\
		&
		= \frac{1}{(2\pi \eps)^{d}} \int_{\R^d_x} \dd y \int_{\R^d_{\xi}} \dd \zeta \, e^{- \frac{i}{\eps}(y - x) \cdot \zeta} \, f \bigl ( \tfrac{1}{2}(x+y) , \zeta \bigr ) \, \varphi(y) 
	\end{align*}
	Identifying the integral with respect to $\zeta$ as partial Fourier transform, we get the second line of equation~\eqref{weyl_calculus:weyl_system:eqn:application_Weyl_quantization}. A variable substitution $y' := \eps y$ yields the first line. 
\end{proof}
\begin{example}
	Although we cannot justify this example rigorously yet, we will only do a quick sanity check whether the quantization of $h(x,\xi) = \frac{1}{2m} \xi^2 + V(x)$ gives the expected result. By linearity, we can consider each of the terms in turn: we have already computed the Fourier transform of $\frac{1}{2m} \xi^2$ in the distributional sense in Chapter~\ref{S_and_Sprime:Sprime}, for all $\varphi \in \Schwartz(\R^d)$ 
	\begin{align*}
		\bigl ( \Op \bigl (\tfrac{1}{2m} \xi^2 \bigr ) \varphi \bigr )(x) &= \frac{1}{(2 \pi)^{\nicefrac{d}{2}}} \int_{\R^d_x} \dd y \, \frac{1}{2m} \, \bigl ( \Fourier (\eps^2 \eta^2) \bigr )(y - x) \, \varphi(y) 
		\\
		&
		= \frac{1}{2m} \frac{1}{(2 \pi)^{\nicefrac{d}{2}}} \int_{\R^d_x} \dd y \, (2\pi)^{\nicefrac{d}{2}} (+i)^2 \eps^2 \Delta_y \delta (y - x) \, \varphi(y) 
		\\
		&= \tfrac{1}{2m} \bigl ( (+i)^2 \eps^2 \Delta_x \delta_x , \varphi \bigr ) 
		= \tfrac{1}{2m} (-i)^2 \eps^2 \Delta_x \varphi(x) 
		= - \tfrac{\eps^2}{2m} \Delta_x \varphi(x) 
	\end{align*}
	holds. We see that this calculation hinges on $\varphi$ being a test function. The other term can also be calculated using $\Fourier e^{+i x \cdot \eta} = (2\pi)^d \delta_x$ where $\delta_x(f) := f(x)$ is the shifted Dirac distribution, 
	\begin{align*}
		\bigl ( \Op(V) \varphi \bigr )(x) &= \frac{1}{(2\pi)^{\nicefrac{d}{2}}} \int_{\R^d_x} \dd y \, (\Fourier \, 1)(y - x) \, V \bigl ( \tfrac{1}{2} (x + y) \bigr ) \, \varphi(y) 
		\\
		&= \frac{1}{(2\pi)^{\nicefrac{d}{2}}} \int_{\R^d_x} \dd y \, (2\pi)^{\nicefrac{d}{2}} \delta(y - x) \, V \bigl ( \tfrac{1}{2} (x + y) \bigr ) \, \varphi(y) 
		= V(x) \, \varphi(x) 
		. 
	\end{align*}
	Hence, the quantization of $h$ yields 
	\begin{align*}
		\Op(h) = \tfrac{1}{2m} (- i \eps \nabla_x)^2 + V(\hat{x}) 
		= \tfrac{1}{2m} \Pe^2 + V(\Qe) 
		, 
	\end{align*}
	exactly what we have expected. We see that it was crucial for this argument to work that $\varphi \in \Schwartz(\R^d)$ and we have to work to extend the integral formula to other $\varphi \in L^2(\R^d) \setminus \Schwartz(\R^d)$. If $V$ is bounded, for instance, then $\Op(V) = V(\Qe)$ defines a bounded operator and we can use Theorem~\ref{operators:bounded:thm:extensions_bounded_operators} to extend it to all of $L^2(\R^d)$. $\Pe^2$, however, is an unbounded operator and thus does not extend to a \emph{bounded} operator on all of $L^2(\R^d)$
\end{example}
In Proposition~\ref{weyl_calculus:weyl_quantization:prop:application_Weyl_quantization}, we have introduced the notion of operator kernel: the operator kernel is a function or distribution that tells us how the operator acts on wave functions. It can be useful in calculating things, \eg if one wants to compute a trace. 

\emph{Every} bounded operator and many unbounded ones have a distributional operator kernel. In general, it is not a function, but only distribution on $\R^d \times \R^d$. The operator kernel of $- i \partial_{x_l}$ is $+ i (2\pi)^{\nicefrac{d}{2}} \partial_{x_l} \delta(x-y)$ since for all $\varphi \in \Schwartz(\R^d) \subset L^2(\R^d)$ 
\begin{align*}
	(- i \partial_{x_l} \varphi)(x) &= \bigl ( (+i) \partial_{x_l} \delta_x , \varphi \bigr ) 
	= \frac{1}{(2\pi)^{\nicefrac{d}{2}}} \int_{\R^d_x} \dd y \, i (2\pi)^{\nicefrac{d}{2}} \partial_{x_l} \delta(x-y) \, \varphi(y) 
	\\
	&
	= \Bigl ( \mathfrak{Int} \bigl ( + i (2\pi)^{\nicefrac{d}{2}} \partial_{x_l} \delta(x-y) \bigr ) \varphi \Bigr )(x) 
\end{align*}
holds. Similarly, the operator kernel associated to $T = \sopro{\psi}{\varphi} = \scpro{\varphi}{\cdot \, } \, \psi$, $\varphi , \psi \in \Schwartz(\R^d)$ is 
\begin{align*}
	(T \phi)(x) &= \sscpro{\varphi}{\phi} \, \psi(x) 
	= \frac{1}{(2\pi)^{\nicefrac{d}{2}}} \int_{\R^d_x} \dd y \, (2\pi)^{\nicefrac{d}{2}} \, \sscpro{\varphi}{\phi} \, \delta(x-y) \, \psi(y) 
	\\
	&
	= \frac{1}{(2\pi)^{\nicefrac{d}{2}}} \int_{\R^d_x} \dd y \int_{\R^d_x} \dd z \, (2\pi)^{\nicefrac{d}{2}} \, \varphi^*(z) \, \phi(z) \, \delta(x-y) \, \psi(y) 
	\\
	&= \frac{1}{(2\pi)^{\nicefrac{d}{2}}} \int_{\R^d_x} \dd z \biggl ( (2\pi)^{\nicefrac{d}{2}} \int_{\R^d_x} \dd y \, \varphi^*(z) \, \delta(x-y) \, \psi(y) \biggr ) \, \phi(z) 
	\\
	&
	=: \frac{1}{(2\pi)^{\nicefrac{d}{2}}} \int_{\R^d} \dd z \, K_T(x,z) \, \phi(z) 
	. 
\end{align*}
Hence, even this well-behaved operator has a distributional operator kernel! \marginpar{\small 2009.12.09}

To get back to the properties of quantization procedures: Weyl quantization is \emph{linear}, \ie 
\begin{align*}
	\Op(f + \alpha g) = \Op(f) + \alpha \Op(g) 
\end{align*}
holds for all $f , g \in \Schwartz(\pspace)$ and $\alpha \in \C$. Furthermore, we can compute the quantization of the constant function $1$ to be 
\begin{align*}
	\bigl ( \Op(f) \varphi \bigr )(x) &= \frac{\eps^{-d}}{(2\pi)^{\nicefrac{d}{2}}} \int_{\R^d_x} \dd y \, (\Fourier_2 \, 1) \bigl ( \tfrac{1}{2}(x + y) , \tfrac{y - x}{\eps} \bigr ) \, \varphi(y) 
	\\
	&= \frac{\eps^{-d}}{(2\pi)^{\nicefrac{d}{2}}} \int_{\R^d_x} \dd y \, \eps^d \, (2 \pi)^{\nicefrac{d}{2}} \, \delta(y - x) \, \varphi(y) 
	\\
	&= \varphi(x) 
	= \bigl ( \id_{L^2} \varphi \bigr )(x) 
	. 
\end{align*}
In the tutorials, we will show that Weyl quantization intertwines complex conjugation and taking adjoints, 
\begin{align*}
	\Op(f^*) = \Op(f)^* 
	. 
\end{align*}
This fact has the important consequence that real-valued functions are potentially mapped onto selfadjoint operators. 
\begin{thm}\label{weyl_calculus:weyl_quantization:thm:properties_weyl_quantization}
	Weyl quantization is linear, maps the constant function $1$ to $\id_{L^2}$ and intertwines complex conjugation with taking adjoints in the sense of operators, \ie $\Op(f^*) = \Op(f)^*$ holds for all $f \in \Schwartz(\pspace)$. 
\end{thm}
%

\section{The Wigner transform} 
\label{weyl_calculus:weyl_calculus:wigner_transform}

One natural question is whether we can rewrite the quantum expectation value $\bscpro{\varphi}{\Op(f) \psi}$ as a phase space average of $f$ with respect to some measure $\mu_{\varphi,\psi}$, 
\begin{align}
	\sscpro{\varphi}{\Op(f) \psi} &= \frac{1}{(2\pi)^d} \int_{\pspace} \dd X \, (\Fs f)(X) \, \bscpro{\varphi}{\WeylSys(X) \psi} 
	\label{weyl_calculus:wigner_transform:eqn:phase_space_exp_value} 
	\\
	&= \frac{1}{(2\pi)^d} \int_{\pspace} \dd X \, f(X) \, \Fs \bigl ( \bscpro{\varphi}{\WeylSys(\cdot) \psi} \bigr )(-X) 
	, 
	\notag 
\end{align}
suggests to look at the symplectic Fourier transform of the expectation value of the Weyl system. Let us start with the first building block: 
\begin{defn}[Fourier-Wigner transform]\label{weyl_calculus:wigner_transform:defn:fourier_wigner_trafo}
	Let $\psi,\varphi \in \Schwartz(\R^d)$. Then we define the Fourier-Wigner transform $\rho(\psi,\varphi)$ to be 
	\begin{align}
		\rho(\psi,\varphi)(X) := (2\pi)^{-\nicefrac{d}{2}} \, \bscpro{\varphi}{\WeylSys(X) \psi}
	\end{align}
\end{defn}
\begin{lem}
	Let $\psi,\varphi \in \Schwartz(\R^d)$. Then it holds 
	\begin{align*}
		\rho(\psi,\varphi)(X) = (2\pi)^{-\nicefrac{d}{2}} \, \bscpro{\varphi}{\WeylSys(X) \psi} 
		&= \frac{1}{(2\pi)^{\nicefrac{d}{2}}} \int_{\R^d_x} \dd y \, e^{-i y \cdot \xi} \, {\varphi}^{\ast} \bigl ( y - \tfrac{\eps}{2} x \bigr ) \, \psi \bigl ( y + \tfrac{\eps}{2} x \bigr ) 
		. 
	\end{align*}
\end{lem}
\begin{proof}
	By formal manipulation, we get 
	\begin{align*}
		\rho(\psi,\varphi)(X) &= (2\pi)^{-\nicefrac{d}{2}} \, \bscpro{\varphi}{\WeylSys(X) \psi} 
		= \frac{1}{(2\pi)^{\nicefrac{d}{2}}} \int_{\R^d_x} \dd y \, {\varphi}^{\ast}(y) \, \bigl ( \WeylSys(X) \psi \bigr )(y) \\
		&= \frac{1}{(2\pi)^{\nicefrac{d}{2}}} \int_{\R^d_x} \dd y \, {\varphi}^{\ast}(y) \, e^{-i \xi \cdot (y + \frac{\eps}{2} x)} \, \psi(y + \eps x) \\ 
		&= \frac{1}{(2\pi)^{\nicefrac{d}{2}}} \int_{\R^d_x} \dd y \, e^{-i y \cdot \xi} \, {\varphi}^{\ast} \bigl ( y - \tfrac{\eps}{2} x \bigr ) \, \psi \bigl ( y + \tfrac{\eps}{2} x \bigr ) 
		. 
	\end{align*}
	Since $\psi$ and $\varphi$ are Schwartz functions, they are also square integrable and the right-hand side exists for all $X \in \pspace$. 
\end{proof}
To write the quantum expectation value as a phase space average, we still have to push over the Fourier transform. 
\begin{defn}[Wigner transform]
	Let $\psi,\varphi \in \Schwartz(\R^d)$. The Wigner transform $\WignerTrafo(\psi,\varphi)$ is defined as the symplectic Fourier transform of $\rho(\psi,\varphi)$. 
	\begin{align*}
		\WignerTrafo(\psi,\varphi)(X) := \bigl ( \Fs \rho(\psi,\varphi) \bigr )(-X) 
		. 
	\end{align*}
\end{defn}
\begin{remark}
	There is a reason why we need to use $\WeylSys(-X)$ and not $\WeylSys(+X)$: the symplectic Fourier transform is unitary on $L^2(\pspace)$ and 
	\begin{align*}
		\bscpro{f}{g}_{L^2(\pspace)} = \bscpro{\Fs f}{\Fs g}_{L^2(\pspace)} = \bigl ( (\Fs f)^* , \Fs g \bigr ) = \bigl ( (\Fs f^*)(- \, \cdot) , \Fs g \bigr )
	\end{align*}
	holds. The extra sign stems from the fact that we are missing complex conjugation in integral~\eqref{weyl_calculus:wigner_transform:eqn:phase_space_exp_value}. 
\end{remark}
\begin{lem}\label{weyl_calculus:wigner_transform:lem:WignerTransform}
	The Wigner transform $\WignerTrafo(\psi,\varphi)$ with respect to $\psi,\varphi \in \Schwartz(\R^d)$ is given by 
	\begin{align*}
		\WignerTrafo(\psi,\varphi)(X) = \frac{1}{(2\pi)^{\nicefrac{d}{2}}} \int_{\R^d_x} \dd y \, e^{- i y \cdot \xi} \, {\varphi}^{\ast} \bigl ( x - \tfrac{\eps}{2} y \bigr ) \, \psi \bigl ( x + \tfrac{\eps}{2} y \bigr ) 
		. 
	\end{align*}
\end{lem}
\begin{proof}
	We plug in the definition of the symplectic Fourier transform and obtain 
	\begin{align*}
		\WignerTrafo(\psi,\varphi)(X) &= \frac{1}{(2\pi)^d} \int_{\pspace} \dd Y \, e^{i \sigma(-X,Y)} \rho(\psi,\varphi) (Y) \\ 
		&= \frac{1}{(2\pi)^{\nicefrac{3d}{2}}} \int_{\pspace} \dd Y \, \int_{\R^d_x} \dd z \, e^{- i (\xi \cdot y - x \cdot \eta)} e^{- i z \cdot \eta} \, \varphi^{\ast} \bigl ( z - \tfrac{\eps}{2} y \bigr ) \, \psi \bigl ( z + \tfrac{\eps}{2} y \bigr ) \\ 
		&= \frac{1}{(2\pi)^{\nicefrac{3d}{2}}} \int_{\R^d_x} \dd y \, \int_{\R^d_x} \dd z \, \int_{\R^d_{\xi}} \dd \eta \, e^{i (x-z) \cdot \eta} e^{- i y \cdot \xi} \, \varphi^{\ast} \bigl ( z - \tfrac{\eps}{2} y \bigr ) \, \psi \bigl ( z + \tfrac{\eps}{2} y \bigr ) \\
		&= \frac{1}{(2\pi)^{\nicefrac{d}{2}}} \int_{\R^d_x} \dd y \, e^{- i y \cdot \xi} \, \varphi^{\ast} \bigl ( x - \tfrac{\eps}{2} y \bigr ) \, \psi \bigl ( x + \tfrac{\eps}{2} y \bigr ) 
		. 
	\end{align*}
	The right-hand side is also well-defined as $\varphi^{\ast} \bigl ( x - \tfrac{\eps}{2} \cdot \bigr ) \, \psi \bigl ( x + \tfrac{\eps}{2} \cdot \bigr ) \in L^1(\R^d)$ for any $x \in \R^d$ and hence its Fourier transform exists. 
\end{proof}
\begin{example}\label{weyl_calculus:wigner_transform:ex:first_excited_state_harmonic_oscillator}
	Take $d = 1$, $\eps = 1$ and consider $\varphi(x) = x \, e^{- \frac{x^2}{4}}$, for instance, the first excited state of the harmonic oscillator. Using 
	\begin{align*}
		(\Fourier e^{- \alpha x^2})(\xi) = \frac{1}{\sqrt{2 \alpha}} e^{- \frac{\xi^2}{4 \alpha}} 
	\end{align*}
	we calculate the Wigner transform to be \marginpar{\small 2009.12.15}
	\begin{align*}
		\WignerTrafo(\varphi,\varphi)(x,\xi) &= \frac{1}{\sqrt{2\pi}} \int_{\R} \dd y \, e^{- i y \cdot \xi} \, \varphi^* \bigl ( x - \tfrac{y}{2} \bigr ) \, \varphi \bigl ( x + \tfrac{y}{2} \bigr ) 
		\\
		&= \frac{1}{\sqrt{2\pi}} \int_{\R} \dd y \, e^{- i y \cdot \xi} \, \bigl ( x - \tfrac{y}{2} \bigr ) \, \bigl ( x + \tfrac{y}{2} \bigr ) \, e^{- \frac{1}{4} [(x - \frac{y}{2})^2 + (x + \frac{y}{2})^2]} 
		\\
		&= 2 e^{- \frac{x^2}{2}} \frac{1}{\sqrt{2\pi}} \int_{\R} \dd y \, e^{- i y \cdot 2 \xi} \, (x^2 - y^2) \, e^{- \frac{y^2}{2}} 
		\\
		&
		= 2 e^{- \frac{x^2}{2}} \, \Bigl ( x^2 e^{- 2 \xi^2} + \tfrac{1}{4} \partial_{\xi}^2 \bigl ( e^{- 2 \xi^2} \bigr ) \Bigr )
		\\
		&
		= 2 \Bigl ( x^2  + (2 \xi)^2 - 1 \Bigr ) \, e^{- 2 x^2} e^{- 2 \xi^2} 
		\not\geq 0 
		. 
	\end{align*}
\end{example}
\begin{cor}\label{weyl_calculus:weyl_calculus:wigner_transform:cor:quantum_classical_expectation_value}
	For all $\psi , \varphi \in \Schwartz(\R^d)$ and $f \in \Schwartz(\pspace)$, we have 
	\begin{align*}
		\bscpro{\varphi}{\Op(f) \psi} = \frac{1}{(2\pi)^{\nicefrac{d}{2}}} \int_{\pspace} \dd X \, f(X) \, \WignerTrafo(\psi,\varphi)(X) 
		. 
	\end{align*}
\end{cor}
The Wigner transform has some very interesting properties: 
\begin{thm}[Properties of the Wigner transform]
	Let $\varphi , \psi \in \Schwartz(\R^d)$, $x \in \R^d_x$ and $\xi \in \R^d_{\xi}$. 
	\begin{enumerate}[(i)]
		\item $\WignerTrafo(\psi,\psi)$ is a real-valued function, but not necessarily positive. 
		\item The marginals of the Wigner transform of $\varphi , \psi \in \Schwartz(\R^d)$ with respect to $x$ and $\xi$ are 
		\begin{align*}
			\frac{1}{(2\pi)^{\nicefrac{d}{2}}} \int_{\R^d_x} \dd x \, \WignerTrafo(\psi,\varphi) (x,\xi) &= \eps^{-d} \, (\Fourier \varphi)^*(\nicefrac{\xi}{\eps}) \, (\Fourier \psi)(\nicefrac{\xi}{\eps}) 
			, 
			\\
			\frac{1}{(2\pi)^{\nicefrac{d}{2}}} \int_{\R^d_{\xi}} \dd \xi \, \WignerTrafo(\psi,\varphi) (x,\xi) &= \varphi^*(x) \, \psi(x) 
			. 
		\end{align*}
		\item $\WignerTrafo(\psi,\varphi) \in L^1(\pspace)$
		\item $\bnorm{\WignerTrafo(\varphi,\psi)}_{L^2(\pspace)} = \eps^{- \nicefrac{d}{2}} \,  \norm{\varphi}_{L^2(\R^d)} \norm{\psi}_{L^2(\R^d)}$
		\item Let $R$ be the reflection operator defined by $(R \varphi)(x) := \varphi(-x)$. Then $\WignerTrafo(R \psi , R \varphi)(X) = \WignerTrafo(\psi , \varphi)(-X)$ holds. 
		\item $\WignerTrafo(\psi^*,\varphi^*)(x,\xi) = \WignerTrafo(\varphi,\psi)(x,-\xi)$ 
		\item $\WignerTrafo \bigl ( U(y) \psi , U(y) \varphi \bigr )(x,\xi) = \WignerTrafo(\psi , \varphi)(x - y,\xi)$ for all $y \in \R_x^d$
	\end{enumerate}
\end{thm}
\begin{proof}
	\begin{enumerate}[(i)]
		\item We have to show $\WignerTrafo(\varphi,\varphi)^* = \WignerTrafo(\varphi,\varphi)$: the complex conjugate of $\WignerTrafo(\varphi,\varphi)$ is 
		\begin{align*}
			\WignerTrafo(\varphi,\varphi)^*(x,\xi) &= \biggl ( \frac{1}{(2\pi)^{\nicefrac{d}{2}}} \int_{\R^d_x} \dd y \, e^{- i y \cdot \xi} \, \varphi^{\ast} \bigl ( x - \tfrac{\eps}{2} y \bigr ) \, \varphi \bigl ( x + \tfrac{\eps}{2} y \bigr )  \biggr )^* 
			\\
			&= \frac{1}{(2\pi)^{\nicefrac{d}{2}}} \int_{\R^d_x} \dd y \, e^{+ i y \cdot \xi} \, \varphi \bigl ( x - \tfrac{\eps}{2} y \bigr ) \, \varphi^* \bigl ( x + \tfrac{\eps}{2} y \bigr ) 
			\\
			&= \frac{1}{(2\pi)^{\nicefrac{d}{2}}} \int_{\R^d_x} \dd y \, e^{- i y \cdot \xi} \, \varphi \bigl ( x + \tfrac{\eps}{2} y \bigr ) \, \varphi^* \bigl ( x - \tfrac{\eps}{2} y \bigr ) 
			= \WignerTrafo(\varphi,\varphi)(x,\xi)
		\end{align*}
		Previously on page~\pageref{weyl_calculus:wigner_transform:ex:first_excited_state_harmonic_oscillator}, we have already given an example where $\WignerTrafo(\varphi,\varphi)$ is not positive. 
		\item If we take the marginals with respect to $x$, then up to a factor of $(2\pi)^{\nicefrac{d}{2}}$ that is due to the choice of convention in Definition~\ref{weyl_calculus:wigner_transform:defn:fourier_wigner_trafo}, we get 
		\begin{align*}
			\int_{\R^d_x} \dd x \, \WignerTrafo(\psi,\varphi)(x,\xi) &= \frac{1}{(2\pi)^{\nicefrac{d}{2}}} \int_{\R^d_x} \dd x \int_{\R^d_x} \dd y \, e^{-i y \cdot \xi} \, \varphi \bigl ( x - \tfrac{\eps}{2} y \bigr )^* \, \psi \bigl ( x + \tfrac{\eps}{2} y \bigr ) 
			\\
			&= \frac{1}{(2\pi)^{\nicefrac{d}{2}}} \int_{\R^d_x} \dd x' \int_{\R^d_x} \dd y \, e^{-i y \cdot \xi} \, \varphi (x')^* \, \psi (x' + \eps y) 
			\\
			&= \frac{\eps^{-d}}{(2\pi)^{\nicefrac{d}{2}}} \int_{\R^d_x} \dd x' \int_{\R^d_x} \dd y' \, e^{-i (y'-x') \cdot \frac{\xi}{\eps}} \, \varphi (x')^* \, \psi (y') 
			\\
			&= \eps^{-d} \, (2\pi)^{\nicefrac{d}{2}} \, (\Fourier \varphi)^*(\nicefrac{\xi}{\eps}) \, (\Fourier \psi)(\nicefrac{\xi}{\eps})
			. 
		\end{align*}
		The other marginal can be obtained analogously. 
		\item This follows immediately from (ii) since $\varphi$ and $\psi$ are square integrable. 
		\item We plug in the definition of the Wigner transform and compute 
		\begin{align*}
			\bnorm{\WignerTrafo(\psi,\varphi)}_{L^2(\pspace)}^2 &= \int_{\R^d_x} \dd x \int_{\R^d_{\xi}} \dd \xi \, \babs{\WignerTrafo(\psi,\varphi)(x,\xi)}^2 
			\\
			&= \frac{1}{(2\pi)^d} \int_{\R^d_x} \dd x \int_{\R^d_{\xi}} \dd \xi \int_{\R^d_{x}} \dd y \, \int_{\R^d_{x}} \dd y' \, e^{+i y \cdot \xi} e^{- i y' \cdot \xi}
			\cdot \\
			&\qquad \qquad \qquad \cdot 
			\varphi^* \bigl ( x - \tfrac{\eps}{2} y \bigr ) \, \psi \bigl ( x + \tfrac{\eps}{2} y \bigr ) \, \varphi \bigl ( x - \tfrac{\eps}{2} y' \bigr ) \, \psi^* \bigl ( x + \tfrac{\eps}{2} y' \bigr ) 
			\\
			&= \frac{1}{(2\pi)^d} \int_{\R^d_x} \dd x \int_{\R^d_{x}} \dd y \, \int_{\R^d_{x}} \dd y' \, \biggl ( \int_{\R^d_{\xi}} \dd \xi \, e^{+i (y - y') \cdot \xi} \biggr ) 
			\cdot \\
			&\qquad \qquad \qquad \cdot 
			\varphi^* \bigl ( x - \tfrac{\eps}{2} y \bigr ) \, \psi \bigl ( x + \tfrac{\eps}{2} y \bigr ) \, \varphi \bigl ( x - \tfrac{\eps}{2} y' \bigr ) \, \psi^* \bigl ( x + \tfrac{\eps}{2} y' \bigr ) 
			\\
			&= \int_{\R^d_x} \dd x \int_{\R^d_{x}} \dd y \, \varphi^* \bigl ( x - \tfrac{\eps}{2} y \bigr ) \, \psi \bigl ( x + \tfrac{\eps}{2} y \bigr ) \, \varphi \bigl ( x - \tfrac{\eps}{2} y \bigr ) \, \psi^* \bigl ( x + \tfrac{\eps}{2} y \bigr ) 
			. 
		\end{align*}
		After two changes of variables, this simplifies to 
		\begin{align*}
			\ldots &= \biggl ( \int_{\R^d_x} \dd x \, \varphi^* (x) \, \varphi (x) \biggr ) \, \biggl ( \eps^{-d} \, \int_{\R^d_x} \dd y \, \psi (y) \, \psi^* (y) \biggr ) 
			\\
			&= \eps^{-d} \, \norm{\varphi}_{L^2(\R^d)}^2 \norm{\psi}_{L^2(\R^d)}^2 
			. 
		\end{align*}
		\item This is a direct consequence of the definition, 
		\begin{align*}
			\WignerTrafo \bigl ( R \psi , R \varphi \bigr )(x,\xi) &= \frac{1}{(2 \pi)^{\nicefrac{d}{2}}} \int_{\R^d_x} \dd y \, e^{- i y \cdot \xi} \, (R \varphi)^* \bigl ( x - \tfrac{\eps}{2} y \bigr ) \, (R \psi) \bigl ( x + \tfrac{\eps}{2} y \bigr ) 
			\\
			&= \frac{1}{(2 \pi)^{\nicefrac{d}{2}}} \int_{\R^d_x} \dd y \, e^{- i y \cdot \xi} \, \varphi^* \bigl ( - x + \tfrac{\eps}{2} y \bigr ) \, \psi \bigl ( - x - \tfrac{\eps}{2} y \bigr ) 
			\\
			&= \frac{1}{(2 \pi)^{\nicefrac{d}{2}}} \int_{\R^d_x} \dd y' \, e^{- i y' \cdot (- \xi)} \, \varphi^* \bigl ( (- x) - \tfrac{\eps}{2} y' \bigr ) \, \psi \bigl ( (- x) + \tfrac{\eps}{2} y' \bigr ) 
			\\
			&= \WignerTrafo(\psi,\varphi)(- x, -\xi) 
			. 
		\end{align*}
		\item Follows directly from the definition of the Wigner transform. 
		\item Follows directly from the definition of the Wigner transform. 
	\end{enumerate}
\end{proof}
The Wigner transform is also the inverse of the Weyl quantization: first, we extend the definition to $\Schwartz(\R^d \times \R^d)$. 
\begin{defn}[Wigner transform on $\Schwartz(\R^d \times \R^d)$]
	We can extend the Wigner transform to functions $K \in \Schwartz(\R^d \times \R^d)$ by setting 
	\begin{align}
		\WignerTrafo K (x,\xi) := \frac{1}{(2\pi)^{\nicefrac{d}{2}}} \int_{\R^d_x} \dd y \, e^{- i y \cdot \xi} \, K \bigl ( x + \tfrac{\eps}{2} y , x - \tfrac{\eps}{2} y \bigr )
		. 
		\label{weyl_calculus:wigner_transform:eqn:Wigner_trafo_kernel} 
	\end{align}
\end{defn}
\begin{lem}\label{weyl_calculus:wigner_transform:lem:Wigner_trafo_unitary}
	$\WignerTrafo$ is an isomorphism between $\Schwartz(\R^d \times \R^d)$ and $\Schwartz(\pspace)$. 
\end{lem}
\begin{proof}
	$\WignerTrafo$ consists of a linear, bijective change of variables and a partial Fourier transform. Both leave $\Schwartz(\R^d \times \R^d)$ invariant (mathematically, $\R^d \times \R^d$ and $\pspace = \R^d_x \times \R^d_{\xi}$ are equivalent). The bijectivity follows from the bijectivity of the Fourier transform in the second argument. 
\end{proof}
The Wigner transform is essentially the inverse of Weyl quantization: 
\begin{prop}\label{weyl_calculus:wigner_transform:prop:inverse_weyl_quantization}
	Let $T \in \Op \bigl ( \Schwartz(\pspace) \bigr ) \subset \mathcal{B} \bigl ( L^2(\R^d) \bigr )$ be an operator with operator kernel $K_T$. The map $\Op^{-1} : \Op \bigl ( \Schwartz(\pspace) \bigr ) \longrightarrow \Schwartz(\pspace)$ defined by 
	\begin{align*}
		\Op^{-1} T := \eps^d \, \WignerTrafo K_T
	\end{align*}
	is the inverse to Weyl quantization, \ie we have $T = \Op \bigl ( \eps^d \, \WignerTrafo K_T \bigr)$ for all $T \in \Op \bigl ( \Schwartz(\pspace) \bigr )$. Conversely, if we take any $f \in \Schwartz(\pspace)$ with Weyl kernel 
	\begin{align*}
		K_f &= \frac{1}{(2\pi)^{\nicefrac{d}{2}}} \int_{\R^d_{\xi}} \dd \eta \, e^{- i (y - x) \cdot \eta} \, f \bigl ( \tfrac{1}{2} (x + y) , \eps \eta \bigr ) 
		= \eps^{-d} \, (\Fourier_2 f) \bigl ( \tfrac{1}{2}(x+y) , \tfrac{y - x}{\eps} \bigr ) 
		, 
	\end{align*}
	then $\Op^{-1} \bigl ( \Op(f) \bigr ) = \eps^d \WignerTrafo K_f = f$ holds. 
\end{prop}
\marginpar{\small 2009.12.16}
\begin{remark}\label{weyl_calculus:wigner_transform:rem:kernel_map}
	The Weyl kernel map $K : \Schwartz(\pspace) \longrightarrow \Schwartz(\R^d \times \R^d)$ which maps $f \in \Schwartz(\pspace)$ onto 
	\begin{align*}
		K_f(x,y) := \eps^{-d} \, (\Fourier_2 f) \bigl ( \tfrac{1}{2}(x+y) , \tfrac{y - x}{\eps} \bigr )
	\end{align*}
	is bijective as it consists of a bijective linear change of variables and a partial Fourier transform. 
\end{remark}
\begin{proof}
	Let $K_T$ be the operator kernel associated to the operator $T \in \Op \bigl ( \Schwartz(\pspace) \bigr )$. Then $K_T$ has to be in $\Schwartz(\R^d \times \R^d)$: as $T$ is an operator that has been obtained by Weyl quantization, there is \emph{a} preimage $f_T \in \Schwartz(\pspace)$ and its Weyl kernel $K_{f_T}$ has to be in $\Schwartz(\R^d \times \R^d)$. Writing $T \varphi$, $\varphi \in \Schwartz(\R^d)$, in two different ways, 
	\begin{align*}
		0 &= T \varphi - T \varphi 
		= \frac{1}{(2\pi)^{\nicefrac{d}{2}}} \int_{\R^d_x} \dd y \, K_T(x,y) \, \varphi(y) - \bigl ( \Op(f_T) \varphi \bigr )(x) 
		\\
		&= \frac{1}{(2\pi)^{\nicefrac{d}{2}}} \int_{\R^d_x} \dd y \, K_T(x,y) \, \varphi(y) - \frac{1}{(2\pi)^{\nicefrac{d}{2}}} \int_{\R^d_x} \dd y \, K_{f_T}(x,y) \, \varphi(y) 
		\\
		&= \frac{1}{(2\pi)^{\nicefrac{d}{2}}} \int_{\R^d_x} \dd y \, \bigl ( K_T(x,y) - K_{f_T}(x,y) \bigr ) \, \varphi(y) 
		, 
	\end{align*}
	we conclude $K_T$ and $K_{f_T}$ have to agree almost everywhere. Without loss of generality, we take the continuous (and thus smooth) representative $K_{f_T} \in \Schwartz(\R^d \times \R^d)$. 
	
	We have to confirm that 
	\begin{align*}
		\bigl ( \Op \bigl ( \Op^{-1}(T) \bigr ) \varphi \bigr )(x) &= 
		\bigl ( \Op \bigl ( \eps^d \, \WignerTrafo K_T \bigr ) \varphi \bigr )(x) = \frac{1}{(2\pi)^{\nicefrac{d}{2}}} \int_{\R^d_x} \dd y \, K_T(x,y) \, \varphi(y) 
		\\
		&= \bigl ( T \varphi \bigr )(x) 
	\end{align*}
	holds. Plugging in the definition and making a change of variables, we get 
	\begin{align*}
		\bigl ( \Op \bigl ( \Op^{-1}(T) \bigr ) \varphi \bigr )(x) &= \frac{1}{(2\pi)^d} \int_{\R^d_x} \dd y \int_{\R^d_{\xi}} \dd \eta \, e^{- i (y - x) \cdot \eta} \, \eps^d \, \bigl ( \WignerTrafo K_T \bigr ) \bigl ( \tfrac{1}{2}(x + y) , \eps \eta \bigr ) \, \varphi(y) 
		\\
		&= \frac{\eps^d}{(2\pi)^{\nicefrac{3d}{2}}} \int_{\R^d_x} \dd y \int_{\R^d_{\xi}} \dd \eta \int_{\R^d_x} \dd z \, e^{- i (y - x) \cdot \eta} e^{-i z \cdot \eps \eta} 
		\cdot \\
		&\qquad \qquad \qquad \qquad \cdot 
		K_T \bigl ( \tfrac{1}{2}(x + y) + \tfrac{\eps}{2} z , \tfrac{1}{2} (x + y) - \tfrac{\eps}{2} z \bigr ) \, \varphi(y) 
		\\
		&= \frac{\eps^d}{(2\pi)^{\nicefrac{3d}{2}}} \int_{\R^d_x} \dd y \int_{\R^d_{\xi}} \dd \eta \int_{\R^d_x} \dd z' \, \eps^{-d} \, e^{+ i (x - y - z') \cdot \eta} 
		\cdot \\
		&\qquad \qquad \qquad \qquad \cdot 
		K_T \bigl ( \tfrac{1}{2}(x + y + z') , \tfrac{1}{2} (x + y - z') \bigr ) \, \varphi(y) 
		. 
	\end{align*}
	Now we integrate over $\eta$ which yields a $\delta$ distribution that can be used to kill one of the remaining integrals, 
	\begin{align*}
		\ldots &= \frac{\eps^{d - d}}{(2\pi)^{\nicefrac{d}{2}}} \int_{\R^d_x} \dd y \int_{\R^d_x} \dd z' \, \delta(x - y - z') 
		\cdot \\
		&\qquad \qquad \qquad \qquad \cdot 
		K_T \bigl ( \tfrac{1}{2}(x + y + z') , \tfrac{1}{2} (x + y - z') \bigr ) \, \varphi(y) 
		\\
		&= \frac{1}{(2\pi)^{\nicefrac{d}{2}}} \int_{\R^d_x} \dd y \, K_T \bigl ( \tfrac{1}{2}(x + y + (x - y)) , \tfrac{1}{2} (x + y - (x - y)) \bigr ) \, \varphi(y) 
		\\
		&= \frac{1}{(2\pi)^{\nicefrac{d}{2}}} \int_{\R^d_x} \dd y \, K_T (x,y) \, \varphi(y) 
		= \bigl ( T \varphi \bigr )(x) 
		. 
	\end{align*}
	On the other hand, let $T = \Op(f)$ be the Weyl quantization of $f \in \Schwartz(\pspace)$. Then $\Op^{-1} T = f$ follows from direct calculation: using 
	\begin{align*}
		K_f \bigl ( x + \tfrac{\eps}{2} y , x - \tfrac{\eps}{2} y \bigr ) &= \eps^{-d} \, (\Fourier_2 f) \bigl ( \tfrac{1}{2} \bigl ( x + \tfrac{\eps}{2} y \bigr ) + \tfrac{1}{2} \bigl ( x - \tfrac{\eps}{2} y \bigr ) , \tfrac{1}{\eps} \bigl ( x - \tfrac{\eps}{2} y \bigr ) - \tfrac{1}{\eps} \bigl ( x + \tfrac{\eps}{2} y \bigr ) \bigr ) 
		\\
		&= \eps^{-d} \, (\Fourier_2 f)(x , -y) 
		= \eps^{-d} \, (\Fourier_2^{-1} f)(x , y) 
		, 
	\end{align*}
	we get 
	\begin{align*}
		\Op^{-1} \bigl ( \Op(f) \bigr ) &= \eps^d \, \WignerTrafo K_f(x,\xi) =  \frac{\eps^d}{(2\pi)^{\nicefrac{d}{2}}} \int_{\R^d_x} \dd y \, e^{- i y \cdot \xi} \, K_f \bigl ( x - \tfrac{\eps}{2} y , x + \tfrac{\eps}{2} y \bigr ) 
		\\
		&= \frac{\eps^{d-d}}{(2\pi)^{d}} \int_{\R^d_x} \dd y \int_{\R^d_{\xi}} \dd \eta \, e^{- i y \cdot \xi} e^{+ i y \cdot \eta} \, f ( x , \eta ) 
		= f(x , \xi) 
		. 
	\end{align*}
	To show that $\Op^{-1}$ maps $\Op(\Schwartz(\pspace))$ onto $\Schwartz(\pspace)$, we invoke Remark~\ref{weyl_calculus:wigner_transform:rem:kernel_map} and Lemma~\ref{weyl_calculus:wigner_transform:lem:Wigner_trafo_unitary} which state that the kernel map $K : \Schwartz(\pspace) \longrightarrow \Schwartz(\R^d \times \R^d)$, $f \mapsto K_f$, and the Wigner transform $\WignerTrafo : \Schwartz(\R^d \times \R^d) \longrightarrow \Schwartz(\pspace)$ are bijective. Hence the composition of the kernel map $K$ and the Wigner transform $\WignerTrafo$ is a bijection as well. In fact, 
	\begin{align*}
		\WignerTrafo \circ K : \Schwartz(\pspace) \longrightarrow \Schwartz(\pspace) 
	\end{align*}
	is the identity map by the above calculation. 
\end{proof}
\begin{example}
	Let us see if we can dequantize $\Pe^2 = \Op(\xi^2)$ to $\xi^2$; for simplicity, set $d = 1$, although the arguments carry over to arbitrary dimension. The operator kernel associated to $\Pe^2$ is 
	\begin{align*}
		K(x,y) = \sqrt{2 \pi} \, (-i \eps)^2 \delta''(x - y) 
	\end{align*}
	since for any $\varphi \in \Schwartz(\R^d) \subset L^2(\R^d)$ we have in the sense of distributions 
	\begin{align*}
		\frac{1}{\sqrt{2 \pi}} \int_{\R_x} \dd y \, K(x,y) \, \varphi(y) &= \frac{\sqrt{2 \pi}}{\sqrt{2 \pi}} \int_{\R_x} \dd y \, (-i \eps)^2 \, \delta''(x - y) \, \varphi(y) 
		\\
		&= (-1)^2 (-i \eps)^2 \, \int_{\R_x} \dd y \, \delta(x - y) \, \partial_y^2 \varphi(y) 
		= (-i \eps)^2 \partial_x^2 \varphi(x) 
		\\
		&= (\Pe^2 \varphi)(x) 
		. 
	\end{align*}
	Formally, we can dequantize $\Pe^2$, 
	\begin{align*}
		\Op^{-1}(\Pe^2) &= \eps \, (\WignerTrafo K)(x,\xi) 
		= \frac{\eps}{\sqrt{2 \pi}} \int_{\R_x} \dd y \, e^{- i y \cdot \xi} \, K \bigl ( x + \tfrac{\eps}{2} y , x - \tfrac{\eps}{2} y \bigr ) 
		\\
		&
		= - \eps^3 \sqrt{2\pi} \, \bigl ( \Fourier \delta''(\eps \cdot) \bigr )(\xi) 
		= - \frac{\eps^3}{\eps} \sqrt{2\pi} \, ( \Fourier \delta'')(\nicefrac{\xi}{\eps}) 
		\\
		&= - \eps^2 \sqrt{2\pi} \frac{1}{\sqrt{2\pi}} i^2 \frac{\xi^2}{\eps^2} 
		= \xi^2 
		. 
	\end{align*}
	This is exactly what we have expected to get. 
\end{example}
\begin{remark}
	In the physics literature, one often uses bra-ket notation to write the Wigner transform. Adopting our conventions regarding factors of $2 \pi$ as well as signs and setting $\eps = 1$, the Wigner transform of the operator $T$ can be written as 
	\begin{align*}
		\WignerTrafo T (x,\xi) = \frac{1}{(2\pi)^{\nicefrac{d}{2}}} \int_{\R^d_x} \dd y \, e^{- i y \cdot \xi} \, \bbra{x + \tfrac{y}{2}} T \bket{x - \tfrac{y}{2}} 
		. 
	\end{align*}
	The term $\bbra{x + \tfrac{y}{2}} T \bket{x - \tfrac{y}{2}}$ is nothing but the operator kernel $K_T \bigl ( x + \tfrac{y}{2} , x - \tfrac{y}{2} \bigr )$. Sometimes $\bbra{x + \tfrac{y}{2}} T \bket{x - \tfrac{y}{2}}$ is expanded in terms of eigenfunctions (which usually cannot be the eigenfunctions of $T$ as the spectrum need not be purely discrete!). 
\end{remark}
So far we have made no attempts to generalize Weyl quantization and Wigner transform beyond Schwartz functions. We shall resist the temptation for now. Weyl quantization, for instance, could immediately defined for functions $f : \pspace \longrightarrow \C$ such that $\Fs f \in L^1(\pspace)$. Similarly, as the Fourier transform extends to a unitary map on $L^2(\R^d)$, we could have extended the Wigner transform to a unitary map between $L^2(\R^d \times \R^d)$ and $L^2(\pspace)$. 


\section{Weyl product} 
\label{weyl_calculus:weyl_product}

The Weyl product $\Weyl$ emulates the operator product on the level of functions on phase space, \ie it satisfies 
\begin{align*}
	\Op \bigl ( f \Weyl g \bigr ) = \Op(f) \, \Op(g) 
\end{align*}
for suitable $f , g : \pspace \longrightarrow \C$. It can be derived from the composition law of the Weyl system (Proposition~\ref{weyl_calculus:weyl_system:prop:composition_law_weyl_system}). 
\begin{thm}\label{weyl_calculus:weyl_product:thm:weyl_product}
	For $f , g \in \Schwartz(\pspace)$, the distribution $f \Weyl g$ which satisfies $\Op \bigl ( f \Weyl g \bigr ) = \Op(f) \, \Op(g)$ is a Schwartz function given by 
	\begin{align}
		( f \Weyl g )(X) &= \frac{1}{(2 \pi)^{2d}} \int_{\pspace} \dd Y \int_{\pspace} \dd Z \, e^{i \sigma(X , Y + Z)} \, e^{i \frac{\eps}{2} \sigma(Y,Z)} \, (\Fs f)(Y) \, (\Fs g)(Z) 
		\label{weyl_calculus:weyl_product:eqn:weyl_product} \\
		&= \frac{1}{(\pi \eps)^{2d}} \int_{\pspace} \dd Y \int_{\pspace} \dd Z \, e^{- i \frac{2}{\eps} \sigma(X - Y, X - Z)} \, f(Y) \, g(Z) 
		\notag 
		. 
	\end{align}
\end{thm}
Before we can prove this statement, we need an auxiliary result: take two operators $T$ and $S$ whose operator kernels $K_T$ and $K_S$ are in $\Schwartz(\R^d \times \R^d)$. Then the operator kernel of $T S$ is given by 
\begin{align*}
	(K_T \diamond K_S)(x,y) := \frac{1}{(2\pi)^{\nicefrac{d}{2}}} \int_{\R^d_x} \dd z \, K_T(x,z) \, K_S(z,y) 
	. 
\end{align*}
\begin{lem}\label{weyl_calculus:weyl_product:lem:composition_of_operator_kernels}
	For any $K_T , K_S \in \Schwartz(\R^d \times \R^d)$, the product $K_T \diamond K_S$ is also in $\Schwartz(\R^d \times \R^d)$, \ie $\diamond : \Schwartz(\R^d \times \R^d) \times \Schwartz(\R^d \times \R^d) \longrightarrow \Schwartz(\R^d \times \R^d)$. 
\end{lem}
\begin{proof}
	We need to estimate the seminorms of $K_T \diamond K_S$: let $a , \alpha , b , \beta \in \N_0^d$ be multiindices and for simplicity define $\Phi : (x,y,z) \mapsto (2\pi)^{- \nicefrac{d}{2}} \, K_T(x,z) \, K_S(z,y)$. Then $\Phi \in \Schwartz(\R^d \times \R^d \times \R^d)$ is a Schwartz function in all three variables. First, we need to show we can exchange differentiation with respect to $x$ and $y$ and integration with respect to $z$, \ie that for fixed $x$ and $y$ 
	\begin{align}
		x^a y^b &\partial_x^{\alpha} \partial_y^{\beta} (K_T \diamond K_S)(x,y) = 
		x^a y^b \partial_x^{\alpha} \partial_y^{\beta} \frac{1}{(2\pi)^{\nicefrac{d}{2}}} \int_{\R^d_x} \dd z \, K_T(x,z) \, K_S(z,y) 
		\notag \\
		&
		= \frac{1}{(2\pi)^{\nicefrac{d}{2}}} \int_{\R^d_x} \dd z \, x^a y^b \partial_x^{\alpha} \partial_y^{\beta} \bigl ( K_T(x,z) \, K_S(z,y) \bigr ) 
		= \int_{\R^d_x} \dd z \, x^a y^b \partial_x^{\alpha} \partial_y^{\beta} \Phi(x,y,z) 
		\label{weyl_calculus:weyl_product:eqn:composition_of_operator_kernels}
	\end{align}
	holds. We will do this by estimating $\babs{x^a y^b \partial_x^{\alpha} \partial_y^{\beta} \Phi(x,y,z)}$ uniformly in $x$ and $y$ by an integrable function $G(z)$ and then invoking Dominated Convergence. For fixed $x$ and $y$, we estimate the $L^1$ norm of $x^a y^b \partial_x^{\alpha} \partial_y^{\beta} \Phi(x,y,\cdot)$ from above by a finite number of seminorms of $\Phi(x,y,\cdot)$ with the help of Lemma~\ref{S_and_Sprime:schwartz_functions:lem:Lp_estimate}, 
	\begin{align*}
		\int_{\R^d} \dd z \, &\babs{x^a y^b \partial_x^{\alpha} \partial_y^{\beta} \Phi(x,y,z)} = \bnorm{x^a y^b \partial_x^{\alpha} \partial_y^{\beta} \Phi(x,y,\cdot)}_{L^1(\R^d)} 
		\\
		&\qquad \leq C_1 \, \sup_{z \in \R^d} \babs{x^a y^b \partial_x^{\alpha} \partial_y^{\beta} \Phi(x,y,z)} + C_2 \, \max_{\abs{c} = 2 n} \sup_{z \in \R^d} \babs{x^a y^b \partial_x^{\alpha} \partial_y^{\beta} \Phi(x,y,z)} 
		\\
		&\qquad = C_1 \, \bnorm{x^a y^b \partial_x^{\alpha} \partial_y^{\beta} \Phi(x,y,\cdot)}_{0 0} + C_2 \, \max_{\abs{c} = 2 n} \bnorm{x^a y^b \partial_x^{\alpha} \partial_y^{\beta} \Phi(x,y,\cdot)}_{c 0} 
		. 
	\end{align*}
	Now we interchange $\sup$ and integration with respect to $z$, 
	\begin{align*}
		\sup_{x,y \in \R^d} \int_{\R^d} \dd z \, \babs{x^a y^b \partial_x^{\alpha} \partial_y^{\beta} \Phi(x,y,z)} &\leq \int_{\R^d} \dd z \, \sup_{x,y \in \R^d} \babs{x^a y^b \partial_x^{\alpha} \partial_y^{\beta} \Phi(x,y,z)} 
		\\
		&= \Bnorm{\sup_{x,y \in \R^d} \babs{x^a y^b \partial_x^{\alpha} \partial_y^{\beta} \Phi(x,y,\cdot)}}_{L^1(\R^d)}
		, 
	\end{align*}
	which can be estimated from above by 
	\begin{align*}
		\Bnorm{\sup_{x,y \in \R^d} \babs{x^a y^b \partial_x^{\alpha} \partial_y^{\beta} \Phi(x,y,\cdot)}}_{L^1(\R^d)} 
		&\leq C_1 \, \sup_{x,y \in \R^d} \sup_{z \in \R^d} \babs{x^a y^b \partial_x^{\alpha} \partial_y^{\beta} \Phi(x,y,z)} 
		+ \\
		&\qquad 
		+ C_2 \, \max_{\abs{c} = 2 n} \sup_{x,y \in \R^d} \sup_{z \in \R^d} \babs{x^a y^b \partial_x^{\alpha} \partial_y^{\beta} \Phi(x,y,z)} 
		\\
		&= C_1 \, \bnorm{\Phi}_{a \alpha b \beta 0 0} + C_2 \, \max_{\abs{c} = 2 n} \bnorm{\Phi}_{a \alpha b \beta c 0} < \infty 
		. 
	\end{align*}
	Here $\bigl \{ \norm{\cdot}_{a \alpha b \beta c \gamma} \bigr \}_{a , \alpha , b , \beta , c , \gamma \in \N_0^d}$ is the family of seminorms associated to $\Schwartz(\R^d \times \R^d \times \R^d)$ which are defined by 
	\begin{align*}
		\norm{\Phi}_{a \alpha b \beta c \gamma} := \sup_{x,y,z \in \R^d} \babs{x^a y^b z^c \partial_x^{\alpha} \partial_x^{\beta} \partial_x^{\gamma} \Phi(x,y,z)} 
		. 
	\end{align*}
	This means we have found an \emph{integrable} function 
	\begin{align*}
		G(z) := \sup_{x,y \in \R^d} \babs{x^a y^b \partial_x^{\alpha} \partial_y^{\beta} \Phi(x,y,z)}
	\end{align*} 
	which dominates $x^a y^b \partial_x^{\alpha} \partial_y^{\beta} \Phi(x,y,z)$ for all $x , y \in \R^d$. Hence, exchanging differentiation and integration in equation~\eqref{weyl_calculus:weyl_product:eqn:composition_of_operator_kernels} is possible and we can bound the $\norm{\cdot}_{a \alpha b \beta}$ seminorm on $\Schwartz(\R^d \times \R^d)$ by 
	\begin{align*}
		\bnorm{K_T \diamond K_S}_{a \alpha b \beta} &= \sup_{x,y \in \R^d} \babs{x^a y^b \partial_x^{\alpha} \partial_y^{\beta} (K_T \diamond K_S)(x,y)} 
		\\
		&\leq C_1 \, \bnorm{\Phi}_{a \alpha b \beta 0 0} + C_2 \, \max_{\abs{c} = 2 n} \bnorm{\Phi}_{a \alpha b \beta c 0} < \infty 
		. 
	\end{align*}
	This means $K_T \diamond K_S \in \Schwartz(\R^d \times \R^d)$. 
\end{proof}
\begin{proof}[Theorem~\ref{weyl_calculus:weyl_product:thm:weyl_product}]
	Using the definition of $\Op$, we get 
	\begin{align*}
		\Op(f) \, \Op(g) &= \frac{1}{(2\pi)^{2d}} \int_{\pspace} \dd Y \int_{\pspace} \dd Z \, (\Fs f)(Y) \, (\Fs g)(Z) \, \WeylSys(Y) \, \WeylSys(Z) 
		\\
		&= \frac{1}{(2\pi)^{2d}} \int_{\pspace} \dd Y \int_{\pspace} \dd Z \, (\Fs f)(Y) \, (\Fs g)(Z) \, e^{i \frac{\eps}{2} \sigma(Y,Z)} \, \WeylSys(Y+Z) 
		\\
		&= \frac{1}{(2\pi)^d} \int_{\pspace} \dd Z \left ( \frac{1}{(2\pi)^d} \int_{\pspace} \dd Y \, e^{i \frac{\eps}{2} \sigma(Y,Z-Y)} \, (\Fs f)(Y) \, (\Fs g)(Z-Y) \right ) \WeylSys(Z) 
		. 
	\end{align*}
	We recognize the inner integral as $\bigl ( \Fs (f \Weyl g) \bigr )(Z)$ and thus we add a Fourier transform to obtain the first of the two equivalent forms of the product formula: 
	\begin{align*}
		(f \Weyl g)(X) &= \Fs \left ( \frac{1}{(2\pi)^d} \int_{\pspace} \dd Y \, e^{i \frac{\eps}{2} \sigma(Y,\, \cdot \, -Y)} \, (\Fs f)(Y) \, (\Fs g)(\, \cdot \, -Y) \right )(X) 
		\\
		&= \frac{1}{(2\pi)^{2d}} \int_{\pspace} \dd Z \int_{\pspace} \dd Y \, e^{i \sigma(X,Z)} \, e^{i \frac{\eps}{2} \sigma(Y,Z-Y)} \, (\Fs f)(Y) \, (\Fs g)(Z-Y)
		\\
		&= \frac{1}{(2\pi)^{2d}} \int_{\pspace} \dd Y \int_{\pspace} \dd Z \, e^{i \sigma(X,Y+Z)} \, e^{i \frac{\eps}{2} \sigma(Y,Z)} \, (\Fs f)(Y) \, (\Fs g)(Z) 
		. 
	\end{align*}
	It remains to show that $f \Weyl g$ is a Schwartz function: from Remark~\ref{weyl_calculus:wigner_transform:rem:kernel_map}, we know that the Weyl kernels $K_f$ and $K_g$ of $f$ and $g$ are in $\Schwartz(\R^d \times \R^d)$. Hence, we can write the operator product of $\Op(f)$ and $\Op(g)$ in terms of the associated integral kernels, 
	\begin{align*}
		\bigl ( \Op(f) \Op(g) \varphi \bigr )(x) &= \frac{1}{(2\pi)^{\nicefrac{d}{2}}} \int_{\R^d} \dd y \, \int_{\R^d} \dd z \, \frac{1}{(2\pi)^{\nicefrac{d}{2}}} \, K_f(x,z) \, K_g(z,y) \, \varphi(y) 
		\\
		&= \frac{1}{(2\pi)^{\nicefrac{d}{2}}} \int_{\R^d} \dd y \, (K_f \diamond K_g)(x,y) \, \varphi(y) 
		\overset{!}{=} \Op \bigl ( f \Weyl g \bigr )
		. 
	\end{align*}
	By Lemma~\ref{weyl_calculus:weyl_product:lem:composition_of_operator_kernels}, $K_f \diamond K_g \in \Schwartz(\R^d \times \R^d)$. We use the inverse of Weyl quantization, Proposition~\ref{weyl_calculus:wigner_transform:prop:inverse_weyl_quantization}, to conclude 
	\begin{align*}
		f \Weyl g &= \Op^{-1} \bigl ( \Op(f) \Op(g) \bigr ) 
		= \eps^d \, \WignerTrafo (K_f \diamond K_g) \in \Schwartz(\pspace)
	\end{align*}
	as $\WignerTrafo$ maps $\Schwartz(\R^d \times \R^d)$ bijectively onto $\Schwartz(\pspace)$. 
	
	To show that the first form of the Weyl product, equation~\eqref{weyl_calculus:weyl_product:eqn:weyl_product}, is equivalent to the second, we have to write out all Fourier tranforms and collect the exponentials properly to obtain 
	\begin{align*}
		(f \Weyl g)(X) &= \frac{1}{(2 \pi)^{4d}} \int_{\pspace} \dd Y \int_{\pspace} \dd Y' \int_{\pspace} \dd Z \int_{\pspace} \dd Z' \, e^{i \sigma(X,Y+Z)} e^{i \frac{\eps}{2} \sigma(Y,Z)} e^{i \sigma(Y,Y')} e^{i \sigma(Z,Z')} 
		\cdot \\
		&\qquad \qquad \qquad \qquad \qquad \qquad \qquad \qquad \cdot 
		f(Y') \, g(Z') 
		\\
		&= \frac{1}{(2 \pi)^{4d}} \int_{\pspace} \dd Y \int_{\pspace} \dd Y' \int_{\pspace} \dd Z \int_{\pspace} \dd Z' \, e^{i \sigma(Y,Y' - X)} e^{i \sigma(Z,Z' - X - \frac{\eps}{2} Y)} \, f(Y') \, g(Z') 
		\\
		&= \frac{(2\pi)^{2d}}{(2 \pi)^{4d}} \int_{\pspace} \dd Y \int_{\pspace} \dd Y' \int_{\pspace} \dd Z' \, e^{i \sigma(Y,Y' - X)} \, \delta \bigl ( Z' - X - \tfrac{\eps}{2} Y \bigr ) \, f(Y') \, g(Z') 
		\\
		&= \frac{1}{(2 \pi)^{2d}} \int_{\pspace} \dd Y \int_{\pspace} \dd Y' \, e^{i \sigma(Y,Y' - X)} \, f(Y') \, g \bigl ( X + \tfrac{\eps}{2} Y \bigr ) 
		. 
	\end{align*}
	Making one last change of variables, $\tilde{Z} := X + \tfrac{\eps}{2} Y$, we get the second form of the Weyl product, 
	\begin{align*}
		(f \Weyl g)(X) &= \frac{1}{(2 \pi)^{2d}} \frac{2^{2d}}{\eps^{2d}} \int_{\pspace} \dd \tilde{Z} \int_{\pspace} \dd Y' \, e^{i \sigma(\frac{2}{\eps}(\tilde{Z} - X),Y' - X)} \, f(Y') \, g(\tilde{Z}) 
		\\
		&= \frac{1}{(\pi \eps)^{2d}} \int_{\pspace} \dd Y \int_{\pspace} \dd Z \, e^{- i \frac{2}{\eps} \sigma(X - Y,X - Z)} \, f(Y) \, g(Z) 
		. 
	\end{align*}
	This concludes the proof. 
\end{proof}
\begin{example}
	Let us compute $x \Weyl \xi$ which we then Weyl quantize using problem~23~(i). Since neither $x$ nor $\xi$ are Schwartz functions, we will dispense with mathematical rigor for now: in the distributional sense, we see from 
	\begin{align*}
		(\xi_l \Weyl x_l)(x,\xi) &= \frac{1}{(2\pi)^{2d}} \int_{\pspace} \dd Y \int_{\pspace} \dd Z \, e^{i \sigma(X,Y+Z)} e^{i \frac{\eps}{2} \sigma(Y,Z)} \, (\Fs \eta'_l)(Y) \, (\Fs z'_l)(Z) 
	\end{align*}
	that we need to compute the symplectic Fourier transforms of the factors. As distributions, their Fourier transforms exist and for the first, we get 
	\begin{align*}
		(\Fs \eta'_l)(Y) &= \frac{1}{(2 \pi)^d} \int_{\pspace} \dd Y' \, e^{i \sigma(Y,Y')} \, \eta'_l 
		= \frac{1}{(2 \pi)^d} \int_{\R^d_x} \dd y' \int_{\R^d_{\xi}} \dd \eta' \, \eta'_l e^{i (\eta \cdot y' - y \cdot \eta')} 
		\\
		&
		= \frac{1}{(2 \pi)^d} \int_{\R^d_x} \dd y' \int_{\R^d_{\xi}} \dd \eta' \, (+i) \partial_{y_l}  e^{i (\eta \cdot y' - y \cdot \eta')} 
		= i \, (2\pi)^d \, \partial_{y_l} \delta(Y) 
		. 
	\end{align*}
	Similarly, we calculate the second Fourier transform to be $(\Fs z'_l)(Z) = -i \, (2\pi)^d \, \partial_{\zeta_l} \delta(Z)$. Plugged into the product formula, we obtain 
	\begin{align*}
		(\xi_l \Weyl x_l)(x,\xi) &= \frac{1}{(2\pi)^{2d}} \int_{\pspace} \dd Y \int_{\pspace} \dd Z \, e^{i \sigma(X,Y+Z)} e^{i \frac{\eps}{2} \sigma(Y,Z)} \, (-i^2) \, (2\pi)^{2d} \, \partial_{y_l} \delta(Y) \, \partial_{\zeta_l} \delta(Z) 
		\\
		&= - \int_{\pspace} \dd Y \int_{\pspace} \dd Z \, \partial_{y_l} \Bigl ( e^{i \sigma(X,Y+Z)} e^{i \frac{\eps}{2} \sigma(Y,Z)} \Bigr ) \, \delta(Y) \, \partial_{\zeta_l} \delta(Z) 
		\\
		&= + \int_{\pspace} \dd Y \int_{\pspace} \dd Z \, \partial_{\zeta_l} \Bigl ( \bigl ( i \xi_l - i \tfrac{\eps}{2} \zeta_l \bigr ) \, e^{i \sigma(X,Y+Z)} e^{i \frac{\eps}{2} \sigma(Y,Z)} \Bigr ) \, \delta(Y) \, \delta(Z) 
		\\
		&= \int_{\pspace} \dd Y \int_{\pspace} \dd Z \, \Bigl ( (-i^2) \bigl ( \xi_l - \tfrac{\eps}{2} \zeta_l \bigr ) \bigl ( x_l + \tfrac{\eps}{2} y_l \bigr ) - i \tfrac{\eps}{2} \Bigr ) \, e^{i \sigma(X,Y+Z)} e^{i \frac{\eps}{2} \sigma(Y,Z)} 
		\cdot \\
		&\qquad \qquad \qquad \quad \cdot 
		\delta(Y) \, \delta(Z) 
		\\
		&= \bigl ( \xi_l - \tfrac{\eps}{2} 0 \bigr ) \bigl ( x_l + \tfrac{\eps}{2} 0 \bigr ) - \eps \tfrac{i}{2} 
		= x_l \, \xi_l - \eps \tfrac{i}{2} 
		. 
	\end{align*}
	By problem~23~(i) and $\Op(1) = \id_{L^2}$, we can compute the Weyl quantization of this function explicitly, 
	\begin{align*}
		\Op \bigl ( \xi \Weyl x \bigr ) &= \Op( x \cdot \xi ) - \eps \, d \tfrac{i}{2} \, \id_{L^2} 
		\\ 
		&= \Qe \cdot \Pe - \eps \, d \tfrac{i}{2} \, \id_{L^2} - \eps \, d \tfrac{i}{2} \, \id_{L^2} 
		\\
		&= \Qe \cdot \Pe - \eps \, d i \, \id_{L^2} 
		= \Pe \cdot \Qe 
		. 
	\end{align*}
	This is exactly what we have expected. \marginpar{\small 2009.12.22}
\end{example}
The Weyl product has the following useful properties which can be easily proven: 
\begin{prop}[Properties of Weyl product]
	Let $f , g , h \in \Schwartz(\pspace)$ and $\alpha \in \C$. Then the Weyl product has the following properties: 
	\begin{enumerate}[(i)]
		\item The Weyl product is bilinear, \ie $\bigl ( f + \alpha g \bigr ) \Weyl h = f \Weyl h + \alpha \bigl ( g \Weyl h \bigr )$ and $f \Weyl \bigl ( g + \alpha h \bigr ) = f \Weyl g + \alpha \bigl ( f \Weyl h \bigr )$ hold. 
		\item The Weyl product is associative, \ie $\bigl ( f \Weyl g \bigr ) \Weyl h = f \Weyl \bigl ( g \Weyl h \bigr )$. 
		\item If $f \equiv f(x)$ and $g \equiv g(x)$ are functions of position, then the Weyl product reduces to the pointwise product, $f \Weyl g = f \, g$. 
		\item If $f \equiv f(\xi)$ and $g \equiv g(\xi)$ are functions of momentum, then the Weyl product reduces to the pointwise product, $f \Weyl g = f \, g$. 
	\end{enumerate}
\end{prop}
The proofs are left as an exercise. 


\section{Asymptotics} 
\label{weyl_calculus:asymptotics}

The most important reason why \emph{physicists} should care about Weyl calculus is that there is an \emph{asymptotic expansion} of the product, \ie we can write 
\begin{align*}
	f \Weyl g \asymp \sum_{n = 0}^{\infty} \eps^n (f \Weyl g)_{(n)} 
\end{align*}
where all the terms $(f \Weyl g)_{(n)}$ are known explicitly. In general, this expansion does not converge to any function, but instead for any $N \in \N$, we can write 
\begin{align*}
	f \Weyl g = \sum_{n = 0}^N \eps^n (f \Weyl g)_{(n)} + \eps^{N+1} R_N = \sum_{n = 0}^N \eps^n (f \Weyl g)_{(n)} + \order(\eps^{N+1}) 
	. 
\end{align*}
As $\eps \to 0$, the remainder converges to $0$ at least as fast as $\eps^{N+1}$. Expansions that do not converge \emph{are the rule} rather than the exception, \eg when one sums up Feynman diagrams of different processes, they are often sorted by powers of some small coupling constant, $\alpha \simeq \nicefrac{1}{137}$, for instance. One has \emph{no reason to believe} that summing over \emph{all} terms, one gets something finite! This is where optimal truncation comes into play: instead of summing \emph{all} terms, one stops at some $N(\eps)$ where the error is minimal. Before we state the main theorem of this section, we will introduce the Landau symbols (``little and big O notation''). 
\begin{defn}[Landau symbols]
	Let $f , g : \R^d \longrightarrow \C$. We say $f(x) = \order \bigl ( g(x) \bigr )$ as $x \to x_0$ if and only if 
	\begin{align*}
		\limsup_{x \to x_0} \frac{\abs{f(x)}}{\abs{g(x)}} < \infty 
	\end{align*}
	or quivalently if there exist $M \geq 0$ and $\delta > 0$ such that $\abs{f(x)} \leq M \abs{g(x)}$ for all $\abs{x - x_0} < \delta$. 
	
	We say $f(x) = o \bigl ( g(x) \bigr )$ as $x \to x_0$ if and only if 
	\begin{align*}
		\lim_{x \to x_0} \frac{f(x)}{g(x)} = 0 
		. 
	\end{align*}
\end{defn}
\begin{example}
	We have $x^2 = \order(x)$ and $x = \order(x)$ as $x \to 0$. However $x \neq o(x)$, but $x^2 = o(x)$ as $x \to 0$ since 
	\begin{align*}
		\lim_{x \to 0} \frac{x^2}{x} = \lim_{x \to 0} x = 0 
		. 
	\end{align*}
\end{example}
Let us now formulate the main theorem of this section: 
\begin{thm}\label{weyl_calculus:asymptotics:thm:asymptotic_expansion}
	Let $f , g \in \Schwartz(\pspace)$ and $\eps < 1$. Then the Weyl product can be expanded asymptotically in $\eps$, \ie for any $N \in \N_0$ we have 
	\begin{align}
		f \Weyl g = \sum_{n = 0}^N \eps^n (f \Weyl g)_{(n)} + \eps^{N+1} R_N = \sum_{n = 0}^N \eps^n (f \Weyl g)_{(n)} + \order(\eps^{N+1}) 
		. 
	\end{align}
	All the terms of the expansion are Schwartz functions, $(f \Weyl g)_{(n)} \in \Schwartz(\pspace)$ for all $n \in \N_0$, and are known explicitly, 
	\begin{align}
		(f \Weyl g)_{(n)}(X) &= \frac{1}{n!} \frac{i^n}{2^n} \bigl ( \nabla_y \cdot \nabla_{\zeta} - \nabla_{\eta} \cdot \nabla_z \bigr )^n f(Y) \, g(Z) \Big \vert_{Y = X = Z} 
		. 
		\label{weyl_calculus:asymptotics:eqn:asymptotic_expansion_nth_order_term}
	\end{align}
	The remainder $R_N \in \Schwartz(\pspace)$ as given by equation~\eqref{weyl_calculus:asymptotics:eqn:remainder} is also a rapidly decreasing function. The first two terms are given by 
	\begin{align}
		f \Weyl g = f \, g - \eps \tfrac{i}{2} \bigl \{ f , g \bigr \} + \order(\eps^2)
		\label{weyl_calculus:asymptotics:eqn:asymptotic_expansion_first_terms}
	\end{align}
	where $\bigl \{ f , g \bigr \} = \sum_{j = 1}^d \bigl ( \partial_{\xi_j} f \, \partial_{x_j} g - \partial_{x_j} f \, \partial_{\xi_j} g \bigr )$ is the Poisson bracket. 
\end{thm}
\begin{proof}
	\textbf{Step 1: Expanding the twister. }
	Let $N \in \N_0$ be arbitrary, but fixed. Then we can expand the twisting factor in they Weyl product,  
	\begin{align*}
		e^{i \frac{\eps}{2} \sigma(Y,Z)} = \sum_{n = 0}^N \eps^n \frac{1}{n!} \frac{i^n}{2^n} \bigl ( \sigma(Y,Z) \bigr )^n + \tilde{R}_N(Y,Z) 
		, 
	\end{align*}
	where 
	\begin{align*}
		\tilde{R}_N(Y,Z) &= \frac{1}{N!} \bigl ( i \tfrac{\eps}{2} \sigma(Y,Z) \bigr )^{N+1} \, \int_0^1 \dd s \, (1 - s)^{N} \, e^{i s \frac{\eps}{2} \sigma(Y,Z)} 
		\\
		&= \eps^{N+1} \frac{1}{N!} \frac{i^{N+1}}{2^{N+1}} \bigl ( \sigma(Y,Z) \bigr )^{N+1} \, \int_0^1 \dd s \, (1 - s)^{N} \, e^{i s \frac{\eps}{2} \sigma(Y,Z)} 
		. 
	\end{align*}
	is the remainder of the Taylor expansion. If we plug this into the product formula, 
	\begin{align}
		(f \Weyl g)(X) &= \frac{1}{(2\pi)^{2d}} \int_{\pspace} \dd Y \int_{\pspace} \dd Z \, e^{i \sigma(X,Y+Z)} \, 
		 \Bigl ( \mbox{$\sum_{n = 0}^N$} \eps^n \tfrac{1}{n!} \tfrac{i^n}{2^n} \bigl ( \sigma(Y,Z) \bigr )^n + \tilde{R}_N(Y,Z) \Bigr ) 
		\cdot \notag \\
		&\qquad \qquad \qquad \qquad \qquad \qquad \qquad \qquad \cdot  
		(\Fs f)(Y) \, (\Fs g)(Z) 
		\notag \\
		&=: \sum_{n = 0}^N \eps^n (f \Weyl g)_{(n)}(X) + \eps^{N+1} R_N(X) 
		\label{weyl_calculus:asymptotics:eqn:remainder}
		, 
	\end{align}
	we can define the $n$th order terms and the remainder. 
	\medskip
	
	\noindent
	\textbf{Step 2: Treating $(f \Weyl g)_{(n)}$. }
	First of all, we note that 
	\begin{align*}
		\bigl ( \sigma(Y,Z) \bigr )^n &= \bigl ( \eta \cdot z - y \cdot \zeta \bigr )^n = \sum_{\abs{a} + \abs{b} = n} \frac{(-1)^{\abs{b}}}{a! b!} \eta^a y^b z^a \zeta^b 
	\end{align*}
	is just a polynomial in $y$, $\eta$, $z$ and $\zeta$. Since $f \in \Schwartz(\pspace)$, by Theorem~\ref{S_and_Sprime:schwartz_functions:thm:Fourier_is_bijection} 
	\begin{align*}
		y^b \eta^a (\Fs f) = \Fs \bigl ( (- i \partial_{\eta'})^b (+ i \partial_{y'})^a f \bigr ) \in \Schwartz(\pspace)
	\end{align*}
	holds for all multiindices $a , b \in \N_0^d$ and similarly for $g$. This means, the integral expression for $(f \Weyl g)_{(n)}$ exists for all $X \in \pspace$ and reduces to 
	\begin{align*}
		(f \Weyl g)_{(n)}(X) &= \frac{1}{n!} \frac{i^n}{2^n} \sum_{\abs{a} + \abs{b} = n} \frac{(-1)^{\abs{b}}}{a! b!} \frac{1}{(2 \pi)^{2d}} \int_{\pspace} \dd Y \int_{\pspace} \dd Z \, e^{i \sigma(X,Y+Z)} 
		\cdot \\
		&\qquad \qquad \qquad \qquad \qquad \qquad \qquad \qquad \quad \cdot 
		\eta^a y^b (\Fs f)(Y) \, z^a \zeta^b (\Fs g)(Z) 
		\\
		&= \frac{1}{n!} \frac{i^n}{2^n} \sum_{\abs{a} + \abs{b} = n} \frac{(-1)^{\abs{b}}}{a! b!} \frac{1}{(2 \pi)^{2d}} \int_{\pspace} \dd Y \int_{\pspace} \dd Z \, e^{i \sigma(X,Y+Z)} 
		\cdot \\
		&\qquad \qquad \quad \cdot 
		\Bigl ( \Fs \bigl ( (- i \partial_{\eta'})^b (+ i \partial_{y'})^a f \bigr ) \Bigr )(Y) 
		\Bigl ( \Fs \bigl ( (- i \partial_{\zeta'})^a (+ i \partial_{z'})^b g \bigr ) \Bigr )(Z) 
		. 
	\end{align*}
	The remaining integral is nothing but the symplectic Fourier transform in $Y$ and $Z$. The symplectic Fourier transform is its own inverse, hence we get expression~\eqref{weyl_calculus:asymptotics:eqn:asymptotic_expansion_nth_order_term}, 
	\begin{align*}
		\ldots &= \frac{1}{n!} \frac{i^n}{2^n} \sum_{\abs{a} + \abs{b} = n} \frac{(-1)^{\abs{b}}}{a! b!} \Bigl ( \Fs^2 \bigl ( (- i \partial_{\eta'})^b (+ i \partial_{y'})^a f \bigr ) \Bigr )(X) \, \Bigl ( \Fs^2 \bigl ( (- i \partial_{\zeta'})^a (+ i \partial_{z'})^b g \bigr ) \Bigr )(X) 
		\\
		&= \frac{1}{n!} \frac{i^n}{2^n} \sum_{\abs{a} + \abs{b} = n} \frac{(-1)^{\abs{b}}}{a! b!} \bigl ( (- i \partial_{\xi})^b (+ i \partial_x)^a f \bigr ) (X) \, \bigl ( (- i \partial_{\xi})^a (+ i \partial_{x})^b g \bigr )(X) 
		\\
		&= \frac{1}{n!} \frac{i^n}{2^n} \bigl ( \nabla_y \cdot \nabla_{\zeta} - \nabla_{\eta} \cdot \nabla_z \bigr )^n f(Y) \, g(Z) \Big \vert_{Y = X = Z} 
		. 
	\end{align*}
	Since each of the factors consists of derivatives of Schwartz functions, $(f \Weyl g)_{(n)} \in \Schwartz(\pspace)$ is also a Schwartz function. 
	\medskip
	
	\noindent
	\textbf{Step 3: Treating the remainder $R_N$. }
	If we combine Theorem~\ref{weyl_calculus:weyl_product:thm:weyl_product}, $f \Weyl g \in \Schwartz(\pspace)$ with Step~2, $(f \Weyl g)_{(n)} \in \Schwartz(\pspace)$, we conclude that the remainder is also a Schwartz function, 
	\begin{align*}
		R_N = f \Weyl g - \sum_{n = 0}^N \eps^n (f \Weyl g)_{(n)} \in \Schwartz(\pspace) 
		. 
	\end{align*}
	From the explicit expression for the remainder, equation~\eqref{weyl_calculus:asymptotics:eqn:remainder}, we also see that it is indeed of order $\order(\eps^{N+1})$. This concludes the proof. 
\end{proof}
\begin{example}
	In the previous section, we have calculated $\xi \Weyl x$. If we use the asymptotic expansion -- which in this case is exact, we can obtain this result with much less work. Pluggin in equation~\eqref{weyl_calculus:asymptotics:eqn:asymptotic_expansion_first_terms}, we get 
	\begin{align*}
		\bigl ( \xi \Weyl x \bigr )(x,\xi) &= \xi \cdot x - \eps \, \frac{i}{2} \sum_{l = 1}^d \bigl \{ \xi_l , x_l \bigr \} + \order(\eps^2) 
		\\
		&
		= x \cdot \xi - \eps \, \frac{i}{2} \sum_{l,j = 1}^d \bigl ( \partial_{\xi_j} \xi_l \, \partial_{x_j} x_l - \partial_{x_j} \xi_l \, \partial_{\xi_j} x_l \bigr ) + \order(\eps^2) 
		\\
		&
		= x \cdot \xi - \eps \, \frac{i}{2} \sum_{l,j = 1}^d \delta_{lj} + \order(\eps^2) 
		= x \cdot \xi - \eps \, d \tfrac{i}{2} + \order(\eps^2) 
		. 
	\end{align*}
	The remainder, however, is exactly $0$: the factor $\bigl ( \sigma(Y,Z) \bigr )^2$ is a polynomial of second order in $y \cdot \zeta$ and $\eta \cdot z$ and leads to \emph{two} derivatives with respect to position and momentum. But $x$ and $\xi$ are linear so their second and higher-order derivatives vanish identically. Hence, we have shown 
	\begin{align*}
		\xi \Weyl x = \xi \cdot x - \eps \, d \tfrac{i}{2} 
		. 
	\end{align*}
	Generally, if one of the factors is a polynomial of $k$th degree, the asymptotic expansion terminates after finitely many terms and is exact. 
\end{example}
In Chapter~\ref{frameworks:properties_of_quantizations}, we have emphasized the importance that the the Moyal commutator 
\begin{align*}
	\bigl [ f , g \bigr ]_{\Weyl} := f \Weyl g - g \Weyl f 
\end{align*}
vanishes as $\eps \to 0$. Using equation~\eqref{weyl_calculus:asymptotics:eqn:asymptotic_expansion_first_terms}, we immediately get 
\begin{align*}
	\bigl [ f , g \bigr ]_{\Weyl} = - i \eps \bigl \{ f , g \bigr \} + \order(\eps^3) 
	. 
\end{align*}
The error is of \emph{third} order as all contributions from even powers vanish. This can be traced back to 
\begin{align*}
	\bigl ( \sigma(Y,Z) \bigr )^{2k} = \bigl ( - \sigma(Z,Y) \bigr )^{2k} = \bigl ( \sigma(Z,Y) \bigr )^{2 k} 
	&&
	\forall k \in \N_0 
\end{align*}
and the commutativity of multiplication in $\C$. We mention the latter explicitly as we will quantize matrix-valued functions in the next section. There, it is not at all clear that even the $0$th order vanishes. \marginpar{\small 2010.01.12}


\section{Application: Diagonalization of the Dirac equation} 
\label{weyl_calculus:diag_dirac}

We will now give an application that is not at all trivial. If it were 1983, we could have published a paper in a prestigious scientific journal \cite{Cordes:pseudodifferential_FW_transform:1983,Cordes:pseudodifferential_FW_transform:2004}. The dynamics of a relativistic spin-$\nicefrac{1}{2}$ particle is described by the Dirac equation \cite{Yndurain:relativisticQM:1996}. The form of this equation follows from three requirements: \emph{gauge-covariance} with respect to the inhomogeneous Lorentz group which is comprised of Lorentz boosts, rotations in space and translations in space-time ($3+3+4 = 10$-dimensional) and the fact we are interested in \emph{spin-$\nicefrac{1}{2}$} particles \emph{with mass $m > 0$} \cite{Wigner:unitary_reps_Lorentz_group:1939,Lerner:derivation_Dirac_eqn:1996,Buitrago:deduction_Dirac:1988}. As long as the energy is below the pair creation threshold of $2 m c^2$, the Dirac equation accurately describes the physics of a spin-$\nicefrac{1}{2}$ particle. Here, $c$ is the speed of light in vacuo. 

We will not attempt to give a derivation here, but only present the result: the Dirac equation can be written in standard form as 
\begin{align*}
	i \frac{\partial}{\partial t} \Psi &= \hat{H}_D \Psi 
	, 
	&& 
	\Psi \in L^2(\R^3;\C^4) 
	, 
\end{align*}
where the free Dirac hamiltonian for a particle with mass $m > 0$ 
\begin{align*}
	\hat{H}_D &= c^2 m \beta + c (- i \hbar \nabla_x) \cdot \alpha 
\end{align*}
contains the $\alpha_j$ and $\beta$ matrices, 
\begin{align*}
	\alpha_j := \left (
	\begin{matrix}
		0 & \sigma_j \\
		\sigma_j & 0 \\
	\end{matrix}
	\right )
	, 
	\qquad 
	\beta := \left (
	\begin{matrix}
		\id_{\C^2} & 0 \\
		0 & - \id_{\C^2} \\
	\end{matrix}
	\right ) 
	, 
	&& 
	j = 1,2,3 
	. 
\end{align*}
The $\sigma_j \in \mathrm{Mat}_{\C}(2)$, $j = 1,2,3$, are the usual Pauli matrices. For convenience, we choose units such that $\hbar = 1$ and rescale the energy by $\nicefrac{1}{c^2}$, 
\begin{align*}
	\hat{H}_0 := \tfrac{1}{c^2} \hat{H}_D = m \beta + \bigl ( - \tfrac{i}{c} \nabla_x \bigr ) \cdot \alpha =: H_0(\Pe) 
	. 
\end{align*}
This suggests to use 
\begin{align*}
	\Pe &:= - \tfrac{i}{c} \nabla_x \\
	\Qe &:= \hat{x}
\end{align*}
as building block operators. Using the relations 
\begin{align*}
	\alpha_j \beta =& - \beta \alpha_j 
	, \\
	[\alpha_k , \alpha_j]_+ :=& \; \alpha_k \alpha_j + \alpha_j \alpha_k = 2 \delta_{kj} \, \id_{\C^4} 
	&&
	k,j = 1,2,3 
	, 
\end{align*}
we can show that 
\begin{align*}
	u_0(\Pe) &:= \frac{1}{\sqrt{2 E(\Pe) (E(\Pe) + m)}} \bigl ( (E(\Pe) + m) \, \id_{\C^4} - (\Pe \cdot \alpha ) \beta \bigr ) 
	, 
\end{align*}
is unitary, \ie it satisfies 
\begin{align*}
	u_0(\Pe) u_0(\Pe)^* = \id_{L^2} = u_0(\Pe)^* u_0(\Pe) 
	, 
\end{align*}
and diagonalizes $H_0(\Pe)$, 
\begin{align*}
	u_0(\Pe) H_0(\Pe) u_0(\Pe)^* &= E(\Pe) \beta := \sqrt{m^2 + \Pe^2} \, \beta 
	. 
\end{align*}
This means, the relativistic kinetic energy governs the free dynamics of relativistic particles. The upper two components of wave functions in the diagonalized representation correspond to particles (``electrons''), the lower two are associated to antiparticles (``positrons''). As each of the energies is two-fold spin-degenerate, the unitary $u_0(\Pe)$ is \emph{not} unique. 

If we add an electric field, we need to add the electrostatic potential, 
\begin{align*}
	H(\Qe,\Pe) := m \beta + \Pe \cdot \alpha + \tfrac{1}{c^2} V(\Qe) 
	=: H_0(\Pe) + \tfrac{1}{c^2} H_2(\Qe) 
	. 
\end{align*}
As long as the particle's energy is below the pair creation threshold, we expect that electronic and positronic degrees of freedom still decouple, at least in an approximate sense. Our choice of building block operators suggests that the \emph{relevant small parameter is $\nicefrac{1}{c}$}. More preoperly, in light of the discussion in Chapter~\ref{frameworks:properties_of_quantizations}, the small parameter should be dimensionless. To remedy this, we should in principle work with the ratio $\nicefrac{v_0}{c}$ where $v_0$ is some characteristic velocity of the system. If the particle is relativistic $v_0 \gtrsim 0.1 c$, but not too fast, say $v_0 \lesssim 0.7 c$, we expect approximate decoupling in the presence of an electric field. 

Since $\Pe$ and $\Qe$ do not commute, their commutator is of the order $\nicefrac{1}{c}$, and $u_0(\Pe)$ no longer diagonalizes $H(\Qe,\Pe)$. The usual recipe found in text books is to use the Foldy-Wouthuysen transform \cite{FoldyW:DiracTheory1950,Thaller:DiracEquation:1992,Yndurain:relativisticQM:1996} which successively diagonalizes $H(\Qe,\Pe)$ up to errors of higher order. The result is the non-relativistic Pauli hamiltonian and higher-order corrections, 
\begin{align*}
	\hat{H}_{\mathrm{Pauli}} = m \beta + \frac{1}{c^2} \biggl ( \frac{1}{2m} (- i \nabla_x)^2 + V(\hat{x}) \biggr ) + \order(\nicefrac{1}{c^3})  
	, 
\end{align*}
and it is often claimed that one gets the Taylor expansion of $\sqrt{m^2 + \xi^2}$ around $\xi = 0$. This is correct, but neither a proof nor a derivation. This equation also does not accurately describe the dynamics of a \emph{relativistic} particle, the kinetic energy is non-relativistic. How do we recover $\sqrt{m^2 + \Pe^2}$ as kinetic energy? The first who has found a way was Cordes in the mid-1980s \cite{Cordes:pseudodifferential_FW_transform:1983}. More sophisticated arguments have since appeared elsewhere \cite{Teufel:adiabaticPerturbationTheory:2003,FuerstLein:nonRelLimitDirac:2008} which allow for the derivation of corrections in a systematic fashion. 

Let us do the calculation using Weyl calculus. In fact, all we need at this stage is 
\begin{align*}
	f \Weyl g &= f \, g - \tfrac{1}{c} \tfrac{i}{2} \{ f , g \} + \order(\nicefrac{1}{c^2}) 
	= f \, g + \order(\nicefrac{1}{c}) 
	. 
\end{align*}
Then we immediately conclude 
\begin{align*}
	u_0(\Pe) H(\Qe,\Pe) u_0(\Pe)^* &= E(\Pe) \beta + \tfrac{1}{c^2} u_0(\Pe) V(\Qe) u_0(\Pe)^*
	\\
	&= E(\Pe) \beta + \tfrac{1}{c^2} \Op(u_0) \Op(V) \Op(u_0^*) 
	= E(\Pe) \beta + \tfrac{1}{c^2} \Op \bigl ( u_0 \Weyl V \Weyl u_0^* \bigr ) 
	\\
	&= E(\Pe) \beta + \tfrac{1}{c^2} \Op \bigl ( u_0 V u_0^* + \order(\nicefrac{1}{c}) \bigr ) 
	= E(\Pe) \beta + \tfrac{1}{c^2} V(\Qe) + \order(\nicefrac{1}{c^3})
	. 
\end{align*}
A slightly more sophisticated argument shows that if in addition to the electric field, a \emph{magnetic} field is added, then the Weyl quantization of $u_0^A(x,\xi) := u_0 \bigl ( x , \xi - \tfrac{1}{c^2} A(x) \bigr )$ diagonalizes 
\begin{align*}
	\hat{H}^A := m \beta + \bigl ( \Pe - \tfrac{1}{c^2} A(\Qe) \bigr ) \cdot \alpha + \tfrac{1}{c^2} V(\Qe)
\end{align*}
up to errors of order $\nicefrac{1}{c^3}$, 
\begin{align*}
	\Op(u_0^A) \hat{H}^A \Op(u_0^A)^* = \Op \bigl ( E \bigl ( \xi - \tfrac{1}{c^2} A(x) \bigr ) \bigr ) + \tfrac{1}{c^2} V(\Qe) + \order(\nicefrac{1}{c^3}) 
	. 
\end{align*}
In Chapter~\ref{multiscale}, we will present a systematic scheme to derive all higher-order corrections. 


\chapter{The semiclassical limit} 
\label{semiclassics}
%
This section is the culmination of what we have done up to now and one of the two main topics of this lecture: we aim to make the fuzzy notion of `semiclassical limit' into a mathematical theorem. The semiclassical limit can be understood as a \emph{scaling limit}. Since even the physics of semiclassical limits is somewhat obscure, we start the topic by repeating the commonly chosen approach, namely the Ehrenfest `theorem.' 

In Chapter~\ref{frameworks:quantum_mechanics}, we have discussed the Heisenberg picture in quantum mechanics where observables $\hat{A}$ are evolved in time, $\hat{A}(t) = e^{+ i \frac{t}{\eps} \hat{H}} \hat{A} e^{- i \frac{t}{\eps} \hat{H}}$, and states remain constant in time. Then $\hat{A}(t)$ obeys the Heisenberg equations of motion, 
\begin{align*}
	\frac{\dd}{\dd t} \hat{A}(t) &= \frac{1}{i \eps} \bigl [ \hat{A}(t) , \hat{H} \bigr ] = \frac{1}{i \eps} e^{+ i \frac{t}{\eps} \hat{H}} \bigl [ \hat{A} , \hat{H} \bigr ] e^{- i \frac{t}{\eps} \hat{H}} 
	. 
\end{align*}
If we consider the standard hamiltonian $\hat{H} = \frac{1}{2m} \Pe^2 + V(\Qe)$ for simplicity, then the equations of motion for position $\Qe$ and momentum $\Pe$ are 
\begin{align}
	\frac{\dd }{\dd t} \Qe(t) &= \tfrac{1}{i \eps} e^{+ i \frac{t}{\eps} \hat{H}} \bigl [ \Qe , \hat{H} \bigr ] e^{- i \frac{t}{\eps} \hat{H}} 
	= e^{+ i \frac{t}{\eps} \hat{H}} \tfrac{\Pe}{m} e^{- i \frac{t}{\eps} \hat{H}} 
	= \frac{\Pe(t)}{m} 
	\label{semiclassics:eqn:Ehrenfest_eom_Q}
\end{align}
and 
\begin{align}
	\frac{\dd }{\dd t} \Pe(t) &= \tfrac{1}{i \eps} e^{+ i \frac{t}{\eps} \hat{H}} \bigl [ \Pe , \hat{H} \bigr ] e^{- i \frac{t}{\eps} \hat{H}} 
	= \tfrac{1}{i \eps} e^{+ i \frac{t}{\eps} \hat{H}} \bigl [ \Pe , V(\Qe) \bigr ] e^{- i \frac{t}{\eps} \hat{H}} 
	\notag \\
	&= - \bigl ( \nabla_x V(\Qe) \bigr )(t) 
	= - \nabla_x V \bigl ( \Qe(t) \bigr ) 
	\label{semiclassics:eqn:Ehrenfest_eom_P}
	, 
\end{align}
respectively. The last equality, $- \bigl ( \nabla_x V(\Qe) \bigr )(t) = - \nabla_x V \bigl ( \Qe(t) \bigr )$, is non-trivial and a consequence of the so-called spectral theorem \cite[Section~VIII.3]{Reed_Simon:bibel_1:1981}. These equations look like Hamilton's equations of motion, equation~\eqref{frameworks:classical_mechanics:eqn:hamiltons_eom}. 

If the particle is in state $\psi \in \mathcal{D}(\hat{H}) \subset L^2(\R^d)$, equations~\eqref{semiclassics:eqn:Ehrenfest_eom_Q} and \eqref{semiclassics:eqn:Ehrenfest_eom_P} make predictions on the measured velocities, momenta and so forth: we simply compute the expectation value. The first line, 
\begin{align}
	\dot{q}(t) := \mathbb{E}_{\psi} \bigl ( \dot{\Qe}(t) \bigr ) = \mathbb{E}_{\psi} \bigl ( \tfrac{\Pe(t)}{m} \bigr ) 
	= \bscpro{\psi}{\tfrac{\Pe(t)}{m} \psi} 
	\label{semiclassics:eqn:Ehrenfest_eom_expval_Q}
	, 
\end{align}
reproduces the relationship between averaged velocity and momentum of a nonrelativistic classical particle. The equation of motion for the expectation value of the momentum reads 
\begin{align}
	\dot{p}(t) := \mathbb{E}_{\psi} \bigl ( \dot{\Pe}(t) \bigr ) = \mathbb{E}_{\psi} \bigl ( - \nabla_x V \bigl ( \Qe(t) \bigr ) \bigr )
	= \bscpro{\psi}{ - \nabla_x V \bigl ( \Qe(t) \bigr ) \psi} 
	. 
	\label{semiclassics:eqn:Ehrenfest_eom_expval_P}
\end{align}
This result is known throughout the physics community as Ehrenfest `theorem.' We have put the word theorem in quotation marks since a mathematically rigorous proof equires one to make assumptions on $\hat{H}$. Although it is possible to make these equations mathematically rigorous (\eg \cite{Friesecke_Schmidt:Ehrenfest_theorem:2010}), we shall make no such attempt. 

Equation~\eqref{semiclassics:eqn:Ehrenfest_eom_expval_P} presents us with the challenge that we do not have direct access to the position variable $q(t) = \mathbb{E}_{\psi} \bigl ( \Qe(t) \bigr )$ as in general 
\begin{align}
	\mathbb{E}_{\psi} \bigl ( - \nabla_x V \bigl ( \Qe(t)) \bigr ) \neq - \nabla_x V \bigl ( \mathbb{E}_{\psi}(\Qe(t)) \bigr ) 
\end{align}
holds. Hence, it is still not clear \emph{in what way} this implies classical behavior. Even in some approximate sense, left- and right-hand side are in general \emph{not equal} as one can see from the red wavepacket in Figure~\ref{semiclassics:figure:Ehrenfest}. 
%
\begin{figure}
	\hfil\includegraphics[height=4cm]{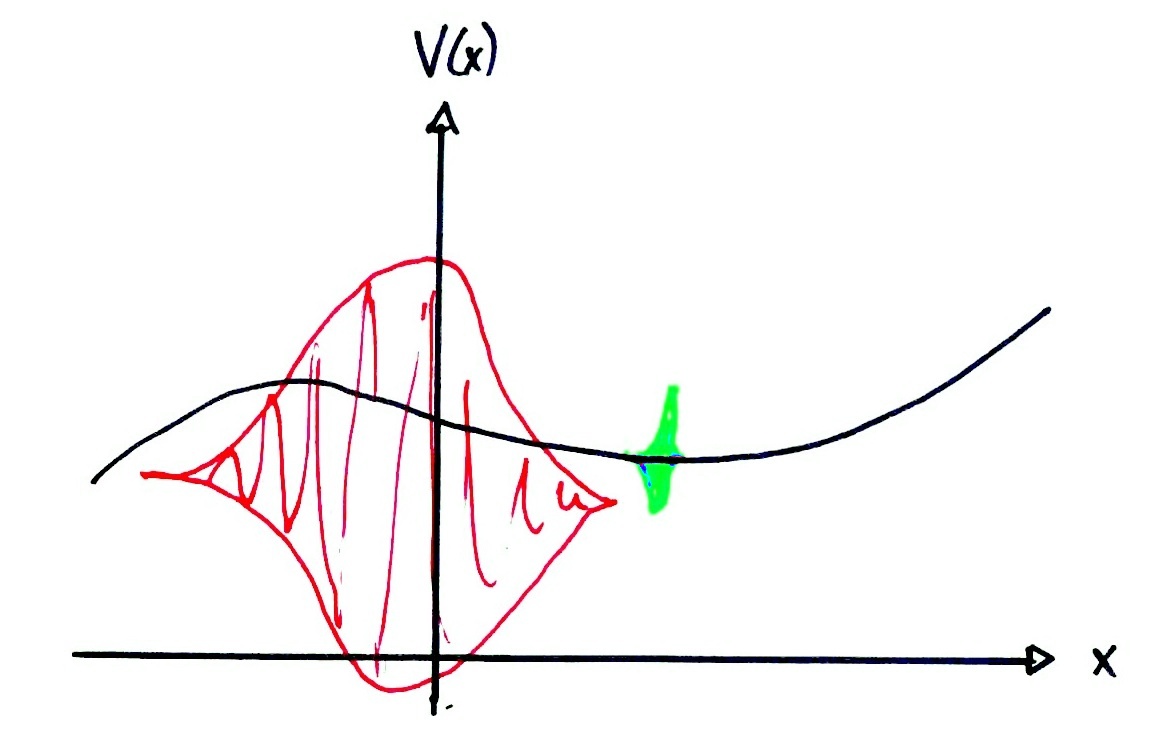}\hfil
	\caption{The green wavepacket is sharply peaked and thus will have a `good' semiclassical behavior. The red wavepacket's spread is too large, its expectation values of position and momentum are not expected to follow the trajectory of a classical point particle. }
	\label{semiclassics:figure:Ehrenfest}
\end{figure}
This means that in order for a semiclassical approximation \emph{in the sense of Ehrenfest} to hold, the wave function must be sharply peaked in comparison to the variation of the potential $V$, \ie one heuristically expects the semiclassical approximation to be more accurate for the sharply peaked green wavepacket rather than the red one. If the wave function has multiple peaks or varies on the same scale as the potential, we cannot expect semiclassical behavior \emph{in this description}. Indeed, this point was very clear to Ehrenfest when he first proposed this approach in a 1927 paper \cite{Ehrenfest:semiclassics:1927}: 
\begin{quote}
	Gleichung~(5) [equations~\eqref{semiclassics:eqn:Ehrenfest_eom_expval_Q} and \eqref{semiclassics:eqn:Ehrenfest_eom_expval_P}] besagt aber offenbar: Jedesmal, wenn die Breite des (Wahrscheinlichkeits-)Wellenpakets $\psi \psi^*$ (im Verhältnis zu \emph{makro\-skopischen} Distanzen) ziemlich klein ist, passt die Beschleunigung (des Schwerpunktes [$q$]) des Wellenpaketes im Sinne der Newtonschen Gleichungen zu der ``am Orte des Wellenpaketes herschenden'' Kraft [$- \nabla_x V$]. 
\end{quote}
However, in most standard textbooks, these crucial remarks are not properly stressed: first of all, the validity of a classical approximation has little to do with `$\hbar$ being small,' but rather with a \emph{separation of length scales}. Secondly, the Ehrenfest `theorem' can serve as a starting point of a semiclassical limit if one restricts oneself to \emph{wave packets} which are sharply peaked in real and reciprocal space. The prototypical example are Gaussian wave packets and their generalizations.\footnote{Interest in states of this form has surged in the 1980s and there is a vast amount of literature on this topic. An early, but well-written review from the point of view of physics, is a work by Littlejohn \cite{Littlejohn:semiclassics_wavepackets:1986}. For a detailed mathematical analysis, a good starting point are the publications of Hagedorn, \eg \cite{Hagedorn:semiclassics_coherent_states:1980}. }

So assume we start with a wave packet which is well-localized in space and reciprocal space, \eg a Gaussian. How long does it remain a sharply focussed wave packet? Even in the absence of a potential, we know that the wave function broadens over time and the maximum of $\abs{\psi(x)}$ scales as $t^{- \nicefrac{d}{2}}$ over time (see Proposition~\eqref{S_and_Sprime:schwartz_functions:prop:free_Schroedinger}). Furthermore, if a potential is present, it is not at all clear if the wave packet breaks up or not. 

Let us consider a step potential $V(x) = V_0 ( 1 - \theta(x))$, $V_0 > 0$, where $\theta$ is the Heavyside function. Then from quantum mechanics, we know that upon hitting the step, part of the wavefunction will be reflected while another part of it will be transmitted, \ie the wavepacket splits in two! Hence, \emph{slow variation of the potential and ``smoothness''} are crucial.\footnote{The smoothness condition can be lifted to a certain degree: the potential only needs to be smooth in regions the wave function can `explore.' } 
\begin{figure}
	\hfil\includegraphics[height=4cm]{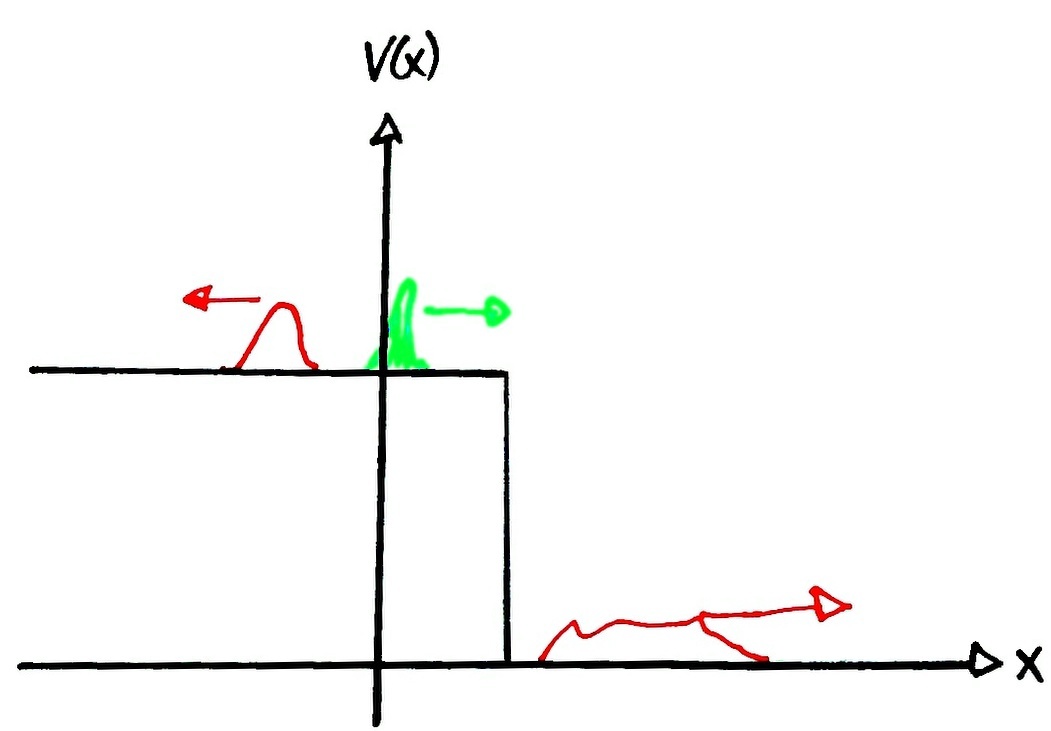}\hfil
	\caption{The incoming wavepacket (green) is split into two outgoing contributions (red). }
\end{figure}
The \emph{microscopic scale} is given by the \emph{de Broglie wavelength} associated to the wavefunction~$\psi$, 
\begin{align*}
	\lambda = \frac{2 \pi \hbar}{\babs{\mathbb{E}_{\psi}(\Pe)}} 
	, 
\end{align*}
whereas the \emph{macroscopic scale} is given by the \emph{characteristic length on which $V$ varies}. In what follows, let $\eps$ denote the \emph{ratio of microscopic to macroscopic length scales}. Thus $\eps$ is a dimensionless quantity and $\ll 1$. Then if $V(x)$ varies on the scale $\order(1)$, $V(\eps x)$ varies on the scale $\order(\eps)$. Hence, the force $F_{\eps}(x) := - \nabla_x \bigl ( V(\eps x) \bigr ) = - \eps \nabla_x V(\eps x)$ generated by $V(\eps x)$ is \emph{weak} and the equations of motion are given by 
\begin{align*}
	\dot{x} &= \frac{p}{m} \\
	\dot{p} &= - \eps \nabla_x V(\eps x) 
	. 
\end{align*}
It is crucial to understand that these equations are written from the perspective of the quantum particle, \ie we use \emph{microscopic units of length!} If we rewrite these equations on the macroscopic length scale, we need to use $q := \eps x$ as position coordinate and we obtain 
\begin{align*}
	\dot{q} &= \eps \dot{x} = \eps \frac{p}{m} \\
	\dot{p} &= - \eps \nabla_q V(q) 
	. 
\end{align*}
Hence, forces are weak and we need to wait \emph{long times} to see their effect, \ie \emph{macroscopic times} $t := \nicefrac{s}{\eps}$. On this time and length scale, the extra factors of $\eps$ in the equations of motion cancel, 
\begin{align*}
	\dot{q} &= \frac{\dd}{\dd s} q = \eps \frac{\dd}{\dd t} q = \eps \frac{p}{m} \\
	\dot{p} &= \frac{\dd }{\dd s} p = \eps \frac{\dd}{\dd t} p = - \eps \nabla_q V(q) 
	. 
\end{align*}
In other words, upon \emph{rescaling space and time to macroscopic units}, we arrive at the usual equations of motion of a classical particle. Now that we have found the correct microscopic and macroscopic scales, we can write down the Schrödinger equation in \emph{microscopic units of length} and \emph{long, macroscopic times}. We will keep the physical constant $\hbar = \unit[1.0546]{Js}$ to emphasize that none of the arguments hinge on `$\hbar$ being small.' Now, the Schrödinger equation reads 
\begin{align*}
	i \eps \hbar \frac{\partial}{\partial t} \psi_{\eps}(t) &= \Bigl ( \frac{1}{2m} \bigl ( - i \hbar \nabla_x \bigr )^2 + V(\eps \hat{x}) \Bigr ) \psi_{\eps}(t) 
	, 
	&& 
	\psi_{\eps}(0) = \psi_0 \in \mathcal{D}(\hat{H}) 
	, 
\end{align*}
and the building block operators are 
\begin{align}
	\Qe_{\mathrm{micro}} &:= \eps \hat{x} \\
	\Pe_{\mathrm{micro}} &:= - i \hbar \nabla_x 
	\notag 
	. 
\end{align}
The index $\eps$ in $\psi_{\eps}$ indicates that the solutions to the Schrödinger equation depend parametrically on $\eps$. Alternatively, we could rescale space and introduce macroscopic position $q := \eps x$ which implies $\nabla_x = \eps \nabla_q$. Then in \emph{macroscopic units}, the Schrödinger equation reads 
\begin{align}
	i \eps \hbar \frac{\partial}{\partial t} \tilde{\psi}_{\eps}(t) &= \Bigl ( \frac{1}{2m} \bigl ( - i \eps \hbar \nabla_q \bigr )^2 + V(\hat{q}) \Bigr ) \tilde{\psi}_{\eps}(t) 
	, 
	&& 
	\tilde{\psi}_{\eps}(0) = \tilde{\psi}_0 \in \mathcal{D}(\hat{H}) 
	\label{semiclassics:eqn:Schroedinger_eqn_macroscopic}
	. 
\end{align}
Since we live in a macroscopic world, this scaling is in fact by far the most common choice. What is the relation between $\psi_{\eps}$ and $\tilde{\psi}_{\eps}$? $\abs{\psi_{\eps}(t,x)}^2$ is a probability \emph{density} and thus scales like $\nicefrac{1}{\mathrm{volume}} = \eps^{-d}$. This suggests to define the unitary rescaling operator $U_{\eps} : L^2(\R^d) \longrightarrow L^2(\R^d)$ as 
\begin{align*}
	\bigl ( U_{\eps} \psi \bigr )(x) := \eps^{- \nicefrac{d}{2}} \, \psi(\nicefrac{x}{\eps}) 
	, 
	&& 
	\psi \in L^2(\R^d)
	. 
\end{align*}
Indeed, the prefactor is consistent with the change of variables formula: for all $\varphi , \psi \in L^2(\R^d)$, we have 
\begin{align*}
	\bscpro{U_{\eps} \varphi}{U_{\eps} \psi} &= \int_{\R^d} \dd q \, \bigl ( U_{\eps} \varphi \bigr )(q)^* \, \bigl ( U_{\eps} \psi \bigr )(q) 
	= \int_{\R^d} \dd q \, \eps^{-d} \, \varphi^*(\nicefrac{q}{\eps}) \, \psi(\nicefrac{q}{\eps}) 
	\\
	&= \int_{\R^d} \dd x \, \varphi^*(x) \, \psi(x) = \scpro{\varphi}{\psi} 
	. 
\end{align*}
In this scaling, the building block operators are 
\begin{align}
	\Qe_{\mathrm{macro}} &:= \eps \hat{x} \\
	\Pe_{\mathrm{macro}} &:= - i \hbar \nabla_x 
	\notag 
	. 
\end{align}
Looking at the Schrödinger equation in the macroscopic scaling, equation~\eqref{semiclassics:eqn:Schroedinger_eqn_macroscopic}, one understands why the semiclassical limit is often thought of as the $\hbar \to 0$ limit. \emph{Formally}, the roles of $\eps$ and $\hbar$ are interchangeable, \emph{physically, they are not!} In fact, the precise origin of $\eps$ says something about the physical mechanism \emph{why} a system behaves almost classically. In the next Chapter, for instance, $\eps := \sqrt{\nicefrac{m}{M}}$ is given by the square root of the ratio of the masses of light to heavy particles and $\eps \to 0$ implies the dynamics of the heavy particles freezes out. 

Put together, we are now prepared to tackle the main questions of this section: 
\begin{enumerate}[(i)]
	\item What is the semiclassical limit? 
	\item Is it necessary to start with special `semiclassical' states? 
	\item How large is the error? 
	\item Do we make the semiclassical limit for observables or for states? 
\end{enumerate}
The first point can be understood by comparing the frameworks (Chapter~\ref{frameworks:comparison}) and using Weyl calculus as developed in Chapter~\ref{weyl_calculus}: we already know how to associate quantum observables to classical observables (Weyl quantization) and quantum states to quasi-classical states (Wigner transform). Hence, the task we are left with is a comparison of classical and quantum dynamics. In the Heisenberg picture where observables are evolved in time and states remain constant, we can thus compare the \emph{quantization of the classically evolved observable} 
\begin{align}
	F_{\mathrm{cl}}(t) := \Op \bigl (f (t) \bigr ) = \Op \bigl ( f \circ \phi_t \bigr ) 
	\label{semiclassics:eqn:F_cl}
\end{align}
and the quantum mechanically evolved observable 
\begin{align}
	F_{\mathrm{qm}}(t) := e^{+ i \frac{t}{\eps} \Op(h)} \, \Op(f) \, e^{- i \frac{t}{\eps} \Op(h)} 
	\label{semiclassics:eqn:F_qm}
\end{align}
Here $\phi_t$ is the hamiltonian flow associated to the hamiltonian function $h : \pspace \longrightarrow \R$ and equation~\eqref{frameworks:classical_mechanics:eqn:hamiltons_eom} and $e^{- i \frac{t}{\eps} \Op(h)}$ is the strongly continuous one-parameter time evolution group generated by $\Op(h)$ over macroscopic times. Ideally, we would like to study dynamics generated by fairly general hamiltonian functions (and associated operators) as well as observables. Typically, they are smooth, but have polynomially growth: the standard hamiltonian $h(x,\xi) = \tfrac{1}{2} \xi^2 + V(x)$, position $x$, momentum $\xi$ and angular momentum $L = x \wedge \xi$ are of this type. A theorem of this generality is out of reach with the tools we have at our disposal, so we restrict ourselves to the case $h , f \in \Schwartz(\pspace)$. The reader rest assured that these somewhat artificial restrictions can be lifted using so-called oscillatory integral techniques \cite{Robert:tour_semiclassique:1987,Robert:semiclassics:1998}. The general strategy remains the same, but each step in the derivation needs to be justified carefully. 
\begin{thm}[Semiclassical limit for observables]\label{semiclassics:thm:semiclassical_limit_observables}
	Let $h \in \Schwartz(\pspace)$ be a classical hamiltonian function and $\phi_t$ the hamiltonian flow associated to $h$ via equations~\eqref{frameworks:classical_mechanics:eqn:hamiltons_eom}. Then for any (real-valued) observable $f \in \Schwartz(\pspace)$, the classical evolution approximates the quantum evolution in the following sense: for any $T > 0$, there exists $C_T > 0$ such that $F_{\mathrm{qm}}(t)$ and $F_{\mathrm{cl}}(t)$ agree up to errors of order $\order(\eps^2)$ for all $t \in [- T , + T]$, 
	\begin{align}
		\bnorm{F_{\mathrm{qm}}(t) - F_{\mathrm{cl}}(t)}_{\mathcal{B}(L^2(\R^d))} &= \Bnorm{e^{+ i \frac{t}{\eps} \Op(h)} \, \Op(f) \, e^{- i \frac{t}{\eps} \Op(h)} - \Op \big ( f \circ \phi_t \bigr )}_{\mathcal{B}(L^2(\R^d))} 
		\notag \\
		&\leq C_T \eps^2 \abs{t}
		. 
	\end{align}
\end{thm}
\begin{remark}
	Question (ii) and (iii) are answered decisively: \emph{at no point do we make assumptions on the specific shape or form of the initial state}, the above estimate holds in the operator norm, \ie \emph{uniformly} for all initial states! Furthermore, we get an explicit upper bound for the error term in the proof. 
	
	It seems that the Heisenberg picture is singled out here, but we will also prove a semiclassical limit for states which follows from Theorem~\ref{semiclassics:thm:semiclassical_limit_observables} and Corollary~\ref{weyl_calculus:weyl_calculus:wigner_transform:cor:quantum_classical_expectation_value}. 
\end{remark}
The main ingredients in the proof are the so-called Duhamel trick (essentially an application of the fundamental theorem of calculus) as well as the asymptotic expansion of the Weyl product. A footnote to the educated reader: we purposely mimic the proof for symbols whose quantization define possibly unbounded selfadjoint operators and ignore some minor simplifications. 
\begin{proof}
	Let $T > 0$ be arbitrary, but fixed. As $h , f \in \Schwartz(\pspace)$ are real-valued and have an integrable symplectic Fourier transform, $\Op(h)$ and $\Op(f)$ define bounded selfadjoint operators on $L^2(\R^d)$ (Proposition~\ref{weyl_calculus:weyl_quantization:prop:application_Weyl_quantization} and Theorem~\ref{weyl_calculus:weyl_quantization:thm:properties_weyl_quantization}) and we do not have to consider domain questions. By Stone's theorem, $U(t) = e^{- i \frac{t}{\eps} \Op(h)}$ is a unitary, strongly continuous one-parameter evolution group. Hence, $F_{\mathrm{qm}}(t)$ defines a bounded operator on $L^2(\R^d)$. 
	
	By assumption on $h$, the components of the hamiltonian vector field and all its derivatives in $x$ and $\xi$ decay rapidly and are smooth. Thus, by Corollary~\ref{frameworks:classical_mechanics:cor:existence_flow} and Theorem~\ref{frameworks:classical_mechanics:thm:smoothness_flow}, the hamiltonian flow $\phi_t$ exists globally in time and is also smooth. By \cite[Lemma~III.9]{Robert:tour_semiclassique:1987} all derivatives of the flow are bounded uniformly in $t \in [-T , +T]$ and the composition $f(t) = f \circ \phi_t$ is again a Schwartz function.\footnote{The idea of the proof is to write down equations of motion for the first-order derivatives of $\phi_t$ which involve second-order derivatives of $h$. These are bounded by assumption and thus imply that all first-order derivatives of the flow $\phi_t$ are bounded by a constant depending only on time. One then proceeds by induction to show the boundedness of higher-order derivatives. } Since products of derivatives of Schwartz functions are Schwartz functions, we also have $\frac{\dd}{\dd t} f(t) = \{ h , f(t) \} \in \Schwartz(\pspace)$. 
	Hence, $F_{\mathrm{cl}}(t)$ also defines a bounded selfadjoint operator for all $t \in \R$. 
	
	We will now use the so-called Duhamel trick to rewrite the difference of $F_{\mathrm{qm}}(t)$ and $F_{\mathrm{cl}}(t)$ as the integral over the derivative of a function via the fundamental theorem of calculus: for any $\varphi \in \Schwartz(\R^d) \subset L^2(\R^d)$, we have 
	\begin{align*}
		F_{\mathrm{qm}}(t) \varphi - F_{\mathrm{cl}}(t) \varphi &= e^{+ i \frac{t}{\eps} \Op(h)} \, \Op(f) \, e^{- i \frac{t}{\eps} \Op(h)} \varphi - \Op \big ( f(t) \varphi \bigr ) 
		\\
		&= \int_0^t \dd s \, \frac{\dd}{\dd s} \left ( e^{+ i \frac{t}{\eps} \Op(h)} \, \Op \bigl ( f(t - s) \bigr ) \, e^{- i \frac{t}{\eps} \Op(h)} \right ) \varphi 
		\\
		&= \int_0^t \dd s \, e^{+ i \frac{s}{\eps} \Op(h)} \, \biggl ( \tfrac{i}{\eps} \Op(h) \, \Op \bigl ( f(t - s) \bigr ) 
		+ \biggr . \\
		\biggl . &\qquad \qquad \quad 
		+ \frac{\dd}{\dd s} \Op \bigl ( f(t - s) \bigr ) - \tfrac{i}{\eps} \Op(h) \, \Op \bigl ( f(t - s) \bigr ) \biggr ) \, e^{- i \frac{s}{\eps} \Op(h)} \varphi 
		. 
	\end{align*}
	The above integral exists as a Bochner integral, \ie an integral of vector-valued functions.\footnote{A measurable function $f : \R \longrightarrow \mathcal{X}$ that takes values in a Banach space $\mathcal{X}$ is called Bochner integrable (or integrable for short) if and only if $\norm{f}_{L^1(\R,\mathcal{X})} := \int_{\R} \dd t \, \norm{f(t)}_{\mathcal{X}} < \infty$. } 
	We can group the first and last term to give a commutator; furthermore, the commutator of the Weyl quantizations of two Schwartz functions is equal to the Weyl quantization of the \emph{Moyal commutator} and once we plug in the asymptotic expansion of the Weyl product, Theorem~\ref{weyl_calculus:asymptotics:thm:asymptotic_expansion}, we obtain 
	\begin{align*}
		\tfrac{i}{\eps} \bigl [ \Op(h) , \Op \bigl ( f(t - s) \bigr ) \bigr ] &= \tfrac{i}{\eps} \Op \bigl ( [ h , f(t-s) ]_{\Weyl} \bigr ) 
		= \tfrac{i}{\eps} \Op \bigl ( - i \eps \{ h , f(t-s) \} + \order(\eps^3) \bigr ) 
		\\
		&= \Op \bigl ( \{ h , f(t-s) \} \bigr ) + \order(\eps^2) 
		. 
	\end{align*}
	Since the Fourier transform $\Fs : \Schwartz(\pspace) \longrightarrow \Schwartz(\pspace)$ is a continuous map and all derivatives of the flow $\partial^{\alpha}_x \partial^{\beta}_{\xi} \phi_t$ are bounded functions, $\abs{\alpha} + \abs{\beta} \geq 1$, (Lemma~III.9 from \cite{Robert:tour_semiclassique:1987}) which can be bounded independently of $t \in [-T , +T]$, we conclude we can also interchange integration implicit in $\Op$ and differentiation with respect to $s$ as $f(t-s)$ and $\frac{\dd}{\dd s} f(t-s)$ can be bounded uniformly in $s \in [0 , t]$ and for all $t \in [-T , +T]$. Hence, we can write the difference between the quantum observable $F_{\mathrm{qm}}(t)$ and the classical observable $F_{\mathrm{cl}}(t)$ as 
	\begin{align*}
		\ldots &= \int_0^t \dd s \, e^{+ i \frac{s}{\eps} \Op(h)} \, \left ( \tfrac{i}{\eps} \bigl [ \Op(h) , \Op \bigl ( f(t - s) \bigr ) \bigr ] + \tfrac{\dd}{\dd s} \Op \bigl ( f(t - s) \bigr ) \right ) \, e^{- i \frac{s}{\eps} \Op(h)} \varphi 
		\\
		&= \int_0^t \dd s \, e^{+ i \frac{s}{\eps} \Op(h)} \, \Op \left ( \tfrac{i}{\eps} \bigl [ h , f(t - s) \bigr ]_{\Weyl} + \tfrac{\dd}{\dd s} \bigl ( f(t - s) \bigr ) \right ) \, e^{- i \frac{s}{\eps} \Op(h)} \varphi 
		. 
	\end{align*}
	By Proposition~\ref{frameworks:classical_mechanics:prop:equations_of_motion_Poisson_bracket}, 
	\begin{align*}
		\frac{\dd}{\dd s} f(t-s) &= - \frac{\dd}{\dd t} f(t - s) = - \bigl \{ h , f(t-s) \bigr \} 
	\end{align*}
	holds. On the other hand, we can expand the Moyal commutator with the help of the asymptotic expansion of the Weyl product, Theorem~\ref{weyl_calculus:asymptotics:thm:asymptotic_expansion}, 
	\begin{align*}
		\bigl [ h , f(t-s) \bigr ]_{\Weyl} &= - i \eps \bigl \{ h , f(t-s) \bigr \} + \order(\eps^3) 
		. 
	\end{align*}
	The error is really of \emph{third} order in $\eps$ as explained at the end of Chapter~\ref{weyl_calculus:asymptotics} and known explicitly. Put together, this implies 
	\begin{align*}
		\tfrac{i}{\eps} \bigl [ h , f(t - s) \bigr ]_{\Weyl} + \tfrac{\dd}{\dd s} \bigl ( f(t - s) \bigr ) &= \tfrac{i}{\eps} (- i \eps) \bigl \{ h , f(t-s) \bigr \} - \bigl \{ h , f(t-s) \bigr \} + \tfrac{i}{\eps} \order(\eps^3) 
		\\
		&
		= \order(\eps^2)
	\end{align*}
	where the right-hand side -- as a sum of Schwartz functions -- is again a Schwartz function. Thus, its quantization is a bounded operator on $L^2(\R^d)$ and the estimate 
	\begin{align*}
		\Bnorm{\Op \Bigl ( \tfrac{i}{\eps} \bigl [ h , f(t - s) \bigr ]_{\Weyl} + \tfrac{\dd}{\dd s} \bigl ( f(t - s) \bigr ) \Bigr )}_{\mathcal{B}(L^2(\R^d))} \leq C_T \eps^2 
	\end{align*}
	holds for all $t \in [-T , +T]$ and some constant $C_T > 0$ that depends on $T$. Taking the operator norm of $F_{\mathrm{qm}}(t) - F_{\mathrm{cl}}(t)$, we conclude 
	\begin{align*}
		\bnorm{F_{\mathrm{qm}}(t) \varphi - F_{\mathrm{cl}}(t)}_{\mathcal{B}(L^2(\R^d))} &\leq \int_0^t \dd s \, \Bnorm{\Op \Bigl ( \tfrac{i}{\eps} \bigl [ h , f(t - s) \bigr ]_{\Weyl} + \tfrac{\dd}{\dd s} \bigl ( f(t - s) \bigr ) \Bigr )}_{\mathcal{B}(L^2(\R^d))} 
		\\
		&\leq \int_0^t \dd s \, C_T \eps^2 = C_T \eps^2 \abs{t} 
		. 
	\end{align*}
	This finishes the proof. 
\end{proof}
\begin{example}
	Even though the following example does not satisfy the assumption of the preceding theorem, it illustrates the semiclassical limit for a realistic hamiltonian and obsersable. Consider $h(x,\xi) = \tfrac{1}{2} \xi^2 + V(x)$ for some smooth bounded potential $V$ with bounded derivatives up to any order. Then the $l$th component of momentum $\xi_j$ has a good semiclassical limit: we can (formally) compute 
	\begin{align*}
		\frac{\dd}{\dd t} \Pe_j(t) &= \frac{i}{\eps} [ \Op(h) , \Pe_j(t) ] = \frac{i}{\eps} e^{+ i \frac{t}{\eps} \Op(h)} \, [\Op(h) , \Pe_j] \, e^{- i \frac{t}{\eps} \Op(h)} 
		= - \partial_{x_j} V(\Qe(t)) 
		\\
		&= \frac{i}{\eps} e^{+ i \frac{t}{\eps} \Op(h)} \, \Op \bigl ( [h , \xi_j]_{\Weyl} \bigr ) \, e^{- i \frac{t}{\eps} \Op(h)} 
		\\
		&
		= \frac{i}{\eps} e^{+ i \frac{t}{\eps} \Op(h)} \, \Op \bigl ( + i \eps \partial_{x_j} V \bigr ) \, e^{- i \frac{t}{\eps} \Op(h)} + \tfrac{i}{\eps} \order(\eps^3) 
		\\
		&= e^{+ i \frac{t}{\eps} \Op(h)} \, \Op \bigl ( - \partial_{x_j} V \bigr ) \, e^{- i \frac{t}{\eps} \Op(h)} + \order(\eps^2) 
	\end{align*}
	and 
	\begin{align*}
		\dot{\xi}_j(t) = \{ h , \xi_j(t) \} = - \partial_{x_j} V(x(t)) 
	\end{align*}
	explicitly. Naïvely, one may think that there is no error term in the Moyal commutator since $\xi_j$ is a polynomial of first order and thus the asymptotic expansion of the Weyl commutator terminates after the first-order correction. However, we plug in the \emph{time-evolved observable} $\xi_j(t)$ which clearly depends on the initial conditions $(x_0,\xi_0) \in \pspace$ chosen. Hence, in general, the higher-order terms in the Moyal commutator $[h , \xi_j(t)]_{\Weyl} = - i \eps \{ h , \xi_j(t) \} + \order(\eps^3)$ are non-zero. 
	
	Repeating the calculations in the proof (and ignoring technical problems), we arrive at 
	\begin{align*}
		&\bnorm{\Pe_{\mathrm{qm} \, j}(t) \varphi - \Pe_{\mathrm{cl} \, j}(t)}_{\mathcal{B}(L^2(\R^d))} \leq \int_0^t \dd s \, \Bnorm{\Op \Bigl ( \tfrac{i}{\eps} \bigl [ h , \xi_j(t - s) \bigr ]_{\Weyl} + \tfrac{\dd}{\dd s} \bigl ( \xi_j(t - s) \bigr ) \Bigr )}_{\mathcal{B}(L^2(\R^d))} 
		\\
		&\qquad \qquad 
		= \int_0^t \dd s \, \Bnorm{\Op \Bigl ( - \partial_{x_j} V(x(t-s)) + \order(\eps^2) + \partial_{x_j} V(x(t-s) \Bigr )}_{\mathcal{B}(L^2(\R^d))} 
		\\
		&\qquad \qquad 
		\leq \int_0^t \dd s \, C_T \eps^2 = C_T \eps^2 \abs{t} 
	\end{align*}
	Since the derivatives of the potential up to third order (which are needed to control the remainder) are bounded, $C_T < \infty$ may be chosen. Hence, at least formally, momentum has a good semiclassical limit. Similarly, position also has a good semiclassical limit. 
\end{example}
The semiclassical limit for states is a consequence of the semiclassical limit for observables. 
\begin{cor}[Semiclassical limit for states]
	Let $h \in \Schwartz(\pspace)$ and $T > 0$. For any $\psi \in L^2(\R^d)$, we define $\psi(t) := e^{- i \frac{t}{\eps} \Op(h)} \psi$. Then the classically evolved signed probability measure 
	\begin{align}
		\mu_{\mathrm{cl}}(t) := (2 \pi)^{- \nicefrac{d}{2}} \, \WignerTrafo(\psi,\psi) \circ \phi_{-t} 
	\end{align}
	approximates the quantum mechanically evolved state 
	\begin{align}
		\mu_{\mathrm{qm}}(t) := (2 \pi)^{- \nicefrac{d}{2}} \, \WignerTrafo(\psi(t),\psi(t)) 
	\end{align}
	up to errors of $\order(\eps^2)$ in the sense that 
	\begin{align}
		\int_{\pspace} \dd X \, f(X) \, \Bigl ( \bigl ( \mu_{\mathrm{qm}}(t) \bigr )(X) - \bigl ( \mu_{\mathrm{cl}}(t) \bigr )(X) \Bigr ) = \order(\eps^2)
		&&
		\forall f \in \Schwartz(\pspace)
		\label{semiclassics:eqn:O_eps2_closeness_states}
	\end{align}
	holds for any $t \in [-T , +T]$. 
\end{cor}
Note that this is a statement on the time evolution of states: the observable $f$ only `tests' the difference between the quantum state $\mu_{\mathrm{qm}}(t) = (2 \pi)^{- \nicefrac{d}{2}} \, \WignerTrafo(\psi(t),\psi(t))$ and the classically evolved state $\mu_{\mathrm{cl}}(t) = (2 \pi)^{- \nicefrac{d}{2}} \, \WignerTrafo(\psi,\psi) \circ \phi_{-t}$. 
\begin{proof}
	By Definition of the Wigner transform and Corollary~\ref{weyl_calculus:weyl_calculus:wigner_transform:cor:quantum_classical_expectation_value}, for any $f \in \Schwartz(\pspace)$, we can rewrite 
	\begin{align*}
		\bscpro{\psi}{F_{\mathrm{qm}}(t) \psi} &= \bscpro{\psi(t)}{\Op(f) \psi(t)} 
		= \frac{1}{(2\pi)^{\nicefrac{d}{2}}} \int_{\pspace} \dd X \, \WignerTrafo(\psi(t),\psi(t)) \, f(X) 
	\end{align*}
	as a classical expectation value with respect to the signed probability measure $\mu_{\mathrm{qm}}(t) = (2\pi)^{- \nicefrac{d}{2}} \, \WignerTrafo(\psi(t),\psi(t))$. Similarly, the expectation value of the classically evolved observable can be rewritten as 
	\begin{align*}
		\bscpro{\psi}{F_{\mathrm{cl}}(t) \psi} &= \bscpro{\psi}{\Op \bigl ( f(t) \bigr ) \psi} 
		= \frac{1}{(2\pi)^{\nicefrac{d}{2}}} \int_{\pspace} \dd X \, f \circ \phi_t (X) \, \bigl ( \WignerTrafo(\psi,\psi) \bigr )(X) 
		. 
	\end{align*}
	An application of Liouville's theorem, Theorem~\ref{frameworks:classical:thm:Liouville}, yields 
	\begin{align*}
		\bscpro{\psi}{F_{\mathrm{cl}}(t) \psi} &= \frac{1}{(2\pi)^{\nicefrac{d}{2}}} \int_{\pspace} \dd X \, f(X) \, \bigl ( \WignerTrafo(\psi,\psi) \circ \phi_{-t} \bigr )(X) 
		= \int_{\pspace} \dd X \, f(X) \, \bigl ( \mu_{\mathrm{cl}}(t) \bigr )(X)
	\end{align*}
	Hence, by Theorem~\ref{semiclassics:thm:semiclassical_limit_observables}, $\mu_{\mathrm{qm}}(t) - \mu_{\mathrm{cl}}(t) = \order(\eps^2)$ in the sense of equation~\eqref{semiclassics:eqn:O_eps2_closeness_states}. 
\end{proof}
One last word on wavepackets: there are three general types of wavepackets \cite[pp.59--60]{Teufel:adiabaticPerturbationTheory:2003}: 
\begin{enumerate}[(i)]
	\item Localization in space \emph{and} reciprocal space: let $\psi \in L^2(\R^d)$ satisfy some mild technical assumptions ($\psi \in \Schwartz(\R^d)$ works, but more general $L^2$-functions are admissible as well). Then the $L^2$-function 
	\begin{align*}
		\psi^{\eps}_{x_0,\xi_0}(x) := \eps^{- \nicefrac{d}{4}} e^{- \frac{i}{\eps} \xi_0 \cdot (x - x_0)} \, \psi \bigl ( \tfrac{x - x_0}{\sqrt{\eps}} \bigr ) 
	\end{align*}
	is sharply peaked at $x_0$ for $\eps$ small while its Fourier transform is peaked around $\xi_0$. \emph{Wave packets of this type track the trajectory associated to the initial conditions $(x_0,\xi_0) \in \pspace$ and energy $E_0 = h(x_0,\xi_0)$.} Typical examples are \emph{Gaussian wavepackets} where $\psi(x) = \frac{2^{\nicefrac{d}{2}}}{\pi^{\nicefrac{d}{2}}} e^{- \frac{x^2}{2}}$. 
	\item Localization in space \emph{or} reciprocal space: for any $\psi \in L^2(\R^d)$, 
	\begin{align*}
		\hat{\psi}^{\eps}_{\xi_0}(\xi) := \hat{\psi}(\xi - \nicefrac{\xi_0}{\eps}) 
	\end{align*}
	is peaked around $\xi_0$ for $\eps \ll 1$ while its probability distribution in real space is given by $\abs{\psi(x)}^2$, \ie 
	\begin{align*}
		\bigl ( \WignerTrafo(\psi^{\eps}_{\xi_0},\psi^{\eps}_{\xi_0}) \bigr )(x,\xi) \xrightarrow{\eps \to 0} (2\pi)^{\nicefrac{d}{2}} \, \abs{\psi(x)}^2 \, \delta(\xi - \xi_0) 
	\end{align*}
	in the weak sense. Similarly, the classical pseudo-probability distribution associated to 
	\begin{align*}
		\psi_{x_0}^{\eps}(x) := \eps^{- \nicefrac{d}{2}} \psi \bigl ( \tfrac{x - x_0}{\eps} \bigr ) 
	\end{align*}
	is $\sabs{\hat{\psi}(\xi)}^2 \delta(x - x_0)$. In both cases, one deals with a \emph{distribution of initial states} in either real or reciprocal space. 
	\item \emph{WKB ansatz:} here, we assume the wave function is of the form 
	\begin{align*}
		\psi_{\eps}(x) := \sqrt{\rho(x)} \, e^{- \frac{i}{\eps} S(x)} 
		. 
	\end{align*}
	The ansatz here is that the wave function associated to the probability distribution $\rho$ has fast oscillations superimposed to a slowly varying enveloping function $\sqrt{\rho}$. Then 
	\begin{align*}
		\bigl ( \WignerTrafo(\psi_{\eps},\psi_{\eps}) \bigr )(x,\xi) \xrightarrow{\eps \to 0} (2\pi)^{\nicefrac{d}{2}} \, \rho(x) \, \delta(\xi - \nabla_x S(x))
	\end{align*}
	holds in the weak sense as has been computed in exercise~27. This means, we start with an initial distribution of states in real space, each of which follows the classical trajectory associated to the classical action $S$. 
\end{enumerate}
%


\chapter{Theory of multiscale systems} 
\label{multiscale}

There are plenty of examples of systems with two or more inherent scales: 
\begin{enumerate}[(i)]
	\item A classical spinning top: the Earth, for instance, rotates around its own axis about once every 24 hours (to be more exact, it is approximately $23$ hours, $56$ minutes and $4.1$ seconds). On the other hand, the axis of rotation precesses around the axis of precession about once ever $25,700-25,800$ years. 
	\item A quantum spinning top is realized by systems with $L \cdot S$ coupling in strong magnetic fields. Here, $L_z$ and $S_z$ are ``almost good quantum numbers'' \cite[p.~311]{Sakurai:modern_qm:1994} if the magnetic field $B$ is strong. 
	\item The Dirac equation as treated in Chapter~\ref{weyl_calculus:diag_dirac} \cite{FuerstLein:nonRelLimitDirac:2008}. 
	\item Currents in crystalline solids including a Kubo formula which explains the quantization of the Hall current \cite{PST:effDynamics:2003}. 
	\item Polarization and piezoelectricity of crystalline solids \cite{Lein:polarization:2005,PanatiSparberTeufel:polarization:2006}. 
	\item Physics of molecules in the Born-Oppenheimer approximation \cite{PST:Born-Oppenheimer:2007}. 
	\item Time-adiabatic problems of quantum mechanics \cite[Chapter~4.3]{PST:sapt:2002}, \cite{Kato:perturbation_theory:1966}. 
\end{enumerate}
Our treatment is based on the seminal work of Panati, Spohn and Teufel \cite{PST:sapt:2002,Teufel:adiabaticPerturbationTheory:2003}, \emph{space-adiabatic perturbation theory} and we will introduce it by a specific example:

\section{Physics of molecules: the Born-Oppenheimer approximation} 
\label{}

Consider a molecule in $\R^3$ consisting of $N$ nuclei with masses $m_n$, $n \in \{ 1 , \ldots , N \}$, and $K$ electrons. If we consider electrons and nuclei as non-relativist particles, then, quantum mechanically, the dynamics is generated by the Born-Oppenheimer hamiltonian: 
\begin{align}
	\opHBO := \sum_{n = 1}^N \frac{1}{2 m_n} (- i \hbar \nabla_{x_n})^2 &+ \sum_{k = 1}^K \frac{1}{2 m_{\mathrm{e}}} (- i \hbar \nabla_{y_k})^2 
	+ \notag \\
	&+ V_{\mathrm{e-e}}(\hat{y}) + V_{\mathrm{e-nuc}}(\hat{x} , \hat{y}) + V_{\mathrm{nuc-nuc}}(\hat{x}) 
	\notag \\
	=: \sum_{n = 1}^N \frac{1}{2 m_n} (- i \hbar \nabla_{x_n})^2 &+ \He(\hat{x})
	\label{multiscale:physics_of_molecules:eqn:naive_Born_Oppenheimer_hamiltonian}
\end{align}
The potential energy has three distinct contributions: an electron-electron interaction term, an electron-nucleon interaction term and a nucleon-nucleon interaction term. Here, $x_n \in \R^3$ and $y_k \in \R^3$ are the coordinates in $\R^3$ of the $n$th nucleon and $k$th electron, respectively. 

If, for instance, one assumes that the particles interact via the Coulomb potential and the nuclei have (effective) charges $Z_n$, $n \in \{ 1 , \ldots , N \}$, then the potentials are given by 
\begin{align*}
	V_{\mathrm{e-e}}(\hat{y}) &:= \frac{1}{2} \sum_{j \neq k} \frac{1}{\hat{y}_j - \hat{y}_k} \\
	V_{\mathrm{e-nuc}}(\hat{x} , \hat{y}) &:= \sum_{k = 1}^K \sum_{n = 1}^N \frac{Z_n}{\hat{y}_k - \hat{x}_n} \\
	V_{\mathrm{nuc-nuc}}(\hat{x}) &:= \frac{1}{2} \sum_{l \neq n} \frac{Z_l Z_n}{\hat{x}_l - \hat{x}_n} 
	. 
\end{align*}
It is often advantageous both, for practical and theoretical considerations, to smudge out the charge so as to get rid of the Coulomb singularities. Also, the bare Coulomb potential is often replaced by an effective nucleonic potential that takes into account the fact that electrons associated to lower shells do not play a role in chemical reactions. We will always assume that the potentials are smooth functions in $x$. 

The time evolution is generated by the Schrödinger equation, 
\begin{align*}
	i \hbar \frac{\partial }{\partial t} \Psi(t) &= \opHBO \Psi(t) 
	, 
	&& 
	\Psi(0) = \Psi_0 \in L^2(\R^{3N} \times \R^{3K}) 
	. 
\end{align*}
To simplify the discussion, we have neglected here that electrons are fermions and as such obey the Pauli exclusion principle, \ie $\Psi_0$ and thus also $\Psi(t)$ should be antisymmetric when exchanging two electron coordinates $y_k$ and $y_l$. Before we continue our analysis, we will simplify the hamiltonian: first of all, we choose units such that $\hbar = 1$ and by a suitable rescaling, we may think of the nuclei having the same mass $m_{\mathrm{nuc}}$. Then the Born-Oppenheimer hamiltonian can be written as 
\begin{align*}
	\opHBO &= \frac{1}{2 m_{\mathrm{e}}} \sum_{n = 1}^N \Bigl ( - i \sqrt{\tfrac{m_{\mathrm{e}}}{m_{\mathrm{nuc}}}} \nabla_{x_n} \Bigr )^2 + \frac{1}{2 m_{\mathrm{e}}} \sum_{k = 1}^K (- i \nabla_{y_k})^2 
	+ \\
	&\qquad \qquad 
	+ V_{\mathrm{e-e}}(\hat{y}) + V_{\mathrm{e-nuc}}(\hat{x} , \hat{y}) + V_{\mathrm{nuc-nuc}}(\hat{x}) 
	. 
\end{align*}
This suggests to introduce $\eps := \sqrt{\nicefrac{m_{\mathrm{e}}}{m_{\mathrm{nuc}}}} \ll 1$ as the small parameter and by choosing units such that $m_{\mathrm{e}} = 1$, we finally obtain 
\begin{align*}
	\opHBO &= \tfrac{1}{2} \bigl ( - i \eps \nabla_x \bigr )^2 + \tfrac{1}{2} \bigl ( - i \nabla_y \bigr )^2 + V_{\mathrm{e-e}}(\hat{y}) + V_{\mathrm{e-nuc}}(\hat{x} , \hat{y}) + V_{\mathrm{nuc-nuc}}(\hat{x}) 
	\\
	&=: \tfrac{1}{2} \Pe^2 + \He(\Qe) 
	. 
\end{align*}
Physically, $\eps$ is always small since a proton is about $2,000$ times heavier than an electron. This means even in the worst case, $\eps \lesssim \nicefrac{1}{44}$. Typical atomic cores contain tens of hadrons, \ie $m_{\mathrm{nuc}} \simeq 10^4-10^5 m_{\mathrm{e}}$, which means that very often the numerical value of $\eps$ is significantly smaller than $\nicefrac{1}{44}$. Hence, we have introduced the usual building block observables 
\begin{align}
	\Pe &:= - i \eps \nabla_x 
	\label{multiscale:physics_of_molecules:eqn:building_block_observables}
	\\
	\Qe &:= \hat{x} \notag 
\end{align}
and we can use the results from Chapter~\ref{weyl_calculus}. 

The physical intuition tells us that electrons adjust instantaneously to the configuration of the nuclei, \ie they are enslaved by the nuclei, and $\Pe$ and $\Qe$ are \emph{approximate constants of motion}, 
\begin{align*}
	[\opHBO , \Pe] &= \order(\eps) \\
	[\opHBO , \Qe] &= \order(\eps) 
	. 
\end{align*}
The strategy to solve this from a physicist's perspective is the following: \marginpar{\small 2010.01.20}
\begin{enumerate}[(i)]
	\item For fixed nucleonic positions, one solves the eigenvalue problem 
	\begin{align*}
		\He(x) \varphi_n(x) &= E_n(x) \varphi_n(x) 
	\end{align*}
	where the nucleonic positions $x = (x_1 , \ldots , x_N) \in \R^{3N}$ are regarded as \emph{parameters} and $\varphi_n(x) \in \Hfast := L^2(\R^{3K})$ is an electronic wave function. More properly, one looks at the spectrum of $\He(x)$ for each value of $x$ which gives rise to energy bands (bound states where the molecule exists) and usually also a part with continuous spectrum (which means the molecule has dissociated). Typically, the energy values $E_n$ are sorted by size, \ie 
	\begin{align*}
		E_1(x) \leq E_2(x) \leq \ldots 
	\end{align*}
	\begin{figure}
		\hfil\includegraphics[height=5cm]{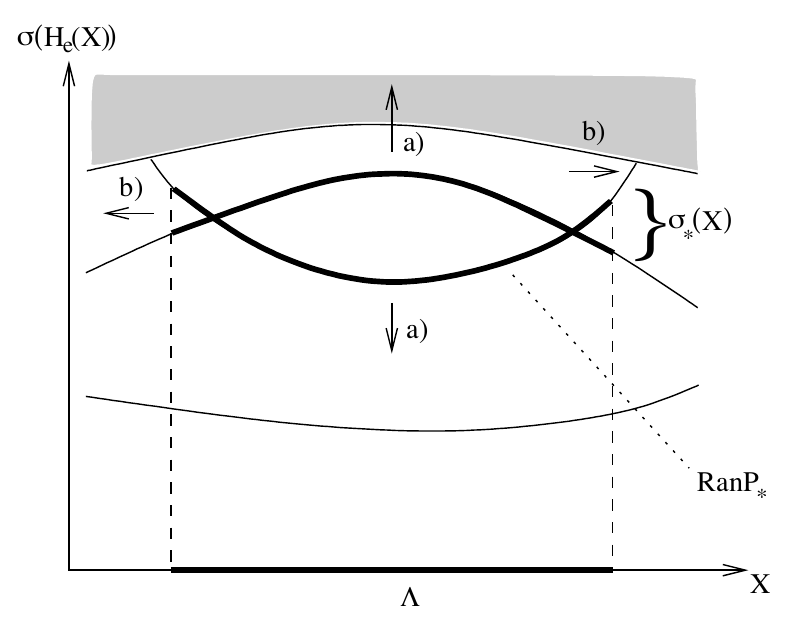}\hfil
		\caption{A simple example of a band spectrum. The relevant part of the spectrum is separated from the remainder by a gap \cite{Spohn_Teufel:time-dep_BO_approx:2001}. }
		\label{multiscale:physics_of_molecules:figure:band_structure}
	\end{figure}
	\item For physical reactions, there are usually \emph{``relevant bands,''} \eg the ground state and the first excited state. These two need to be \emph{separated by a gap from the rest of the spectrum} so that one may neglect band transitions from the ``relevant bands'' to the rest. 
	\item The last step is an \emph{``adiabatic approximation:''} one can neglect band transitions to other bands if the total energy is bounded independently of $\eps$, \eg if 
	\begin{align*}
		\norm{\Pe \Psi_0} = \norm{\eps \nabla_x \Psi_0} = \order(1) 
		. 
	\end{align*}
	Then there are several options: 
	\begin{enumerate}[(a)]
		\item One treats the nucleons \emph{classically}, \eg in case of an isolated relevant band, one considers the hamilton function 
		\begin{align*}
			H_{\mathrm{cl}}(x,\xi) := \tfrac{1}{2} \xi^2 + E_1(x) 
		\end{align*}
		where the ground state energy band $E_1$ acts as an \emph{effective potential}. For us, the semiclassical approximation is a distinct and separate step (see Chapter~\ref{multiscale:semiclassics}). 
		\item Or one treats nucleons quantum mechanically. Usually one does a product ansatz here: let $\chi \in L^2(\R^{3N})$ be a nucleonic wave function and $\varphi_1(x) \in L^2(\R^{3K})$ be the electronic ground state. Then we define the product wave function 
		\begin{align*}
			\Psi(x,y) := \chi(x) \, \varphi_1(x,y) 
			. 
		\end{align*}
		Related to this, we define the projection onto the electronic ground state 
		\begin{align*}
			\pi_0(x) := \sopro{\varphi_1(x)}{\varphi_1(x)}_{L^2(\R^{3K})} 
			. 
		\end{align*}
		For each $x$, $\pi_0(x)$ is an operator on the electronic subspace $L^2(\R^{3K}) = \Hfast$. Its quantization $\pi_0(\Qe)$ is an operator on the full Hilbert space $L^2(\R^{3N} \times \R^{3K})$ which maps each $\Phi \in L^2(\R^{3N} \times \R^{3K})$ onto a product wave function, 
		\begin{align*}
			\bigl ( \pi_0(\Qe) \Phi \bigr )(x,y) = \underbrace{\bscpro{\varphi_1(x)}{\Phi(x,\cdot)}_{L^2(\R^{3K})}}_{\in L^2(\R^{3N})} \, \varphi_1(x,y) 
			. 
		\end{align*}
		In order to compute how $\opHBO$ acts on product wave functions 
		\begin{align*}
			\Psi(x,y) = \chi(x) \, \varphi_1(x,y) 
			, 
		\end{align*}
		we compute the expectation value: since 
		\begin{align*}
			\bscpro{\Psi}{\opHBO \Psi} &= \tfrac{1}{2} \bscpro{\Psi}{\Pe^2 \Psi} + \bscpro{\Psi}{\He(\Qe) \Psi} 
			= \tfrac{1}{2} \bscpro{\Pe \Psi}{\Pe \Psi} + \bscpro{\Psi}{\He(\Qe) \Psi} 
		\end{align*}
		holds, we can compute the expectation values separately for kinetic and potential energy: the potential energy term can be inferred from 
		\begin{align*}
			\bigl ( \He(\Qe) \chi \varphi_1 \bigr )(x,y) &= \chi(x) \, \He(x) \varphi_1(x,y) = E_1(x) \, \chi(x) \varphi_1(x,y) 
			\\
			&= \bigl ( E_1(\Qe) \chi \varphi_1 \bigr )(x,y) 
			. 
		\end{align*}
		The kinetic energy term splits in four: using the symmetry of $\Pe$, we get 
		\begin{align*}
			\bscpro{\Pe \Psi}{\Pe \Psi} &= \int_{\R^{3N}} \dd x \int_{\R^{3K}} \dd y \, \bigl ( \Pe \chi \varphi_1 \bigr )^* (x,y) \, \bigl ( \Pe \chi \, \varphi_1 \Bigr )(x,y) 
			\\
			&= \int_{\R^{3N}} \dd x \int_{\R^{3K}} \dd y \, \Bigl ( \bigl ( \Pe \chi \bigr )^*(x) \, \varphi_1^*(x,y) + \chi^*(x) \, \bigl ( \Pe \varphi_1 \bigr )^*(x,y) \Bigr ) 
			\cdot \\
			&\qquad \qquad \qquad \qquad \cdot 
			\Bigl ( \bigl ( \Pe \chi \bigr )(x) \, \varphi_1(x,y) + \chi(x) \, \bigl ( \Pe \varphi_1 \bigr )(x,y) \Bigr ) 
			. 
		\end{align*}
		Now we rewrite this expression by introducing the scalar product on $L^2(\R^{3K})$, using that the fast state $\varphi_1(x)$ is normalized for all $x \in \R^{3N}$, 
		\begin{align*}
			\scpro{\varphi_1(x)}{\varphi_1(x)}_{L^2(\R^{3K})} = 1 
			, 
		\end{align*}
		and by introducing the \emph{Berry phase} \cite{Berry:quantal_phase_factors:1984,Simon:holonomy:1983,Bohm_Mostafazadeh_Koizumi_Niu_Zwanziger:geometric_phase:2003} 
		\begin{align}
			\mathcal{A}(x) := i \scpro{\varphi_1(x)}{\nabla_x \varphi_1(x)}_{L^2(\R^{3K})} 
			. 
		\end{align}
		This leads to the following simplifications if we complete the square: 
		\begin{align*}
			\bscpro{\Pe \Psi}{\Pe \Psi} &= \int_{\R^{3N}} \dd x \, \Bigl ( \bigl ( \Pe \chi \bigr )^*(x) \, \bigl ( \Pe \chi \bigr )(x) 
			+ \Bigr . \\
			&\qquad \qquad \qquad 
			- \bigl ( \Pe \chi \bigr )^*(x) \cdot \eps \mathcal{A}(x) \, \chi(x) 
			- \chi^*(x) \, \eps \mathcal{A}^*(x) \cdot \bigl ( \Pe \chi \bigr )(x) 
			+ \\
			&\qquad \qquad \qquad 
			+ \eps^2 \chi^*(x) \, \bscpro{\nabla_x \varphi_1(x)}{\cdot \nabla_x \varphi_1(x)}_{L^2(\R^{3K})} \, \chi(x)
			\Bigr ) 
			\\
			&= \bscpro{\chi}{\bigl ( \Pe - \eps \mathcal{A}(\Qe) \bigr )^2 \chi}_{L^2(\R^{3N})} 
			+ \\
			&\qquad 
			- \eps^2 \bscpro{\chi}{\mathcal{A}(\Qe)^2 \chi}_{L^2(\R^{3N})} + \eps^2 \Bscpro{\chi}{\scpro{\nabla_x \varphi_1(\Qe)}{\cdot \nabla_x \varphi_1(\Qe)} \chi}_{L^2(\R^{3N})} 
			\\
			&= \bscpro{\chi}{\bigl ( \Pe - \eps \mathcal{A}(\Qe) \bigr )^2 \chi}_{L^2(\R^{3N})} 
			+ \\
			&\qquad 
			+ \eps^2 \Bscpro{\chi}{\bscpro{\nabla_x \varphi_1(\Qe)}{(\id - \sopro{\varphi_1(\Qe)}{\varphi_1(\Qe)}) \nabla_x \varphi_1(\Qe)} \chi}_{L^2(\R^{3N})} 
			. 
		\end{align*}
		This suggests to propose 
		\begin{align}
			\pi_0(\Qe) \opHBO \pi_0(\Qe) &= \tfrac{1}{2} \bigl ( \Pe - \eps \mathcal{A}(\Qe) \bigr )^2 + E_1(\Qe) + \tfrac{\eps^2}{2} \phi(\Qe) 
			\label{multiscale:physics_of_molecules:eqn:guess_effective_hamiltonian}
		\end{align}
		as \emph{effective hamiltonian} for the nuclei where the last term, 
		\begin{align}
			\phi(x) := \Bscpro{\nabla_x \varphi_1(x)}{\cdot (\id - \sopro{\varphi_1(x)}{\varphi_1(x)}) \nabla_x \varphi_1(x)}_{L^2(\R^{3K})} 
		\end{align}
		is the so-called \emph{Born-Huang potential}. 
	\end{enumerate}
\end{enumerate}
What are the problems associated to this simplistic approach? 
\begin{enumerate}[(i)]
	\item The nuclei move slowly, their typical velocities are of order $\order(\eps)$, and we have to wait for times of order $\order(\nicefrac{1}{\eps})$ to see motion on the scale $\order(1)$. Hence, it is necessary to include all $\order(\eps)$ corrections. 
	\item Since $[\pi_0(\Qe) , \opHBO] = \order(\eps)$, after times of order $\order(\nicefrac{1}{\eps})$, has the initial product wave function evolved into an approximate product wave function? If so, how big is the error after macroscopic times? In other words, does 
	\begin{align*}
		\bigl ( e^{- i \frac{t}{\eps} \opHBO} \chi \varphi_1 \bigr )(x,y) = \tilde{\chi}(x) \, \varphi_1(x,y) + \order(\eps)
	\end{align*}
	hold for some $\tilde{\chi} \in L^2(\R^{3N})$? Naïvely, the answer would be no: if we compare the full evolution $e^{- i \frac{t}{\eps} \opHBO}$ to that generated by $\pi_0(\Qe) \opHBO \pi_0(\Qe)$, then a naïve approximation yields 
	\begin{align*}
		\Bigl ( e^{- i \frac{t}{\eps} \opHBO} &- e^{- i \frac{t}{\eps} \pi_0(\Qe) \opHBO \pi_0(\Qe)} \Bigr ) \pi_0(\Qe) = 
		\\
		&= \int_0^{\nicefrac{t}{\eps}} \dd s \, \frac{\dd }{\dd s} \Bigl ( e^{- i s \opHBO} e^{- i (\frac{t}{\eps} - s) \pi_0(\Qe) \opHBO \pi_0(\Qe)} \Bigr ) \pi_0(\Qe) 
		\\
		&= \int_0^{\nicefrac{t}{\eps}} \dd s \, e^{- i s \opHBO} \, \Bigl ( - i \opHBO \pi_0(\Qe) + i \pi_0(\Qe) \opHBO \pi_0(\Qe)^2 \Bigr ) 
		\, 
		e^{- i (\frac{t}{\eps} - s) \pi_0(\Qe) \opHBO \pi_0(\Qe)} 
		\\
		&= \int_0^{\nicefrac{t}{\eps}} \dd s \, e^{- i s \opHBO} \, i \underbrace{\bigl [ \opHBO , \pi_0(\Qe) \bigr ]}_{= \order(\eps)} \, \pi_0(\Qe) e^{- i (\frac{t}{\eps} - s) \pi_0(\Qe) \opHBO \pi_0(\Qe)} 
		\\
		&= \order(\nicefrac{1}{\eps}) \cdot \order(\eps) = \order(1) 
		. 
	\end{align*}
	This is \emph{not good enough} for our intents and purposes since the errors may sum up to order $\order(1)$. In reality, the integral is highly oscillatory and we will show in Chapter~\ref{multiscale:effective_quantum_dynamics:effective_dynamics} that the errors stay $\order(\eps)$ small. The way one should think about it is this: the wave function oscillates quickly in and out of the space of product states $\pi_0(\Qe) L^2(\R^{3N} \times \R^{3K})$ (see Figure~\ref{multiscale:physics_of_molecules:figure:oscillations}). 
	\begin{figure}
		\hfil\includegraphics[height=4cm]{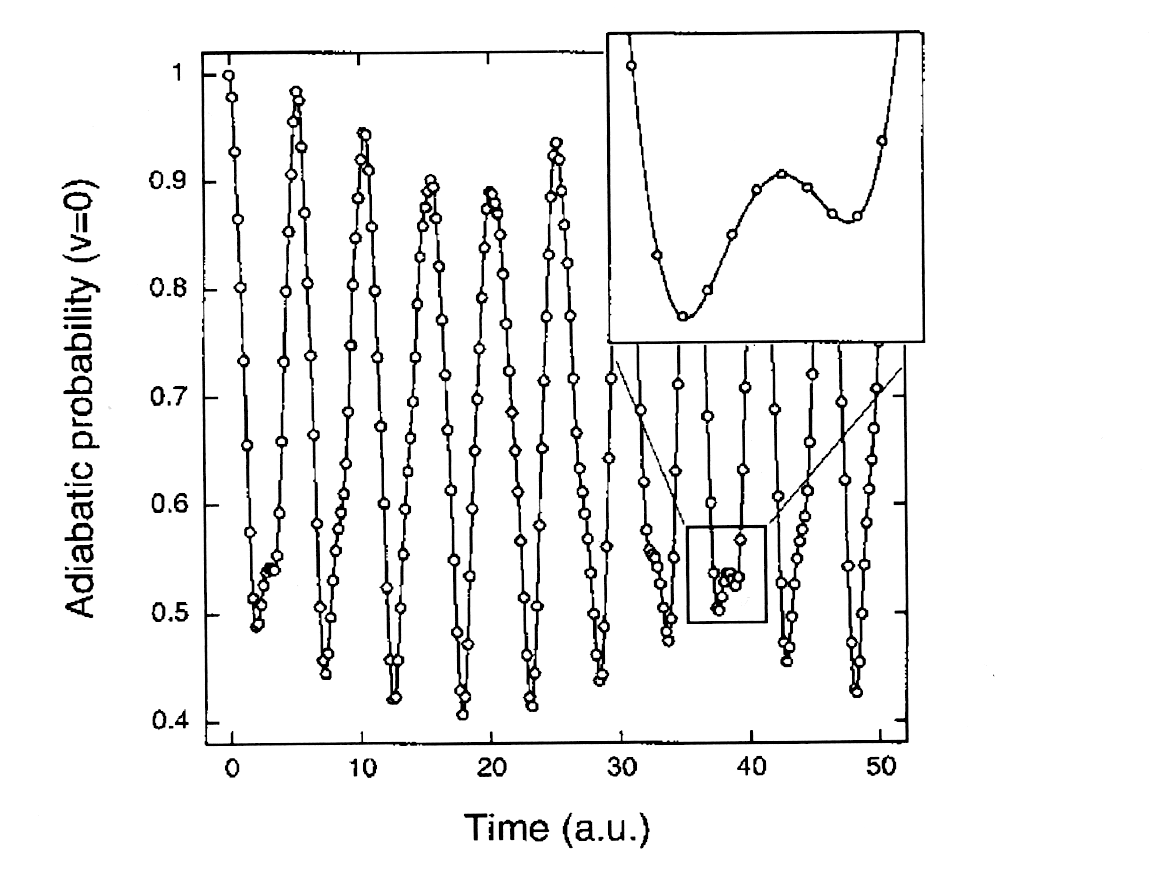}\hfil
		\caption{The wave functions oscillates in and out of the relevant subspace \cite[p.~162]{Wyatt:qd_trajectories:2005}. }
		\label{multiscale:physics_of_molecules:figure:oscillations}
	\end{figure}
	\item Last, but not least, what about higher-order corrections? Are we able to improve our approximation by some scheme? 
\end{enumerate}
%


\section{The adiabatic trinity} 
\label{multiscale:adiabatic_trinity}
Born-Oppenheimer systems have three basic features that are shared by many other multiscale systems (\eg the treatment of the relativistic spin-$\nicefrac{1}{2}$ particle presented in Section~\ref{weyl_calculus:diag_dirac}). We will call them \emph{adiabatic trinity}: 
\begin{enumerate}[(i)]
	\item \emph{A distinction between slow and fast degrees of freedom. }
	The nuclei move much slower than the electrons. This suggests to split the total, physical Hilbert space $\Hphys := L^2(\R^{3N} \times \R^{3K})$ into slow (nucleonic) and fast (electronic) Hilbert space, 
	\begin{align*}
		\Hslow \otimes \Hfast := L^2(\R^{3N}) \otimes L^2(\R^{3K}) 
		. 
	\end{align*}
	\item \emph{A small dimensionless parameter $\eps \ll 1$ which quantifies the separation of scales.}
	In our situation, $\eps := \sqrt{\nicefrac{m_{\mathrm{e}}}{m_{\mathrm{nuc}}}}$ quantifies the ratio of typical velocities of electrons and nuclei if they are assumed to have similar momentum. 
	\item \emph{A relevant part of the spectrum of $H_{\eps = 0}(x,\xi)$ (an operator on $\Hfast$) separated by a gap from the remainder. }
	Here, we are interested in a collection of bands 
	\begin{align*}
		\sigma_{\rel}(x,\xi) := \bigcup_{j \in \mathcal{J}} \bigl \{ E_j(x) \bigr \} 
	\end{align*}
	that are physically relevant, \eg the ground state band. Here, $\mathcal{J}$ is a finite index set, $L := \abs{\mathcal{J}} < \infty$. There may be band crossings between relevant bands, but the relevant bands must not cross, merge or touch the remainder of the spectrum $\sigma_{\rel}^c(x,\xi) := \sigma \bigl ( H_{\eps = 0}(x,\xi) \bigr ) \setminus \sigma_{\rel}(x,\xi)$. In other words, there must be a gap 
	\begin{align*}
		C_g := \inf_{(x,\xi) \in \pspace} \mathrm{dist} \, \bigl ( \sigma_{\rel}(x,\xi) , \sigma_{\rel}^c(x,\xi) \bigr ) > 0 
		. 
	\end{align*}
	This gap suppresses transitions from relevant bands to states in the remainder of the spectrum \emph{exponentially}. Hence, we can consider the contributions from the relevant part of the spectrum separately. \marginpar{\small 2010.01.26}
\end{enumerate}
Since this is a conceptual problem, we will use ``conceptual notation,'' \eg $\Hphys$ instead of $L^2(\R^{3N} \times \R^{3K})$ and $\Hslow$ instead of $L^2(\R^{3N})$. We are initially given a hamiltonian operator 
\begin{align*}
	\hat{H}_{\eps} = H_{\eps}(\Qe,\Pe) := \sum_{n = 0}^{\infty} \eps^n H_n(\Qe,\Pe) := \opHBO 
\end{align*}
which acts on $\Hphys = L^2(\R^{3N} \times \R^{3K})$. We allow the more general situation where $\hat{H}_{\eps}$ is a power series in $\eps$. The idea is that for $\eps > 0$, the nucleonic kinetic energy acts as a ``perturbation'' for the case $\eps = 0$. If the nuclei are infinitely more heavy than the electrons, $\eps = 0$, their dynamics freezes out and we can arrange all bits and pieces into the following diagram: 
\begin{align}
	\bfig
		\node L2full(-1300,0)[\Hphys]
		\node piL2full(-1300,-600)[\hat{\pi}_0 \Hphys]
		\node HslowHfast(0,0)[\Hslow \otimes \Hfast]
		\node Href(0,-600)[\Hslow \otimes \C^{\abs{\mathcal{J}}}]
		\arrow[L2full`HslowHfast;\hat{u}_0]
		\arrow[L2full`piL2full;\hat{\pi}_0]
		\arrow[HslowHfast`Href;\Piref]
		\arrow/-->/[piL2full`Href;]
		\Loop(-1300,0){\Hphys}(ur,ul)_{e^{-i t \hat{H}_{\eps = 0}}} 
		\Loop(0,0){\Hslow \otimes \Hfast}(ur,ul)_{e^{-i t \hat{u}_0 \hat{H}_{\eps = 0} \hat{u}_0^*}} 
		\Loop(-1300,-600){\hat{\pi}_0 \Hphys}(dr,dl)^{e^{-i t \hat{\pi}_0 \hat{H}_{\eps = 0} \hat{\pi}_0}} 
		\Loop(0,-600){\Hslow \otimes \C^{\abs{\mathcal{J}}}}(dr,dl)^{e^{-i t \Piref \hat{u}_0 \hat{H}_{\eps = 0} \hat{u}_0^* \Piref}} 
	\efig
	\label{multiscale:adiabatic_trinity:diagram:unperturbed}
\end{align}
In the unperturbed case, the dynamics on the physical Hilbert space $\Hphys = L^2(\R^{3N} \times \R^{3K})$ is generated by $e^{- i t \hat{H}_{\eps = 0}} = e^{- i t \He(\Qe)}$. Associated to the family of \emph{physically relevant bands}, property~(iii) of the adiabatic trinity, 
\begin{align*}
	\sigma_{\mathrm{rel}}(x,\xi) := \bigcup_{j \in \Index} \{ E_j(x) \} 
	, 
\end{align*}
we can defined the associated \emph{family of projections}
\begin{align}
	\pi_0(x) := \sum_{j \in \Index} \sopro{\varphi_j(x)}{\varphi_j(x)} : \Hfast \longrightarrow \Hfast 
	. 
	\label{multiscale:adiabatic_trinity:eqn:pi0}
\end{align}
Here, the $\varphi_j(x)$ are the eigenfunctions of the electronic hamiltonian $\He(x)$ associated to $E_j(x)$, 
\begin{align*}
	\He(x) \varphi_j(x) = E_j(x) \, \varphi_j(x) 
	. 
\end{align*}
Fear each $x \in \R^{3N}$, $\He(x) \in \mathcal{B}(\mathcal{D},\Hfast)$ is a selfadjoint operator with dense domain $\mathcal{D} \subset \Hfast$, \ie $\He$ is an operator-valued function. Without loss of generality, we may assume the $\varphi_j(x)$, $j \in \Index$, are normalized to $1$ for all $x \in \R^{3N}$, 
\begin{align*}
	\scpro{\varphi_j(x)}{\varphi_j(x)}_{\Hfast} = 1
	. 
\end{align*}
If a band $E_j$ is $k$-fold degenerate, the index set $\Index$ is such that the band is counted $k$ times, \ie we associate $k$ orthonormal vectors to this band which span the eigenspace associated to $E_j(x)$. 

If we quantize $\pi_0$ to $\hat{\pi}_0 := \pi_0(\Qe)$, it becomes an operator on $\Hphys$. The subspace 
\begin{align*}
	\hat{\pi}_0 \Hphys \subset \Hphys 
\end{align*}
corresponds to those states where the electronic (fast) part belongs to the relevant subspace. \emph{By assumption}, we are interested in the dynamics of initial states that are associated to the relevant electronic bands. Since the unperturbed hamiltonian $\hat{H}_{\eps = 0}$ commutes with $\hat{\pi}_0$, $[ \hat{H}_{\eps = 0} , \hat{\pi}_0 ] = 0$, the operator $\hat{\pi}_0$ is a \emph{constant of motion} and thus the dynamics generated by the unperturbed hamiltonian $\hat{H}_{\eps = 0}$ leave $\hat{\pi}_0 \Hphys$ invariant. This means, the unitary time evolution on the relevant subspace $\hat{\pi}_0 \Hphys$ is generated by 
\begin{align*}
	\hat{\pi}_0 e^{- i t \hat{H}_{\eps = 0}} \hat{\pi}_0 = e^{- i t \hat{\pi}_0 \hat{H}_{\eps = 0} \hat{\pi}_0} : \hat{\pi}_0 \Hphys \longrightarrow \hat{\pi}_0 \Hphys 
	. 
\end{align*}
In our case, the projected hamiltonian is equal to 
\begin{align*}
	\hat{\pi}_0 \hat{H}_{\eps = 0} \hat{\pi}_0 &= \sopro{\varphi_j(\Qe)}{\varphi_j(\Qe)} \He(\Qe) \sopro{\varphi_{j'}(\Qe)}{\varphi_{j'}(\Qe)} 
	= E_j(\Qe) \sopro{\varphi_j(\Qe)}{\varphi_j(\Qe)} 
\end{align*}
where Einstein's summation convention is used (\ie repeated indices in a product are implicitly summed over). 

So far, we have only used property~(iii) of the adiabatic trinity, the existence of physically relevant bands. With this, we have explained the left column of the diagram. The right column is associated to property~(i), \ie the existence of slow (nucleonic) and fast (electronic) degrees of freedom. We now introduce a unitary $\hat{u}_0$ which facilitates the separation into slow and fast degrees of freedom. While it is superfluous in the unperturbed case, it is \emph{essential} once the perturbation is switched on. Let $\{ \chi_j \}_{j \in \Index} \subset \Hfast$ be a set of orthonormal vectors in the fast Hilbert space. These vectors are arbitrary and nothing of the subsequent construction will depend on a specific choice of these vectors. For all $x \in \R^{3N}$, we define the unitary operator 
\begin{align}
	u_0(x) := \sum_{j \in \Index} \sopro{\chi_j}{\varphi_j(x)} + u_0^{\perp}(x) : \Hfast \longrightarrow \Hfast 
	\label{multiscale:adiabatic_trinity:eqn:u0}
\end{align}
where $u_0^{\perp}(x)$ is associated to the remainder of the spectrum and is such that $u_0(x)$ is a unitary on $\Hfast$. $u_0(x)$ maps the $x$-dependent vector $\varphi_j(x)$, $j \in \Index$, onto the $x$-\emph{in}dependent vector $\chi_j$, 
\begin{align*}
	u_0(x) \varphi_j(x) = \chi_j = u_0(x') \varphi_j(x') 
	. 
\end{align*}
The quantization $\hat{u}_0 := u_0(\Qe)$ is a unitary map from $\Hphys$ to $\Hslow \otimes \Hfast = L^2(\R^{3N}) \otimes L^2(\R^{3K})$. The simple product state $\phi \, \varphi_j$, $\phi \in L^2(\R^{3N})$, $j \in \Index$, is mapped onto 
\begin{align*}
	\bigl ( \hat{u}_0 \phi \varphi_j \bigr )(x,y) &= \bigl ( u_0(\Qe) \phi \varphi_j \bigr )(x,y) 
	= \phi(x) \, u_0(x) \varphi_j(x,y) = \phi(x) \chi_j(y) 
	. 
\end{align*}
In the new representation, product states are mapped onto even \emph{simpler} product states where the electronic component is \emph{independent of $x$}. Since $\hat{u}_0$ is a unitary operator, the full dynamics can be transported to the unitarily equivalent space, 
\begin{align*}
	\hat{u}_0 e^{- i t \hat{H}_{\eps = 0}} \hat{u}_0^* = e^{- i t \hat{u}_0 \hat{H}_{\eps = 0} \hat{u}_0^*} : \Hslow \otimes \Hfast \longrightarrow \Hslow \otimes \Hfast 
	. 
\end{align*}
This is \emph{not an approximation}, just a convenient rewriting by choosing a suitable representation. The states associated to the relevant bands can also be written in the new representation: if we view $\hat{\pi}_0 \Hphys$ as a subspace of $\Hphys$, it makes sense to write 
\begin{align*}
	\hat{u}_0 \hat{\pi}_0 \Hphys = \hat{u}_0 \hat{\pi}_0 \hat{u}_0^* \, \hat{u}_0 \Hphys = \hat{u}_0 \hat{\pi}_0 \hat{u}_0^* ( \Hslow \otimes \Hfast ) 
	. 
\end{align*}
We dub this unitarily transformed projection \emph{reference projection} $\Piref = \id_{\Hslow} \otimes \piref$, 
\begin{align*}
	\piref :=& \; u_0(x) \pi_0(x) u_0(x)^* 
	\\
	=& \; \bigl ( \sopro{\chi_j}{\varphi_j(x)} + u_0^{\perp}(x) \bigr ) \sopro{\varphi_k(x)}{\varphi_k(x)} \bigl ( \sopro{\varphi_j(x)}{\chi_j} + u_0^{\perp}(x)^* \bigr ) 
	\\
	=& \; \sopro{\chi_j}{\chi_j} 
	. 
\end{align*}
By construction, the reference projection does not depend on the slow (nucleonic) coordinates. We can and will view 
\begin{align*}
	\Piref (\Hslow \otimes \Hfast) = \Hslow \otimes (\piref \Hfast) \cong \Hslow \otimes \C^{\abs{\Index}} 
\end{align*}
either as subspace of $\Hslow \otimes \Hfast$ or, more commonly, as $\Hslow \otimes \C^{\abs{\Index}}$. In case of the latter, each $\chi_j \in \Hfast$ is identified with a \emph{canonical base vector} of $\C^{\abs{\Index}}$. Hence, operators on $\Piref (\Hslow \otimes \Hfast) \cong \Hslow \otimes \C^{\abs{\Index}}$ can be thought of as \emph{matrix-valued operators} just like the Dirac hamiltonian. The dynamics on $\Hslow \otimes \C^{\abs{\Index}}$ is generated by 
\begin{align*}
	\Piref \, \hat{u}_0 \, e^{- i t \hat{H}_{\eps = 0}} \, \hat{u}_0^* \, \Piref = e^{- i t \Piref \, \hat{u}_0 \hat{H}_{\eps = 0} \hat{u}_0^* \, \Piref} : \Hslow \otimes \C^{\abs{\Index}} \longrightarrow \Hslow \otimes \C^{\abs{\Index}} 
	. 
\end{align*}
These dynamics describe the full unperturbed dynamics exactly if we start with initial conditions associated to the relevant bands. 
\medskip

\noindent
For the unperturbed case, \ie $\eps = 0$, we have found a way to systematically exploit the structure of adiabatic systems. What happens if we switch on the perturbation, \ie we let $\eps > 0$? Let us postulate that if the separation of scales is big enough, \ie $\eps \ll 1$ small enough, there exists an analogous diagram 
\begin{align}
	\bfig
		\node L2full(-1300,0)[\Hphys]
		\node piL2full(-1300,-600)[\hat{\Pi}_{\eps} \Hphys]
		\node HslowHfast(0,0)[\Hslow \otimes \Hfast]
		\node Href(0,-600)[\Hslow \otimes \C^{\abs{\mathcal{J}}}]
		\arrow[L2full`HslowHfast;\hat{U}_{\eps}]
		\arrow[L2full`piL2full;\hat{\Pi}_{\eps}]
		\arrow[HslowHfast`Href;\Piref]
		\arrow/-->/[piL2full`Href;]
		\Loop(-1300,0){\Hphys}(ur,ul)_{e^{-i t \hat{H}_{\eps}}} 
		\Loop(0,0){\Hslow \otimes \Hfast}(ur,ul)_{e^{-i t \hat{U}_{\eps} \hat{H}_{\eps} \hat{U}_{\eps}^*}} 
		\Loop(-1300,-600){\hat{\Pi}_{\eps} \Hphys}(dr,dl)^{e^{-i t \hat{\Pi}_{\eps} \hat{H}_{\eps} \hat{\Pi}_{\eps}}} 
		\Loop(0,-600){\Hslow \otimes \C^{\abs{\mathcal{J}}}}(dr,dl)^{e^{-i t \Piref \hat{U}_{\eps} \hat{H}_{\eps} \hat{U}_{\eps}^* \Piref}} 
	\efig
	\label{multiscale:adiabatic_trinity:diagram:perturbed}
\end{align}
where some $\eps$-dependent unitary $\hat{U}_{\eps} : \Hphys \longrightarrow \Hslow \otimes \Hfast$ facilitates the splitting into slow and fast degrees of freedom and a projection $\hat{\Pi}_{\eps} : \Hphys \longrightarrow \Hphys$ is somehow still associated to the relevant bands $\sigma_{\mathrm{rel}}(x) = \bigcup_{j \in \Index} \{ E_j(x) \}$. \emph{If} the nucleonic kinetic energy can be regarded as a perturbation, it would be natural to ask that 
\begin{align*}
	\hat{\Pi}_{\eps} &= \hat{\pi}_0 + \order(\eps) \\
	\hat{U}_{\eps} &= \hat{u}_0 + \order(\eps) 
\end{align*}
holds in some sense. Implicitly, we have postulated that $\hat{\Pi}_{\eps}$ and $\hat{U}_{\eps}$ have the following four properties: 
\begin{align}
	\hat{\Pi}_{\eps}^2 &= \hat{\Pi}_{\eps} 
	&&
	[ \hat{H}_{\eps} , \hat{\Pi}_{\eps} ] = \order(\eps^{\infty}) 
	\label{multiscale:adiabatic_trinity:eqn:four_properties_operators}
	\\
	\hat{U}_{\eps}^* \hat{U}_{\eps} &= \id_{\Hphys} , \; 
	\hat{U}_{\eps} \hat{U}_{\eps}^* = \id_{\Hslow \otimes \Hfast} 
	&&
	\hat{U}_{\eps} \hat{\Pi}_{\eps} \hat{U}_{\eps}^* = \Piref 
	\notag 
\end{align}
The first column merely tells us that $\hat{\Pi}_{\eps}$ and $\hat{U}_{\eps}$ are a projection and a unitary while the second tells us the two are \emph{adapted to the problem}. Ultimately, we are interested in the dynamics generated by the \emph{effective hamiltonian} 
\begin{align*}
	\hat{h}_{\mathrm{eff}} = \Piref \hat{U}_{\eps} \hat{H}_{\eps} \hat{U}_{\eps}^* \Piref + \order(\eps^{\infty})
\end{align*}
which generates the dynamics on the reference space $\Hslow \otimes \C^{\abs{\Index}}$. It describes the dynamics accurately if we start with an initial state in the subspace $\hat{\Pi}_{\eps} \Hphys$. Compared to the naïve subspace $\hat{\pi}_0 \Hphys$, this subspace is tiled by an angle of order $\order(\eps)$. 
\medskip

\noindent
Panati, Spohn and Teufel have constructed $\hat{\Pi}_{\eps}$ and $\hat{U}_{\eps}$ recursively order-by-order via a ``defect construction.'' Before we go into detail, let us introduce Weyl calculus for operator-valued functions which is then used to reformulate equation~\eqref{multiscale:adiabatic_trinity:eqn:four_properties_operators}. This allows us to make use of the asymptotic expansion of the Weyl product proven in Chapter~\ref{weyl_calculus:asymptotics}. 


\section{Intermezzo: Weyl calculus for operator-valued functions} 
\label{multiscale:weyl_calculus_operator_valued_functions}

We need to modify Weyl calculus to accommodate \emph{operator-valued} functions. Let $\Hil$ be a separable Hilbert space and $\mathcal{B}(\Hil)$ the Banach space of bounded operators on $\Hil$. Our goal is to quantize functions $f : \pspace \longrightarrow \mathcal{B}(\Hil)$ that take values in the bounded operators on $\Hil$. The simplest example which we have already gotten to know in Chapter~\ref{weyl_calculus:diag_dirac} are matrix-valued symbols. For all intents and purposes, the reader may think of $\Hil = \C^N$. 

As usual, we start by defining the Weyl system: for each $X \in \pspace$ 
\begin{align*}
	\WeylSys(X) := e^{- i \sigma(X,(\Qe,\Pe))} \otimes \id_{\Hil}
\end{align*}
defines a unitary operator on $L^2(\R^d , \Hil) \cong L^2(\R^d) \otimes \Hil$. With this, for any $f \in \Schwartz \bigl ( \pspace,\mathcal{B}(\Hil) \bigr )$, the Weyl quantization 
\begin{align*}
	\Op(f) := \frac{1}{(2\pi)^d} \int_{\pspace} \dd X \, (\Fs f)(X) \, \WeylSys(X) 
	= \frac{1}{(2\pi)^d} \int_{\pspace} \dd X e^{- i \sigma(X,(\Qe,\Pe))} \otimes (\Fs f)(X) 
\end{align*}
defines a bounded operator on $L^2(\R^d,\Hil)$. The proof is virtually identical to that of Proposition~\ref{weyl_calculus:weyl_quantization:prop:application_Weyl_quantization}. This also induces a Weyl product $\Weyl$ that is given by the \emph{same} formula as before, 
\begin{align*}
	(f \Weyl g)(X) = \frac{1}{(2\pi)^{2d}} \int_{\pspace} \dd Y \int_{\pspace} \dd Z \, e^{i \sigma(X,Y + Z)} \, e^{i \frac{\eps}{2} \sigma(Y,Z)} \, (\Fs f)(Y) \, (\Fs g)(Z) 
	, 
\end{align*}
which needs to be interpreted slightly differently, though. This product has the \emph{same asymptotic expansion}, 
\begin{align*}
	f \Weyl g &= f \, g - \eps \tfrac{i}{2} \{ f , g \} + \order(\eps^2) 
	. 
\end{align*}
Here, $f \, g$ is defined pointwise via the operator product on $\mathcal{B}(\Hil)$, 
\begin{align*}
	(f \, g)(X) := f(X) \, g(X) 
	. 
\end{align*}
This has some profound consequences: for instance, the Poisson bracket of an operator-valued function with itself does not vanish, 
\begin{align*}
	\{ f , f \} = \sum_{l = 1}^d \bigl ( \partial_{\xi_l} f \, \partial_{x_l} f - \partial_{x_l} f \, \partial_{\xi_l} f \bigr ) 
	= \sum_{l = 1}^d [ \partial_{\xi_l} f , \partial_{x_l} f ]
\end{align*}
since the derivative of $f$ with respect to $\xi$ in general does not commute with the derivative of $f$ with respect to $x$. Hence the Weyl commutator of $f$ and $g$ has non-trivial terms from the zeroth order on, 
\begin{align*}
	[ f , g ]_{\Weyl} &= f \Weyl g - g \Weyl f = [ f , g ] - \eps \tfrac{i}{2} \bigl ( \{ f , g \} - \{ g , f \} \bigr ) + \order(\eps^2) 
	. 
\end{align*}
This will be important in Chapter~\ref{multiscale:semiclassics} where we discuss the semiclassical limit of multiscale systems. 


\section{Effective quantum dynamics: adiabatic decoupling to all orders} 
\label{multiscale:effective_quantum_dynamics}

Now that we have introduced Weyl quantization for operator-valued functions, let us rewrite the defining properties of 
\begin{align*}
	\hat{\Pi}_{\eps} = \Op_{\eps}(\pi) + \order(\eps^{\infty})
\end{align*}
and 
\begin{align*}
	\hat{U}_{\eps} = \Op_{\eps}(u) + \order(\eps^{\infty})
\end{align*}
which are now assumed to be quantizations of the $\mathcal{B}(\Hfast)$-valued functions $\pi$ and $u$,\footnote{For technical reasons, we must distinguish between $\Op_{\eps}(\pi)$ and $\hat{\Pi}_{\eps}$ and similarly between $\Op_{\eps}(u)$ and $\hat{U}_{\eps}$. One can recover $\hat{\Pi}_{\eps}$ and $\hat{U}_{\eps}$ from $\Op_{\eps}(\pi)$ and $\Op_{\eps}(u)$, see~\cite[p.82~f.,~p.87~f.]{Teufel:adiabaticPerturbationTheory:2003} for details. In computations, one need not distinguish between these objects. } 
\begin{align}
	\pi \Weyl \pi &= \pi + \order(\eps^{\infty}) 
	&&
	[H_{\eps} , \pi]_{\Weyl} = \order(\eps^{\infty}) 
	\label{multiscale:adiabatic_trinity:eqn:four_properties_symbols}
	\\
	u \Weyl u^* &= 1 + \order(\eps^{\infty}) , \; u^* \Weyl u = 1 + \order(\eps^{\infty}) 
	&&
	u \Weyl \pi \Weyl u^* = \piref + \order(\eps^{\infty}) 
	\notag 
	. 
\end{align}
We will expand $\pi \asymp \sum_{n = 0}^{\infty} \eps^n \pi_n$ and $u \asymp \sum_{n = 0}^{\infty} \eps^n u_n$ asymptotically. We guess that the zeroth order terms are -- as suggested by the notation -- $\pi_0$ and $u_0$ as defined by equations~\eqref{multiscale:adiabatic_trinity:eqn:pi0} and \eqref{multiscale:adiabatic_trinity:eqn:u0} since 
\begin{align*}
	\pi_0 \Weyl \pi_0 &= \pi_0 + \order(\eps) 
	&&
	[H_{\eps} , \pi_0]_{\Weyl} = \order(\eps) 
	\\
	u_0 \Weyl u_0^* &= 1 + \order(\eps) , \; u_0^* \Weyl u_0 = 1 + \order(\eps) 
	&&
	u_0 \Weyl \pi_0 \Weyl u_0^* = \piref + \order(\eps)
	\notag 
	. 
\end{align*}
follows immediately from the asymptotic expansion of $\Weyl$. In case of the first equation, for instance, we use that $\pi_0(x)$ is a projection on $\Hfast$ for each $x \in \R^{3N}$ and obtain 
\begin{align*}
	\pi_0 \Weyl \pi_0 = \pi_0 \pi_0 + \order(\eps) = \pi_0 + \order(\eps) 
	. 
\end{align*}
Our goal is to construct the \emph{effective hamiltonian}
\begin{align}
	h_{\mathrm{eff}} := \piref \Weyl u \Weyl H_{\eps} \Weyl u^* \Weyl \piref 
	= \piref \, u \Weyl H_{\eps} \Weyl u^* \, \piref 
	\label{multiscale:effective_quantum_dynamics:eqn:effective_hamiltonian}
\end{align}
since its quantization gives the effective dynamics of states associated to the relevant bands. Let us start with the construction of the tilted projection.

\subsection{Construction of the projection} 
\label{multiscale:effective_quantum_dynamics:projection}

We can split each term $\pi_n$ of the projection into a (block) diagonal term 
\begin{align*}
	\pi_n^{\mathrm{D}} := \pi_0 \pi_n \pi_0 + (1 - \pi_0) \pi_n (1 - \pi_0) 
\end{align*}
and an offdiagonal term 
\begin{align*}
	\pi_n^{\mathrm{OD}} := \pi_0 \pi_n (1 - \pi_0) + (1 - \pi_0) \pi_n \pi_0 
	. 
\end{align*}
One can check easily that $\pi_n = \pi_n^{\mathrm{D}} + \pi_n^{\mathrm{OD}}$. Let us focus on the first-order correction $\pi_1$:
it turns out that the first equation, the \emph{projection defect} 
\begin{align*}
	\pi_0 \Weyl \pi_0 - \pi_0 =: \eps G_1 + \order(\eps^2) = - \tfrac{i}{2} \{ \pi_0 , \pi_0 \} + \order(\eps^2)
	, 
\end{align*}
only determines the \emph{diagonal} term. To see that, let us substitute $\pi_0 + \eps \pi_1^{\mathrm{OD}}$ into the above equation and do the computation: 
\begin{align*}
	\bigl ( \pi_0 + \eps \pi_1^{\mathrm{OD}} \bigr ) &\Weyl \bigl ( \pi_0 + \eps \pi_1^{\mathrm{OD}} \bigr ) - \bigl ( \pi_0 + \eps \pi_1^{\mathrm{OD}} \bigr ) 
	= \\
	&
	= \eps \bigl ( \underbrace{\pi_0 \pi_1^{\mathrm{OD}} + \pi_1^{\mathrm{OD}} \pi_0}_{= \pi_1^{\mathrm{OD}}} - \tfrac{i}{2} \{ \pi_0 , \pi_0 \} - \pi_1^{\mathrm{OD}} \bigr ) + \order(\eps^2) 
	\\
	&
	= - \tfrac{i}{2} \{ \pi_0 , \pi_0 \} + \order(\eps^2)
\end{align*}
This means we could have substituted \emph{any} offdiagonal $\pi_1^{\mathrm{OD}}$ into the equation, the term cancels out and its choice does not affect the projection defect $G_1$. One can show that the diagonal part can be computed from $G_1$ via 
\begin{align*}
	\pi_1^{\mathrm{D}} = - \pi_0 G_1 \pi_0 + (1 - \pi_0) G_1 (1 - \pi_0) 
	. 
\end{align*}
The diagonal part needs to be taken into account when computing the \emph{commutation defect} 
\begin{align*}
	\bigl [ H_{\eps} , \pi_0 + \eps \pi_1^{\mathrm{D}} \bigr ]_{\Weyl} &=: \eps F_1 + \order(\eps^2) 
	\\
	&= \eps \Bigl ( [H_1 , \pi_0] + [H_0 , \pi_1^{\mathrm{D}}] - \tfrac{i}{2} \{ H_0 , \pi_0 \} + \tfrac{i}{2} \{ \pi_0 , H_0 \} \Bigr ) + \order(\eps^2) 
\end{align*}
from which the offdiagonal part can be determined. Since by assumption 
\begin{align*}
	\order(\eps^2) &= \bigl [ H_{\eps} , \pi_0 + \eps \pi_1 \bigr ]_{\Weyl} = \bigl [ H_{\eps} , \pi_0 + \eps \pi_1^{\mathrm{D}} \bigr ]_{\Weyl} + \bigl [ H_{\eps} , \eps \pi_1^{\mathrm{OD}} \bigr ]_{\Weyl} 
	\\
	&= \eps \bigl ( F_1 + [H_0 , \pi_1^{\mathrm{OD}}] \bigr ) + \order(\eps^2) 
\end{align*}
holds, we conclude that the offdiagonal part $\pi_1^{\mathrm{OD}}$ must satisfy 
\begin{align}
	\bigl [ H_0 , \pi_1^{\mathrm{OD}} \bigr ] = - F_1 
	\label{multiscale:effective_quantum_dynamics:projection:eqn:pi_1_od_equation} 
	. 
\end{align}
In the presence of band crossings within the relevant bands, equation~\eqref{multiscale:effective_quantum_dynamics:projection:eqn:pi_1_od_equation} cannot be solved in closed form and one has to resort to numerical schemes and an alternative ansatz \cite[p.~79]{Teufel:adiabaticPerturbationTheory:2003}. If there are no band crossings, the solution is given by 
\begin{align*}
	\pi_1^{\mathrm{OD}} = - \sum_{j \in \Index} (1 - \pi_0) (H_0 - E_j)^{-1} F_1 \sopro{\varphi_j}{\varphi_j} - \sum_{j \in \Index} \sopro{\varphi_j}{\varphi_j} F_1 (H_0 - E_j)^{-1} (1 - \pi_0) 
	. 
\end{align*}
Higher-order terms can be devised in a similar manner by recursion: assume we have computed the first $n$ terms starting from a given $\pi_0$, 
\begin{align*}
	\pi^{(n)} := \sum_{k = 0}^n \eps^k \pi_k 
	. 
\end{align*}
Then we can get the diagonal part $\pi_{n+1}^{\mathrm{D}}$ from the projection defect 
\begin{align*}
	\pi^{(n)} \Weyl \pi^{(n)} - \pi^{(n)} =: \eps^{n+1} G_{n+1} + \order(\eps^{n+2}) 
\end{align*}
by the same equation as before, 
\begin{align}
	\pi_{n+1}^{\mathrm{D}} = - \pi_0 G_{n+1} \pi_0 + (1 - \pi_0) G_{n+1} (1 - \pi_0) 
	. 
	\label{multiscale:effective_quantum_dynamics:projection:eqn:pi_n_d_equation} 
\end{align}
The difficulty to solve 
\begin{align}
	[H_0 , \pi_{n+1}^{\mathrm{OD}}] = - F_{n+1} 
	\label{multiscale:effective_quantum_dynamics:projection:eqn:pi_n_od_equation} 
\end{align}
for the offdiagonal part $\pi_{n+1}^{\mathrm{OD}}$ explicitly remains where the right-hand side is the commutation defect 
\begin{align*}
	\bigl [ H_{\eps} , \pi^{(n)} + \eps^{n+1} \pi_{n+1}^{\mathrm{D}} \bigr ]_{\Weyl} =: \eps^{n+1} F_{n+1} + \order(\eps^{n+2})
	. 
\end{align*}
%


\subsection{Construction of the intertwining unitary} 
\label{multiscale:effective_quantum_dynamics:unitary}

The first order of the unitary is easy to construct once $\pi_1$ has been computed. We make the ansatz that we write $u_1 = (a_1 + b_1) u_0$ as the sum of a symmetric part $a_1 = a_1^*$ and an antisymmetric part $b_1 = - b_1^*$ which is akin to decomposing a complex number into real and imaginary part. Hence, this is not a restriction or a special assumption: any bounded operator $A : \Hfast \longrightarrow \Hfast$ can be written as 
\begin{align*}
	C = \tfrac{1}{2} (C + C^*) + \tfrac{1}{2} (C - C^*) =: A + B 
	. 
\end{align*}
The symmetric part $a_1$ follows from the \emph{unitarity defect} 
\begin{align*}
	u_0 \Weyl u_0^* - 1 =: \eps A_1 + \order(\eps^2) 
\end{align*}
and is given by 
\begin{align*}
	a_1 = - \tfrac{1}{2} A_1 
	. 
\end{align*}
Equivalently, $u_0^* \Weyl u_0 - 1$ could have been used to deduce $a_1$. We have to plug in $a_1$ into the \emph{intertwining defect} 
\begin{align*}
	\bigl ( u_0 + \eps a_1 u_0 \bigr ) \Weyl \bigl ( \pi_0 + \eps \pi_1 \bigr ) \Weyl \bigl ( u_0 + \eps a_1 u_0 \bigr )^* - \piref =: \eps B_1 + \order(\eps^2) 
\end{align*}
which can also be solved for $b_1$ explicitly, 
\begin{align}
	b_1 := [\piref , B_1] 
	. 
\end{align}
This recipe immediately generalizes to higher-order terms: if we have computed $u$ up to errors of order $\order(\eps^{n+1})$, 
\begin{align*}
	u^{(n)} := \sum_{k = 0}^n \eps^k u_k 
	, 
\end{align*}
we can solve for 
\begin{align}
	a_{n+1} = - \tfrac{1}{2} A_{n+1}
\end{align}
where $A_{n+1}$ is the unitarity defect 
\begin{align*}
	u^{(n)} \Weyl {u^{(n)}}^* - 1 =: \eps^{n+1} A_{n+1} + \order(\eps^{n+2}) 
	. 
\end{align*}
The antisymmetric part $b_{n+1}$ is again related to the intertwining defect 
\begin{align*}
	\eps^{n+1} B_{n+1} + \order(\eps^{n+2}) := \bigl ( u^{(n)} + \eps^{n+1} a_{n+1} u_0 \bigr ) \Weyl \pi^{(n+1)} \Weyl \bigl ( u^{(n)} + \eps^{n+1} a_{n+1} u_0 \bigr )^* - \piref
\end{align*}
through the equation 
\begin{align}
	b_{n+1} = [\piref , B_{n+1}] 
	. 
\end{align}
Put together, these give the $n+1$th order term of the unitary, \marginpar{\small 2010.02.02}
\begin{align}
	u_{n+1} = (a_{n+1} + b_{n+1}) u_0 
	. 
\end{align}
%


\subsection{The effective hamiltonian} 
\label{multiscale:effective_quantum_dynamics:effective_hamiltonian}

The zeroth order of the effective hamiltonian 
\begin{align*}
	h_{\mathrm{eff} \, 0} = \piref \, u_0 H_0 u_0^* \, \piref 
\end{align*}
is obtained by simply replacing the Weyl product with the pointwise product and $u$ as well as $H_{\eps}$ with the zeroth-order terms in their respective asymptotic expansions. To compute terms of higher orders, it is useful to compute the unitarily transformed hamiltonian first, 
\begin{align}
	h := u \Weyl H_{\eps} \Weyl u^* 
	. 
\end{align}
Then a closed form for $h_1$ can be deduced from expanding the left-hand side of 
\begin{align}
	u \Weyl H_{\eps} - h_0 \Weyl u = \eps h_1 \Weyl u + \order(\eps^2) = \eps h_1 u_0 + \order(\eps^2) 
	\label{multiscale:effective_quantum_dynamics:effective_hamiltonian:eqn:closed_expression_h_n_plus_1}
\end{align}
to first order. The main simplification is that the left-hand side consists of single Weyl products instead of double Weyl products which are much more difficult to expand. This relation is key to computing higher-order terms in $h$ and $h_{\mathrm{eff}}$ without making mistakes by leaving out terms. If $h^{(n)} := \sum_{k = 0}^n \eps^k h_k$ is the transformed hamiltonian where all terms up to $n$th order have been computed, then a closed expression for $h_{n+1}$ can be found by computing the left-hand side of 
\begin{align*}
	u \Weyl H_{\eps} - h^{(n)} \Weyl u = \eps^{n+1} h_{n+1} u_0 + \order(\eps^{n+2}) 
	. 
\end{align*}
Once we collect all terms of first order, $h_1$ can be expressed as 
\begin{align}
	h_1 &= u_1 H_0 u_0^* + u_0 H_1 u_0^* - h_0 u_1 u_0^* - \tfrac{i}{2} \{ u_0 , H_0 \} u_0^* + \tfrac{i}{2} \{ h_0 , u_0 \} u_0^* 
	\notag \\
	&= [u_1 u_0^* , h_0] + u_0 H_1 u_0^* + \tfrac{i}{2} \bigl ( \{ h_0 , u_0 \} - \{ u_0 , H_0 \} \bigr ) u_0^* 
	. 
	\label{multiscale:effective_quantum_dynamics:effective_hamiltonian:eqn:closed_expression_h_1}
\end{align}
We emphasize that it helps \emph{tremendously} to collect terms that are alike because there are typically a lot of non-obvious cancellations. The effective hamiltonian is obtained by conjugating with $\piref$, 
\begin{align}
	h_{\mathrm{eff} \, 1} &= \piref h_1 \piref 
	\notag \\
	&= \piref [u_1 u_0^* , h_0] \piref + \piref u_0 H_1 u_0^* \piref - \tfrac{i}{2} \piref \bigl ( \{ h_0 , u_0 \} - \{ u_0 , H_0 \} \bigr ) u_0^* \piref 
	. 
	\label{multiscale:effective_quantum_dynamics:effective_hamiltonian:eqn:closed_expression_h_eff_1}
\end{align}
Let us go back to the original problem of the Born-Oppenheimer system where the relevant bands consist of an isolated, non-degenerate band, $\sigma_{\mathrm{rel}}(x) = \{ E_*(x) \}$. Then the hamiltonian is the Weyl quantization of 
\begin{align*}
	H_{\eps}(x,\xi) = \tfrac{1}{2} \xi^2 \otimes \id_{\Hfast} + \He(x) \equiv \tfrac{1}{2} \xi^2 + \He(x) = H_0(x,\xi) 
\end{align*}
and $\piref = \sopro{\chi}{\chi}$ for some $\chi \in \Hfast$ with norm $1$. The leading order term of the effective hamiltonian is exactly what we expect from equation~\eqref{multiscale:physics_of_molecules:eqn:guess_effective_hamiltonian}, 
\begin{align}
	h_{\mathrm{eff} \, 0}(x,\xi) &= \piref \, u_0(x) H_0(x,\xi) u_0^*(x) \, \piref 
	\notag \\
	&= \sopro{\chi}{\chi} \bigl ( \sopro{\chi}{\varphi_*(x)} + u_0^{\perp}(x) \bigr ) \bigl ( \tfrac{1}{2} \xi^2 + \He(x) \bigr ) \bigl ( \sopro{\varphi_*(x)}{\chi} + u_0^{\perp}(x)^* \bigr ) \sopro{\chi}{\chi} 
	\notag \\
	&= \bigl ( \tfrac{1}{2} \xi^2 + E_*(x) \bigr ) \sopro{\chi}{\chi} 
	. 
\end{align}
To compute the first-order correction, we need to work a little: we need to compute only two terms since $u_0 H_1 u_0^* = 0$. The first one vanishes identically as well as we are interested in the case of a single band, 
\begin{align*}
	\piref \bigl [ u_1 u_0^* , h_0 \bigr ] \piref &= \piref \, u_1 u_0^* h_0 \, \piref - \piref \, h_0 u_1 u_0^* \, \piref 
	\\
	&= \piref \, u_1 u_0^* E_* \, \piref - E_* \piref \, u_1 u_0^* \, \piref 
	\\
	&= E_* \, \bigl ( \piref \, u_1 u_0^* \, \piref - \piref \, u_1 u_0^* \, \piref \bigr )
	= 0 
	. 
\end{align*}
Here, the usefulness of grouping terms properly shows. Similarly, the second term also simplifies since the Weyl product is bilinear and the Weyl product of two functions of position only reduces to the pointwise product. This means, we get 
\begin{align*}
	\tfrac{i}{2} \piref &\bigl ( \{ h_0 , u_0 \} - \{ u_0 , H_0 \} \bigr ) u_0^* \piref = 
	\tfrac{i}{2} \tfrac{1}{2} \piref \bigl ( \{ \xi^2 , u_0 \} - \{ u_0 , \xi^2 \} \bigr ) u_0^* \piref 
	\\
	&= \tfrac{i}{2} \piref \{ \xi^2 , u_0 \} u_0^* \piref 
	= \sum_{l = 1}^d i \xi_l \, \sopro{\chi}{\chi} \, \sopro{\chi}{\partial_{x_l} \varphi_*} \, \sopro{\varphi_*}{\chi} \, \sopro{\chi}{\chi} 
	\\
	&= \xi \cdot i \scpro{\nabla_x \varphi_*}{\varphi_*}_{\Hfast} \, \sopro{\chi}{\chi} 
	. 
\end{align*}
Since $\varphi_*(x)$ has norm $1$ for each $x \in \R^{3N}$, we can differentiate 
\begin{align*}
	\scpro{\varphi_*(x)}{\varphi_*(x)}_{\Hfast} = 1 
\end{align*}
with respect to $x$ to conclude $\scpro{\nabla_x \varphi_*}{\varphi_*}_{\Hfast} = - \scpro{\varphi_*}{\nabla_x \varphi_*}_{\Hfast}$. Thus, we can express $h_{\mathrm{eff} \, 1}$ in terms of the Berry phase, 
\begin{align*}
	h_{\mathrm{eff} \, 1} &= \tfrac{i}{2} \piref \bigl ( \{ h_0 , u_0 \} - \{ u_0 , H_0 \} \bigr ) u_0^* \piref 
	= - \xi \cdot \mathcal{A} \, \piref 
	. 
\end{align*}
This term contains the geometric phase associated to the band $E_*$. If we complete the square (which gives another error of order $\eps^2$), we see that $\mathcal{A}$ acts like a ``magnetic vector potential,''
\begin{align*}
	h_{\mathrm{eff}} &= h_{\mathrm{eff} \, 0} + \eps h_{\mathrm{eff} \, 1} + \order(\eps^2) 
	= \bigl ( \tfrac{1}{2} \xi^2 + E_* \bigr ) \, \piref - \eps \xi \cdot \mathcal{A} \, \piref + \order(\eps^2) 
	\\
	&= \bigl ( \tfrac{1}{2} (\xi - \eps \mathcal{A})^2 + E_*(x) \bigr ) \, \piref + \order(\eps^2)
	. 
\end{align*}
Associated to $\mathcal{A}$, there is a pseudo-magnetic field $\Omega := \dd \mathcal{A}$, the so-called \emph{Berry curvature}. This field modifies quantum and classical dynamics alike if time-reversal symmetry of the system is broken by \eg a magnetic field $B \neq 0$ or in the presence of band crossings. In our particular case, $\sigma_{\mathrm{rel}}(x) = \{ E_*(x) \}$ and $B = 0$, the pseudo-magnetic field $\Omega = 0$ can be shown to vanish identically. 


\subsection{Effective Dynamics} 
\label{multiscale:effective_quantum_dynamics:effective_dynamics}

The last -- and most important -- piece of the puzzle is to approximate the full dynamics generated by $e^{- i t \hat{H}_{\eps}}$ on $\Hphys$ for states in the relevant subspace $\hat{\Pi}_{\eps} \Hphys$ by the dynamics generated by the effective hamiltonian $\hat{h}_{\mathrm{eff}} := \Op_{\eps}(h_{\mathrm{eff}})$ on $\Hslow \otimes \C^{\abs{\Index}}$. 

As a first step, let us show that it suffices to compute only finitely many terms of $\hat{h}_{\mathrm{eff}}$: assume we are interested to approximate the dynamics for \emph{times of the order $\order(\nicefrac{1}{\eps})$} and we would like to \emph{make a total error of at most $\order(\eps)$}. We claim that it suffices to compute only $\hat{h}_{\mathrm{eff} \, 0}$ and $\hat{h}_{\mathrm{eff} \, 1}$. By the usual trick, we rewrite the difference of the two evolution groups as an integral and then reduce the difference between the evolution groups to the difference of their respective generators: 
\begin{align}
	e^{- \frac{i}{\eps} t \hat{h}_{\mathrm{eff}}} - &e^{- \frac{i}{\eps} t (\hat{h}_{\mathrm{eff} \, 0} + \eps \hat{h}_{\mathrm{eff} \, 1})} = \int_0^{\nicefrac{t}{\eps}} \dd s \, \frac{\dd }{\dd s} \Bigl (  e^{- i s \hat{h}_{\mathrm{eff}}} \, e^{- i (\frac{t}{\eps} - s) (\hat{h}_{\mathrm{eff} \, 0} + \eps \hat{h}_{\mathrm{eff} \, 1})} \Bigr ) 
	\notag \\
	&= \int_0^{\nicefrac{t}{\eps}} \dd s \, e^{- i s \hat{h}_{\mathrm{eff}}} \underbrace{\bigl (  
	- i \hat{h}_{\mathrm{eff}} + i (\hat{h}_{\mathrm{eff} \, 0} + \eps \hat{h}_{\mathrm{eff} \, 1})
	\bigr )}_{= \order(\eps^2)} e^{- i (\frac{t}{\eps} - s) (\hat{h}_{\mathrm{eff} \, 0} + \eps \hat{h}_{\mathrm{eff} \, 1})} 
	\notag \\
	&= \order(\eps^2) \cdot \order(\nicefrac{1}{\eps}) = \order(\eps) 
	\label{multiscale:effective_quantum_dynamics:effective_dynamics:eqn:dynamics_h_eff_compared_h0_h1}
\end{align}
This means even for macroscopic times ($\order(\nicefrac{1}{\eps})$ on the microscopic scale), the difference between $e^{- \frac{i}{\eps} t \hat{h}_{\mathrm{eff}}} \psi$ and $e^{- \frac{i}{\eps} t (\hat{h}_{\mathrm{eff} \, 0} + \eps \hat{h}_{\mathrm{eff} \, 1})} \psi$ for \emph{any} initial state $\psi \in \Hslow \otimes \C^{\abs{\Index}}$ is small, \ie of order $\order(\eps)$. 

Now let us explain in what way the effective hamiltonian can be used to approximate the full dynamics: let us reconsider our favorite diagram~\eqref{multiscale:adiabatic_trinity:diagram:perturbed} where we have defined $\hat{h} := \Op_{\eps}(h)$ and $\hat{h}_{\mathrm{eff}} := \Op_{\eps}(h_{\mathrm{eff}})$ to simplify notation. 
\begin{align*}
	\bfig
		\node L2full(-1300,0)[\Hphys]
		\node piL2full(-1300,-600)[\hat{\Pi}_{\eps} \Hphys]
		\node HslowHfast(0,0)[\Hslow \otimes \Hfast]
		\node Href(0,-600)[\Hslow \otimes \C^{\abs{\mathcal{J}}}]
		\arrow[L2full`HslowHfast;\hat{U}_{\eps}]
		\arrow[L2full`piL2full;\hat{\Pi}_{\eps}]
		\arrow[HslowHfast`Href;\Piref]
		\arrow/-->/[piL2full`Href;]
		\Loop(-1300,0){\Hphys}(ur,ul)_{e^{-i t \hat{H}_{\eps}}} 
		\Loop(0,0){\Hslow \otimes \Hfast}(ur,ul)_{e^{-i t \hat{h}}} 
		\Loop(0,-600){\Hslow \otimes \C^{\abs{\mathcal{J}}}}(dr,dl)^{e^{-i t \hat{h}_{\mathrm{eff}}}} 
	\efig
\end{align*}
Now we show that the dynamics in the upper-left corner, the physical Hilbert space, by the dynamics in the upper-right space by using $\hat{U}_{\eps} = \Op_{\eps}(u) + \order(\eps^{\infty})$, 
\begin{align*}
	u^* \Weyl h \Weyl u - H_{\eps} &= u^* \Weyl u \Weyl H_{\eps} \Weyl u^* \Weyl u - H_{\eps} = \order(\eps^{\infty}) 
	, 
\end{align*}
and the usual Duhamel trick: 
\begin{align}
	e^{- i t \hat{H}_{\eps}} - &\hat{U}_{\eps}^* e^{- i t \hat{h}} \hat{U}_{\eps} = e^{- i t \hat{H}_{\eps}} - \Op_{\eps}(u)^* e^{- i t \hat{h}} \Op_{\eps}(u) + \order(\eps^{\infty}) 
	\notag \\
	&= \int_0^t \dd s \, \frac{\dd }{\dd s} \Bigl ( e^{- i s \hat{H}_{\eps}} \, \Op_{\eps}(u)^* e^{- i (t - s) \hat{h}} \Op_{\eps}(u) \Bigr ) + \order(\eps^{\infty}) 
	\notag 
\end{align}
Expanding the terms yields that the two evolution groups differ only by $\order(\eps^{\infty})$
\begin{align}
	\ldots &= \int_0^t \dd s \, \Bigl ( e^{- i s \hat{H}_{\eps}} (- i \hat{H}_{\eps}) \Op_{\eps}(u)^* e^{- i (t - s) \hat{h}} \Op_{\eps}(u) 
	+ \Bigr . \notag \\
	&\qquad \qquad \qquad \qquad \Bigl . 
	+ e^{- i s \hat{H}_{\eps}} \, \Op_{\eps}(u)^* (+ i \hat{h}) e^{- i (t - s) \hat{h}} \Op_{\eps}(u) 
	\Bigr ) + \order(\eps^{\infty})
	\notag \\
	&= \int_0^t \dd s \, e^{- i s \hat{H}_{\eps}} \, i \Op_{\eps} \bigl ( u^* \Weyl h \Weyl u - H_{\eps} \bigr ) \, \Op_{\eps}(u)^* e^{- i (t - s) \hat{h}} \Op_{\eps}(u) + \order(\eps^{\infty}) 
	\notag \\
	&= \order(\eps^{\infty}) 
	\label{multiscale:effective_quantum_dynamics:effective_dynamics:eqn:full_dynamics_dynamics_h}
\end{align}
Furthermore, if we start with a state in the relevant subspace, \ie $\hat{\Pi}_{\eps} \Hphys$, then we can use the \emph{effective hamiltonian $\hat{h}_{\mathrm{eff}}$} to approximate the full dynamics: by equation~\eqref{multiscale:effective_quantum_dynamics:effective_dynamics:eqn:full_dynamics_dynamics_h} and $\hat{\Pi}_{\eps} = \Op_{\eps}(\pi) + \order(\eps^{\infty})$, we get 
\begin{align}
	e^{- i t \hat{H}_{\eps}} \hat{\Pi}_{\eps} &= \Op_{\eps}(u)^* e^{- i t \hat{h}} \, \Op_{\eps}(u) \, \Op_{\eps}(\pi) + \order(\eps^{\infty}) 
	\notag \\
	&= \Op_{\eps}(u)^* e^{- i t \hat{h}} \, \Piref \, \Op_{\eps}(u) + \order(\eps^{\infty}) 
	\notag \\
	&= \Op_{\eps}(u)^* e^{- i t \hat{h}_{\mathrm{eff}}} \, \Piref \, \Op_{\eps}(u) + \order(\eps^{\infty}) 
	. 
	\label{multiscale:effective_quantum_dynamics:effective_dynamics:eqn:full_dynamics_effective_dynamics}
\end{align}
Finally, since we are interested in times of order $\order(\nicefrac{1}{\eps})$ and in a total precision of $\order(\eps)$, we combine  equation~\eqref{multiscale:effective_quantum_dynamics:effective_dynamics:eqn:dynamics_h_eff_compared_h0_h1} with equation~\eqref{multiscale:effective_quantum_dynamics:effective_dynamics:eqn:full_dynamics_effective_dynamics} to get the final result, 
\begin{align}
	e^{- i \frac{t}{\eps} \hat{H}_{\eps}} &= \Op_{\eps}(u)^* e^{- i \frac{t}{\eps} \hat{h}_{\mathrm{eff}}} \, \Piref \, \Op_{\eps}(u) + \order(\eps^{\infty}) 
	\notag \\
	&= \Op_{\eps}(u)^* e^{- i \frac{t}{\eps} (\hat{h}_{\mathrm{eff} \, 0} + \eps \hat{h}_{\mathrm{eff} \, 1})} \, \Piref \, \Op_{\eps}(u) + \order(\eps) 
	\notag \\
	&= \Op_{\eps}(u_0)^* e^{- i \frac{t}{\eps} (\hat{h}_{\mathrm{eff} \, 0} + \eps \hat{h}_{\mathrm{eff} \, 1})} \, \Piref \, \Op_{\eps}(u_0) + \order(\eps) 
	\notag \\
	&= \hat{u}_0^* \, e^{- i \frac{t}{\eps} (\hat{h}_{\mathrm{eff} \, 0} + \eps \hat{h}_{\mathrm{eff} \, 1})} \, \Piref \, \hat{u}_0 + \order(\eps) 
	\label{multiscale:effective_quantum_dynamics:effective_dynamics:eqn:full_dynamics_effective_dynamics_first_order} 
	. 
\end{align}
This result tells us that if we start with states that are associated to the relevant bands, we can approximate the dynamics generated by the Born-Oppenheimer hamiltonian 
\begin{align*}
	\opHBO = \tfrac{1}{2} \Pe^2 + \He(\Qe)
\end{align*}
with that of the first two terms of effective hamiltonian, \ie the quantization of \marginpar{\small 2010.02.03}
\begin{align*}
	h_{\mathrm{eff} \, 0}(x,\xi) + \eps h_{\mathrm{eff} \, 1}(x,\xi) = \tfrac{1}{2} \bigl ( \xi - \eps \mathcal{A}(x) \bigr )^2 + E_*(x) + \order(\eps^2) 
	. 
\end{align*}
%



\section{Semiclassical limit} 
\label{multiscale:semiclassics}
As a last step in the analysis, we would like to make a semiclassical limit. We emphasize the semiclassical approximation is \emph{distinct} from the adiabatic approximation. These two are very often mixed in other approaches. 

Although there is a semiclassical theory for degenerate bands and several bands \cite{Robert:semiclassics:1998,Teufel:adiabaticPerturbationTheory:2003}, for simplicity we consider the case where the relevant part of the spectrum consists of an \emph{isolated, non-degenerate band}, $\sigma_{\mathrm{rel}}(x) := \{ E_*(x) \}$. We have shown in Chapter~\ref{multiscale:effective_quantum_dynamics:effective_hamiltonian} that the effective hamiltonian is to leading order the quantization of 
\begin{align*}
	h_{\mathrm{eff}}(x,\xi) = \tfrac{1}{2} \bigl ( \xi - \eps \mathcal{A}(x) \bigr )^2 + E_*(x) + \order(\eps^2) 
	. 
\end{align*}
Now we would like to make a semiclassical approximation for the dynamics of observables. It turns out that this is only possible for observables that are compatible with the splitting into slow and fast degrees of freedom. 
\begin{defn}[Macroscopic observable]
	A macroscopic observable $F_{\eps} \in \mathcal{B}(\Hphys)$ is the quantization of a $\mathcal{B}(\Hfast)$-valued function 
	\begin{align*}
		f \in \mathcal{S} \bigl ( \pspace , \mathcal{B}(\Hfast) \bigr ) 
	\end{align*}
	such that $f$ commutes pointwise with the hamiltonian $H_{\eps}$, \ie 
	\begin{align*}
		\bigl [ f(x,\xi) , H_{\eps}(x,\xi) \bigr ] &= f(x,\xi) \, H_{\eps}(x,\xi) - H_{\eps}(x,\xi) \, f(x,\xi) = 0 \in \mathcal{B}(\Hfast)
	\end{align*}
	holds for all $(x,\xi) \in \pspace$. 
\end{defn}
This implies $f$ is compatible with the band structure of $H_{\eps}$. With respect to the fast electronic degrees of freedom, it is a constant of motion. Even with respect to the slow nucleonic degrees of freedom, it is still an \emph{approximate constant of motion}, 
\begin{align*}
	\bigl [ H_{\eps} , f \bigr ]_{\Weyl} = \underbrace{\bigl [ H_{\eps} , f \bigr ]}_{= 0} + \order(\eps) = \order(\eps)
	. 
\end{align*}
In a way, this implies $F_{\eps}$ is oblivious to the details of the fast dynamics. Mathematically, this key property is needed to ensure that the Weyl commutator which appears in the Egorov theorem (Theorem~7.1) is small. Associated to a macroscopic observable $f$ (we will use $f$ synonymously with its quantization $F_{\eps}$) is its effective observable 
\begin{align*}
	f_{\mathrm{eff}} := \piref \, u \Weyl f \Weyl u^* \, \piref 
\end{align*}
defined analogously to the effective hamiltonian (equation~\eqref{multiscale:effective_quantum_dynamics:eqn:effective_hamiltonian}). This effective observable describes the physics if we are interested in initial conditions associated to the relevant band $E_*$ in the representation fascilitated by $\hat{U}_{\eps} = \Op_{\eps}(u) + \order(\eps^{\infty})$. The condition $[H_{\eps}(x,\xi) , f(x,\xi)] = 0$ ensures that $f_{\mathrm{eff}}$ is compatible with this procedure, \ie 
\begin{align*}
	[ \hat{\Pi}_{\eps} , \Op_{\eps}(f) ] = \order(\eps) 
\end{align*}
and 
\begin{align*}
	[ \Piref , \Op_{\eps}(f_{\mathrm{eff}}) ] = 0 
\end{align*}
hold. If we identify $\Hslow \otimes \C^1$ with $\Hslow$, then according to the diagram, 
\begin{align*}
	\bfig
		\node L2full(-1300,0)[\Hphys]
		\node piL2full(-1300,-600)[\hat{\Pi}_{\eps} \Hphys]
		\node HslowHfast(0,0)[\Hslow \otimes \Hfast]
		\node Href(0,-600)[\Hslow]
		\arrow[L2full`HslowHfast;\hat{U}_{\eps}]
		\arrow[L2full`piL2full;\hat{\Pi}_{\eps}]
		\arrow[HslowHfast`Href;\Piref]
		\arrow/-->/[piL2full`Href;]
		\Loop(-1300,0){\Hphys}(ur,ul)_{e^{-i t \hat{H}_{\eps}}} 
		\Loop(0,0){\Hslow \otimes \Hfast}(ur,ul)_{e^{-i t \hat{h}}} 
		\Loop(0,-600){\Hslow}(dr,dl)^{e^{-i t \hat{h}_{\mathrm{eff}}}} 
	\efig 
	, 
\end{align*}
there is a clear interpretation of these objects: $\Op_{\eps}$ is the \emph{physical observable} we seek to measure in experiments. This physical observable lives on the physical Hilbert space $\Hphys$ and its dynamics is governed by $e^{-i t \hat{H}_{\eps}}$. The \emph{effective observable} lives on the reference space $\Hslow$ and its dynamics is governed by the effective time evolution $e^{- i t (\hat{h}_{\mathrm{eff} \, 0} + \eps \hat{h}_{\mathrm{eff} \, 1})}$ up to errors of order $\order(\eps^2)$. 

Let us compute the first few terms of the effective observable $f_{\mathrm{eff}}$: the leading-order term is given by the partial average over the electronic state $\varphi_*$, 
\begin{align}
	f_{\mathrm{eff} \, 0}(x,\xi) &= \piref \, u_0(x) f(x,\xi) u_0^*(x) \, \piref = \scpro{\varphi_*(x)}{f(x,\xi) \varphi_*(x)}_{\Hfast} 
	\label{multiscale:semiclassics:eqn:effective_observable_0}
	. 
\end{align}
The first-order term also simplifies as we are interested in an isolated non-degenerate band: by equation~\eqref{multiscale:effective_quantum_dynamics:effective_hamiltonian:eqn:closed_expression_h_eff_1} and arguments completely analogous to those in Chapter~\ref{multiscale:effective_quantum_dynamics:effective_hamiltonian}, the subleading term computes to be 
\begin{align}
	f_{\mathrm{eff} \, 1} &= \underbrace{\piref \, [u_1 u_0^* , f] \, \piref}_{= 0} + 0 + \tfrac{i}{2} \piref \bigl ( \{ f_{\mathrm{eff} \, 0} , u_0 \} - \{ u_0 , f \} \bigr ) u_0^* \piref 
	\notag \\
	&= - i \nabla_{\xi} f_{\mathrm{eff} \, 0} \cdot \mathcal{A} 
	\label{multiscale:semiclassics:eqn:effective_observable_1}
\end{align}
where we have used 
\begin{align*}
	\piref \, i \partial_{x_l} u_0 u_0^* \, \piref &= i \scpro{\partial_{x_l} \varphi_*}{\varphi_*}_{\Hfast} 
	= - i \scpro{\varphi_*}{\partial_{x_l} \varphi_*}_{\Hfast} 
	\\
	&= - \mathcal{A}_l 
	. 
\end{align*}
\begin{example}
	Two exceptionally important observables are our \emph{effective} building block observables 
	\begin{align}
		x_{\mathrm{eff}} &:= \piref \, u \Weyl x \Weyl u^* \piref = x + \order(\eps^2) \\
		\xi_{\mathrm{eff}} &:= \piref \, u \Weyl \xi \Weyl u^* \piref = \xi - \eps \mathcal{A}(x) + \order(\eps^2) 
		\notag 
		. 
	\end{align}
	At least up to errors of order $\order(\eps^2)$ they look like the usual position observable and \emph{kinetic} momentum associated to the ``magnetic'' field $\eps \Omega_{lj} := \partial_{x_l} (\eps \mathcal{A}_j) - \partial_{x_j} (\eps \mathcal{A}_l)$. Since there is a prefactor of $\eps$, this means, the ``magnetic'' field is weak. It is not surprising that this pseudomagnetic field $\eps \Omega$ appears on the level of the classical equations of motion if we increase our precision to include terms of $\order(\eps)$. 
\end{example}
Let us define the map $T_{\eps} : \pspace \longrightarrow \pspace$, 
\begin{align*}
	T_{\eps} : (x,\xi) \mapsto \bigl ( x , \xi - \eps \mathcal{A}(x) \bigr ) + \order(\eps^2)
\end{align*}
which replaces $(x,\xi)$ by the effective observables $\bigl ( x_{\mathrm{eff}} , \xi_{\mathrm{eff}} \bigr ) = \bigl ( x , \xi - \eps \mathcal{A}(x) \bigr )$. This is in fact akin to adding a magnetic field to classical mechanics by minimal substitution, \ie replacing momentum $\xi$ with kinetic momentum $\xi - A(x)$. We propose the following 
\begin{lem}
	Let $f$ be a macroscopic observable. Then up to errors of order $\order(\eps^2)$, $f_{\mathrm{eff}}$ is equal to the minimally substituted observable $\piref \, u_0 (f \circ T_{\eps}) u_0^* \, \piref$, 
	\begin{align}
		f_{\mathrm{eff}}(x,\xi) &= (f \circ T_{\eps})(x,\xi) + \order(\eps^2) 
		= f \bigl ( x , \xi - \eps \mathcal{A}(x) \bigr ) + \order(\eps^2) 
		. 
	\end{align}
\end{lem}
\begin{proof}
	We Taylor expand $f \circ T_{\eps}$ around $\eps = 0$ and get 
	\begin{align*}
		(f \circ T_{\eps})(x,\xi) &= f \bigl ( x  + \order(\eps^2) , \xi - \eps \mathcal{A}(x)  + \order(\eps^2) \bigr ) 
		\\
		&
		= f(x,\xi) - \eps \nabla_{\xi} f(x,\xi) \cdot \mathcal{A}(x) + \order(\eps^2) 
		. 
	\end{align*}
	If we multiply the above expression with $\piref \, u_0$ from the left and $u_0^* \, \piref$ from the right, it coincides with $f_{\mathrm{eff} \, 0} + \eps f_{\mathrm{eff} \, 1}$ as given by equations~\eqref{multiscale:semiclassics:eqn:effective_observable_0} and~\eqref{multiscale:semiclassics:eqn:effective_observable_1}. 
\end{proof}
This means, $f$ and $f_{\mathrm{eff}}$ are related by minimal substitution and averaging over the fast degrees of freedom. 
\begin{thm}[Semiclassical limit]
	Let $f$ be a macroscopic observable. Then the full quantum dynamics of a Born-Oppenheimer system for an isolated non-degenerate band $E_*$ can be approximated by the flow $\phi_t$ generated by 
	\begin{align}
		\left (
		\begin{matrix}
			0 & \id \\
			- \id & 0 \\
		\end{matrix}
		\right ) \left (
		\begin{matrix}
			\dot{x} \\
			\dot{\xi} \\
		\end{matrix}
		\right ) = \left (
		\begin{matrix}
			\nabla_{x} \\
			\nabla_{\xi} \\
		\end{matrix}
		\right ) \bigl ( \tfrac{1}{2} \xi^2 + E_*(x) \bigr ) 
		\label{multiscale:semiclassics:eqn:classical_eom}
	\end{align}
	in the following sense: 
	\begin{align}
		\hat{\Pi}_{\eps} \bigl ( F_{\mathrm{qm}}(t) - F_{\mathrm{cl}}(t) \bigr ) \hat{\Pi}_{\eps} &:= 
		\hat{\Pi}_{\eps} e^{+ i \frac{t}{\eps} \hat{H}_{\eps}} \Op_{\eps}(f) e^{- i \frac{t}{\eps} \hat{H}_{\eps}} \hat{\Pi}_{\eps} - \hat{u}_0^* \, \Piref \, \Op_{\eps}(f_{\mathrm{eff} \, 0} \circ \phi_t) \, \Piref \, \hat{u}_0
		\notag \\
		&= \order(\eps^2) 
		\label{multiscale:semiclassics:eqn:semiclassical_limit}
	\end{align}
\end{thm}
\begin{proof}
	First of all, from Chapter~\ref{multiscale:effective_quantum_dynamics:effective_dynamics} (equation~\eqref{multiscale:effective_quantum_dynamics:effective_dynamics:eqn:full_dynamics_effective_dynamics_first_order}), we know we can replace the full dynamics $e^{- i \frac{t}{\eps}} \hat{\Pi}_{\eps}$ by 
	\begin{align*}
		e^{- i \frac{t}{\eps} \hat{H}_{\eps}} &= \hat{u}_0^* e^{- i \frac{t}{\eps} (\hat{h}_{\mathrm{eff} \, 0} + \eps \hat{h}_{\mathrm{eff} \, 1})} \Piref \, \hat{u}_0 + \order(\eps) 
		. 
	\end{align*}
	This can be plugged into 
	\begin{align*}
		\hat{\Pi}_{\eps} F_{\mathrm{qm}}(t) \hat{\Pi}_{\eps} &= \bigl ( \hat{u}_0^* e^{- i \frac{t}{\eps} (\hat{h}_{\mathrm{eff} \, 0} + \eps \hat{h}_{\mathrm{eff} \, 1})} \Piref \, \hat{u}_0 + \order(\eps) \bigr )^* \Op_{\eps}(f) 
		\cdot \\
		&\qquad \qquad \qquad \qquad \qquad \qquad \cdot 
		\bigl ( \hat{u}_0^* e^{- i \frac{t}{\eps} (\hat{h}_{\mathrm{eff} \, 0} + \eps \hat{h}_{\mathrm{eff} \, 1})} \Piref \, \hat{u}_0 + \order(\eps) \bigr ) 
		\\
		&= \hat{u}_0^* e^{- i \frac{t}{\eps} (\hat{h}_{\mathrm{eff} \, 0} + \eps \hat{h}_{\mathrm{eff} \, 1})} \underbrace{\Piref \, \hat{u}_0 \, \Op_{\eps}(f) \, \hat{u}_0^* \, \Piref}_{= f_{\mathrm{eff} \, 0} + \order(\eps)} e^{- i \frac{t}{\eps} (\hat{h}_{\mathrm{eff} \, 0} + \eps \hat{h}_{\mathrm{eff} \, 1})} \hat{u}_0 + \order(\eps) 
		\\
		&= \hat{u}_0^* e^{- i \frac{t}{\eps} (\hat{h}_{\mathrm{eff} \, 0} + \eps \hat{h}_{\mathrm{eff} \, 1})} \Op_{\eps}(f_{\mathrm{eff} \, 0}) e^{- i \frac{t}{\eps} (\hat{h}_{\mathrm{eff} \, 0} + \eps \hat{h}_{\mathrm{eff} \, 1})} \hat{u}_0 + \order(\eps) 
		. 
	\end{align*}
	Now we make the semiclassical approximation: since we can regard $\Op_{\eps} \bigl (\hat{h}_{\mathrm{eff} \, 0} + \eps \hat{h}_{\mathrm{eff} \, 1} \bigr )$ and $\Op_{\eps}(f_{\mathrm{eff} \, 0})$ as operators on $\Hslow \otimes \C^1 \cong \Hslow = L^2(\R^{3N})$, we can invoke Theorem~7.1, the Egorov-type theorem: by the usual trick, we relate quantum and classical dynamics: first of all, since $(x_{\mathrm{eff}} , \xi_{\mathrm{eff}}) = (x,\xi) + \order(\eps^2)$, we can replace the effective variables by usual variables in the classical equations of motion, equation~\eqref{multiscale:semiclassics:eqn:classical_eom}. We also see that the hamiltonian in these equations of motion is nothing but $h_{\mathrm{eff} \, 0}(x,\xi) = \tfrac{1}{2} \xi^2 + E_*(x)$. 
	
	Now we can make a semiclassical argument to replace the full quantum mechanics generated by $\hat{h}_{\mathrm{eff} \, 0} + \eps \hat{h}_{\mathrm{eff} \, 1}$ and compare it to the quantization of the classically evolved observable, 
	\begin{align*}
		e^{+ i \frac{t}{\eps} (\hat{h}_{\mathrm{eff} \, 0} + \eps \hat{h}_{\mathrm{eff} \, 1})} &\Op_{\eps}(f_{\mathrm{eff} \, 0}) e^{- i \frac{t}{\eps} (\hat{h}_{\mathrm{eff} \, 0} + \eps \hat{h}_{\mathrm{eff} \, 1})} - \Op_{\eps}(f_{\mathrm{eff} \, 0} \circ \phi_t) 
		= \\
		&= \int_0^t \dd s \, \frac{\dd}{\dd s} \Bigl ( e^{+ i \frac{s}{\eps} (\hat{h}_{\mathrm{eff} \, 0} + \eps \hat{h}_{\mathrm{eff} \, 1})} \Op_{\eps}(f_{\mathrm{eff} \, 0} \circ \phi_{t - s}) e^{- i \frac{s}{\eps} (\hat{h}_{\mathrm{eff} \, 0} + \eps \hat{h}_{\mathrm{eff} \, 1})} \Bigr ) 
		\\
		&= \int_0^t \dd s \, e^{+ i \frac{s}{\eps} (\hat{h}_{\mathrm{eff} \, 0} + \eps \hat{h}_{\mathrm{eff} \, 1})} \Op_{\eps} \Bigl ( \tfrac{i}{\eps} \bigl [ h_{\mathrm{eff} \, 0} + \eps h_{\mathrm{eff} \, 1} , f_{\mathrm{eff} \, 0}(t - s) \bigr ]_{\Weyl} 
		+ \\
		&\qquad \qquad \qquad \qquad \qquad \qquad \qquad 
		- \frac{\dd}{\dd t} f_{\mathrm{eff} \, 0}(t - s) \Bigr ) e^{- i \frac{s}{\eps} (\hat{h}_{\mathrm{eff} \, 0} + \eps \hat{h}_{\mathrm{eff} \, 1})} 
		. 
	\end{align*}
	To show that the integral is small, we need to show that 
	\begin{align*}
		\tfrac{i}{\eps} \bigl [ h_{\mathrm{eff} \, 0} + \eps h_{\mathrm{eff} \, 1} , &f_{\mathrm{eff} \, 0}(t - s) \bigr ]_{\Weyl} - \bigl \{ h_{\mathrm{eff} \, 0} , f_{\mathrm{eff} \, 0}(t - s) \bigr \} 
		= \\
		&
		= \tfrac{i}{\eps} \bigl [ h_{\mathrm{eff} \, 0} + \eps h_{\mathrm{eff} \, 1} , f_{\mathrm{eff} \, 0}(t - s) \bigr ] + \order(\eps) 
	\end{align*}
	is of order $\order(\eps)$. In the single-band case, $h_{\mathrm{eff} \, 0} + \eps h_{\mathrm{eff} \, 1}$ and $f_{\mathrm{eff} \, 0}$ are \emph{scalar}-valued functions and hence the pointwise commutator vanishes, 
	\begin{align*}
		\bigl [ h_{\mathrm{eff} \, 0} + \eps h_{\mathrm{eff} \, 1} , f_{\mathrm{eff} \, 0}(t - s) \bigr ] = 0 
		. 
	\end{align*}
	Thus, we have shown that equation~\eqref{multiscale:semiclassics:eqn:semiclassical_limit} holds, 
	\begin{align*}
		\hat{\Pi}_{\eps} F_{\mathrm{qm}}(t) \hat{\Pi}_{\eps} &= \hat{u}_0^* \Piref \, \Op_{\eps} \bigl ( f_{\mathrm{eff}} \circ \phi_t \bigr ) \, \Piref \hat{u}_0 + \order(\eps) 
		. 
	\end{align*}
	This concludes the proof. 
\end{proof}
The proof already contains the core idea how to push the error to $\order(\eps^2)$: one has to compute $h_{\mathrm{eff} \, 2}$ and then include the pseudomagnetic field $\eps \Omega$ in the semiclassical equations of motion. However, since the pseudomagnetic field is weak, on our time scale and with our level of precision, it will not affect the dynamics. 



\backmatter

\bibliographystyle{alpha}
\bibliography{./quantization_semiclassics}

\end{document}